\documentclass[11pt]{article}
\usepackage{geometry}
\usepackage{setspace}
\usepackage[round]{natbib}

\usepackage{amsmath}
\usepackage{amsfonts}
\usepackage{amssymb}
\usepackage{graphicx}
\usepackage{float}
\usepackage{subcaption}
\usepackage{multirow}
\usepackage{enumitem}

\usepackage[mathlines]{lineno}
\usepackage{caption}

\usepackage{siunitx}
\usepackage{mathtools}
\usepackage{bm}
\usepackage{tikz}
\usetikzlibrary{decorations.pathreplacing}
\usepackage{booktabs}
\usepackage{hyperref}
\hypersetup{colorlinks=true, linkcolor=blue, citecolor=magenta}
\usepackage{indentfirst}
\usepackage{adjustbox}
\usepackage{lmodern}

\usepackage{amsthm}

\theoremstyle{plain}
\newtheorem{theorem}{Theorem}
\newtheorem{proposition}{Proposition}
\newtheorem{lemma}{Lemma}

\theoremstyle{definition}
\newtheorem{definition}{Definition}
\newtheorem{assumption}{Assumption}
\theoremstyle{remark}
\newtheorem{remark}{Remark}

\usepackage[page]{appendix}

\makeatletter
\def\definerm#1{\expandafter\def\csname rm#1\endcsname{\mathrm{#1}}}
\def\definebb#1{\expandafter\def\csname bb#1\endcsname{\mathbb{#1}}}
\def\definecal#1{\expandafter\def\csname cal#1\endcsname{\mathcal{#1}}}
\def\definesf#1{\expandafter\def\csname sf#1\endcsname{\mathsf{#1}}}
\def\@letterlist{A,B,C,D,E,F,G,H,I,J,K,L,M,N,O,P,Q,R,S,T,U,V,W,X,Y,Z}
\@for\@letter:=\@letterlist\do{\expandafter\definerm\expandafter{\@letter}}
\@for\@letter:=\@letterlist\do{\expandafter\definebb\expandafter{\@letter}}
\@for\@letter:=\@letterlist\do{\expandafter\definecal\expandafter{\@letter}}
\@for\@letter:=\@letterlist\do{\expandafter\definesf\expandafter{\@letter}}
\makeatother

\newcommand{\R}{\bbR}
\newcommand{\Z}{\bbZ}

\newcommand{\E}{\bbE}
\newcommand{\pr}{\bbP}
\newcommand{\T}{\sfT}
\newcommand{\ud}{\mathrm{d}}
\newcommand{\normal}{\calN}
\newcommand{\var}{\mathrm{Var}}
\newcommand{\cov}{\mathrm{Cov}}
\newcommand{\corr}{\mathrm{Corr}}
\newcommand{\unif}{\mathrm{Unif}}

\newcommand{\eps}{\varepsilon}
\newcommand{\convp}{\xrightarrow{p}}
\newcommand{\convd}{\xrightarrow{d}}
\newcommand{\convst}{\xrightarrow{\rm st}}
\newcommand{\convucp}{\xrightarrow{\rm u.c.p.}}
\newcommand{\true}{\langle X, \sigma^2 \rangle}
\newcommand{\estx}{\widetilde{\true}}
\newcommand{\esty}{\widehat{\true}}

\allowdisplaybreaks

\usepackage{authblk}

\title{\textbf{Holistic Multi-Scale Inference of the Leverage Effect:
Efficiency under Dependent Microstructure Noise}}

\author[1]{Ziyang Xiong}
\author[1]{Zhao Chen\thanks{Corresponding author: zchen\_fdu@fudan.edu.cn}}
\author[2]{Christina Dan Wang\thanks{Corresponding author:
christina.wang@nyu.edu}}

\affil[1]{School of Data Science, Fudan University}
\affil[2]{Business Division, New York University Shanghai}

\begin{document}

\maketitle

\begin{abstract}
  This paper addresses the long-standing challenge of estimating the
  leverage effect from high-frequency data contaminated by dependent,
  non-Gaussian microstructure noise. We depart from the conventional
  reliance on pre-averaging or volatility ``plug-in'' methods by
  introducing a holistic multi-scale framework that operates directly
  on the leverage effect. We propose two novel estimators: the
  Subsampling-and-Averaging Leverage Effect (SALE) and the
  Multi-Scale Leverage Effect (MSLE). Central to our approach is a
  shifted window technique that constructs a noise-unbiased base
  estimator, significantly simplifying the multi-scale architecture.
  We provide a rigorous theoretical foundation for these estimators,
  establishing central limit theorems and stable convergence results
  that remain valid under both noise-free and dependent-noise
  settings. The primary contribution to estimation efficiency is a
  specifically designed weighting strategy for the MSLE estimator. By
  optimizing the weights based on the asymptotic covariance structure
  across scales and incorporating finite-sample variance corrections,
  we achieve substantial efficiency gains over existing benchmarks.
  Extensive simulation studies and an empirical analysis of 30 U.S.
  assets demonstrate that our framework consistently yields smaller
  estimation errors and superior performance in realistic, noisy
  market environments.
\end{abstract}


\section{Introduction}
The leverage effect, or the observed negative correlation between
asset returns and their volatility changes, is a prominent stylized
fact in financial econometrics \citep{black1976StudiesStockPrice,
christie1982StochasticBehaviorCommon}. It captures the asymmetry in
volatility responses to positive and negative shocks in asset prices
and is widely attributed to mechanisms such as financial leverage and
asymmetric information in markets. Accurate estimation of the
leverage effect is not only central to our understanding of asset
price dynamics but also has critical implications for the pricing and
hedging of derivative securities, especially in the presence of
volatility skews.

However, obtaining reliable estimates of the leverage effect in
high-frequency data is complicated by market microstructure (MMS)
noise. The well-known ``volatility signature plot'' demonstrates a
key challenge: with high-frequency observations, simple realized
volatility estimators are largely biased and thus inconsistent due to
noise \citep{zhou1996HighFrequencyDataVolatility,
  andersen2000GreatRealizations, aitsahalia2005HowOftenSample,
  patton2011DatabasedRankingRealised,
aitsahalia2019HausmanTestPresence}; simple leverage effect estimators
suffer from a similar issue. Existing methods for leverage effect
estimation have primarily focused on mitigating the impact of such
noise by either adopting the pre-averaging technique
\citep{jacod2009MicrostructureNoiseContinuous,
  podolskij2009BipowertypeEstimationNoisy,
mykland2016DataCleaningInference} under the assumption of independent
and identically distributed (i.i.d.) Gaussian noise or using extra
information. For instance, \cite{wang2014EstimationLeverageEffect}
and \cite{aitsahalia2017EstimationContinuousDiscontinuous} utilize
the pre-averaging methods to estimate leverage effect under i.i.d.
Gaussian noise, \cite{yuan2020LeverageEffectHighfrequency} adopts the
parametric setting for microstructure noise proposed by
\cite{li2016EfficientEstimationIntegrated} that incorporates trading
information to address noise,
\cite{chong2024VolatilityVolatilityLeverage} utilizes high-frequency
short-dated option to recover spot volatility process, thereby
estimating leverage effect, while some other works
\citep[\emph{e.g.}][]{bandi2012TimevaryingLeverageEffects,
  kalnina2017NonparametricEstimationLeverage,
curato2022StochasticLeverageEffect,yang2023EstimationLeverageEffect}
estimate leverage effect without explicitly addressing microstructure
noise. While these methods represent significant advances, their
reliance on i.i.d. Gaussian noise-related assumptions or auxiliary
data remains restrictive in practice. Specifically, empirical
evidence has consistently shown that microstructure noise exhibits
serial dependence and higher-order moments
\citep{jacod2017StatisticalPropertiesMicrostructure,
  aitsahalia2019HausmanTestPresence,
  li2020DependentMicrostructureNoise, da2021WhenMovingAverageModels,
li2022ReMeDIMicrostructureNoise}. These features not only violate
standard modeling assumptions, but also exacerbate bias and variance
in leverage effect estimation, underscoring the need for more
flexible methodologies that can accommodate complex noise structures.

In this paper, we propose a novel multi-scale framework for
estimating the leverage effect that explicitly accounts for
microstructure noise exhibiting more flexible noise structures,
allowing for stationary, dependent noise with nontrivial higher-order
moments, which are commonly observed in empirical financial data.
Specifically, we introduce two new estimators: the
Subsampling-and-Averaging Leverage Effect (SALE) estimator and the
Multi-Scale Leverage Effect (MSLE) estimator. In particular, the MSLE
estimator aggregates a series of weighted SALE estimators computed
across multiple time scales, exploiting their complementary
properties to achieve improved convergence rates. Our methodology
draws inspiration from the principles underlying the Two-Scale
Realized Volatility (TSRV) and Multi-Scale Realized Volatility (MSRV)
estimators \citep{zhang2005TaleTwoTime,
  zhang2006EfficientEstimationStochastic,
aitsahalia2011UltraHighFrequency}, but adapts and extends them for
the specific task of estimating the leverage effect. Crucially, we do
not merely use their methods as plug-in estimators for spot
volatility; we develop a holistic multi-scale approach for the
leverage effect itself. A key innovation in our construction is the
use of a shifted window for estimating spot volatility, as
illustrated schematically in Figure~\ref{fig:base-estimators-1} and
\ref{fig:base-estimators-2}. This shift is important as it not only
helps decouple noise components to achieve unbiasedness and variance
reduction with respect to noise, but also fundamentally simplifies
certain aspects of the multi-scale estimation procedure compared with
the classical constructions in TSRV and MSRV.

The second primary contribution of this work is the explicit
demonstration of multi-scale benefits in efficiency beyond mere noise
mitigation. Specifically, for the noise-free setting, we show that
our MSLE estimator, through a proper weighting scheme, achieves a
notable reduction in asymptotic variance compared with the base
estimator by exploiting the covariance structure of SALE estimators
of different subsampling scales. Furthermore, in the noisy setting,
this efficiency advantage becomes even more pronounced when
benchmarked against the pre-averaging estimator, particularly at
lower noise levels that are more representative of practical scenarios.

Recognizing a potential gap between standard theoretical assumptions
and empirical reality regarding MMS noise, another major contribution
of this work lies in the in-depth study of optimal weight assignment
under diverse and realistic noise magnitudes. Classical asymptotic
analyses often assume that noise variance is of constant order, thus
dominating the shrinking latent increments as the sampling frequency
increases. However, empirical evidence suggests a more complex
picture. For instance, \cite{aitsahalia2019HausmanTestPresence} finds
that improvements in market liquidity allow employing simple
volatility estimators at higher frequencies, while
\cite{kalnina2008EstimatingQuadraticVariation} and
\cite{da2021WhenMovingAverageModels} also explore the case of
shrinkage MMS noise. A recent work by
\cite{chong2025WhenFrictionsAre} investigates the rough noise model
that captures a more subtle interplay between the latent price
process and noise. This implies that weighting schemes based solely
on the strict noise-dominance assumption might suffer from modeling
error when applied to real data. Motivated by this, we conduct a
fine-grained analysis of the MSLE's asymptotic variance components
across different subsampling scales and noise conditions. Moreover,
implementing truly optimal weights would require precise values of
asymptotic covariances between SALE estimators at sampling scales,
which are infeasible in practice. To overcome this, we develop a
computationally efficient approximate weighting strategy that works
for a wide spectrum of noise conditions, yielding robust
finite-sample performance as demonstrated in Monte Carlo simulations.

Our theoretical analysis establishes central limit theorems and
stable convergence results for both the noise-free and noisy
settings, demonstrating that the proposed estimators achieve
consistency and asymptotic normality. Under the noise-free setting,
both estimators attain the optimal convergence rate of order
$n^{-1/4}$. In the presence of MMS noise, the MSLE estimator achieves
a convergence rate of order $n^{-1/9}$. While this theoretical rate
is slightly slower than the $n^{-1/8}$ achieved by the pre-averaging
approach, our simulations consistently demonstrate that the MSLE
estimator outperforms the pre-averaging method in practical settings
across a wide range of noise levels and sample sizes. To support
feasible inference, we construct consistent estimators for the
asymptotic variances in both noise-free and noisy regimes, enabling
the implementation of feasible central limit theorems. Monte Carlo
simulations, employing both independent and dependent noise with
various distributions, validate both the feasible and infeasible
central limit theorems.

To assess the finite-sample performance of the proposed estimators,
we conduct another extensive simulation study that encompasses a
variety of settings with realistic time horizons and noise
conditions. This study verifies the outstanding efficiency of the
proposed MSLE estimator with the approximate weighting strategy. The
result shows that: (i) under the noise-free setting, our method
outperforms the base estimator, whereas (ii) under the noisy setting,
our method consistently outperforms the pre-averaging approach, and
this advantage is particularly pronounced in empirically relevant
scenarios. This superior finite-sample behavior is attributable to a
combination of factors, detailed in
Section~\ref{sec:practical-weight}: (i) the more advantageous
noise-free asymptotic variance of the SALE estimator, (ii) the
relatively smaller impact of noise under realistic settings, and
(iii) the enhanced performance of MSLE over the individual SALE estimators.

To examine the performance of our estimators in practical
applications, we also conduct an empirical analysis using
high-frequency financial data. The high-frequency trading data of 30
assets including ETFs and individual stocks from the U.S. stock
market are utilized. The leverage effects are estimated adaptively
based on the microstructure noise characteristics in each period, and
general negative correlation between the returns and volatility
changes is verified. This study demonstrates the practical
flexibility and effectiveness of the MSLE estimator in capturing the
leverage effect under realistic market conditions, further validating
the advantages observed in the simulation experiments.

The remaining paper is arranged as follows. Section~\ref{sec:method}
introduces the model settings, notations, and estimators.
Section~\ref{sec:results} states the main theoretical results,
including limit theorems for the SALE and MSLE estimators under both
noise-free and noisy conditions. Section~\ref{sec:practical}
discusses the issue of variance reduction, including variance
approximations under both noise-free and noisy settings, and proposes
practical strategies for optimizing the performance of MSLE.
Section~\ref{sec:simulation} provides a detailed simulation study to
validate the theoretical properties, examining the asymptotic
behavior and finite-sample performance under various settings of
microstructure noise. Section~\ref{sec:empirical} reports the results
of an empirical study using real-world high-frequency data to
demonstrate the practical utility of the proposed methods.
Proofs, feasible central limit theorems, and further elaborations on
Section~\ref{sec:practical} to \ref{sec:empirical} are provided in
Supplementary Material.

\section{Methodology}
\label{sec:method}

\subsection{Model Settings}

The noise-contaminated log-price $Y_{t_i}$ is observed at $t_i = i\Delta_n
= iT/n$ for $i = 0, 1, \dots, n$:
\begin{align}\label{eq:observation}
  Y_{t_i} = X_{t_i} + \eps_i,
\end{align}
where $X_{t_i}$ denotes the latent log-price, and $\eps_i$ denotes
the noise. For simplicity, for any stochastic process $V$, we denote:
(i) $V_i \coloneqq V_{t_i}$,
(ii) $\Delta V_i \coloneqq V_{i+1} - V_i$, and
(iii) $\Delta_k V_i \coloneqq V_{i+k} - V_i$.

\begin{assumption}[Underlying processes]\label{ass:process}
  Let both log-price process $(X_t)_{t \geq 0}$ and volatility process
  $(\sigma_t)_{t \geq 0}$ be Itô processes defined on a filtered
  probability space $(\Omega, \calF, (\calF_t)_{t \geq 0}, \pr)$:
  \begin{align}
    X_t &= X_0 + \int_0^t \mu_s \ud s + \int_0^t \sigma_s \ud W_s, \\
    \sigma_t &= \sigma_0 + \int_0^t a_s \ud s + \int_0^t f_s \ud W_s +
    \int_0^t g_s \ud B_s,
  \end{align}
  where $W_t$ and $B_t$ are independent Brownian motions, and $\mu_t,
  a_t, f_t, g_t$ are adapted càdlàg locally bounded processes. In
  addition, $f_t$ and $g_t$ are Itô processes and the volatility path
  $\sigma_t^2$ is bounded away from zero.
\end{assumption}

\begin{assumption}[Noise]\label{ass:noise}
  Let $\{\eps_i\}_{i=0}^n$ be mean-zero, identically distributed
  random variables independent of $\mathcal{F}$, with $\nu_k
  \coloneqq \E[\eps_i^k]$ and $\nu_4 < \infty$. In terms of serial
  dependency, $\{\eps_i\}_{i=0}^n$ are specified as either:
  \begin{enumerate}[label=(\alph*)]
    \item \label{ass:noise-iid} independent ($\eps_i \perp \eps_j$
      for all $i \neq j$); or
    \item \label{ass:noise-dep} $q$-dependent ($\eps_i \perp \eps_j$
      for $|i-j| > q$) and stationary up to the fourth moment.
  \end{enumerate}
\end{assumption}

Let $\true_T = \int_0^T 2\sigma_s^2 f_s \ud s$ denote the true
leverage effect parameter. For any estimator of interest, we maintain
a clear distinction between its infeasible version $\estx_T$ based on
latent data $\{X_i\}_{i=0}^n$, and its feasible version $\esty_T$
based on observed data $\{Y_i\}_{i=0}^n$. The statistical properties
of the estimator are investigated from two aspects:
\begin{enumerate}
  \item Asserting its \emph{unbiasedness with respect to noise}:
    \begin{align}
      \underbrace{
        \E\bigl( \esty_T  \big| \calF \bigr)
        -
        \estx_T
      }_{\text{bias due to noise}}
      = 0.
    \end{align}
  \item Assessing its total variance by decomposing it into the
    \emph{variance due to discretization} and the expected
    \emph{variance due to noise}:
    \begin{align}\label{eq:variance-decomposition}
      \var\bigl(\esty_T \bigr)
      =
      \underbrace{
        \var\bigl( \estx_T \bigr)
      }_{\text{variance due to discretization}}
      + \:
      \E\bigl[\,
        \underbrace{
          \var\bigl( \esty_T \big| \calF \bigr)
        }_{\text{variance due to noise}}
      \,\bigr].
    \end{align}
\end{enumerate}

\begin{remark}
  While some literature \citep[for
  example,][]{wang2014EstimationLeverageEffect,
  kalnina2017NonparametricEstimationLeverage} defines the leverage
  effect as $\langle X, F(\sigma^2) \rangle_T = \int_0^T 2
  F'(\sigma_s^2) \sigma_s^2 f_s \ud s$ for a general function $F \in
  \bbC^2$, we focus on the canonical case where $F(x) = x$. Beyond
  notational simplicity, this allows us to concentrate on a primary
  challenge addressed in this work: robust estimation under complex,
  dependent noise, where a general $F$ function poses significant
  additional difficulties. However, it is worth noting that the
  results of our noise-free estimators
  (Theorems~\ref{thm:SALE-clt-noisefree} and
  \ref{thm:MSLE-clt-noisefree}) can be easily extended to the general case.
\end{remark}

\subsection{Estimators and The Robustness to Noise}

For simplicity, the following notations are used for spot volatility estimation:
\begin{align}
  \label{eq:spot-vol-right}
  \widehat{\sigma}_+^2 (i, H, k, s)
  =
  \frac{1}{kH\Delta_n} \sum_{j=s+1}^{k+s} (\Delta_H Y_{i+jH})^2
  & \quad \text{for} \quad \sigma_{i+H}^2,
  \\
  \label{eq:spot-vol-left}
  \widehat{\sigma}_-^2 (i, H, k, s)
  =
  \frac{1}{kH\Delta_n} \sum_{j=-k-s}^{-s-1} (\Delta_H Y_{i+jH})^2
  & \quad \text{for} \quad \sigma_i^2,
  \\
  \label{eq:spot-vol-delta}
  \widehat{\delta} (i, H, k, s)
  =
  \widehat{\sigma}_+^2 (i, H, k, s) - \widehat{\sigma}_-^2 (i, H, k, s)
  & \quad \text{for} \quad \Delta_H \sigma_i^2.
\end{align}
Here, $H$ represents the subsampling scale, whereas $k$ and $s$
represent the size and shift of the windows for estimating spot
volatility. Specifically, the last parameter $s$ can be omitted when
$s = 1$, as this work primarily focuses on this case.

\begin{figure}[!ht]
  \centering
  \begin{subfigure}[c]{\textwidth}
    \centering
    \resizebox{\textwidth}{!}{
      \begin{tikzpicture}
        \draw[-latex] (-15, 0) -- (17, 0);
        \foreach \x in {-14, -13} {\draw (\x, 0) -- (\x, 0.1);}
        \foreach \x in {-10, ..., -6} {\draw (\x, 0) -- (\x, 0.1);}
        \foreach \x in {-3, ..., 4} {\draw (\x, 0) -- (\x, 0.1);}
        \foreach \x in {7, ..., 11} {\draw (\x, 0) -- (\x, 0.1);}
        \foreach \x in {14, ..., 16} {\draw (\x, 0) -- (\x, 0.1);}
        \node[above] at (-11.5, -0.15) {$\cdots\cdots$};
        \node[above] at (-4.5, -0.15) {$\cdots\cdots$};
        \node[above] at (5.5, -0.15) {$\cdots\cdots$};
        \node[above] at (12.5, -0.15) {$\cdots\cdots$};
        \draw[decorate,decoration={brace,amplitude=5pt,mirror}] (1,
        0.2) -- (0, 0.2);
        \node[above] at (0.5, 0.3) {$\Delta Y_i$};
        \node[below] at (0, 0) {\scriptsize{$i$}};
        \node[below] at (1, 0) {\scriptsize{$i+1$}};
        \draw[decorate,decoration={brace,amplitude=5pt,mirror}] (0,
        0.2) -- (-7, 0.2);
        \node[above] at (-3.5, 0.3) {$\widehat{\sigma}_-^2 (i, 1, k_n, 0)$};
        \node[below] at (-7, 0) {\scriptsize{$i-k_n$}};
        \draw[decorate,decoration={brace,amplitude=5pt,mirror}] (8,
        0.2) -- (1, 0.2);
        \node[above] at (4.5, 0.3) {$\widehat{\sigma}_+^2 (i, 1, k_n, 0)$};
        \node[below] at (8, 0) {\scriptsize{$i+k_n+1$}};
        \foreach \xStart/\xEnd in {-7/0, 1/8} {
          \fill [fill=red, opacity=0.2] (\xStart,0.08) -- (\xStart,
          0.0) -- (\xEnd, 0.0) -- (\xEnd, 0.08) -- cycle;
        }
        \foreach \xStart/\xEnd in {0/1} {
          \fill [fill=blue, opacity=0.2] (\xStart,0.08) -- (\xStart,
          0.0) -- (\xEnd, 0.0) -- (\xEnd, 0.08) -- cycle;
        }
      \end{tikzpicture}
    }
    \caption{The continuous leverage effect estimator in
    \cite{aitsahalia2017EstimationContinuousDiscontinuous}}
    \label{fig:base-estimators-1}
  \end{subfigure}
  \begin{subfigure}[c]{\textwidth}
    \centering
    \resizebox{\textwidth}{!}{
      \begin{tikzpicture}
        \draw[-latex] (-15, 0) -- (17, 0);
        \foreach \x in {-14, -13} {\draw (\x, 0) -- (\x, 0.1);}
        \foreach \x in {-10, ..., -6} {\draw (\x, 0) -- (\x, 0.1);}
        \foreach \x in {-3, ..., 4} {\draw (\x, 0) -- (\x, 0.1);}
        \foreach \x in {7, ..., 11} {\draw (\x, 0) -- (\x, 0.1);}
        \foreach \x in {14, ..., 16} {\draw (\x, 0) -- (\x, 0.1);}
        \node[above] at (-11.5, -0.15) {$\cdots\cdots$};
        \node[above] at (-4.5, -0.15) {$\cdots\cdots$};
        \node[above] at (5.5, -0.15) {$\cdots\cdots$};
        \node[above] at (12.5, -0.15) {$\cdots\cdots$};
        \draw[decorate,decoration={brace,amplitude=5pt,mirror}] (1,
        0.2) -- (0, 0.2);
        \node[above] at (0.5, 0.3) {$\Delta Y_i$};
        \node[below] at (0, 0) {\scriptsize{$i$}};
        \node[below] at (1, 0) {\scriptsize{$i+1$}};
        \draw[decorate,decoration={brace,amplitude=5pt,mirror}] (-1,
        0.2) -- (-8, 0.2);
        \node[above] at (-4.5, 0.3) {$\widehat{\sigma}_-^2 (i, 1, k_n, 1)$};
        \node[below] at (-1, 0) {\scriptsize{$i-1$}};
        \node[below] at (-8, 0) {\scriptsize{$i-k_n-1$}};
        \draw[decorate,decoration={brace,amplitude=5pt,mirror}] (9,
        0.2) -- (2, 0.2);
        \node[above] at (5.5, 0.3) {$\widehat{\sigma}_+^2 (i, 1, k_n, 1)$};
        \node[below] at (2, 0) {\scriptsize{$i+2$}};
        \node[below] at (9, 0) {\scriptsize{$i+k_n+2$}};
        \foreach \xStart/\xEnd in {-8/-1, 2/9} {
          \fill [fill=red, opacity=0.2] (\xStart,0.08) -- (\xStart,
          0.0) -- (\xEnd, 0.0) -- (\xEnd, 0.08) -- cycle;
        }
        \foreach \xStart/\xEnd in {0/1} {
          \fill [fill=blue, opacity=0.2] (\xStart,0.08) -- (\xStart,
          0.0) -- (\xEnd, 0.0) -- (\xEnd, 0.08) -- cycle;
        }
      \end{tikzpicture}
    }
    \caption{The all-observation leverage effect estimator defined in
    Equation~\eqref{eq:all-observation}}
    \label{fig:base-estimators-2}
  \end{subfigure}
  \begin{subfigure}[c]{\textwidth}
    \centering
    \resizebox{\textwidth}{!}{
      \begin{tikzpicture}
        \draw[-latex] (-15, 0) -- (17, 0);
        \foreach \x in {-14, -13} {\draw (\x, 0) -- (\x, 0.1);}
        \foreach \x in {-10, ..., -6} {\draw (\x, 0) -- (\x, 0.1);}
        \foreach \x in {-3, ..., 4} {\draw (\x, 0) -- (\x, 0.1);}
        \foreach \x in {7, ..., 11} {\draw (\x, 0) -- (\x, 0.1);}
        \foreach \x in {14, ..., 16} {\draw (\x, 0) -- (\x, 0.1);}
        \node[above] at (-11.5, -0.15) {$\cdots\cdots$};
        \node[above] at (-4.5, -0.15) {$\cdots\cdots$};
        \node[above] at (5.5, -0.15) {$\cdots\cdots$};
        \node[above] at (12.5, -0.15) {$\cdots\cdots$};
        \draw[decorate,decoration={brace,amplitude=5pt,mirror}] (1,
        0.2) -- (0, 0.2);
        \node[above] at (0.5, 0.3) {$\Delta Y_i$};
        \node[below] at (0, 0) {\scriptsize{$i$}};
        \node[below] at (1, 0) {\scriptsize{$i+1$}};
        \draw[decorate,decoration={brace,amplitude=5pt,mirror}] (-2,
        0.2) -- (-9, 0.2);
        \node[above] at (-5.5, 0.3) {$\widehat{\sigma}_-^2 (i, 1, k_n, 2)$};
        \node[below] at (-2, 0) {\scriptsize{$i-2$}};
        \node[below] at (-9, 0) {\scriptsize{$i-k_n-2$}};
        \draw[decorate,decoration={brace,amplitude=5pt,mirror}] (10,
        0.2) -- (3, 0.2);
        \node[above] at (6.5, 0.3) {$\widehat{\sigma}_+^2 (i, 1, k_n, 2)$};
        \node[below] at (3, 0) {\scriptsize{$i+3$}};
        \node[below] at (10, 0) {\scriptsize{$i+k_n+3$}};
        \foreach \xStart/\xEnd in {-9/-2, 3/10} {
          \fill [fill=red, opacity=0.2] (\xStart,0.08) -- (\xStart,
          0.0) -- (\xEnd, 0.0) -- (\xEnd, 0.08) -- cycle;
        }
        \foreach \xStart/\xEnd in {0/1} {
          \fill [fill=blue, opacity=0.2] (\xStart,0.08) -- (\xStart,
          0.0) -- (\xEnd, 0.0) -- (\xEnd, 0.08) -- cycle;
        }
      \end{tikzpicture}
    }
    \caption{The shifted all-observation leverage effect estimator
    (see Remark~\ref{rem:shifted-spot-vol})}
    \label{fig:base-estimators-3}
  \end{subfigure}
  \begin{subfigure}[c]{\textwidth}
    \centering
    \resizebox{\textwidth}{!}{
      \begin{tikzpicture}
        \draw[-latex] (-15, 0) -- (17, 0);
        \foreach \x in {-14, -13} {\draw (\x, 0) -- (\x, 0.1);}
        \foreach \x in {-10, ..., -6} {\draw (\x, 0) -- (\x, 0.1);}
        \foreach \x in {-3, ..., 4} {\draw (\x, 0) -- (\x, 0.1);}
        \foreach \x in {7, ..., 11} {\draw (\x, 0) -- (\x, 0.1);}
        \foreach \x in {14, ..., 16} {\draw (\x, 0) -- (\x, 0.1);}
        \node[above] at (-11.5, -0.15) {$\cdots\cdots$};
        \node[above] at (-4.5, -0.15) {$\cdots\cdots$};
        \node[above] at (5.5, -0.15) {$\cdots\cdots$};
        \node[above] at (12.5, -0.15) {$\cdots\cdots$};
        \draw[decorate,decoration={brace,amplitude=5pt,mirror}] (2,
        0.2) -- (0, 0.2);
        \node[above] at (1, 0.3) {$\Delta Y_i$};
        \node[below] at (0, 0) {\scriptsize{$i$}};
        \node[below] at (2, 0) {\scriptsize{$i+2$}};
        \draw[decorate,decoration={brace,amplitude=5pt,mirror}] (-2,
        0.2) -- (-14, 0.2);
        \node[above] at (-8, 0.3) {$\widehat{\sigma}_-^2 (i, 2, k_n', 1)$};
        \node[below] at (-2, 0) {\scriptsize{$i-2$}};
        \node[below] at (-14, 0) {\scriptsize{$i-2k_n'-2$}};
        \draw[decorate,decoration={brace,amplitude=5pt,mirror}] (16,
        0.2) -- (4, 0.2);
        \node[above] at (10, 0.3) {$\widehat{\sigma}_+^2 (i, 2, k_n', 1)$};
        \node[below] at (4, 0) {\scriptsize{$i+4$}};
        \node[below] at (16, 0) {\scriptsize{$i+2k_n'+4$}};
        \foreach \xStart/\xEnd in {-14/-2, 4/16} {
          \fill [fill=red, opacity=0.2] (\xStart,0.08) -- (\xStart,
          0.0) -- (\xEnd, 0.0) -- (\xEnd, 0.08) -- cycle;
        }
        \foreach \xStart/\xEnd in {0/2} {
          \fill [fill=blue, opacity=0.2] (\xStart,0.08) -- (\xStart,
          0.0) -- (\xEnd, 0.0) -- (\xEnd, 0.08) -- cycle;
        }
      \end{tikzpicture}
    }
    \caption{The subsampling leverage effect estimator defined in
    Equation~\eqref{eq:subsample-noisy} with $H=2$}
    \label{fig:subsample-estimators}
  \end{subfigure}
  \caption{
    Increments in the base and subsampling estimators. The increments
    in the base estimators shown in panels
    (\subref{fig:base-estimators-1}) to
    (\subref{fig:base-estimators-3}) are used as proxies for
    $\int_{t_i}^{t_{i+1}} \ud \true_t$, whereas the increment in the
    subsampling estimator shown in panel
    (\subref{fig:subsample-estimators}) is used as a proxy for
    $\int_{t_i}^{t_{i+H}} \ud \true_t$.
  }
  \label{fig:base-estimators}
\end{figure}

\subsubsection{All-Observation Estimator}

To start with, consider an \emph{all-observation Leverage Effect}
(LE) estimator that directly utilizes all noisy observations,
\begin{align}\label{eq:all-observation}
  \esty^{\rm (all)}_T
  =
  \sum_{i=k_n+1}^{n-k_n-2} (\Delta Y_{i})
  \: \widehat{\delta} (i, 1, k_n).
\end{align}
Here, the window size $k_n$ satisfies that $k_n \to \infty$ and $k_n
\Delta_n \to 0$ as $n\to \infty$. The estimator is very similar to
the continuous leverage effect estimator\footnote{Since jumps are not
  included, the truncation term in their original estimator is omitted
here.} studied
by~\citet{aitsahalia2017EstimationContinuousDiscontinuous}. The only
difference is that our estimator shifts the windows for spot
volatility estimates outward by $\Delta_{n}$, as illustrated by
Figure~\ref{fig:base-estimators-1}~and~\ref{fig:base-estimators-2}.
To see the reason, consider applying their estimator directly to
noisy observations, and it follows that there is a divergent bias due
to noise of
\begin{align}
  \E\bigl(
    \esty^{[\text{AFLWY17}]}_T \big| \calF
  \bigr)
  -
  \estx^{[\text{AFLWY17}]}_T
  =
  2 \nu_3 T^{-1} k_n^{-1} n^2
  + O_p(n),
\end{align}
when $\nu_3$ is not strictly zero. In contrast, the estimator in
Equation~\eqref{eq:all-observation} is unbiased due to noise with a
smaller variance, as described by the next
Proposition~\ref{prop:all-observation-noise}.

\begin{proposition}\label{prop:all-observation-noise}
  Under Assumptions~\ref{ass:process} and
  \ref{ass:noise}\ref{ass:noise-iid}, as $n\to\infty$, we have
  \begin{align}
    \E \big(
      \esty^{\rm (all)}_T
      \big| \calF
    \bigr)
    &=
    \estx^{\rm (all)}_T,
    \\
    \Delta_n^3 k_n^2
    \var\bigl( \esty^{\rm (all)}_T \big| \calF \bigr)
    &\convp
    (8 \nu_2 \nu_4 + 16 \nu_2^3 + 8 \nu_3^2) T.
    \label{eq:all-observation-noise-variance}
  \end{align}
\end{proposition}

With the window shift, the bias due to noise is eliminated, and the
variance due to noise is reduced, while the asymptotic behavior of
the noise-free estimator $\estx_T^{\rm (all)}$ remains the same as
$\estx^{[\text{AFLWY17}]}_{T}$. As a direct result of this
unbiasedness, a debias step in TSRV or MSRV is no longer needed.
However, the all-observation estimator is not consistent in the
presence of noise: the variance due to noise is $O(n^3 k_n^2)$ while
the noise due to discretization is $O(k_n^{-1} + n^{-1} k_n)$,
resulting in an exploding total variance.

\subsubsection{SALE: Subsampling-and-Averaging Estimator}

To mitigate the impact of noise, we employ a subsampling procedure: a
subsample of observations is denoted by a pair of integers $(H, h)$,
where $H \geq 1$ and $1 \leq h \leq H$ are the scale and
index of the subsample. The $j$th observation in subsample $(H, h)$
corresponds to the original index $j_{H,h} = jH + h - 1$, where $j =
0, 1, \dotsc, n_{H,h}$ and $n_{H,h} = \lfloor (n - h + 1)/H \rfloor$.
Assume that $k_n\to\infty$, $H_n\to\infty$ and $k_nH_n\Delta_n \to 0$
as $n\to\infty$. A subsampling estimator is constructed by applying
the all-observation estimator to a subsample of observations:
\begin{align}\label{eq:subsample-noisy}
  \esty^{(H_n, h)}_T
  =
  \sum_{i=k_n+1}^{n_{H_n,h}-k_n-2}
  (\Delta_{H_n} Y_{i_{H_n, h}})
  \: \widehat{\delta} (i_{H_n, h}, H_n, k_n).
\end{align}
As a variant of the all-observation estimator, the subsampling
estimator is not consistent. Nor is it statistically sound, as it
fails to utilize the full information in the observed data.
Therefore, by taking average over the subsampling estimators with the
same scale, we obtain an \emph{Subsampling-and-Averaging Leverage
Effect} (SALE) estimator:
\begin{align}\label{eq:SALE-noisy}
  \esty^{(H_n)}_T
  =
  \frac{1}{H_n}
  \sum_{h=1}^{H_n} \esty^{(H_n, h)}_T
  =
  \frac{1}{H_n}
  \sum_{i=(k_n+1)H_n}^{n-(k_n+2)H_n}
  (\Delta_{H_n} Y_i)
  \: \widehat{\delta} (i, H_n, k_n),
\end{align}
which can be viewed as summation of overlapping sparse increments.
The SALE estimator is consistent under Assumption~\ref{ass:noise}
with proper choices of $k_n$ and $H_n$. This is primarily attributed
to the fact that the correlation due to noise between different
subsampling estimators in Equation~\eqref{eq:SALE-noisy} can be
controlled by the dependence level of noise. Specifically, if
Assumption~\ref{ass:noise}\ref{ass:noise-iid} holds, this correlation
becomes zero. Thus, the variance due to noise significantly reduces,
as the following proposition describes.
\begin{proposition}\label{prop:SALE-noise}
  Suppose that Assumptions~\ref{ass:process} and
  \ref{ass:noise}\ref{ass:noise-dep} hold, and that $H_n > 2q$.
  Define the generalized autocorrelation functions (ACFs) of noise
  for any $l \in \Z$ as
  \begin{align}\label{eq:general-acf}
    \rho_2(l) &= \corr(\eps_i, \eps_{i+l})
    = \frac{\E [ \eps_{i}\eps_{i+l} ]}{\nu_2}, \\
    \rho_3(l) &= \corr(\eps_i, \eps_{i+l}^2)
    = \frac{\E [ \eps_{i}\eps_{i+l}^{2} ]}{\sqrt{\nu_2(\nu_4 - \nu_2^2)}}, \\
    \rho_4(l) &= \corr( \eps_{i}^{2}, \eps_{i+l}^{2} )
    = \frac{\E[ \eps_{i}^{2} \eps_{i+l}^{2} ] - \nu_2^2}{\nu_4 - \nu_2^2}.
  \end{align}
  As $n\to\infty$, we have
  \begin{gather}
    \label{eq:SALE-var-noise}
    \Delta_n^3 H_n^4 k_n^2
    \var \bigl( \esty^{(H_n)}_T \big|
    \calF \bigr)
    \convp
    \Phi T, \\
    \label{eq:Phi}
    \text{where }
    \Phi = 8 \nu_2 (\nu_4 - \nu_2^2)
    \sum_{l=-q}^{q} \bigl( \rho_2(l)\rho_4(l) + \rho_3(l) \rho_3(-l) \bigr)
    + 24 \nu_2^3 \sum_{l=-q}^{q} \rho_2^3(l).
  \end{gather}
  Specifically, under Assumptions~\ref{ass:process} and
  \ref{ass:noise}\ref{ass:noise-iid},
  Equation~\eqref{eq:SALE-var-noise} holds with
  $\Phi = 8 \nu_2 \nu_4 + 16 \nu_2^3 + 8 \nu_3^2$.
\end{proposition}

The SALE estimator has a variance due to noise of $O(n^3 H_n^{-4}
k_n^{-2})$, and a variance due to discretization of $O(k_n^{-1} +
n^{-1} H_n k_n)$. Suppose $H_n \propto n^a$ and $k_n \propto
(n/H)^b$, the optimal total variance $O(n^{-1/7})$ is achieved at
$a=5/7$ and $b=1/2$. Despite being consistent, this is far from the
optimal rate $O(n^{-1/4})$ of pre-averaging approach.

\subsubsection{MSLE: Multi-Scale Estimator}\label{sec:MSLE}

Consider a set of scales $1 \leq H_1 < \dots < H_{M_n} \leq n^a$ for
some $a\in(0,1)$, where $M_n>0$. An \emph{Multi-Scale Leverage
Effect} (MSLE) estimator is defined as a weighted average of SALE
estimators at different scales:
\begin{align}\label{eq:MSLE-noisy}
  \esty^{\rm (MS)}_T = \sum_{p=1}^{M_n} w_{p} \esty^{(H_p)}_T,
\end{align}
where the weight vector $\bm w=(w_1, \dotsc, w_{M_n})$ satisfies $\bm
w^T \bm 1_{M_n} = 1$ with $\|\bm w\|_1$ bounded.

When the noise is i.i.d., the covariance due to noise between a pair
of SALE estimators are non-zero only when a scale is double the
other, and the corresponding correlation coefficient is small (for
example, 0.1 for Gaussian noise). For serial dependent noise, the
case is more complicated and analytical results are hard to derive.
Instead, numerical calculation is available (see Supplementary Material).

\begin{proposition}\label{prop:MSLE-noise}
  Under Assumptions~\ref{ass:process} and
  \ref{ass:noise}\ref{ass:noise-iid}, and suppose $k_p = \lfloor
  \beta \lfloor n/H_p\rfloor^b \rfloor$ holds for all $p\in\{1,
  \dotsc, M_n\}$ with some constant $\beta>0$ and $b\in(0,1)$. For
  any $p, q \in \{1, \dotsc, M_n\}$, as $n\to\infty$, we have
  \begin{gather}
    \Delta_n^3 H_p^2 H_q^2 k_p k_q
    \cov \bigl(
      \esty^{(H_p)}_T, \esty^{(H_q)}_T
      \big| \calF
    \bigr)
    \convp
    F_{p,q} T. \\
    \label{eq:MSLE-noise-Fpq}
    \text{where }
    F_{p,q} = (8\nu_2\nu_4 + 16\nu_2^3 + 8\nu_3^2) 1_{\{p=q\}} +
    2\nu_2(\nu_4-\nu_2^2) (1_{\{H_p/H_q=2\}} + 1_{\{H_q/H_p=2\}}),
  \end{gather}
  and thus
  \begin{align}
    \frac
    {\var \bigl( \esty^{\rm (MS)}_T
    \big| \calF \bigr)}
    {\displaystyle \Delta_n^{-3} \sum_{p=1}^{M_n} \sum_{q=1}^{M_n}
    \frac{w_p}{H_p^2 k_p} \cdot F_{p,q} \cdot \frac{w_q}{H_q^2 k_q}}
    \convp
    T.
  \end{align}
\end{proposition}

\begin{remark}\label{rem:shifted-spot-vol}
  The ``double scale'' terms in $F_{p,q}$ can be removed by using a
  further shifted spot volatility window with $s=2$ in
  Equation~\eqref{eq:spot-vol-delta}, as illustrated in
  Figure~\ref{fig:base-estimators-3}. A similar proposition can be
  established with $F_{p,q} = (8\nu_2\nu_4 + 8\nu_2^3 + 8\nu_3^2)
  1_{\{p=q\}}$, eliminating the cross terms and reducing the variance.
\end{remark}

The MSLE estimator effectively reduces the variance due to noise. For
example, setting $H_p=p$ for $p=1, \dotsc, M_n$, $M_n = \lfloor n^a
\rfloor$, and $w_p \propto p^{4-2b}$, the variance due to noise is
$O(n^{3-5a-2b+2ab})$, while the variance due to discretization is
$O(n^{-(1-a)(b \land (1-b))})$. Thus, by selecting $a=5/9$ and $b =
1/2$, an optimal convergence rate of $n^{1/9}$ is achieved for MSLE,
close to the optimal rate $n^{1/8}$ of pre-averaging approach.

\section{Main Results}\label{sec:results}

\subsection{Central Limit Theorems for SALE}

We start by establishing the following theorem for the noise-free
version of SALE. Two scenarios for the scale $H_n$ are considered:
either it is fixed, or it goes to infinity as $n \to \infty$.
Hereafter, we use $\convst$ to denote stable convergence in law.

\begin{assumption}\label{ass:SALE-para}
  Suppose that $H_n$ and $k_n$ satisfy one of the following conditions:
  \begin{enumerate}[label=(\alph*)]
    \item \label{ass:SALE-para-finite} $H_n=H$ is a given positive
      integer, $k_n = \lfloor \beta \lfloor n/H \rfloor^b \rfloor$
      for some $\beta > 0$ and $b \in (0,1)$.
    \item \label{ass:SALE-para-asym} $H_n = \lfloor \alpha n^a \rfloor$,
      $k_n = \lfloor \beta
      \lfloor n/H_n \rfloor^b \rfloor$ for some $\alpha, \beta > 0$ and
      $a, b \in (0,1)$.
  \end{enumerate}
\end{assumption}

\begin{theorem}\label{thm:SALE-clt-noisefree}
  \begin{enumerate}[label=(\arabic*)]
    \item [] %
    \item \label{thm:SALE-clt-noisefree-1} Under
      Assumptions~\ref{ass:process} and
      \ref{ass:SALE-para}\ref{ass:SALE-para-finite}, let $u_n =
      \sqrt{k_n \land (k_n H \Delta_n)^{-1}}$. There exist a standard
      Brownian motion $(W_{1,t})_{t\geq 0}$ independent of $\calF$
      and a predictable process $(\zeta_{1,t})_{t\geq 0}$ such that,
      as $n\to\infty$,
      \begin{gather}
        u_n \bigl( \estx^{(H)}_T - \true_T \bigr)
        \convst
        \int_0^T \zeta_{1,t} \ud W_{1,t},
        \\
        \label{eq:avar-SALE-clt-noisefree-1}
        \int_0^T \zeta_{1,t}^2 \ud t =
        \frac{u_n^2}{k_n} \left(\frac{8}{3} + \frac{4}{3H^2}\right)
        \int_0^T \sigma_s^6 \ud t +
        u_n^2 k_n H \Delta_n \frac{2}{3} \int_0^T \sigma_t^2 \ud
        \langle \sigma^2, \sigma^2 \rangle_t.
      \end{gather}
    \item \label{thm:SALE-clt-noisefree-2} Under
      Assumptions~\ref{ass:process} and
      \ref{ass:SALE-para}\ref{ass:SALE-para-asym}, there exist a
      standard Brownian motion $(W_{1,t})_{t\geq 0}$ independent of
      $\calF$ and a predictable process $(\zeta_{1,t})_{t\geq 0}$
      such that, as $n\to\infty$,
      \begin{gather}
        n^{\frac{1}{2}(1-a)(b\land (1-b))}
        \bigl( \estx^{(H_n)}_T - \true_T \bigr)
        \convst
        \int_0^T \zeta_{1, t} \ud W_{1, t},
        \\
        \label{eq:avar-SALE-clt-noisefree-2}
        \int_{0}^{T} \zeta_{1, t}^{2} \ud t
        =
        \frac{8 \alpha^b}{3 \beta} \int_0^T \sigma_t^6 \ud t
        \cdot 1_{(0, 1/2]}(b)
        +
        \frac{2 \alpha^{1-b} \beta T}{3} \int_0^T \sigma_t^2 \ud
        \langle \sigma^2, \sigma^2 \rangle_t
        \cdot 1_{[1/2, 1)}(b).
      \end{gather}
  \end{enumerate}
\end{theorem}

\begin{remark}
  Taking $H=1$ in
  Theorem~\ref{thm:SALE-clt-noisefree}\ref{thm:SALE-clt-noisefree-1}
  yields the central limit theorem for the all-observation estimator,
  which has a same asymptotic variance as the continuous leverage
  effect estimator in~\citet{aitsahalia2017EstimationContinuousDiscontinuous}.
\end{remark}

\begin{remark}\label{remark:spot-vol-errors}
  Similar to existing work on leverage effect
  estimation~\citep{wang2014EstimationLeverageEffect,
    aitsahalia2014HighFrequencyFinancialEconometrics,
    aitsahalia2017EstimationContinuousDiscontinuous,
    kalnina2017NonparametricEstimationLeverage,
  yang2023EstimationLeverageEffect}, the asymptotic variance is
  determined by the spot volatility estimation, which consists of two
  sources of errors: the \emph{price variation error} and the
  \emph{volatility variation
  error}~\citep{aitsahalia2017EstimationContinuousDiscontinuous},
  corresponding to the first and second terms in
  Equation~\eqref{eq:avar-SALE-clt-noisefree-1} or
  \eqref{eq:avar-SALE-clt-noisefree-2}. Intuitively, increasing $k_n$
  leads to a wider spot volatility estimation window and thus a
  larger sample size for that estimation, which reduces the price
  variation error. On the other hand, this increases the volatility
  variation error, because the estimated volatility becomes less
  ``local''. The optimal choice of $k_n$ is thus a trade-off between
  these two sources of error.
\end{remark}

Next, we establish the following theorem for the noisy version of
SALE. Notably,
Assumption~\ref{ass:SALE-para}\ref{ass:SALE-para-finite} is not
considered, as a finite $H$ leads to a divergent variance due to
noise and thus an inconsistent estimator.

\begin{theorem}\label{thm:SALE-clt-noisy}
  Under Assumptions~\ref{ass:process},
  \ref{ass:noise}\ref{ass:noise-dep} and
  \ref{ass:SALE-para}\ref{ass:SALE-para-asym}, suppose that $H_n >
  2q$ and $4a + 2b - 2ab > 3$. Let $r = [(1-a)(b\land (1-b))] \land
  [4a+2b-2ab-3]$ and let $\Phi$ be as defined in
  Equation~\eqref{eq:Phi}. There exist a standard Brownian motion
  $(W_{2, t})_{t\geq 0}$ independent of $\calF$ and a predictable
  process $(\zeta_{2, t})_{t\geq 0}$, such that, as $n\to\infty$,
  \begin{align}
    & n^{\frac{1}{2}r}
    \bigl(
      \esty^{(H_n)}_T - \true_T
    \bigr)
    \convst
    \int_{0}^{T} \zeta_{2, t} \ud W_{2, t},
    \\
    \int_{0}^{T} \zeta_{2, t}^{2} \ud t
    & =
    \frac{8 \alpha^{b}}{3 \beta} \int_{0}^{T} \sigma_{t}^{6} \ud t
    \cdot 1_{\{(1-a)b\}}(r)
    \notag
    \\ & \qquad
    +
    \frac{2 \alpha^{1-b} \beta T}{3} \int_{0}^{T} \sigma_{t}^{2}
    \ud \langle \sigma^{2}, \sigma^{2} \rangle_t
    \cdot 1_{\{(1-a)(1-b)\}}(r)
    \notag
    \\ & \qquad
    + \frac{1}{\alpha^{4-2b} \beta^2 T^3}
    \int_{0}^{T} \Phi \ud t
    \cdot 1_{\{4a+2b-2ab-3\}}(r).
  \end{align}
\end{theorem}

\subsection{Central Limit Theorems for MSLE}

To establish the limit theorems for MSLE, the following conditions on
scales, window sizes and weights are established, and the covariances
due to discretization between SALE estimators are given by
Proposition~\ref{prop:SALE-acov-noisefree} and \ref{prop:adj-factor-asym}.

\begin{assumption}\label{ass:MSLE-para}
  The scales $\{H_p\}_{p=1}^{M_n}$ satisfy $1 \leq H_1 < \dotsc <
  H_{M_n} \leq n^a$ for some $a\in(0,1)$. The window sizes are $k_p =
  \lfloor \beta \lfloor n/H_p \rfloor^b \rfloor$ for all $p \in \{1,
  \dotsc, M_n\}$ for some $\beta > 0$ and $b \in (0, 1)$. The weight
  vector $\bm w=(w_1, \dotsc, w_{M_n})$ satisfies that $\bm w^T \bm
  1_{M_n} = 1$ and that $\|\bm w\|_1$ is bounded.
\end{assumption}

\begin{proposition}\label{prop:SALE-acov-noisefree}
  Suppose that Assumptions~\ref{ass:process} and \ref{ass:MSLE-para}
  hold. For any $1 \leq q \leq p \leq M_n$, let $u_n = \sqrt{k_p
  \land (k_p H_p \Delta_n)^{-1}}$. There exist $v_{p,q}^{(1)},
  v_{p,q}^{(2)} \in [0, 1)$ that depend on $H_p, H_q$ and $n$, such that,
  as $n\to\infty$,
  \begin{gather}
    u_n
    \begin{pmatrix}
      \estx^{(H_p)}_T - \true_T \\
      \estx^{(H_q)}_T - \true_T
    \end{pmatrix}
    \convst
    \normal\left(0,
      u_n^2
      \begin{bmatrix}
        \Sigma_{p,p}^{\mathrm{(disc)}} &
        \Sigma_{p,q}^{\mathrm{(disc)}} \\
        \Sigma_{p,q}^{\mathrm{(disc)}} &
        \Sigma_{q,q}^{\mathrm{(disc)}}
    \end{bmatrix}\right),
    \\
    \label{eq:SALE-acov-disc}
    \text{where }
    \Sigma_{p,q}^{\mathrm{(disc)}} =
    \frac{1}{k_p} \cdot 4 v_{p,q}^{(1)} \frac{H_q}{H_p}
    \cdot \int_0^T \sigma_t^6 \ud t
    +
    k_pH_p\Delta_n \cdot \frac{2}{3} v_{p,q}^{(2)}
    \left(\frac{k_qH_q}{k_pH_p}\right)^2
    \cdot \int_0^T \sigma_t^2 \ud \langle \sigma^2, \sigma^2 \rangle_t,
  \end{gather}
  and the limiting process is independent of $\calF$.
\end{proposition}

\begin{remark}
  Factors $v_{p,q}^{(1)}$ and $v_{p,q}^{(2)}$, arising from the grid
  structures of scales $H_p$ and $H_q$, contribute to price variation
  error and volatility variation error, respectively. Their
  definitions are detailed in Supplementary Material.
\end{remark}

\begin{proposition}\label{prop:adj-factor-asym}
  Suppose that the conditions of
  Proposition~\ref{prop:SALE-acov-noisefree} hold. Consider two
  sequences of scales $H_p$ and $H_q$ indexed by $n$, satisfying that
  (i) $H_p \geq H_q$ for all $n$, (ii) as $n\to\infty$, $H_p, H_q \to
  \infty$ and $H_{q} / H_{p} \to \rho$ for some constant $\rho \in (0, 1]$.
  Then, as $n\to\infty$, we have
  \begin{align}\label{eq:adj-factor-asym}
    v_{p,q}^{(1)} \to 1 - \frac{\rho}{3}
    \quad \text{and} \quad
    v_{p,q}^{(2)} \to 1.
  \end{align}
\end{proposition}

For the asymptotic behavior of MSLE, a specific set of consecutive
scales are considered, and the weights are defined with a continuous
bounded function.

\begin{assumption}\label{ass:MSLE-para-detail}
  Suppose that $\{H_p\}_{p=1}^{M_n}$ and $\bm{w}$ satisfy the
  following conditions:
  \begin{enumerate}[label=(\alph*)]
    \item \label{ass:MSLE-para-detail-scale} $H_p = m_n+p$ for all $p
      \in \{1, \dotsc, M_n\}$. Defining $H_n^* = H_{M_n}$, the
      sequences of positive integers $m_n$ and $M_n$ are selected
      such that, as $n\to\infty$, $H_n^* / n^a \to \alpha$, $m_n /
      H_n^* \to c$ for some constants $\alpha > 0$, $a\in(0,1)$, $c\in(0, 1)$.
    \item \label{ass:MSLE-para-detail-weight} $w_p =
      \frac{1-c}{M_n}\phi(c+\frac{p}{H_n^*})$ for all $p \in \{1,
      \dotsc, M_n\}$, where $\phi: [c, \infty) \to \R$ is a
      continuous bounded function satisfying that $\int_c^1 \phi(x)
      \ud x = 1$.
  \end{enumerate}
\end{assumption}

\begin{theorem}\label{thm:MSLE-clt-noisefree}
  Suppose that Assumptions~\ref{ass:process}, \ref{ass:MSLE-para} and
  \ref{ass:MSLE-para-detail} hold. There exist a standard Brownian
  motion $(W_{3, t})_{t\geq 0}$ independent of $\calF$ and a
  predictable process $(\zeta_{3, t})_{t\geq 0}$, such that, as $n\to\infty$,
  \begin{align}
    & \quad
    n^{\frac{1}{2}(1-a)(b\land (1-b))}
    \bigl( \estx^{\rm (MS)}_T - \true_T \bigr)
    \convst
    \int_0^T \zeta_{3, t} \ud W_{3, t},
    \\
    \int_0^T \zeta_{3, t}^{2} \ud t
    &=
    \frac{8\alpha^b}{\beta} \int_c^1 \int_c^x \phi(x) \phi(y)
    x^b \left(\frac{y}{x} - \frac{y^2}{3x^2}\right) \ud y \ud x
    \cdot \int_0^T \sigma_t^6 \ud t \cdot 1_{(0, 1/2]}(b)
    \notag
    \\ & \quad
    + \frac{4\alpha^{1-b}\beta T}{3} \int_c^1 \int_c^x \phi(x)
    \phi(y) \frac{y^{2(1-b)}}{x^{1-b}} \ud y \ud x \cdot
    \int_0^T \sigma_t^2 \ud \langle \sigma^2, \sigma^2
    \rangle_t \cdot 1_{[1/2, 1)}(b).
    \label{eq:avar-MSLE-clt-noisefree-2}
  \end{align}
\end{theorem}

\begin{theorem}\label{thm:MSLE-clt-noisy}
  Under Assumptions~\ref{ass:process},
  \ref{ass:noise}\ref{ass:noise-iid}, \ref{ass:MSLE-para} and
  \ref{ass:MSLE-para-detail}, suppose that $5a+2b-2ab > 3$, and let
  $r = [(1-a)(b\land (1-b))] \land [5a+2b-2ab-3]$, $F_1 = 8\nu_2\nu_4
  + 16\nu_2^3 + 8\nu_3^2$ and $F_2 = 2\nu_2(\nu_4-\nu_2^2)$. There
  exist a standard Brownian motion $(W_{4, t})_{t\geq 0}$ independent
  of $\calF$ and a predictable process $(\zeta_{4, t})_{t\geq 0}$,
  such that, as $n\to\infty$,
  \begin{align}
    & \qquad \qquad
    n^{\frac{1}{2}r}
    \bigl( \esty^{\rm (MS)}_T - \true_T \bigr)
    \convst \int_0^T \zeta_{4, t} \ud W_{4, t},
    \\
    \notag
    \int_0^T \zeta_{4, t}^{2} \ud t
    &=
    \frac{8\alpha^b}{\beta} \int_c^1 \int_c^x \phi(x) \phi(y) x^b
    \left(\frac{y}{x} - \frac{y^2}{3x^2}\right) \ud y \ud x \cdot
    \int_0^T \sigma_t^6 \ud t \cdot 1_{\{(1-a)b\}}(r)
    \\ & \quad \notag
    + \frac{4\alpha^{1-b}\beta T}{3} \int_c^1 \int_c^x \phi(x)
    \phi(y) \frac{y^{2(1-b)}}{x^{1-b}} \ud y \ud x \cdot \int_0^T
    \sigma_t^2 \ud \langle \sigma^2, \sigma^2 \rangle_t \cdot
    1_{\{(1-a)(1-b)\}}(r)
    \\ & \quad \notag
    + \frac{1}{\alpha^{5-2b}\beta^2 T^3} \int_c^1 \phi^2(x)
    x^{-(4-2b)} \ud x \cdot \int_0^T F_1 \ud t \cdot 1_{\{5a+2b-2ab-3\}}(r)
    \\ & \quad
    + \frac{1_{(0, 1/2]}(c)}{2^{2-b}\alpha^{5-2b}\beta^2 T^3}
    \int_c^{1/2} \phi(x) \phi(2x) x^{-(4-2b)} \ud x \cdot \int_0^T
    F_2 \ud t \cdot 1_{\{5a+2b-2ab-3\}}(r).
  \end{align}
\end{theorem}

Note that the asymptotic variances in
Theorems~\ref{thm:SALE-clt-noisefree} to \ref{thm:MSLE-clt-noisy} are
unobservable. Their consistent estimators and feasible central limit
theorems are detailed in Supplementary Material.

\section{Practical Aspects: Variances and Weights}
\label{sec:practical}

\subsection{Asymptotic Variances in Practice}
\label{sec:practical-variance}

Accurate asymptotic variance is important for parameter tuning. For
SALE, it helps pin down the optimal scale; for MSLE, it helps decide
the optimal weight distributions. Despite theoretical correctness,
the accuracy of derived variances may be affected by two situations
in practice: (i) small noise and (ii) violation of conditions in
Proposition~\ref{prop:adj-factor-asym}.

The asymptotic variances due to noise in
Proposition~\ref{prop:all-observation-noise}, \ref{prop:SALE-noise}
and \ref{prop:MSLE-noise} are established based on non-shrinkaging noise
assumptions. However, a small noise correction could be necessary in
practice, as some terms of small order become more pronounced as
noise becomes smaller. For all-observation estimators, we have
\begin{align}
  \Delta_n^3 & k_n^2
  \var\bigl( \esty^{\rm (all)} \big| \calF \bigr)
  \convp
  \notag
  \\
  & \Bigl(
    (8\nu_2\nu_4 + 16\nu_2^3 + 8\nu_3^2)
    + \frac{k_n}{n^2} \cdot (8\nu_4 + 16\nu_2^2) \int_0^T \sigma_t^2 \ud t
    + \frac{1}{n^2} \cdot 8T \nu_2 \int_0^T \sigma_t^4 \ud t
  \Bigr) T.
  \label{eq:all-observation-noise-variance-corrected}
\end{align}
Figure~\ref{fig:avar-noise} compares the simulated performance of
Equation~\eqref{eq:all-observation-noise-variance-corrected} and
Equation~\eqref{eq:all-observation-noise-variance}. Similar
correction for SALE estimators are provided in Supplementary Material.

\begin{figure}[!htp]
  \centering
  \centering
  \includegraphics[width=0.5\textwidth]{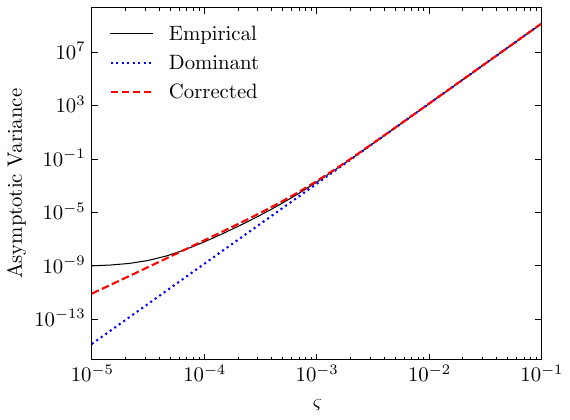}
  \caption{
    Simulated asymptotic variance due to noise of all-observation
    estimators under different noise levels. A fixed Heston path from
    Section~\ref{sec:simulation} is used for $X$, whereas i.i.d.
    $\normal(0, \varsigma^2)$ random variables are used for $\eps$.
    ``Empirical'': average of 5000 realizations of $\bigl( \esty^{\rm
    (all)}_T - \estx^{\rm (all)}_T \bigr)^2$. ``Dominant'': variance
    calculated with
    Equation~\eqref{eq:all-observation-noise-variance}.
    ``Corrected'': variance calculated with
    Equation~\eqref{eq:all-observation-noise-variance-corrected}.
  }
  \label{fig:avar-noise}
\end{figure}

As for asymptotic variance due to discretization, the limit
expression in Equation~\eqref{eq:adj-factor-asym} does not work for
$H_q / H_p \to 0$, and can be inaccurate when any of $n, H_p, H_q$ is
not large enough
(Theorem~\ref{thm:SALE-clt-noisefree}\ref{thm:SALE-clt-noisefree-1}
is an example with fixed $H_p = H_q$). Consequently,
Proposition~\ref{prop:SALE-acov-noisefree} will improve the accuracy
of asymptotic variances in such cases.

\subsection{Approximate Weights of MSLE Estimators}
\label{sec:practical-weight}

Suppose that Assumption~\ref{ass:MSLE-para} holds. Let $\bm{\Sigma}
\in \R^{M_n \times M_n}$ denote the total asymptotic covariance
matrix between scales, the optimal weight assignment can be
obtained by solving
\begin{align}
  \underset{\bm{w} \in \R^{M_n}}{\text{minimize}} & \quad V(\bm{w})
  = \bm{w}^\T \bm{\Sigma} \bm{w}, \\
  \text{subject to} & \quad \bm{w}^\T \bm{1}_{M_n} = 1.
\end{align}
The solution and the corresponding minimum are given by
\begin{align}\label{eq:weight-optimization-solution}
  \bm{w}^*
  =
  \frac
  {\bm{\Sigma}^{-1} \bm{1}_{M_n}}
  {\bm{1}_{M_n}^\T \bm{\Sigma}^{-1} \bm{1}_{M_n}},
  \quad
  V(\bm{w}^*)
  =
  \frac{1}{\bm{1}_{M_n}^\T \bm{\Sigma}^{-1} \bm{1}_{M_n}}
  =
  \Biggl(\sum_{p=1}^{M_n}\sum_{q=1}^{M_n}
  (\bm{\Sigma}^{-1})_{p,q}\Biggr)^{-1}.
\end{align}
However, direct application of
Equation~\eqref{eq:weight-optimization-solution} faces challenges in
practice: (i) the covariance matrix cannot be observed and therefore
estimated values are needed; (ii) the solution $\bm w^*$ can be
numerically unstable, sensitive to estimation errors of $\bm \Sigma$;
and (iii) calculating
Equation~\eqref{eq:weight-optimization-solution} requires matrix
inversion, which has an expensive time complexity of $O(M_n^3)$.
To address these challenges, we construct approximate weights for
MSLE estimators. For simplicity, hereafter in this section, we suppose that
Assumption~\ref{ass:MSLE-para-detail}\ref{ass:MSLE-para-detail-scale}
holds, and let
\begin{align}
  s_1 = \frac{4}{\beta} \int_0^T \sigma_t^6 \ud t, \quad
  s_2 = \frac{2\beta T}{3}
  \int_0^T \sigma_t^2 \ud \langle \sigma^2, \sigma^2 \rangle_t, \quad
  s_3 = \frac{8\nu_2\nu_4 + 16\nu_2^3 + 8\nu_3^2}{\alpha^{9/2} \beta^2 T^2}.
\end{align}
Detailed derivations for Equations~\eqref{eq:approx-weights} and
\eqref{eq:fredholm} presented in this section can be found in
Supplementary Material.

\subsubsection{Approximation in the Noise-Free Case}
\label{sec:practical-weight-noisefree}

In the absence of noise, the MSLE estimator can be used to enhance
the statistical efficiency of the all-observation estimator by
employing a more optimal weight assignment, rather than only
allocating all weight to the $H=1$ scale. For generality,
Definition~\ref{def:approx-weights} provides a closed-form expression
for the approximate weights applicable to all scales within $(m_n,
m_n+M_n]$. It is obtained by taking the limit $m_n \to \infty$ in
Equation~\eqref{eq:weight-optimization-solution}, where $\bm\Sigma$
is approximated by using Proposition~\ref{prop:adj-factor-asym} and
assuming that $s_2 \ll s_1$.

\begin{definition}[Approximate weights]\label{def:approx-weights}
  Suppose that Assumptions~\ref{ass:MSLE-para} and
  \ref{ass:MSLE-para-detail}\ref{ass:MSLE-para-detail-scale} hold.
  The approximate weights are given by $\bm{\widetilde w} = (\bm
  1_{M_n}^\T \bm{\widetilde\omega})^{-1} \bm{\widetilde\omega}$, where
  \begin{align}\label{eq:approx-weights}
    \bm{\widetilde\omega}_p =
    \begin{cases}
      2(m_n+1)^{-1/2}, & p = 1, \\
      (m_n+p)^{-3/2}, & p = 2, \dots, M_n-1, \\
      (m_n+M_n)^{-1/2}, & p = M_n.
    \end{cases}
  \end{align}
\end{definition}

Apart from being computationally and statistically efficient,
$\bm{\widetilde{w}}$ is numerically stable: since $\widetilde{w}_p >
0$ holds for any $p = 1, \dots, M_n$, we always have
$\sum_{p=1}^{M_n} |\widetilde{w}_p| = \sum_{p=1}^{M_n} \widetilde{w}_p = 1$.

\subsubsection{Approximation in the Noisy Case}
\label{sec:practical-weight-noisy}

Let $\varphi(x) = x^{-3/2} \phi(x)$, $\gamma=s_1/(3(s_1+s_2))$,
$\lambda=-(s_1+s_2)/s_3$, and
\begin{align}\label{eq:fredholm-kernel}
  K(x,y) = x^{3/2} y^{3/2} (x\lor y)^{1/2} \Biggl(
    \frac{x\land y}{x\lor y} - \gamma \left(\frac{x\land y}{x\lor y}\right)^2
  \Biggr).
\end{align}
The optimal $\phi(x)$ in
Assumption~\ref{ass:MSLE-para-detail}\ref{ass:MSLE-para-detail-weight}
is related to a Fredholm integral equation:
\begin{align}\label{eq:fredholm}
  \varphi(x) = kx^{3/2} + \lambda \int_c^1 K(x,y) \varphi(y) \ud y,
\end{align}
where $k \in \R$ is a constant such that $\int_c^1 \varphi(x) x^{3/2}
\ud x = 1$.
Equation~\eqref{eq:fredholm} relies on several simplifications:
(i) the conditions of Theorem~\ref{thm:MSLE-clt-noisy} hold with
$a=5/9$ and $b=1/2$, corresponding to the optimal convergence rate
of MSLE estimators;
(ii) the sparse off-diagonal terms of covariance due to noise are
omitted, as explained in Section~\ref{sec:MSLE}; and
(iii) finite-sample corrections in
Section~\ref{sec:practical-variance} are not considered.
These simplifications are made to isolate the dominant asymptotic
behavior for analytical tractablity.

\begin{figure}[!htp]
  \centering
  \begin{subfigure}[b]{0.24\linewidth}
    \includegraphics[width=\linewidth]{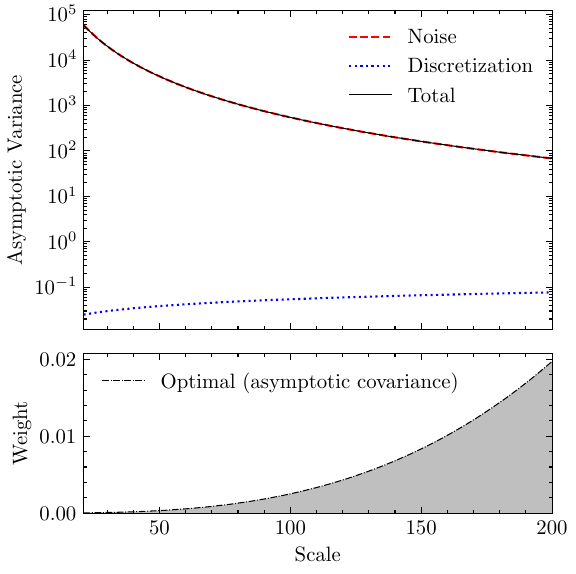}
    \caption{$\lambda=-1$}
  \end{subfigure}
  \hfill
  \begin{subfigure}[b]{0.24\linewidth}
    \includegraphics[width=\linewidth]{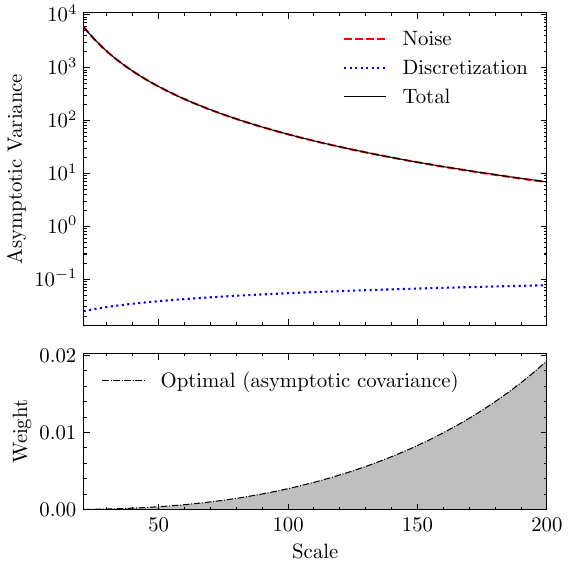}
    \caption{$\lambda=-10$}
  \end{subfigure}
  \hfill
  \begin{subfigure}[b]{0.24\linewidth}
    \includegraphics[width=\linewidth]{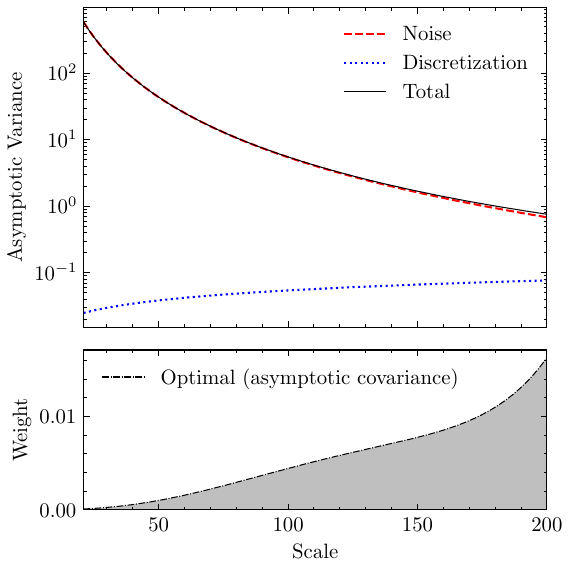}
    \caption{$\lambda=-10^2$}
  \end{subfigure}
  \hfill
  \begin{subfigure}[b]{0.24\linewidth}
    \includegraphics[width=\linewidth]{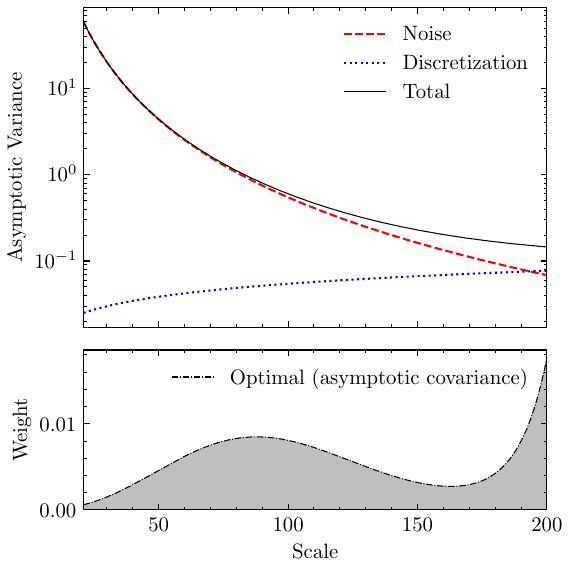}
    \caption{$\lambda=-10^3$}
  \end{subfigure}

  \begin{subfigure}[b]{0.24\linewidth}
    \includegraphics[width=\linewidth]{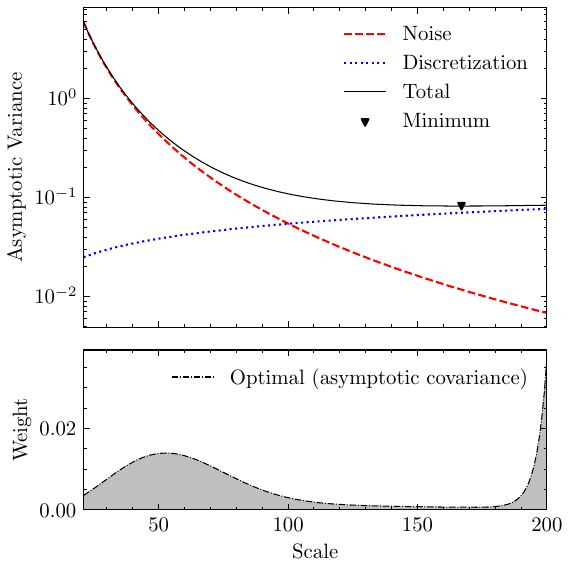}
    \caption{$\lambda=-10^4$}
  \end{subfigure}
  \hfill
  \begin{subfigure}[b]{0.24\linewidth}
    \includegraphics[width=\linewidth]{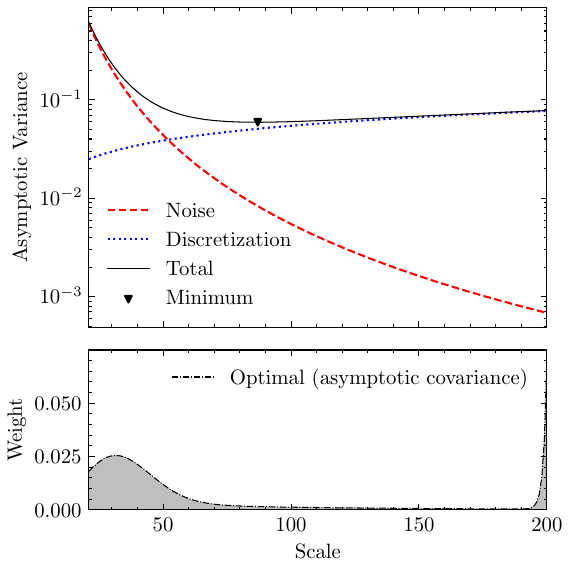}
    \caption{$\lambda=-10^5$}
  \end{subfigure}
  \hfill
  \begin{subfigure}[b]{0.24\linewidth}
    \includegraphics[width=\linewidth]{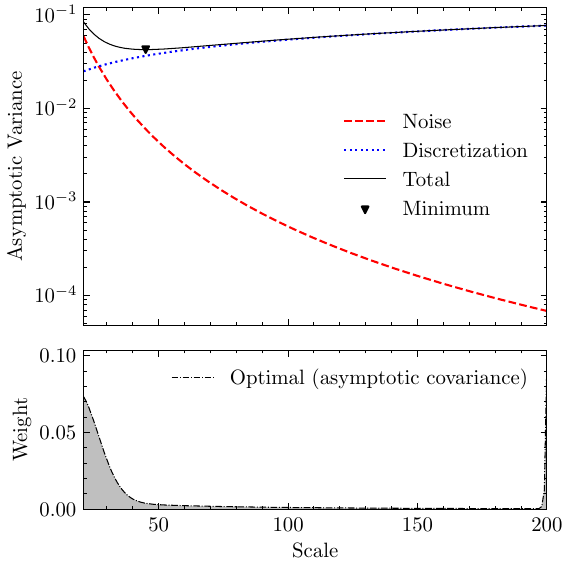}
    \caption{$\lambda=-10^6$}
  \end{subfigure}
  \hfill
  \begin{subfigure}[b]{0.24\linewidth}
    \includegraphics[width=\linewidth]{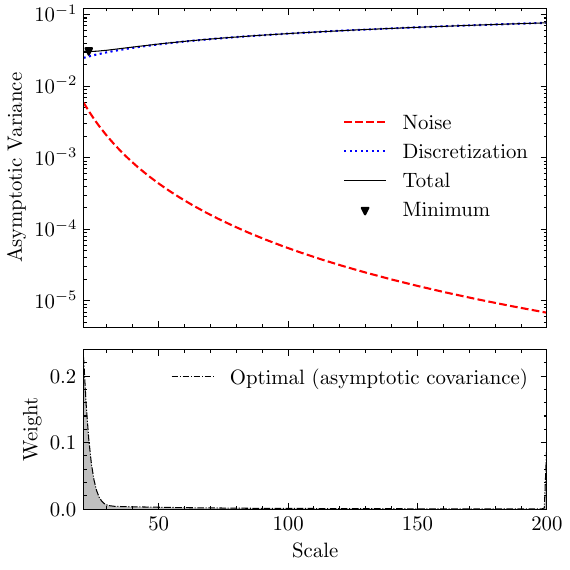}
    \caption{$\lambda=-10^7$}
  \end{subfigure}
  \caption{
    The effect of varying noise levels on the asymptotic variances of
    SALE estimators (top panels) and the resulting optimal weights
    for the MSLE estimator (bottom panels). Each subfigure
    corresponds to a different noise level, parameterized by
    $\lambda$, while the range of scales is fixed to isolate the
    effect of noise. Asymptotic variances are shown on a logarithmic
    scale. Parameters: $n=23400$, $m_n=20$, $M_n=180$, $H_n^*=200$,
    $s_1=1/2$, $s_2 = 1/2$, $s_3 = -1/\lambda$.
  }
  \label{fig:avar-weight}
\end{figure}

Figure~\ref{fig:avar-weight} shows the impact of the noise level in
this simplified situation. The asymptotic variances of SALEs and the
optimal weights of the MSLE are presented on a fixed set of scales.
Note that a smaller $|\lambda|$ represents a larger noise magnitude.
For a small $|\lambda|$, the variances due to noise dominate in most
scales. Specifically, as $|\lambda| \to 0$, the integral term in
Equation~\eqref{eq:fredholm} vanishes, and the optimal weights are
given by $\phi(x) \to 4x^3$. Despite having a closed-form expression,
the result is not useful in practice, as noise always dominates the
total variance, so the estimation error is too large. On the other
hand, as $|\lambda| \to \infty$, noise has negligible contributions
to the total variances, and Equation~\eqref{eq:fredholm} becomes
ill-posed. However, the proposed weights in
Definition~\ref{def:approx-weights} offer good approximations in this case.

Beyond these simplifications, the intuition behind our multi-scale
approach and approximate weighting strategy is illustrated by the
signature plot in Figure~\ref{fig:signature}, which compares
different estimators across various scales (or pre-averaging window
lengths) for a simulated path under realistic noise. While
illustrative, these patterns are systematic and confirmed by the
extensive simulations in Section~\ref{sec:simulation-finite-sample}.

\begin{figure}[!ht]
  \centering
  \includegraphics[width=0.7\textwidth]{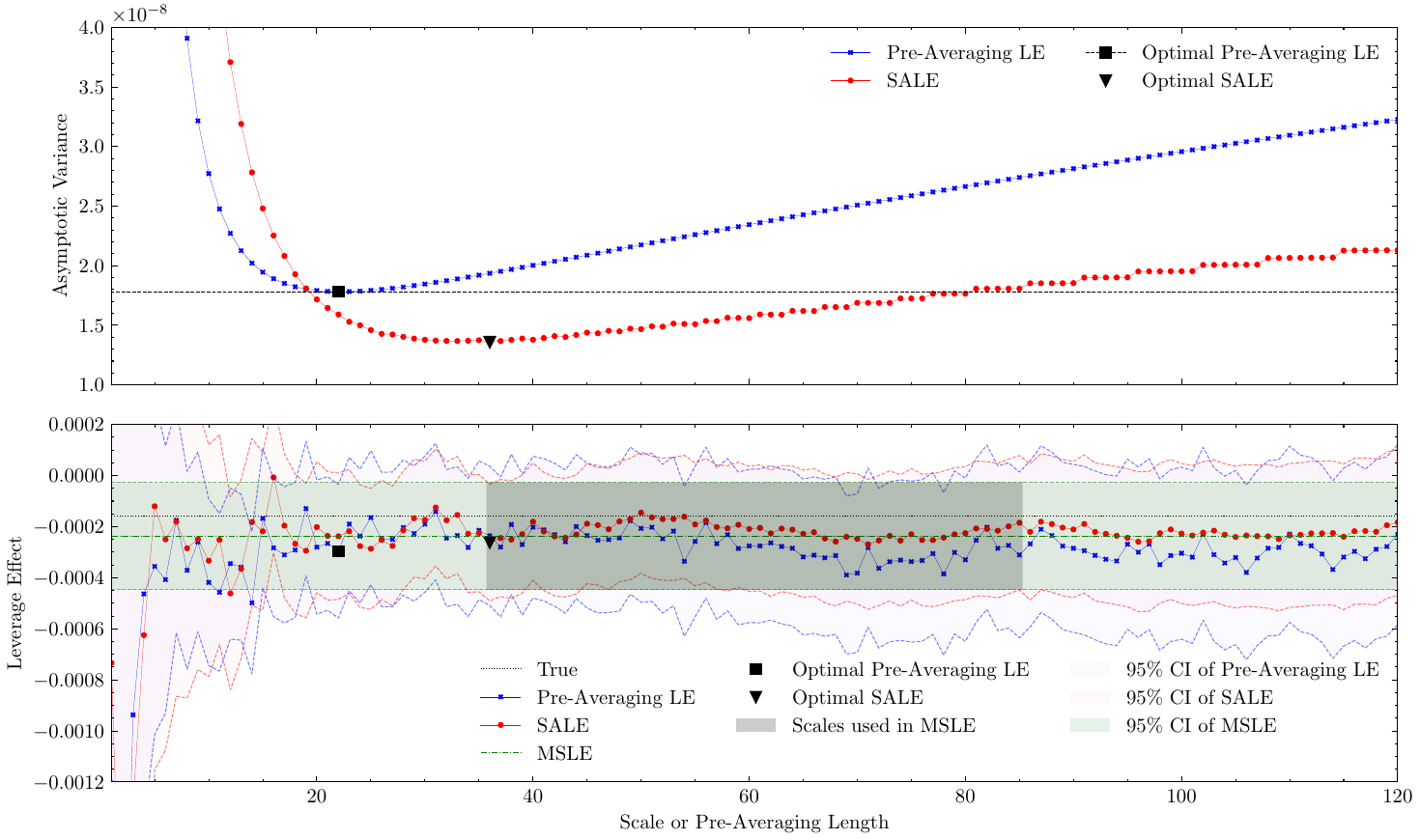}
  \caption{
    Signature plot for a simulated path with time horizon $T=5/252$
    and noise scale $\varsigma=3\times10^{-4}$.
  }
  \label{fig:signature}
\end{figure}

The plot reveals two key advantages. First, SALE outperforms the
pre-averaging estimator, as its minimum asymptotic variance (top
panel) is smaller, leading to a more efficient optimal
estimator.\footnote{The fundamental reason is that the SALE estimator
  has a smaller coefficient of price variation error as shown in
  Theorem~\ref{thm:SALE-clt-noisefree}, which is $8/3$, compared to a
larger coefficient $4$ in the pre-averaging estimator.} Second, MSLE
can improve upon SALE by averaging SALE estimates across an
appropriate range of scales (bottom panel) using an effective
weighting strategy, yielding a more accurate estimate with a tighter
confidence interval.

Based on this, we propose the following weighting method: (i) find an
optimal scale $\overline{H}_n$ for SALE estimators by minimizing the
total variance, using the corrections in
Section~\ref{sec:practical-variance}; and (ii) allocate weights by
applying Definition~\ref{def:approx-weights} with $m_n =
\overline{H}_n - 1$. This data-driven strategy ensures that MSLE
leverages the most informative scales, improving upon the optimal
SALE in two ways: the variance due to discretization is reduced
through the approximate weights, and the variance due to noise is
reduced because SALE estimators at subsequent scales exhibit smaller
noise contributions.

\section{Monte Carlo Simulations}\label{sec:simulation}

\subsection{Data Generating Processes}

The Heston model \citep{heston1993ClosedFormSolutionOptions} is used
to generate the discrete values of the underlying continuous
processes. The model is defined as
\begin{align} \ud X_t &= \left(\mu - \frac{\sigma_t^2}{2}\right) \ud
  t + \sigma_t \ud W_t, \\ \ud \sigma_t^2 &= \kappa (\theta -
  \sigma_t^2) \ud t + \gamma \sigma_t \left(\rho \ud W_t +
  \sqrt{1-\rho^2} \ud B_t \right),
\end{align} where $W_{t}$ and $B_{t}$ are independent Brownian
motions, with the leverage effect being $\true_T = \gamma \rho
\int_{0}^{T} V_{t} \ud t$. The parameters are set as follows:
$\mu=0.02, \kappa=5, \theta=0.04, \gamma=0.5, \rho=-0.7$. The initial
values are set as $X_0 = 0, V_0 = 0.02$.

Let the variance of noise random variables $\{\eps_i\}_{i=0}^n$ be
$\varsigma^2$. For independent noises, three distributions are considered:
(i) normal: $\eps_i \sim \normal(0, \varsigma^2)$;
(ii) uniform: $\eps_i \sim \unif(-\sqrt{3}\varsigma, \sqrt{3}\varsigma)$; and
(iii) skew-normal: $\eps_i$ has a PDF of $f(x) =
2\omega^{-1}\phi_0(\omega^{-1}(x-\xi)) \Phi_0(\alpha \omega^{-1}
(x-\xi))$, where $\phi_0$ and $\Phi_0$ are the PDF and CDF of
$\normal(0, 1)$, and $\xi=-\omega\delta\sqrt{2/\pi}$,
$\omega=\varsigma (1-2\delta^2/\pi)^{-1/2}$, $\delta=\alpha
(1+\alpha^2)^{-1/2}$, with the shape parameter $\alpha = 1$.
For dependent noises, consider
\begin{align}
  \text{(i) MA(2) process:}
  & \quad
  \eps_{i} = e_{i} + \theta_{1}e_{i-1} + \theta_{2}e_{i-2}, \quad
  e_i \overset{\text{i.i.d.}}{\sim} \normal(0,
  \varsigma^2(1+\theta_1^2+\theta_2^2)^{-1});
  \text{ and}
  \label{eq:ma2}
  \\
  \text{(ii) AR(1) process:}
  & \quad
  \eps_i = \phi \eps_{i-1} + e_i, \quad
  e_i \overset{\text{i.i.d.}}{\sim} \normal(0, \varsigma^2\sqrt{1-\phi^2})
  \text{ with } \phi \in (-1, 1).
  \label{eq:ar1}
\end{align}
Specifically, the AR(1) process is included to evaluate the
robustness of the proposed estimators. While AR(1) noise is not
$q$-dependent and thus technically violates
Assumption~\ref{ass:noise}, it serves as a benchmark model for
persistent, serially correlated noise
\citep{jacod2017StatisticalPropertiesMicrostructure,
li2022ReMeDIMicrostructureNoise}. As our subsequent simulations will
confirm, the proposed estimators are indeed robust to this moderate
violation, preserving their asymptotic normality and superior
finite-sample efficiency.

\subsection{Asymptotic Normality}
\label{sec:simulation-asym-normal}

This section validates the central limit theorems by examining the
distribution of standardized estimation errors. For each estimator,
the error is standardized using both its infeasible and feasible
asymptotic variance. According to the results in
Section~\ref{sec:results}, these standardized errors should converge
to a standard normal distribution.

We simulate 5000 paths for each scenario, covering noise-free,
independent noise, and dependent noise. For data generation, we set
$T=1/252$, $n=23400$, $\varsigma = 0.005$, $\theta_1=\pm 0.7$,
$\theta_2=0.5$, and $\phi=0.7$. For estimators, we set $\beta=1/2$,
$b=1/2$, with scales and weights detailed in Table~\ref{tab:scales-weights}.

\begin{table}[!ht]
  \centering
  \caption{Scales and weights used in SALE and MSLE.}
  \label{tab:scales-weights}
  \begin{tabular}{llll}
    \toprule
    \multirow{2.5}{*}{\textbf{Noise}} &
    \multicolumn{1}{c}{\textbf{SALE}} & \multicolumn{2}{c}{\textbf{MSLE}} \\
    \cmidrule(lr){2-2} \cmidrule(lr){3-4}
    & $H$ & $\{H_p\}_{p=1}^{M_n}$ & $\{w_p\}_{p=1}^{M_n}$ \\
    \midrule
    None & 1, 15 & $\{1, 2, \dotsc, 15\}$ & $w_p \propto H_p^{-3/2}$ \\
    Independent & 1, 15 & $\{1, 2, \dotsc, 15\}$ & $w_p \propto H_p^3$ \\
    Dependent & 15 & $\{11, 12, \dotsc, 15\}$ & $w_p \propto H_p^3$ \\
    \bottomrule
  \end{tabular}
\end{table}

Table~\ref{tab:standardized-errors} presents the summary statistics
for the standardized errors. Across all scenarios, these statistics
closely match those of a standard normal distribution, corroborating
our theoretical results (Theorems~\ref{thm:SALE-clt-noisefree} to
\ref{thm:MSLE-clt-noisy} and their feasible versions). Additional Q-Q
plots in the Supplementary Material further support these findings.

\begin{table}[!htb]
  \centering
  \caption{
    Summary statistics for standardized estimation errors under
    different settings. The standardized error is computed as
    $(\text{Estimate} - \text{True Value}) / \sqrt{\text{Asymptotic
    Variance}}$, using both infeasible and feasible versions of the
    asymptotic variance. The mean, standard deviation and the 25th,
    50th and 75th percentiles are reported. ``LE'' in the
    ``Estimator'' column represents the all-observation estimator
    (SALE with $H=1$), whereas ``SALE'' represents SALE with $H=15$.
  }
  \label{tab:standardized-errors}
  \resizebox{0.95\textwidth}{!}{
    \begin{tabular}{lllrrrrrrrrrr}
      \toprule
      \multicolumn{2}{l}{\multirow{2}{*}{\textbf{Noise Setting}}} &
      \multirow{2}{*}{\textbf{Estimator}} &
      \multicolumn{5}{c}{\textbf{Infeasible}} &
      \multicolumn{5}{c}{\textbf{Feasible}} \\
      \cmidrule(lr){4-8} \cmidrule(lr){9-13}
      & & & \textbf{Mean} & \textbf{Std} & $\bm{Q_1}$ & $\bm{Q_2}$ &
      $\bm{Q_3}$ & \textbf{Mean} & \textbf{Std} & $\bm{Q_1}$ &
      $\bm{Q_2}$ & $\bm{Q_3}$ \\
      \midrule
      \multicolumn{2}{l}{\multirow{3}{*}{\textbf{None}}} & LE &
      -0.016 & 1.013 & -0.693 & -0.013 & 0.683 &
      -0.014 & 1.013 & -0.693 & -0.013 & 0.684 \\
      & & SALE & 0.012 & 0.990 & -0.650 & 0.012 & 0.685 &
      0.013 & 0.996 & -0.652 & 0.012 & 0.687 \\
      & & MSLE & -0.012 & 1.013 & -0.694 & 0.005 & 0.678 &
      -0.010 & 1.014 & -0.696 & 0.005 & 0.676 \\
      \cmidrule(lr){1-13}
      \multirow{9}{*}{\textbf{Independent}} & \multirow{3}{*}{Normal} & LE
      & -0.022 & 1.004 & -0.693 & -0.030 & 0.650 &
      -0.022 & 0.998 & -0.691 & -0.030 & 0.643 \\
      & & SALE & -0.008 & 1.003 & -0.693 & -0.007 & 0.673 &
      -0.007 & 0.991 & -0.684 & -0.007 & 0.666 \\
      & & MSLE & -0.005 & 1.047 & -0.718 & 0.000 & 0.702 &
      -0.005 & 1.019 & -0.699 & 0.000 & 0.685 \\
      \cmidrule(lr){2-13}
      & \multirow{3}{*}{Uniform} & LE & 0.008 & 1.009 & -0.653 &
      0.019 & 0.680 & 0.008 & 1.002 & -0.651 & 0.019 & 0.676 \\
      & & SALE & 0.002 & 0.998 & -0.679 & -0.009 & 0.673 &
      0.002 & 0.983 & -0.667 & -0.009 & 0.663 \\
      & & MSLE & 0.004 & 1.023 & -0.707 & 0.008 & 0.689 &
      0.004 & 0.989 & -0.683 & 0.007 & 0.668 \\
      \cmidrule(lr){2-13}
      & \multirow{3}{*}{Skew-normal} & LE & 0.010 & 1.004 & -0.662
      & 0.018 & 0.685 & 0.010 & 0.999 & -0.655 & 0.018 & 0.678 \\
      & & SALE & 0.008 & 1.010 & -0.681 & 0.015 & 0.691 &
      0.008 & 0.999 & -0.676 & 0.014 & 0.686 \\
      & & MSLE & 0.010 & 1.048 & -0.699 & 0.001 & 0.716 &
      0.010 & 1.021 & -0.680 & 0.001 & 0.696 \\
      \cmidrule(lr){1-13}
      \multirow{6}{*}{\textbf{Dependent}} & \multirow{2}{*}{MA(2)
      ($\theta_1=0.7$)} & SALE & -0.013 & 1.012 & -0.706 & -0.024 & 0.676 &
      -0.012 & 1.004 & -0.698 & -0.024 & 0.671 \\
      & & MSLE & -0.008 & 1.032 & -0.727 & -0.007 & 0.696 &
      -0.008 & 1.022 & -0.721 & -0.007 & 0.689 \\
      \cmidrule(lr){2-13}
      & \multirow{2}{*}{MA(2) ($\theta_1=-0.7$)} & SALE & 0.001 &
      1.022 & -0.682 & 0.008 & 0.680 & 0.001 & 1.015 & -0.677 &
      0.008 & 0.678 \\
      & & MSLE & -0.001 & 1.073 & -0.717 & 0.004 & 0.727 &
      -0.001 & 1.063 & -0.709 & 0.004 & 0.723 \\
      \cmidrule(lr){2-13}
      & \multirow{2}{*}{AR(1)} & SALE & -0.011 & 1.006 & -0.693 &
      -0.015 & 0.666 & -0.010 & 0.996 & -0.690 & -0.015 & 0.657 \\
      & & MSLE & -0.005 & 1.040 & -0.705 & -0.008 & 0.685 &
      -0.003 & 1.028 & -0.695 & -0.008 & 0.680 \\
      \midrule
      \multicolumn{3}{l}{\textbf{Asymptotic Value (Standard Normal)}}
      & 0.000 & 1.000 & -0.674 & 0.000 & 0.674 &
      0.000 & 1.000 & -0.674 & 0.000 & 0.674 \\
      \bottomrule
    \end{tabular}
  }
\end{table}

\subsection{Finite-Sample Performance: Superior Efficiency}
\label{sec:simulation-finite-sample}

\subsubsection{Efficiency in the Noise-Free Case}

The finite-sample efficiency of the MSLE estimator, using both
optimal and approximate weights, is compared against the
all-observation estimator. To evaluate the performance across
different sample sizes, four common time horizons are considered for
$T$: one day $(T=1/252)$, one week $(T=5/252)$, two weeks
$(T=10/252)$ and one month $(T=22/252)$. For each $T$, the sample
size is set to $n = 23400 \times 252T$, and 1000 paths are simulated.
The MSLE estimators are computed with scales $H_p = 1, 2, \dots,
\lfloor 0.5 n^{0.5}\rfloor$, and the optimal weights are given by
Equation~\eqref{eq:weight-optimization-solution} and
\eqref{eq:SALE-acov-disc}.

\begin{figure}[!ht]
  \centering
  \begin{subfigure}{0.24\textwidth}
    \includegraphics[width=\textwidth]{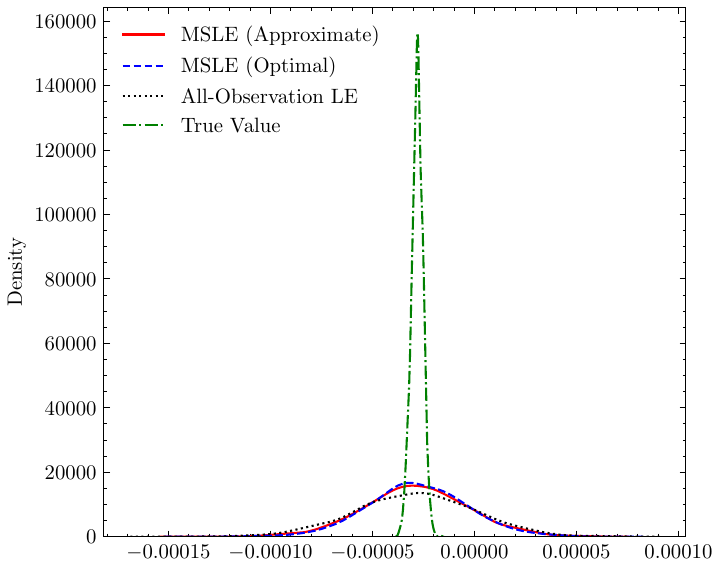}
    \caption{1 day}
  \end{subfigure}
  \hfill
  \begin{subfigure}{0.24\textwidth}
    \includegraphics[width=\textwidth]{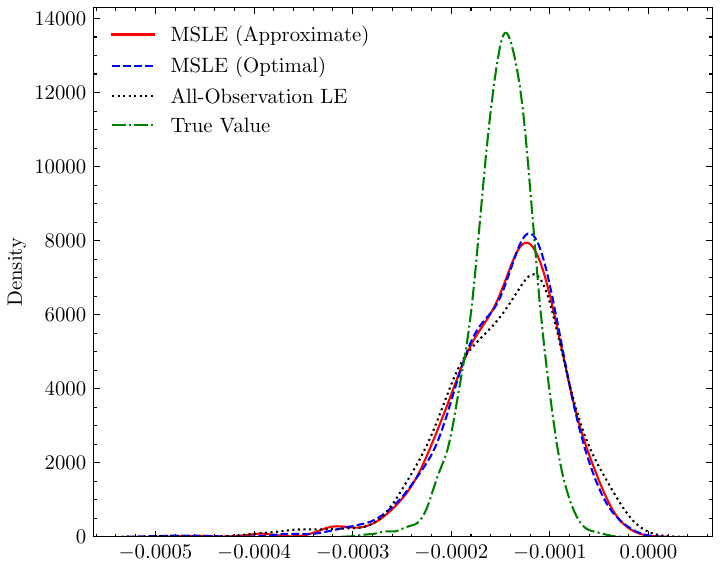}
    \caption{5 days}
  \end{subfigure}
  \hfill
  \begin{subfigure}{0.24\textwidth}
    \includegraphics[width=\textwidth]{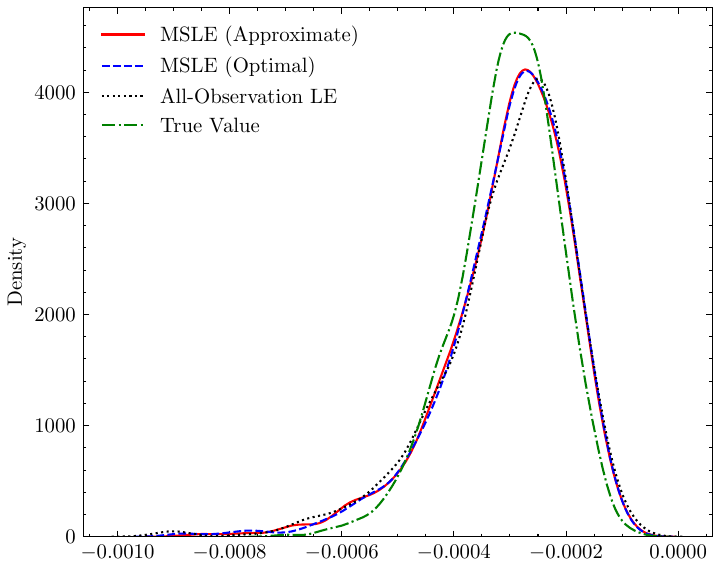}
    \caption{10 days}
  \end{subfigure}
  \hfill
  \begin{subfigure}{0.24\textwidth}
    \includegraphics[width=\textwidth]{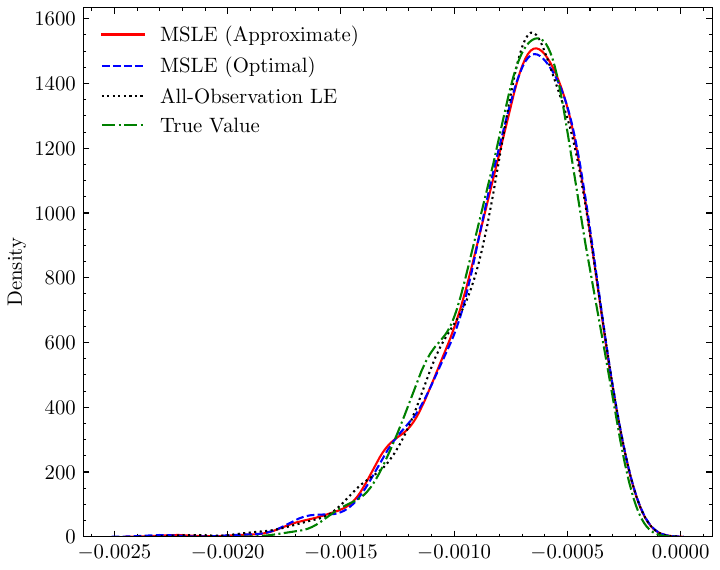}
    \caption{22 days}
  \end{subfigure}
  \begin{subfigure}{0.24\textwidth}
    \includegraphics[width=\textwidth]{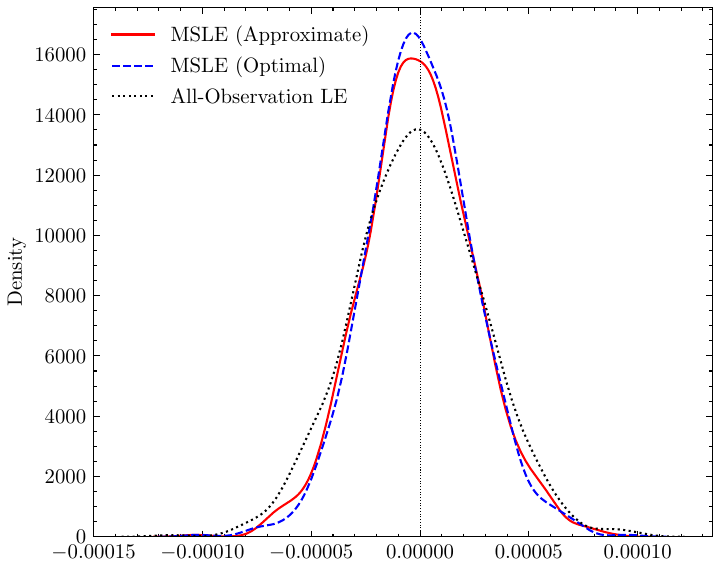}
    \caption{1 day}
  \end{subfigure}
  \hfill
  \begin{subfigure}{0.24\textwidth}
    \includegraphics[width=\textwidth]{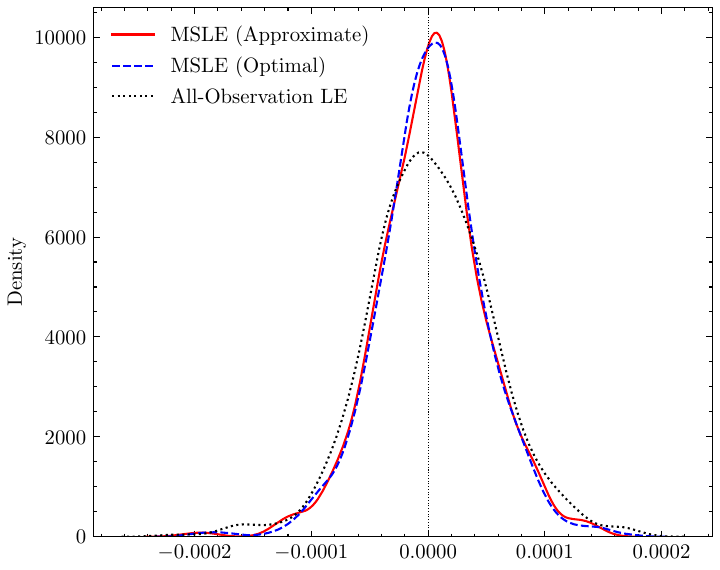}
    \caption{5 days}
  \end{subfigure}
  \hfill
  \begin{subfigure}{0.24\textwidth}
    \includegraphics[width=\textwidth]{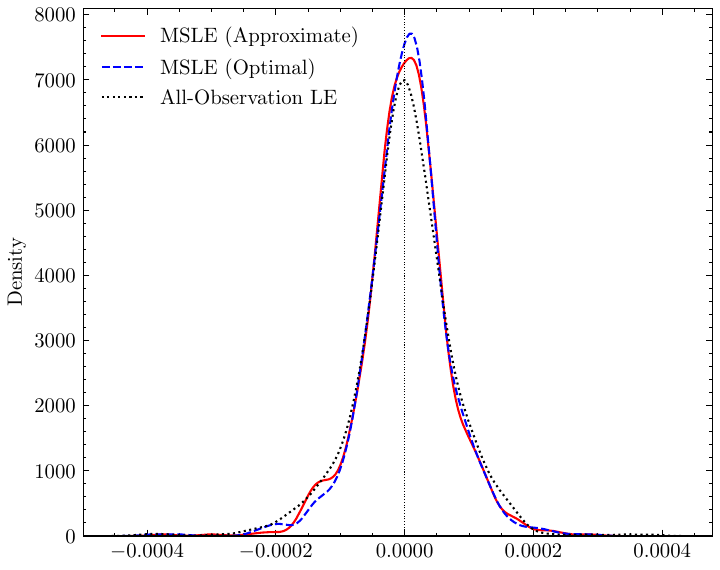}
    \caption{10 days}
  \end{subfigure}
  \hfill
  \begin{subfigure}{0.24\textwidth}
    \includegraphics[width=\textwidth]{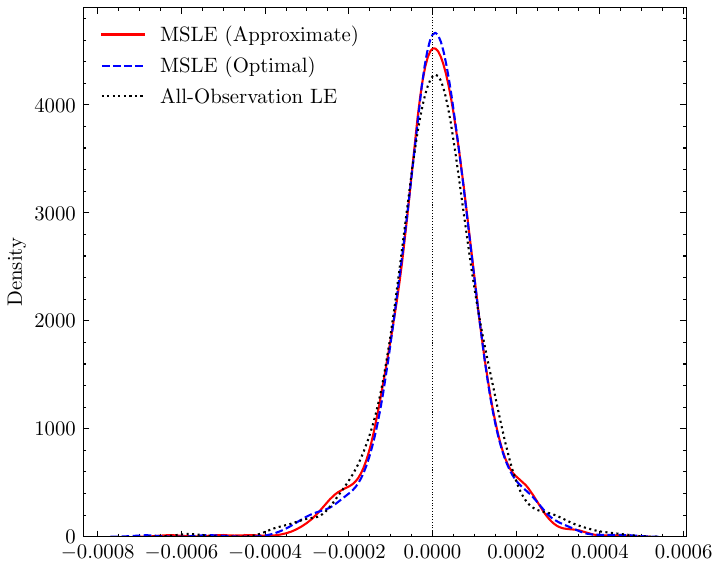}
    \caption{22 days}
  \end{subfigure}
  \begin{subfigure}{0.24\textwidth}
    \includegraphics[width=\textwidth]{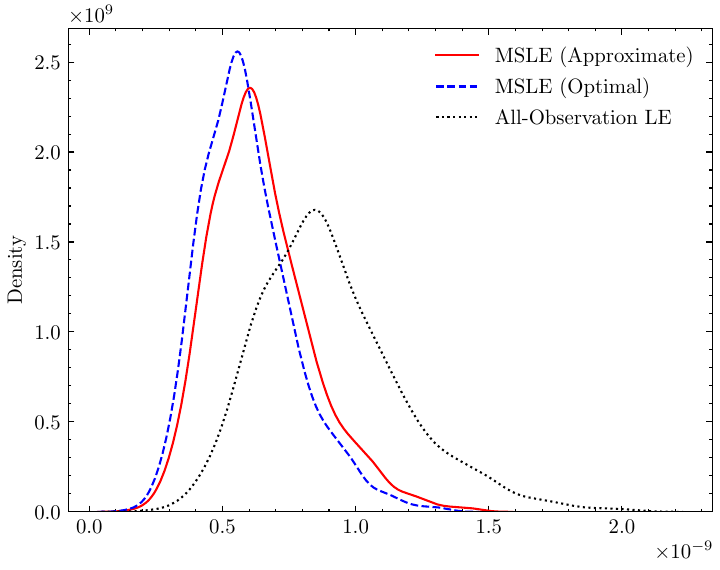}
    \caption{1 day}
  \end{subfigure}
  \hfill
  \begin{subfigure}{0.24\textwidth}
    \includegraphics[width=\textwidth]{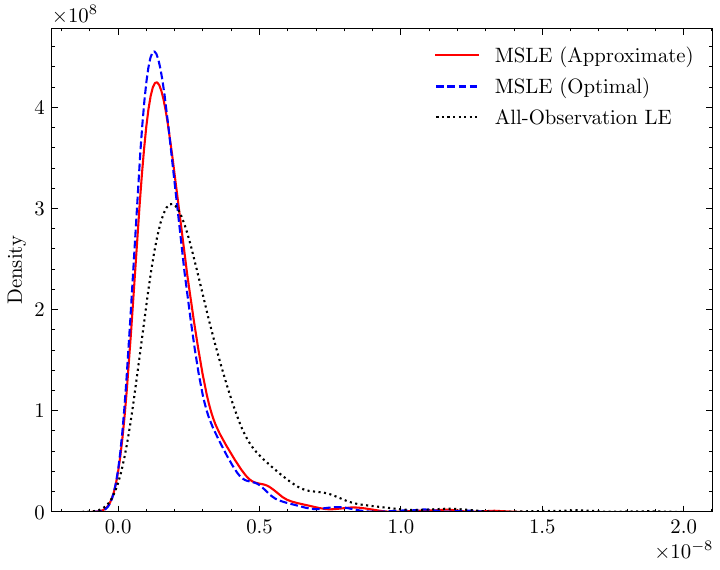}
    \caption{5 days}
  \end{subfigure}
  \hfill
  \begin{subfigure}{0.24\textwidth}
    \includegraphics[width=\textwidth]{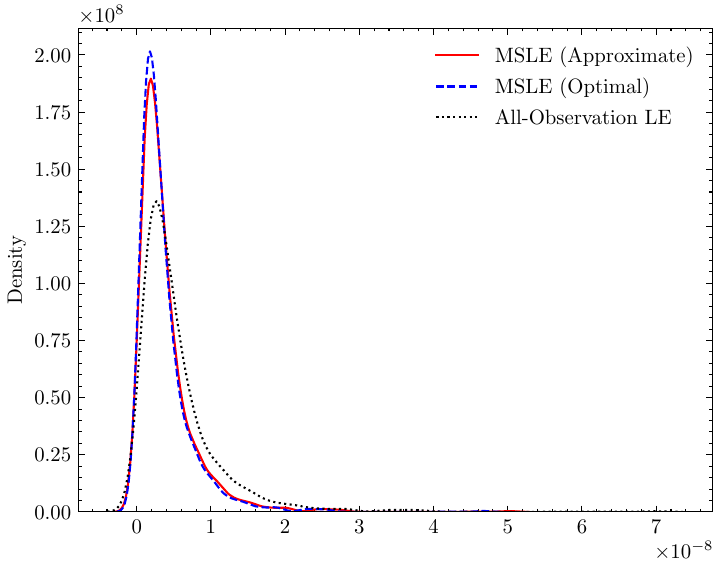}
    \caption{10 days}
  \end{subfigure}
  \hfill
  \begin{subfigure}{0.24\textwidth}
    \includegraphics[width=\textwidth]{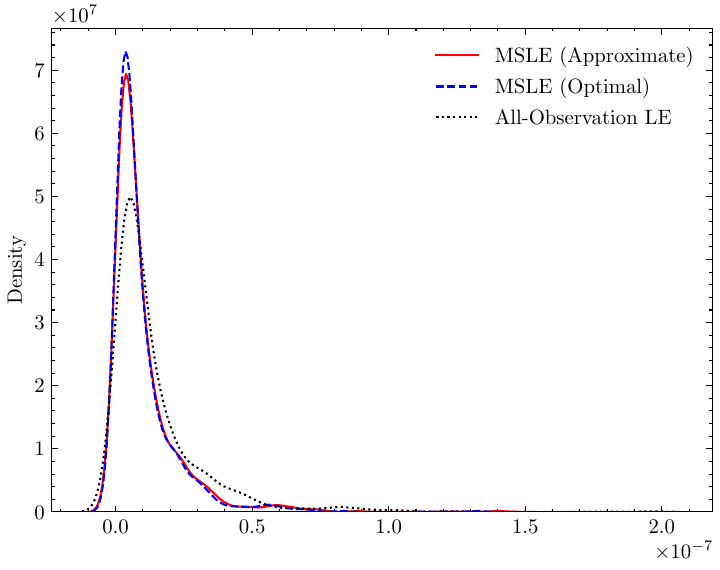}
    \caption{22 days}
  \end{subfigure}
  \caption{
    The performances of the MSLE and all-observation estimators for
    each setting of $T$ in the noise-free setting. The first row
    shows the true and estimated values of leverage effect, the
    second row shows the estimation errors, and the third row shows
    the asymptotic variances.
  }
  \label{fig:noisefree-estimate}
\end{figure}

\begin{table}[!ht]
  \centering
  \caption{
    Finite-sample performances of the MSLE and all-observation
    estimators in the noise-free setting. The finite-sample relative
    efficiency is compared with the all-observation estimator.
  }
  \label{tab:noisefree-rmse}
  \resizebox{0.95\textwidth}{!}{
    \begin{tabular}{rrrrrrrrrrrr}
      \toprule
      \multicolumn{1}{c}{\multirow{2.5}{*}{\textbf{Days}}} &
      \multicolumn{2}{c}{\textbf{True Value}} &
      \multicolumn{3}{c}{\textbf{RMSE}} &
      \multicolumn{3}{c}{\textbf{Relative Efficiency}} &
      \multicolumn{3}{c}{$\bm{\|w\|_1}$} \\
      \cmidrule(lr){2-3} \cmidrule(lr){4-6} \cmidrule(lr){7-9}
      \cmidrule{10-12}
      & \multicolumn{1}{c}{Mean} & \multicolumn{1}{c}{Std} &
      \multicolumn{1}{c}{Approximate} & \multicolumn{1}{c}{Optimal} &
      \multicolumn{1}{c}{LE} & \multicolumn{1}{c}{Approximate} &
      \multicolumn{1}{c}{Optimal} & \multicolumn{1}{c}{LE} &
      \multicolumn{1}{c}{Approximate} & \multicolumn{1}{c}{Optimal} &
      \multicolumn{1}{c}{LE} \\
      \midrule
      1 & \num{-2.80e-05} & \num{2.77e-06} & \num{2.55e-05} &
      \num{2.44e-05} & \num{2.93e-05} & 1.32 & 1.44 & 1.00 & 1.00 &
      13.77 & 1.00 \\
      5 & \num{-1.48e-04} & \num{3.10e-05} & \num{4.57e-05} &
      \num{4.50e-05} & \num{5.32e-05} & 1.36 & 1.40 & 1.00 & 1.00 &
      24.77 & 1.00 \\
      10 & \num{-3.06e-04} & \num{8.98e-05} & \num{6.43e-05} &
      \num{6.44e-05} & \num{7.28e-05} & 1.28 & 1.28 & 1.00 & 1.00 &
      28.53 & 1.00 \\
      22 & \num{-7.44e-04} & \num{2.72e-04} & \num{1.06e-04} &
      \num{1.07e-04} & \num{1.14e-04} & 1.16 & 1.13 & 1.00 & 1.00 &
      34.83 & 1.00 \\
      \bottomrule
    \end{tabular}
  }
\end{table}

Figure~\ref{fig:noisefree-estimate} and
Table~\ref{tab:noisefree-rmse} summarize the results. The findings
clearly demonstrate the superiority of the MSLE estimator over the
all-observation estimator, evidenced by its smaller asymptotic
variance, lower finite-sample RMSE, and higher finite-sample efficiency.

A key practical insight is that the MSLE with approximate weights
achieves efficiency nearly identical to that with optimal weights,
but offers significantly better numerical stability. By construction,
the approximate weights are non-negative, thus ensuring $\|\bm{w}\|_1
= 1$. In contrast, the optimal weights can become negative, causing
their $L^1$-norm to grow with the sample size. This not only
increases numerical instability, but also potentially violates our
theoretical requirement in Assumption~\ref{ass:MSLE-para}, making the
approximate weighting scheme a more robust choice for practical
implementation.

\subsubsection{Efficiency in the Noisy Case}

The finite-sample efficiency of the MSLE estimator using approximate
weights is compared against the pre-averaging LE estimator
in~\citet{aitsahalia2017EstimationContinuousDiscontinuous} in a noisy
setting. While the pre-averaging estimator has a slightly faster
theoretical convergence rate ($n^{-1/8}$) than the MSLE estimator
($n^{-1/9}$), we demonstrate that MSLE achieves superior
finite-sample efficiency, especially in realistic scenarios. To this
end, we again vary $T$ to assess the performance across different
sample sizes.

The simulation uses dependent AR(1) noise with $\phi = 0.7$, and
three noise levels: small ($\varsigma=10^{-4}$), medium
($\varsigma=10^{-3.5}$), and large ($\varsigma=10^{-3}$). Notably,
empirical evidence suggests that real-world noise levels are closer
to the ``small'' setting~\citep{christensen2014FactFrictionJumps},
which is further supported by our empirical study in
Section~\ref{sec:empirical}. For the MSLE estimator, the noise ACF is
truncated at $q=3$, the scales are set to $H_p=7, 8, \dots, \lfloor
n^{5/9} \rfloor$, and for simplicity, the same weight allocation is
used for all paths in each case. For the pre-averaging estimator,
since a closed-form optimal tuning parameter is unavailable, we grant
it an advantage by \emph{ex-post} selecting the pre-averaging window
that yields the minimum RMSE, from a wide grid of candidates ($5, 10,
30, 60, 90, 120, 180, 240,$ and $300$).

\begin{figure}[!ht]
  \centering
  \begin{subfigure}{0.24\textwidth}
    \includegraphics[width=\textwidth]{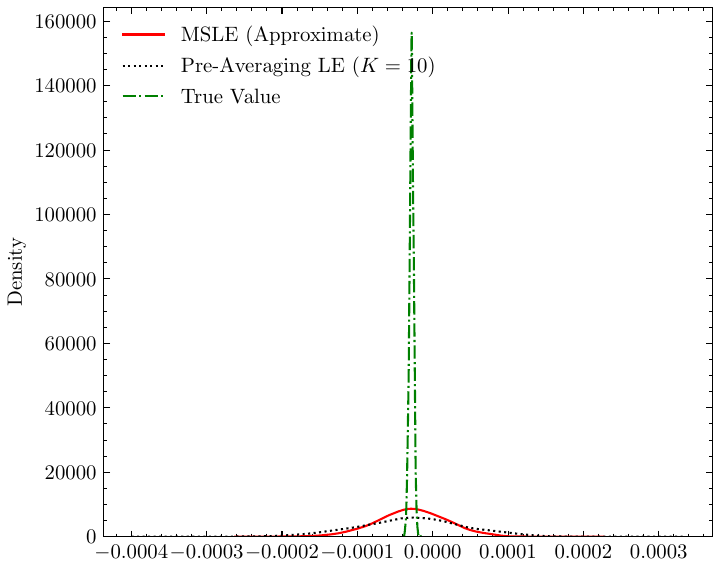}
    \caption{1 day}
  \end{subfigure}
  \hfill
  \begin{subfigure}{0.24\textwidth}
    \includegraphics[width=\textwidth]{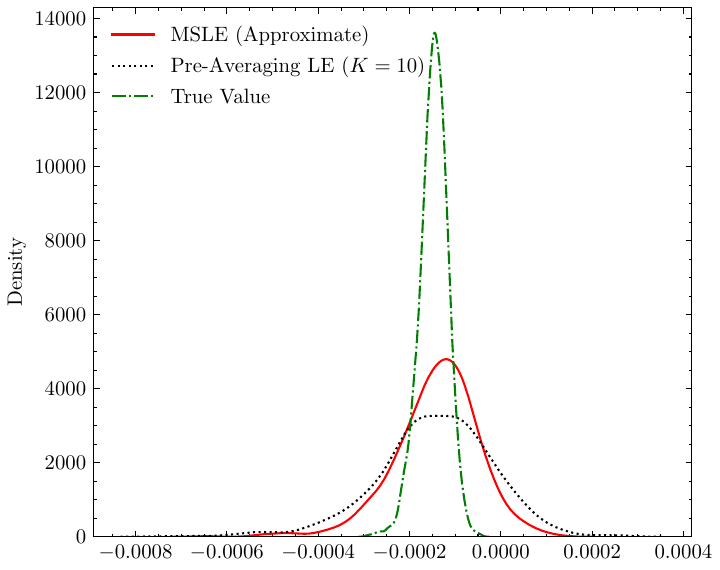}
    \caption{5 days}
  \end{subfigure}
  \hfill
  \begin{subfigure}{0.24\textwidth}
    \includegraphics[width=\textwidth]{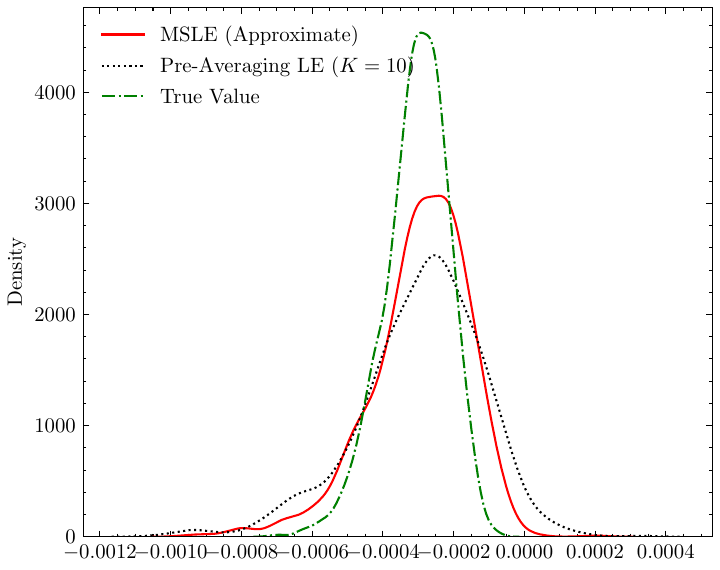}
    \caption{10 days}
  \end{subfigure}
  \hfill
  \begin{subfigure}{0.24\textwidth}
    \includegraphics[width=\textwidth]{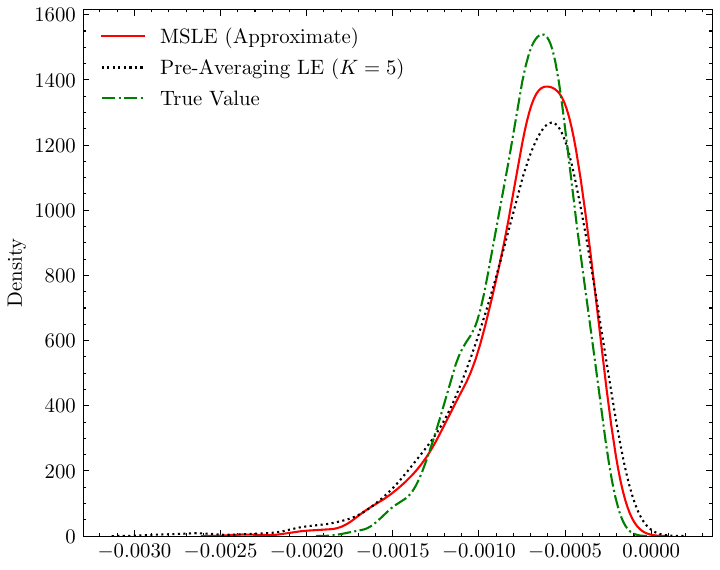}
    \caption{22 days}
  \end{subfigure}
  \begin{subfigure}{0.24\textwidth}
    \includegraphics[width=\textwidth]{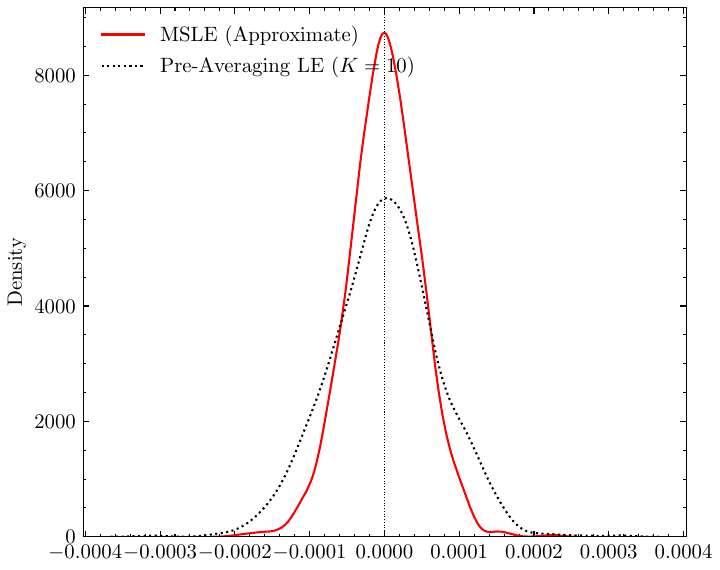}
    \caption{1 day}
  \end{subfigure}
  \hfill
  \begin{subfigure}{0.24\textwidth}
    \includegraphics[width=\textwidth]{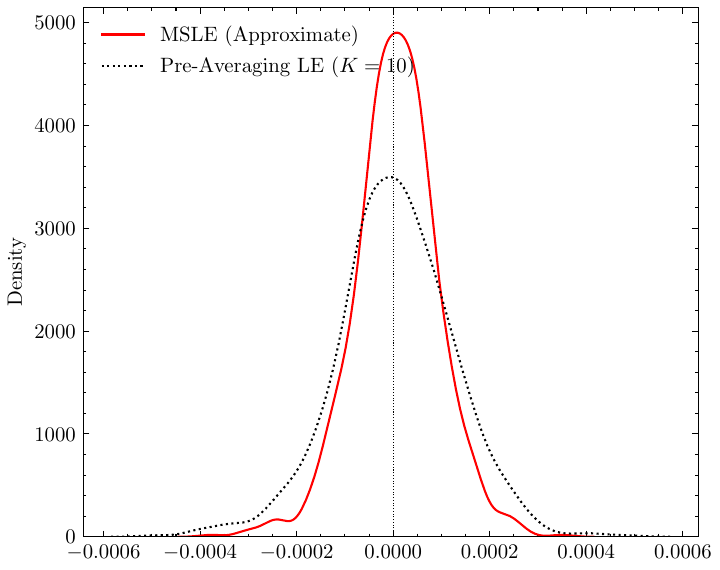}
    \caption{5 days}
  \end{subfigure}
  \hfill
  \begin{subfigure}{0.24\textwidth}
    \includegraphics[width=\textwidth]{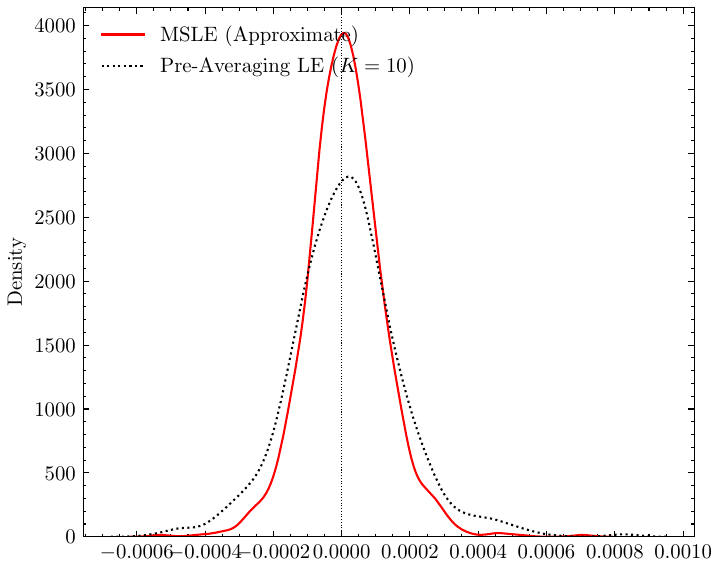}
    \caption{10 days}
  \end{subfigure}
  \hfill
  \begin{subfigure}{0.24\textwidth}
    \includegraphics[width=\textwidth]{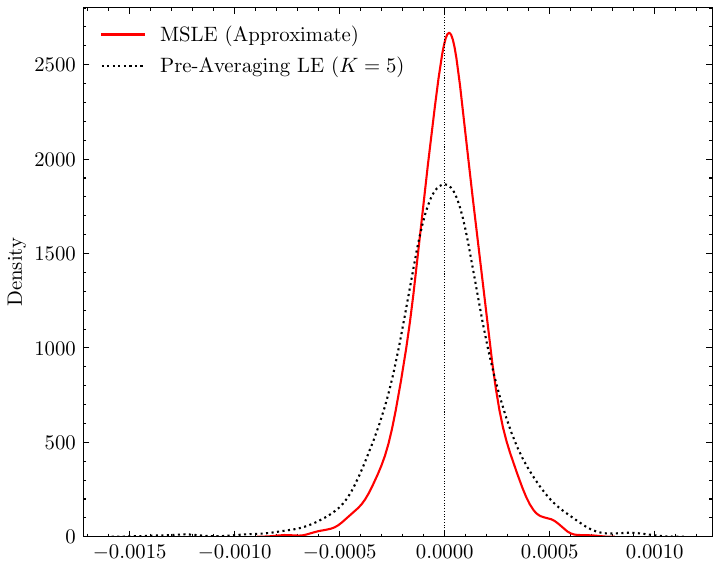}
    \caption{22 days}
  \end{subfigure}
  \caption{The performances of the MSLE and pre-averaging LE
    estimators for each setting of $T$ in the dependent noise setting
    ($\varsigma=10^{-4}$). The first row shows the true and estimated
  values of leverage effect, and the second row shows the estimation error.}
  \label{fig:dependent-estimate}
\end{figure}

\begin{table}[!ht]
  \centering
  \caption{
    Finite-sample performances of the MSLE and pre-averaging LE
    estimators in the dependent noise setting. The finite-sample
    relative efficiency is compared with the pre-averaging estimator.
  }
  \label{tab:dependent-rmse}
  \resizebox{0.95\textwidth}{!}{
    \begin{tabular}{lrrrrrrr}
      \toprule
      \multicolumn{1}{c}{\multirow{2.5}{*}{$\bm{\varsigma}$}} &
      \multicolumn{1}{c}{\multirow{2.5}{*}{\textbf{Days}}} &
      \multicolumn{2}{c}{\textbf{True Value}} &
      \multicolumn{2}{c}{\textbf{RMSE}} &
      \multicolumn{2}{c}{\textbf{Relative Efficiency}} \\
      \cmidrule(lr){3-4} \cmidrule(lr){5-6} \cmidrule(lr){7-8}
      & & \multicolumn{1}{c}{Mean} & \multicolumn{1}{c}{Std} &
      \multicolumn{1}{c}{MSLE} & \multicolumn{1}{c}{Pre-Averaging LE}
      & \multicolumn{1}{c}{MSLE} & \multicolumn{1}{c}{Pre-Averaging LE} \\
      \midrule
      \multirow{4}{*}{$10^{-4}$} & 1 & \num{-2.80e-05} &
      \num{2.77e-06} & \num{4.78e-05} & \num{7.11e-05} & 2.21 & 1.00 \\
      & 5 & \num{-1.48e-04} & \num{3.10e-05} & \num{8.67e-05} &
      \num{1.21e-04} & 1.95 & 1.00 \\
      & 10 & \num{-3.06e-04} & \num{8.98e-05} & \num{1.14e-04} &
      \num{1.62e-04} & 2.00 & 1.00 \\
      & 22 & \num{-7.44e-04} & \num{2.72e-04} & \num{1.74e-04} &
      \num{2.50e-04} & 2.05 & 1.00 \\
      \cmidrule(lr){1-8}
      \multirow{4}{*}{$10^{-3.5}$} & 1 & \num{-2.80e-05} &
      \num{2.77e-06} & \num{6.65e-05} & \num{9.28e-05} & 1.95 & 1.00 \\
      & 5 & \num{-1.48e-04} & \num{3.10e-05} & \num{1.18e-04} &
      \num{1.62e-04} & 1.86 & 1.00 \\
      & 10 & \num{-3.06e-04} & \num{8.98e-05} & \num{1.63e-04} &
      \num{2.20e-04} & 1.81 & 1.00 \\
      & 22 & \num{-7.44e-04} & \num{2.72e-04} & \num{2.67e-04} &
      \num{3.65e-04} & 1.87 & 1.00 \\
      \cmidrule(lr){1-8}
      \multirow{4}{*}{$10^{-3}$} & 1 & \num{-2.80e-05} &
      \num{2.77e-06} & \num{9.63e-05} & \num{1.18e-04} & 1.50 & 1.00 \\
      & 5 & \num{-1.48e-04} & \num{3.10e-05} & \num{1.82e-04} &
      \num{2.27e-04} & 1.56 & 1.00 \\
      & 10 & \num{-3.06e-04} & \num{8.98e-05} & \num{2.53e-04} &
      \num{3.07e-04} & 1.48 & 1.00 \\
      & 22 & \num{-7.44e-04} & \num{2.72e-04} & \num{3.82e-04} &
      \num{4.86e-04} & 1.62 & 1.00 \\
      \bottomrule
    \end{tabular}
  }
\end{table}

Figure~\ref{fig:dependent-estimate} (for $\varsigma=10^{-4}$) and
Table~\ref{tab:dependent-rmse} (for all noise levels) present the
results. The findings confirm that the MSLE estimator consistently
and substantially outperforms the pre-averaging estimator in terms of
finite-sample RMSE and efficiency across all sample sizes and noise
levels. Crucially, the advantage is most pronounced in the
small-noise setting, which is the most empirically relevant scenario.
Even as the time horizon increases to one month, MSLE's lead remains
significant, demonstrating that the theoretical convergence rate is
not the only determinant of the finite-sample performance. This
highlights the practical power of the proposed estimators and the
approximate weighting strategy. Furthermore, this superior
performance is achieved without resorting to the infeasible
\emph{ex-post} parameter tuning that was granted to the pre-averaging
estimator, underscoring the robustness and practical utility of our methods.

An additional study for the i.i.d. noise case, along with
supplementary information of the simulation details, are provided in
Supplementary Material.

\section{Empirical Study}\label{sec:empirical}

The high-frequency trading data for a selection of assets, covering
the regular trading hours from 2014 to 2023 (2,516 trading days), are
collected from the TAQ database. The data are cleaned before analysis,
retaining only regular trades and removing erroneous
entries.\footnote{A practical and detailed guideline on
  high-frequency data cleaning is offered by
  \cite{barndorffnielsen2009RealizedKernelsPractice}. While we follow
  most of the steps therein, some are omitted. For example, the entries
  with \emph{Sale Conditions} `I' (odd lot trade) and `C' (cash trade)
  are retained because of their significant contribution in our
  dataset. We also remove the ``bounceback'' outliers described by
\cite{aitsahalia2011UltraHighFrequency}.} The dataset consists of 15
ETFs and 15 individual stocks, as listed in Table~\ref{tab:emp}. The
ETFs track the performance of the S\&P 500, NASDAQ 100, Dow Jones
Industrial Average, Russell 2000 indices, as well as the 11 sectors
of the S\&P 500. The individual stocks are selected to represent a
range of liquidity and volatility levels, covering various sectors
such as technology, consumer goods, healthcare, and entertainment,
thus providing a diverse set of assets for the empirical study. Among
the 30 assets, XLC and XLRE were issued partway through the sample
period. Therefore, our analysis for them begins at the start of their
second year, in 2017 and 2019, respectively. After data cleaning, we
resample the data to obtain 1-second and 5-second returns.

We apply the jump test proposed by
\cite{aitsahalia2012TestingJumpsNoisy} to identify and remove trading
days with the presence of jumps for each stock. This test is a
robustified version of the test introduced by
\cite{aitsahalia2009TestingJumpsDiscretely}, incorporating the
pre-averaging method to deal with the MMS noise. After computing the
standardized statistics with 5-second intraday data, we apply the
universal threshold technique proposed by
\cite{bajgrowicz2016JumpsHighFrequencyData} to eliminate spurious
jump detections. This method is more stringent than the FDR procedure
and is designed to asymptotically remove all spurious detections,
thereby minimizing data loss. As a result, 909 asset-days, comprising
1.2\% of the entire dataset of 73,910 asset-days, are identified as
containing jumps and excluded from further analysis. The numbers of
days with jumps for each asset are listed in Table~\ref{tab:emp}.

We estimate the leverage effects for both weekly (defined as every
five trading days) and monthly (defined as natural months) periods
using both 1-second and 5-second data for each stock. The estimation
proceeds in several steps, showcasing the flexibility and robustness
of our framework in handling real-world data complexities.
\begin{enumerate}
  \item The ReMeDI estimator proposed by
    \cite{li2022ReMeDIMicrostructureNoise} is used to estimate the
    moments $\nu_2$, $\nu_4$ and the generalized acfs $\rho_2(l)$,
    $\rho_3(l)$, and $\rho_4(l)$ of the MMS noise in each period. The
    existence and the dependence level of noise are determined by its
    autocovariances. The results show that the noise in the dataset
    is small, while the dependencies are common. For example, the
    analysis of noise in monthly 5-second frequency data shows that:
    (i) for the ETFs, 33.9\% of asset-months exhibit significant
    noise, among which 71.5\% are dependent and the average noise
    scale is $\varsigma=9.6\times 10^{-5}$; while (ii) for the
    stocks, 63.5\% of the asset-months exhibit significant noise,
    among which 60.8\% are dependent and the average noise scale is
    $\varsigma=2.0\times 10^{-4}$.
  \item The scale $\overline{H}_n$ in the approximate weights of the
    MSLE estimator is determined by minimizing the total asymptotic
    variances of SALE estimators, where the asymptotic variances are
    estimated using 1-minute pre-averaging return data. With the
    existence of noise, additional lower bounds for $\overline{H}_n$
    are applied, such that: (i) $\overline{H}_n$ satisfies
    $\overline{H}_n \geq 2\hat q + 1$, where $\hat q$ is the
    estimated dependence level; and (ii) the minimum values of
    $\overline{H}_n$ are 20 for 1-second data (corresponding to 20
    seconds) and 12 for 5-second data (corresponding to 60 seconds).
    The former is a condition for the proposed SALE and MSLE
    estimators, while the latter is a rather conservative manual
    intervention, which mitigates numerical instability at the cost
    of larger asymptotic variance.
  \item The leverage effect and its asymptotic variance are estimated
    using the MSLE estimator with approximate weights, where the
    number of scales is set to $M_n=50$ for computational
    efficiency.
\end{enumerate}

\begin{figure}[!ht]
  \centering
  \begin{subfigure}{\textwidth}
    \includegraphics[width=\textwidth]{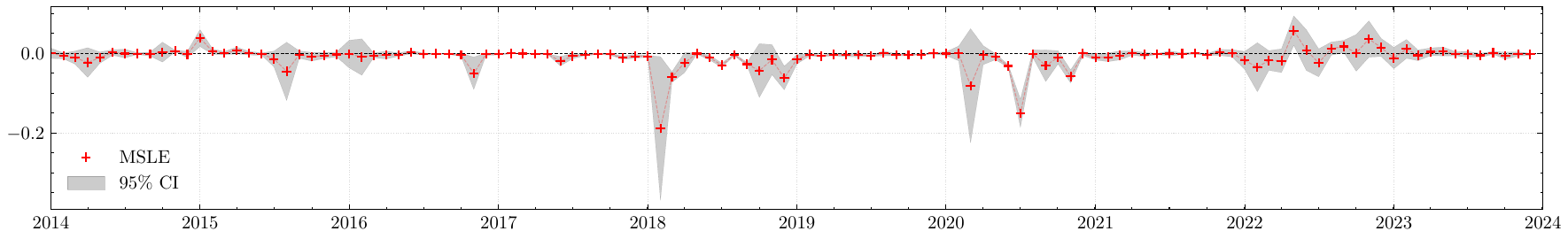}
    \caption{Month data, 5-second frequency}
  \end{subfigure}
  \begin{subfigure}{\textwidth}
    \includegraphics[width=\textwidth]{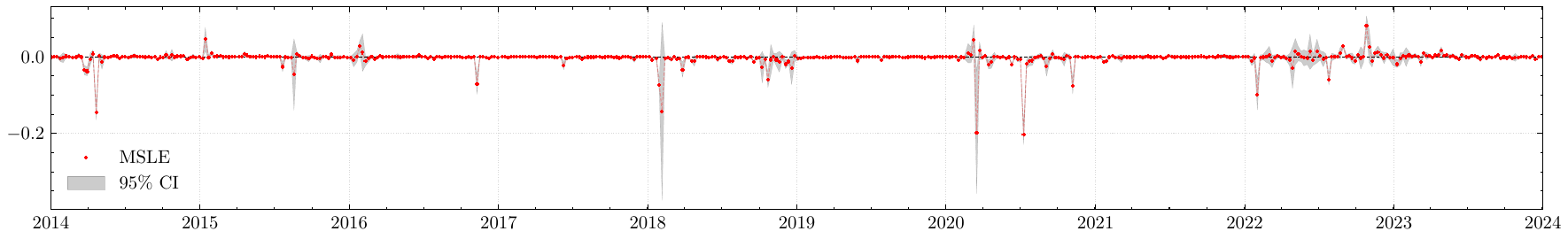}
    \caption{Week data, 5-second frequency}
  \end{subfigure}
  \caption{Leverage effect estimation for AMZN (Amazon.com, Inc.).}
  \label{fig:estimation-amzn}
\end{figure}

Figure~\ref{fig:estimation-amzn} illustrates the dynamic nature of
the leverage effect for AMZN, showcasing our estimator's ability to
capture its time-varying behavior. The monthly estimates reveal
significant fluctuations, clearly capturing major market stress
events such as the February 2018 ``Volpocalypse'' and the COVID-19
sell-off in early 2020. The weekly estimates, while more volatile,
provide a higher-resolution view of these dynamics.

\begin{table}[!ht]
  \centering
  \caption{Data descriptions and results of empirical study.}
  \label{tab:emp}
  \resizebox{1.00\textwidth}{!}{
    \begin{tabular}{lllrrrrrrrrrrrrrr}
      \toprule
      \multicolumn{1}{c}{\multirow{3.5}{*}{\textbf{Type}}} &
      \multicolumn{1}{c}{\multirow{3.5}{*}{\textbf{Ticker}}} &
      \multicolumn{1}{c}{\multirow{3.5}{*}{\textbf{Name}}} &
      \multicolumn{1}{c}{\multirow{3.5}{*}{\textbf{\shortstack{Average
      \\ Daily \\ Observations}}}} &
      \multicolumn{1}{c}{\multirow{3.5}{*}{\textbf{\shortstack{Average
      \\ Daily \\ Volume}}}} &
      \multicolumn{1}{c}{\multirow{3.5}{*}{\textbf{\shortstack{Annualized
      \\ Volatility}}}} &
      \multicolumn{1}{c}{\multirow{3.5}{*}{\textbf{\shortstack{Trading
      \\ Days}}}} &
      \multicolumn{1}{c}{\multirow{3.5}{*}{\textbf{\shortstack{Days
      \\ with \\ Jumps}}}} & \multicolumn{8}{c}{\textbf{Signs of
      Leverage Effects (\%)}} \\
      \cmidrule(lr){9-16}
      & & & & & & & & \multicolumn{2}{c}{\textbf{M, 1-sec}} &
      \multicolumn{2}{c}{\textbf{M, 5-sec}} &
      \multicolumn{2}{c}{\textbf{W, 1-sec}} &
      \multicolumn{2}{c}{\textbf{W, 5-sec}} \\
      \cmidrule(lr){9-10} \cmidrule(lr){11-12} \cmidrule(lr){13-14}
      \cmidrule(lr){15-16}
      & & & & & & & & \multicolumn{1}{c}{$\bm{-}$} &
      \multicolumn{1}{c}{$\bm{+}$} & \multicolumn{1}{c}{$\bm{-}$} &
      \multicolumn{1}{c}{$\bm{+}$} & \multicolumn{1}{c}{$\bm{-}$} &
      \multicolumn{1}{c}{$\bm{+}$} & \multicolumn{1}{c}{$\bm{-}$} &
      \multicolumn{1}{c}{$\bm{+}$} \\
      \midrule
      \multirow{15}{*}{\textbf{ETF}} & SPY & SPDR S\&P 500 ETF Trust
      & 428321 & \num{9.27e+07} & 0.175 & 2516 & 5 & 89.2 & 10.8 &
      93.3 & 6.7 & 83.3 & 16.7 & 86.7 & 13.3 \\
      & QQQ & Invesco QQQ Trust & 219580 & \num{4.14e+07} & 0.215 &
      2516 & 3 & 88.3 & 11.7 & 90.8 & 9.2 & 84.9 & 15.1 & 86.7 & 13.3 \\
      & DIA & SPDR Dow Jones Industrial Average ETF Trust & 35501 &
      \num{4.53e+06} & 0.174 & 2516 & 11 & 86.7 & 13.3 & 90.0 & 10.0
      & 77.4 & 22.6 & 82.1 & 17.9 \\
      & IWM & iShares Russell 2000 ETF & 153098 & \num{2.98e+07} &
      0.221 & 2516 & 5 & 86.7 & 13.3 & 89.2 & 10.8 & 77.0 & 23.0 &
      80.4 & 19.6 \\
      & XLC & The Communication Services Select Sector SPDR ETF Fund
      & 27626 & \num{4.65e+06} & 0.242 & 1393 & 6 & 88.3 & 11.7 &
      86.7 & 13.3 & 75.1 & 24.9 & 76.3 & 23.7 \\
      & XLY & The Consumer Discretionary Select Sector SPDR Fund &
      49127 & \num{5.63e+06} & 0.208 & 2516 & 12 & 88.3 & 11.7 & 90.0
      & 10.0 & 74.2 & 25.8 & 79.8 & 20.2 \\
      & XLP & The Consumer Staples Select Sector SPDR Fund & 45530 &
      \num{1.19e+07} & 0.146 & 2516 & 21 & 70.8 & 29.2 & 75.8 & 24.2
      & 62.5 & 37.5 & 61.9 & 38.1 \\
      & XLE & The Energy Select Sector SPDR Fund & 112759 &
      \num{2.10e+07} & 0.298 & 2516 & 24 & 71.7 & 28.3 & 75.8 & 24.2
      & 62.5 & 37.5 & 64.9 & 35.1 \\
      & XLF & The Financial Select Sector SPDR Fund & 72205 &
      \num{5.54e+07} & 0.221 & 2516 & 15 & 81.7 & 18.3 & 87.5 & 12.5
      & 71.2 & 28.8 & 71.6 & 28.4 \\
      & XLV & The Health Care Select Sector SPDR Fund & 63047 &
      \num{9.89e+06} & 0.170 & 2516 & 19 & 82.5 & 17.5 & 88.3 & 11.7
      & 69.0 & 31.0 & 68.8 & 31.2 \\
      & XLI & The Industrial Select Sector SPDR Fund & 65217 &
      \num{1.15e+07} & 0.196 & 2516 & 13 & 80.8 & 19.2 & 90.8 & 9.2 &
      71.4 & 28.6 & 73.8 & 26.2 \\
      & XLB & The Materials Select Sector SPDR Fund & 38370 &
      \num{6.33e+06} & 0.206 & 2516 & 20 & 73.3 & 26.7 & 77.5 & 22.5
      & 69.0 & 31.0 & 68.8 & 31.2 \\
      & XLRE & The Real Estate Select Sector SPDR Fund & 15748 &
      \num{4.21e+06} & 0.214 & 2069 & 30 & 73.8 & 26.2 & 75.0 & 25.0
      & 59.8 & 40.2 & 63.5 & 36.5 \\
      & XLK & The Technology Select Sector SPDR Fund & 60643 &
      \num{1.04e+07} & 0.226 & 2516 & 7 & 89.2 & 10.8 & 91.7 & 8.3 &
      83.3 & 16.7 & 83.7 & 16.3 \\
      & XLU & The Utilities Select Sector SPDR Fund & 66683 &
      \num{1.48e+07} & 0.191 & 2516 & 48 & 64.2 & 35.8 & 65.0 & 35.0
      & 55.4 & 44.6 & 57.1 & 42.9 \\
      \cmidrule(lr){1-16}
      \multirow{15}{*}{\textbf{Stock}} & AAPL & Apple Inc. & 368263 &
      \num{1.37e+08} & 0.284 & 2516 & 22 & 76.7 & 23.3 & 75.8 & 24.2
      & 70.2 & 29.8 & 73.6 & 26.4 \\
      & AMC & AMC Entertainment Holdings, Inc. & 104278 &
      \num{2.77e+06} & 1.352 & 2516 & 56 & 50.8 & 49.2 & 48.3 & 51.7
      & 47.8 & 52.2 & 49.6 & 50.4 \\
      & AMZN & Amazon.com, Inc. & 164947 & \num{8.02e+07} & 0.332 &
      2516 & 21 & 78.3 & 21.7 & 76.7 & 23.3 & 68.7 & 31.3 & 68.8 & 31.2 \\
      & CLX & The Clorox Company & 16352 & \num{1.18e+06} & 0.227 &
      2516 & 70 & 52.5 & 47.5 & 55.8 & 44.2 & 49.6 & 50.4 & 51.2 & 48.8 \\
      & CPB & The Campbell's Company & 18902 & \num{2.32e+06} & 0.236
      & 2516 & 66 & 51.7 & 48.3 & 49.2 & 50.8 & 54.8 & 45.2 & 53.6 & 46.4 \\
      & GME & GameStop Corp. & 51874 & \num{1.82e+07} & 1.080 & 2516
      & 52 & 43.3 & 56.7 & 49.2 & 50.8 & 47.2 & 52.8 & 48.2 & 51.8 \\
      & KO & The Coca-Cola Company & 78632 & \num{1.42e+07} & 0.180 &
      2516 & 44 & 65.8 & 34.2 & 67.5 & 32.5 & 55.4 & 44.6 & 59.1 & 40.9 \\
      & MRK & Merck \& Co., Inc. & 68200 & \num{1.08e+07} & 0.214 &
      2516 & 52 & 64.2 & 35.8 & 56.7 & 43.3 & 56.0 & 44.0 & 50.8 & 49.2 \\
      & MSFT & Microsoft Corporation & 231381 & \num{3.02e+07} &
      0.271 & 2516 & 16 & 72.5 & 27.5 & 73.3 & 26.7 & 65.5 & 34.5 &
      66.3 & 33.7 \\
      & NVDA & NVIDIA Corporation & 204440 & \num{4.58e+08} & 0.464 &
      2516 & 30 & 71.7 & 28.3 & 76.7 & 23.3 & 68.1 & 31.9 & 68.7 & 31.3 \\
      & PEP & PepsiCo, Inc. & 44764 & \num{4.69e+06} & 0.184 & 2516 &
      44 & 66.7 & 33.3 & 68.3 & 31.7 & 52.0 & 48.0 & 54.6 & 45.4 \\
      & PFE & Pfizer Inc. & 113288 & \num{2.79e+07} & 0.226 & 2516 &
      40 & 65.8 & 34.2 & 61.7 & 38.3 & 57.3 & 42.7 & 57.1 & 42.9 \\
      & PG & The Procter \& Gamble Company & 61423 & \num{8.27e+06} &
      0.182 & 2516 & 61 & 69.2 & 30.8 & 65.8 & 34.2 & 59.1 & 40.9 &
      61.5 & 38.5 \\
      & TAP & Molson Coors Beverage Company & 16751 & \num{1.83e+06}
      & 0.282 & 2516 & 67 & 60.0 & 40.0 & 59.2 & 40.8 & 52.0 & 48.0 &
      52.6 & 47.4 \\
      & TSLA & Tesla, Inc. & 368486 & \num{1.13e+08} & 0.557 & 2516 &
      29 & 75.8 & 24.2 & 74.2 & 25.8 & 63.7 & 36.3 & 63.7 & 36.3 \\
      \bottomrule
    \end{tabular}
  }
\end{table}

\begin{figure}[!ht]
  \centering
  \includegraphics[width=0.5\textwidth]{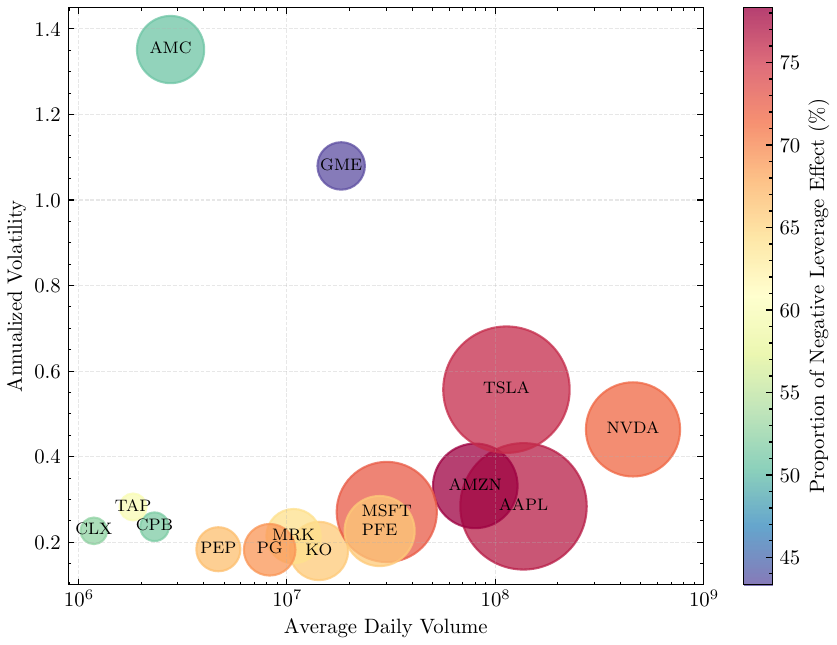}
  \caption{Individual stocks included in the empirical study. The
    sizes of the circles represent the average number of daily
    observations in our dataset, while the colors represent the
    proportion of negative leverage effect detected using monthly data
  sampled at 1-second frequency.}
  \label{fig:stocks}
\end{figure}

As summarized in Table~\ref{tab:emp} and visualized in
Figure~\ref{fig:stocks}, a negative leverage effect is predominant
across most assets, particularly within established large-cap and
defensive stocks, consistent with financial theory. The notable
exceptions are the ``meme stocks'' AMC and GME, where
retail-investor-driven speculative trading results in extreme
volatility and a weaker or non-negative leverage effect. This
demonstrates our method's ability to uncover such asset-specific idiosyncrasies.

\section{Conclusion}

We introduce a multi-scale framework for the robust and efficient
estimation of the leverage effect from high-frequency data
contaminated by complex, serially dependent microstructure noise. We
construct two estimators, SALE and MSLE, by combining the shifted
window, subsampling, and multi-scale techniques. We develop the
asymptotic theory, and design an effective weighting strategy for the
MSLE estimator. Beyond noise robustness, a central merit of our
framework is its superior efficiency. In the absence of noise, the
MSLE estimator improves the efficiency of the base estimator. Under a
realistic setting of noise and sample size, the SALE estimator
already outperforms the standard pre-averaging estimator, and the
MSLE estimator further improves upon this, delivering consistently
more accurate and reliable estimates. Extensive simulations and
empirical applications have validated the asymptotic theory,
finite-sample performance, and the practical robustness and
flexibility of the proposed methods.

\bibliographystyle{plainnat}
\bibliography{ref.bib}

\pagebreak

\begin{appendices}

  \renewcommand{\thesection}{S\arabic{section}}
  \renewcommand{\thetable}{S\arabic{table}}
  \renewcommand{\thefigure}{S\arabic{figure}}
  \renewcommand{\theequation}{S\arabic{equation}}
  \renewcommand{\thetheorem}{S\arabic{theorem}}
  \renewcommand{\theproposition}{S\arabic{proposition}}
  \renewcommand{\thelemma}{S\arabic{lemma}}
  \renewcommand{\thecorollary}{S\arabic{corollary}}
  \renewcommand{\theexample}{S\arabic{example}}
  \renewcommand{\thedefinition}{S\arabic{definition}}
  \renewcommand{\theassumption}{S\arabic{assumption}}
  \renewcommand{\theremark}{S\arabic{remark}}

  \setcounter{section}{0}
  \setcounter{table}{0}
  \setcounter{figure}{0}
  \setcounter{equation}{0}
  \setcounter{theorem}{0}
  \setcounter{proposition}{0}
  \setcounter{lemma}{0}
  \setcounter{corollary}{0}
  \setcounter{example}{0}
  \setcounter{definition}{0}
  \setcounter{assumption}{0}
  \setcounter{remark}{0}

  \section{Technical Proofs}

  \subsection{Preliminary Results}\label{app:preliminary}

  \begin{lemma}\label{lem:err-oracle-leverage}
    Suppose that Assumption~\ref{ass:process} holds, and let $s < t$.
    As $(t-s) \to 0$, we have
    \begin{equation}
      \E\left[\left.
        \left(
          \left(X_{t} - X_{s}\right)\left(\sigma_{t}^{2} - \sigma_{s}^{2}\right)
          - \int_{s}^{t} 2\sigma_r^{2} f_r \ud r
        \right)^{2}
        \right| \calF_{s}
      \right]
      =
      4\sigma_{s}^{4} \left(2f_{s}^{2} + g_{s}^{2}\right) (t-s)^{2}
      + O_p\left((t-s)^{5/2}\right).
    \end{equation}
  \end{lemma}

  \begin{proof}
    Note that $\ud \langle X,\sigma^2 \rangle_r = 2\sigma_r^2 f_r$, and that
    \begin{equation}
      \left(X_{t} - X_{s}\right)\left(\sigma_{t}^{2} - \sigma_{s}^{2}\right)
      =
      \int_s^t (\sigma_r^2 - \sigma_s^2) \ud X_r
      + \int_s^t (X_r - X_s) \ud \sigma_r^2
      + \int_s^t \ud \langle X, \sigma^2 \rangle_r.
    \end{equation}
    According the It\^{o} isometry, it follows that
    \begin{align}
      &
      \E\left[\left.
        \left(
          \left(X_{t} - X_{s}\right)\left(\sigma_{t}^{2} - \sigma_{s}^{2}\right)
          - \int_{s}^{t} 2\sigma_r^{2} f_r \ud r
        \right)^{2}
        \right| \calF_{s}
      \right]
      \notag
      \\
      &=
      \E\left[\left.
        \left(
          \int_s^t (\sigma_r^2 - \sigma_s^2) \ud X_r
          + \int_s^t (X_r - X_s) \ud \sigma_r^2
        \right)^{2}
        \right| \calF_{s}
      \right]
      \notag
      \\
      &=
      \E\Biggl[
        \int_s^t (\sigma_r^2 - \sigma_s^2)^2 \ud \langle X, X \rangle_r
        + \int_s^t (X_r - X_s)^2 \ud \langle \sigma^2, \sigma^2 \rangle_r
        \notag
        \\
        &\qquad \qquad \qquad
        + 2 \int_s^t (\sigma_r^2 - \sigma_s^2) (X_r - X_s) \ud
        \langle X, \sigma^2 \rangle_r
        \Bigg| \calF_{s}
      \Biggr].
    \end{align}
    Denote $\partial \langle A, B \rangle_t = \frac{\ud \langle A, B
    \rangle_r}{\ud r}|_{r=t}$ for processes $A$ and $B$. For the first
    term, we have
    \begin{align}
      &
      \E\Biggl[
        \int_s^t (\sigma_r^2 - \sigma_s^2)^2 \ud \langle X, X \rangle_r
        \Bigg| \calF_{s}
      \Biggr]
      \notag
      \\ &=
      \Bigl(\partial \langle X, X \rangle_s + O_p((t-s)^{1/2})\Bigr)
      \E\Biggl[
        \int_s^t (\sigma_r^2 - \sigma_s^2)^2 \ud r
        \Bigg| \calF_{s}
      \Biggr]
      \notag
      \\ &=
      \Bigl(\partial \langle X, X \rangle_s + O_p((t-s)^{1/2})\Bigr)
      \E\Biggl[
        \int_s^t (t - r) \ud (\sigma_r^2 - \sigma_s^2)^2
        \Bigg| \calF_{s}
      \Biggr]
      \notag
      \\ &=
      \Bigl(\partial \langle X, X \rangle_s + O_p((t-s)^{1/2})\Bigr)
      \Bigl(\partial \langle \sigma^2, \sigma^2 \rangle_s +
      O_p((t-s)^{1/2})\Bigr)
      \frac{1}{2} (t-s)^2
      \notag
      \\ &=
      \partial \langle X, X \rangle_s
      \partial \langle \sigma^2, \sigma^2 \rangle_s
      \frac{(t-s)^2}{2}
      + O_p\left((t-s)^{5/2}\right).
    \end{align}
    The same approach applies to the remaining two terms, leading to a
    final result of
    \begin{align}
      &
      \E\left[\left.
        \left(
          \left(X_{t} - X_{s}\right)\left(\sigma_{t}^{2} - \sigma_{s}^{2}\right)
          - \int_{s}^{t} 2\sigma_r^{2} f_r \ud r
        \right)^{2}
        \right| \calF_{s}
      \right]
      \notag
      \\
      &=
      \Bigl[
        \partial \langle X, X \rangle_s \partial \langle \sigma^2,
        \sigma^2 \rangle_s
        + \left(\partial \langle X, \sigma^2 \rangle_s\right)^2
      \Bigr] (t-s)^2
      + O_p\left((t-s)^{5/2}\right)
      \notag
      \\
      &=
      4\sigma_{s}^{4} \left(2f_{s}^{2} + g_{s}^{2}\right) (t-s)^{2}
      + O_p\left((t-s)^{5/2}\right).
    \end{align}
    This completes the proof.
  \end{proof}

  \begin{lemma}\label{lem:overlap-expectation}
    Suppose that Assumption~\ref{ass:process} holds, and let $v < s$,
    $u < t$, $d = (s-v)\lor(t-u) $.
    \begin{enumerate}[label=(\alph*)]
      \item If $v \leq u < s \leq t$, as $d \to 0$, we have
        \begin{align}
          &\E\left[\left.
            \left( \int_{v}^{s} (X_r - X_v) \ud X_r \right)
            \left( \int_{u}^{t} (X_r - X_u) \ud X_r \right)
          \right| \calF_v \right]
          =
          \frac{1}{2}\sigma_v^4 (s-u)^2 + O_p(d^{5/2}), \\
          &\E\left[\left.
            \left( \int_{v}^{s} (\sigma_r^2 - \sigma_v^2) \ud r \right)
            \left( \int_{u}^{t} (\sigma_r^2 - \sigma_u^2) \ud r \right)
          \right| \calF_v \right]
          =
          \frac{1}{3} \left.\frac{\ud\langle\sigma^2,
          \sigma^2\rangle_t}{\ud t}\right|_{t=v} (s-u)^3 + O_p(d^{7/2}).
        \end{align}
      \item If $v \leq u < t \leq s$, as $d \to 0$, we have
        \begin{align}
          &\E\left[\left.
            \left( \int_{v}^{s} (X_r - X_v) \ud X_r \right)
            \left( \int_{u}^{t} (X_r - X_u) \ud X_r \right)
          \right| \calF_v \right]
          =
          \frac{1}{2}\sigma_v^4 (t-u)^2 + O_p(d^{5/2}), \\
          &\E\left[\left.
            \left( \int_{v}^{s} (\sigma_r^2 - \sigma_v^2) \ud r \right)
            \left( \int_{u}^{t} (\sigma_r^2 - \sigma_u^2) \ud r \right)
          \right| \calF_v \right]
          =
          \frac{1}{3} \left.\frac{\ud\langle\sigma^2,
          \sigma^2\rangle_t}{\ud t}\right|_{t=v} (t-u)^3 + O_p(d^{7/2}).
        \end{align}
    \end{enumerate}
  \end{lemma}

  \begin{proof}
    Let $a < b$. By the Burkholder-Davis-Gundy inequality, there is a
    constant $C_1 > 0$ such that (ignoring the drift term of $O_p(b-a)$)
    \begin{equation}
      \E \Biggl[\left.
        \sup_{a \leq r \leq b} | X_r - X_a |
      \right| \calF_a \Biggr]
      \leq
      C_1
      \E \Biggl[\left.
        \left( \int_a^b \sigma_r^2 \ud r \right)^{1/2}
      \right| \calF_a \Biggr]
      \leq
      C_1 \sigma_+ (b-a)^{1/2}.
    \end{equation}
    Consequently, we have
    \begin{equation}
      \E \Biggl[\left.
        \Bigl| \int_a^b (X_r - X_a) \mu_r \ud r \Bigr|
      \right| \calF_a \Biggr]
      \leq
      C_1 \sigma_+ \mu_+ (b-a)^{3/2},
    \end{equation}
    where $\sigma_+$ and $\mu_+$ are some local upper bounds of
    $|\sigma_r|$ and $|\mu_r|$. It follows that
    \begin{equation}
      \E \Biggl[\left.
        \int_a^b (X_r - X_a) \ud X_r
      \right| \calF_a \Biggr]
      =
      \E \Biggl[\left.
        \int_a^b (X_r - X_a) \mu_r \ud r
      \right| \calF_a \Biggr]
      =
      O_p \left( (b-a)^{3/2} \right).
    \end{equation}
    Obviously, $\E [X_b - X_a|\calF_a] = O_p(b-a)$. According to the
    It\^o isometry, we also have
    \begin{equation}
      \E \Biggl[\left.
        (X_b - X_a) \int_a^b (X_r - X_a) \ud X_r
      \right| \calF_a \Biggr]
      =
      \E \Biggl[\left.
        \int_a^b (X_r - X_a) \sigma_r^2 \ud r
      \right| \calF_a \Biggr]
      =
      O_p \left( (b-a)^{3/2} \right).
    \end{equation}
    For the first equality in (a), decompose the integral as
    \begin{align}
      \int_v^s (X_r - X_v) \ud X_r
      &=
      \int_v^s (X_r - X_v) \ud X_r + (X_u - X_v) (X_s - X_u) + \int_u^s
      (X_r - X_u) \ud X_r, \\
      \int_u^t (X_r - X_u) \ud X_r
      &=
      \int_u^s (X_r - X_u) \ud X_r + (X_s - X_u) (X_t - X_s) + \int_s^t
      (X_r - X_s) \ud X_r.
    \end{align}
    By utilizing the nested property $\E[\cdot|\calF_v] =
    \E[\E(\E(\cdot|\calF_s)|\calF_u)|\calF_v]$ and applying the
    previous results, one can verify that the only dominant term turns out to be
    \begin{align}
      \E \Biggl[\left.
        \left(\int_u^s (X_r - X_u) \ud X_r\right)^2
      \right| \calF_v \Biggr]
      &=
      \E \Biggl[\left.
        \int_u^s (X_r - X_u)^2 \sigma_r^2 \ud r
      \right| \calF_v \Biggr]
      \notag
      \\
      &=
      \sigma_v^4 \int_u^s (r-u) \ud r + O_p\left((s-u)^{5/2}\right)
      \notag
      \\
      &=
      \frac{1}{2}\sigma_v^4 (s-u)^2 + O_p(d^{5/2}),
    \end{align}
    and all other terms are $O_p(d^{5/2})$. The rest of the proof is
    similar after using integration by parts:
    \begin{equation}
      \int_a^b (\sigma_r^2 - \sigma_a^2) \ud r
      =
      -\int_a^b (\sigma_r^2 - \sigma_b^2) \ud (b - r)
      =
      \int_a^b (b-r) \ud \sigma_r^2.
    \end{equation}
    This completes the proof.
  \end{proof}

  Without loss of generality, consider grids of two scales, with
  intervals being $1$ for the larger scale and $l \in (0, 1)$ for the
  smaller one. For a given larger scale interval, the grids of smaller
  scale partition the interval into several subintervals, with their
  lengths denoted as $(\alpha_1, \dotsc, \alpha_p)$. Define $\alpha_1$
  as a uniformly distributed random variable on $(0, l]$, so that $p$
  is also a random variable related to $\alpha_1$ and $l$. The
  parameter of interest is the summation of $k$th powers of
  $\alpha_i$'s as $S_k = S_k(\alpha_1, l) = \sum_{i=1}^{p} \alpha_i^k$,
  and its expectation with respect to $\alpha_1 \sim \unif (0, l]$ is
  discussed in the following lemma.

  \begin{lemma}\label{lem:grid-integral}
    For any integer $k \geq 1$, we have
    \begin{equation}
      \E_{\alpha_1\sim\unif(0,l]}\Bigl[ S_k (\alpha_1, l) \Bigr]
      =
      l^{k-1} \left(1 - \frac{k-1}{k+1} l \right)
      =
      l^{k-1} - \frac{k-1}{k+1} l^k.
    \end{equation}
  \end{lemma}

  \begin{proof}
    It is equivalent to prove that
    \begin{equation}\label{eq:grid-integral-I}
      I(l)
      = \int_0^l S_k (\alpha_1, l) \ud \alpha_1
      = l^k - \frac{k-1}{k+1} l^{k+1}.
    \end{equation}
    Denote $(\alpha_1, \alpha_2, \dotsc, \alpha_p) = (x, a_1, \dotsc,
    a_m)$, where $m = \lceil (1-x) / l \rceil = p - 1$, $x = \alpha_1
    \in (0, l]$, $a_1 = \dotsb, = a_{m-1} = l$, and $a_m = 1-x-(m-1)l$.
    The summation can be rewritten as
    \begin{equation}
      S_k(\alpha_1, l) = x^k + T_k(x, l), \quad \text{where } T_k(x, l)
      = \sum_{i=1}^m a_i^k,
    \end{equation}
    and the target integral is
    \begin{equation}
      I(l) = \frac{1}{k+1}l^{k+1} + J(l), \quad \text{where } J(l) =
      \int_0^l T_k(x, l) \ud x.
    \end{equation}
    Therefore, it is further equivalent to prove that
    \begin{equation}\label{eq:grid-integral-J}
      J(l)
      = \int_0^l T_k(x, l) \ud x
      = l^k - \frac{k}{k+1} l^{k+1}.
    \end{equation}
    It is easy to see that $\lim_{l\to 0^+} J(l) = 0$, which is
    consistent with Equation~\eqref{eq:grid-integral-J}. According to
    the Leibniz's rule, the derivative of $J(l)$ is
    \begin{equation}\label{eq:grid-integral-J-derivative}
      \frac{\ud}{\ud l} J(l)
      = T_k(l, l) + \int_0^l \frac{\partial}{\partial l} T_k(x, l) \ud x,
    \end{equation}
    so we only need to prove that for any $l \in (0, 1)$,
    \begin{equation}\label{eq:grid-integral-J-derivative-target}
      \frac{\ud}{\ud l} J(l)
      = k l^{k-1} - k l^k.
    \end{equation}
    For the first term in
    Equation~\eqref{eq:grid-integral-J-derivative}, we have
    \begin{equation}
      T_k (l, l) = (m-1)l^k + a_m^k.
    \end{equation}
    For the second term in
    Equation~\eqref{eq:grid-integral-J-derivative}, we need to derive
    the detailed form of $\frac{\partial}{\partial l} T_k(x, l)$. For a
    given $a_m \in (0, l]$, consider an arbitrarily small $\Delta l >
    0$, so that $m$ remains unchanged. In this case, $a_i = l \to l +
    \Delta l$, $\forall i=1,\dotsc, m-1$; while $a_m \to a_m - (m-1)l$.
    Therefore, the corresponding change in $T_k(x, l)$ is
    \begin{align}
      \Delta T_k(x, l)
      &=
      T_k(x, l+\Delta l) - T_k(x, l)
      \notag
      \\
      &=
      (m-1)\Bigl[ (l+\Delta l)^k - l^k \Bigr] + \Bigl[\left( a_m -
      (m-1)\Delta l \right)^k - a_m^k\Bigr]
      \notag
      \\
      &=
      (m-1) k (l^{k-1} - a_m^{k-1}) \Delta l + o(\Delta l),
    \end{align}
    so we have
    \begin{equation}\label{eq:grid-integral-derivative-Tk}
      \frac{\partial}{\partial l} T_k(x, l)
      = (m-1) k \left(l^{k-1} - a_m^{k-1}\right).
    \end{equation}
    Next, we need to consider the specific value of $a_m$ as $x$ ranges
    from $0$ to $l$. In fact, for a given $l$, there exists $r = \lceil
    1/l \rceil$ and $x_0 = 1-(r-1)l$, such that for $x \in (0, l]$, we have
    \begin{equation}
      m =
      \begin{cases}
        r, & \text{if } 0 < x < x_0, \\
        r - 1, & \text{if } x_0 \leq x \leq l,
      \end{cases}
    \end{equation}
    \begin{equation}
      a_m =
      \begin{cases}
        - x + x_0, & \text{if } 0 < x < x_0, \\
        - x + x_0 + l, & \text{if } x_0 \leq x \leq l,
      \end{cases}
    \end{equation}
    \begin{equation}
      \frac{\partial T_k}{\partial l} =
      \begin{cases}
        (r - 1) k \left(l^{k-1} - (x_0 - x)^{k-1}\right), & \text{if }
        0 < x < x_0, \\
        (r - 2) k \left(l^{k-1} - (x_0 + l - x)^{k-1}\right), &
        \text{if } x_0 \leq x \leq l.
      \end{cases}
    \end{equation}
    As a result, the first term in
    Equation~\eqref{eq:grid-integral-J-derivative} is
    \begin{equation}
      T_k (l, l) =  (r-2) l^k + x_0^k,
    \end{equation}
    and the second term in Equation~\eqref{eq:grid-integral-J-derivative} is
    \begin{align}
      \int_0^l \frac{\partial}{\partial l} T_k(x, l) \ud x
      &=
      \int_0^{x_0} \frac{\partial}{\partial l} T_k(x, l) \ud x
      + \int_{x_0}^l \frac{\partial}{\partial l} T_k(x, l) \ud x
      \notag
      \\
      &=
      (r-1) k \Bigl[ x_0 l^{k-1} - \frac{1}{k} x_0^k \Bigr]
      + (r-2) k \Bigl[ (l-x_0)l^{k-1} - \frac{1}{k} l^k + \frac{1}{k}
      x_0^k \Bigr]
      \notag
      \\
      &=
      (r-2)(k-1) l^k + kx_0 l^{k-1} - x_0^k.
    \end{align}
    Therefore, the derivative of $J(l)$ is
    \begin{align}
      \frac{\ud}{\ud l} J(l)
      &=
      T_k(l, l) + \int_0^l \frac{\partial}{\partial l} T_k(x, l) \ud x
      \notag
      \\
      &=
      (r-2) l^k + x_0^k
      + (r-2)(k-1) l^k + kx_0 l^{k-1} - x_0^k
      \notag
      \\
      &=
      (r-2) kl^k + k [1-(r-1)l] l^{k-1}
      \notag
      \\
      &=
      kl^{k-1} - kl^k,
    \end{align}
    which is the same as
    Equation~\eqref{eq:grid-integral-J-derivative-target}. This
    completes the proof of Lemma~\ref{lem:grid-integral}.
  \end{proof}

  \subsection{Proof of Proposition~\ref{prop:all-observation-noise}}
  \label{sec:app-all-observation-noise}

  \subsubsection{Notations and Unbiasedness}
  \label{sec:app-all-observation-noise-notation-bias}

  Rewrite the noisy all-observation estimator as
  \begin{align}
    \esty^{\rm (all)}_T
    =
    \sum_{i \in I}
    (\Delta Y_i)
    \: \widehat{\delta} (i, 1, k_n)
    =
    \frac{1}{k_n \Delta_n}
    \sum_{i \in I}
    u_i,
  \end{align}
  where $I = \{k_n+1, \dotsc, n-k_n-2\}$. Denote the noise-free version
  of $\widehat{\delta}(i, H, k, s)$ defined in
  Equation~\eqref{eq:spot-vol-delta} as $\overline{\delta}(i, H, k,
  s)$. Decompose $u_i$ as
  \begin{equation}
    u_i = \left( a_i + A_i \right) \left( b_i + B_i \right),
  \end{equation}
  where
  \begin{align}
    \textcolor{blue}{a_i} + \textcolor{red}{A_i} & = (\Delta Y_{i}) =
    \textcolor{blue}{(\Delta X_{i})}
    +
    \textcolor{red}{\langle \theta_{A_i}, \eps \rangle}, \\
    \textcolor{blue}{b_i} + \textcolor{red}{B_i}
    & =
    k_n \Delta_n \widehat{\delta} (i, 1, k_n)
    =
    \textcolor{blue}{
      k_n \Delta_n \overline{\delta} (i, 1, k_n)
    }
    +
    \textcolor{red}{
      \langle \theta_{B_i}, \eps \rangle
      + \langle \phi_{B_i}, \eps^{2} \rangle
      + \langle \psi_{B_i}, \eps\eps_{+} \rangle
    }.
  \end{align}
  The lowercase letters $a_i$ and $b_i$ represent pure process terms
  $u_i$, while the capital letters $A_i$ and $B_i$ represent
  noise-related terms. The latter includes three types: $\eps_j$,
  $\eps_j^2$, and $\eps_j \eps_{j+1}$ (for valid $j$ values). For
  simplicity, we denote the corresponding noise vectors as $\eps$,
  $\eps^2$, and $\eps\eps_+$, defined by
  \begin{equation}
    (\eps)_{j} = \eps_{j}, \quad (\eps^{2})_{j} = \eps_{j}^{2}, \quad
    (\eps\eps_{+})_{j} = \eps_{j}\eps_{j+1}I(j \neq n),
  \end{equation}
  for $j \in \{0, 1, \dotsc, n\}$. The non-zero values of the
  coefficient vectors $\theta_{A_i}$, $\theta_{B_i}$, $\phi_{B_i}$, and
  $\psi_{B_i}$ are listed in Table~\ref{tab:coeff-all-observation}. We
  also adopt the convention that the index of coefficient vectors is
  written as a superscript, so we have, for instance, $\langle
  \psi_{B_i}, \eps\eps_+ \rangle = \sum_{j=0}^{n} \psi_{B_i}^{j}
  \eps_{j}\eps_{j+1}$.

  \begin{table}[!ht]
    \centering
    \caption{Values of coefficient vectors in all-observation estimator}
    \label{tab:coeff-all-observation}
    \resizebox{\textwidth}{!}{
      \begin{tabular}{ccccccccccccc}
        \toprule
        ~ & $i-k_n-1$ & $i-k_n$ & $\cdots$ & $i-2$ & $i-1$ & $i$ &
        $i+1$ & $i+2$ & $i+3$ & $\cdots$ & $i+k_n+1$ & $i+k_n+2$ \\
        \midrule
        $\theta_{A_i}$ & \textcolor{blue}{$0$} & \textcolor{blue}{$0$}
        & $\cdots$ & \textcolor{blue}{$0$} & \textcolor{blue}{$0$} &
        $-1$ & $+1$ & \textcolor{blue}{$0$} & \textcolor{blue}{$0$} &
        $\cdots$ & \textcolor{blue}{$0$} & \textcolor{blue}{$0$} \\
        $\theta_{B_i}$ & $+2\Delta X_{i-k_n-1}$ & $+2\Delta^2
        X_{i-k_n-1}$ & $\cdots$ & $+2\Delta^2 X_{i-3}$ & $-2\Delta
        X_{i-2}$ & \textcolor{blue}{$0$} & \textcolor{blue}{$0$} &
        $-2\Delta X_{i+2}$ & $-2\Delta^2 X_{i+2}$ & $\cdots$ &
        $-2\Delta^2 X_{i+k_n}$ & $+2\Delta X_{i+k_n+1}$ \\
        $\phi_{B_i}$ & $-1$ & $-2$ & $\cdots$ & $-2$ & $-1$ &
        \textcolor{blue}{$0$} & \textcolor{blue}{$0$} & $+1$ & $+2$ &
        $\cdots$ & $+2$ & $+1$ \\
        $\psi_{B_i}$ & $+2$ & $+2$ & $\cdots$ & $+2$ &
        \textcolor{blue}{$0$} & \textcolor{blue}{$0$} &
        \textcolor{blue}{$0$} & $-2$ & $-2$ & $\cdots$ & $-2$ &
        \textcolor{blue}{$0$} \\
        \bottomrule
      \end{tabular}
    }
  \end{table}

  Similarly, rewrite the noise-free all-observation estimator as
  \begin{align}
    \estx_T^{\rm (all)}
    =
    \sum_{i\in I}
    (\Delta X_i)
    \: \overline{\delta} (i, H_n, k_n)
    =
    \frac{1}{k_n \Delta_n}
    \sum_{i\in I}
    v_i,
  \end{align}
  where $v_i = a_i b_i$. Note that $\E[\eps_i] = 0$, so $\E [A_i |
  \calF] = 0$; and that $\E[\eps_i \eps_{i+1}] = 0$, $\sum_j
  \phi_{B_i}^j = 0$, so $\E [B_i | \calF] = 0$. Also note that $(A_i
  \perp B_i) \mid \calF$. We have thus
  \begin{align}
    \E[u_i | \calF]
    &=
    \E \Bigl(
      (a_i + A_i) (b_i + B_i)
      \Big| \calF
    \Bigr)
    \\
    &=
    a_i b_i
    + a_i \E [B_i | \calF]
    + b_i \E [A_i | \calF]
    + \E [A_i | \calF] \E [B_i | \calF]
    \\
    &=
    v_i.
  \end{align}
  Therefore, the all-observation estimator is unbiased due to noise:
  \begin{equation}
    \E \left( \esty^{\rm (all)}_T \middle| \calF \right)
    =
    \estx^{\rm (all)}_T.
  \end{equation}

  \subsubsection{Variance due to noise}
  \label{sec:app-all-observation-noise-variance}

  The variance due to noise is given by
  \begin{align}
    \var\left( \esty^{\rm (all)}_T \middle| \calF \right)
    &=
    \E\left( \left( \esty^{\rm (all)}_T \right)^2 \middle| \calF \right)
    -
    \left[\E\left( \esty^{\rm (all)}_T \middle| \calF \right)\right]^2
    \\
    &=
    \E\left( \left( \esty^{\rm (all)}_T \right)^2 \middle| \calF \right)
    -
    \left( \estx^{\rm (all)}_T \right)^2
    \\
    &=
    \frac{1}{k_n^2 \Delta_n^2}
    \E\left(\left.
      \sum_{k \in I} \sum_{l \in I}
      \left( u_k u_l - v_k v_l \right)
    \right| \calF \right).
  \end{align}
  At this stage, we obtain an important decomposition:
  \begin{equation}
    \begin{aligned}
      &
      k_n^2 \Delta_n^2 \var\left( \esty^{\rm (all)}_T \middle| \calF \right)
      \\
      &=
      \sum_{k \in I} \sum_{l \in I}
      \E\Bigl[
        \left( a_k + A_k \right) \left( b_k + B_k \right)
        \left( a_l + A_l \right) \left( b_l + B_l \right)
        -
        a_k b_k a_l b_l
      \Big| \calF \Bigr] \\
      &=
      \sum_{k \in I} \sum_{l \in I}
      \E\Bigl[
        N_1(k,l) + N_2(k,l) + N_3(k,l) + N_4(k,l)
      \Big| \calF \Bigr].
    \end{aligned}
  \end{equation}
  Note that in the above equation, the pure process-related term $a_k
  b_k a_l b_l$ is elimininated, while the remaining terms may contain
  one to four noise-related terms. For simplicity, we omit the
  subscript $k$ in $(a_k, b_k, A_k, B_k)$ and denote $(a_l, b_l, A_l,
  B_l)$ as $(c, d, C, D)$ when these terms appear inside $\sum_{k,l}$.
  Thus, we have
  \begin{align}
    N_1(k, l) &= \sum_{\rm comb} Abcd = Abcd + aBcd + abCd + abcD, \\
    N_2(k, l) &= \sum_{\rm comb} ABcd = ABcd + AbCd + AbcD + aBCd +
    aBcD + abCD, \\
    N_3(k, l) &= \sum_{\rm comb} ABCd = ABCd + ABcD + AbCD + aBCD, \\
    N_4(k, l) &= ABCD.
  \end{align}
  We further adopt the conventions that noise terms in $A, B, C, D$ are
  indexed by $i, \alpha, j, \beta$ respectively, for example, $A =
  \langle \theta_{A}, \eps \rangle = \sum_{i=0}^{n} \theta_{A}^{i} \eps_{i}$.

  To summarize, the variance due to noise can be calculated by
  \begin{equation}\label{all-observation-noise-variance-four-terms}
    \var\left( \esty^{\rm (all)}_T \middle| \calF \right)
    =
    \sum_{m=1}^{4}
    \frac{1}{k_n^2 \Delta_n^2}
    \sum_{k \in I} \sum_{l \in I}
    \E \Bigl(
      N_{m}(k, l)
    \Big| \calF \Bigr),
  \end{equation}
  and we will proceed by analyzing the contribution of $m=1, 2, 3, 4$,
  respectively.

  The following observation, while elementary, will be repeatedly used
  in the calculations.
  \begin{lemma}\label{lem:app-noise-prod-expect}
    Let $\{X_{i}\}_{i=1}^{l}$ be mean zero, independent and identically
    distributed random variables, and let $\nu_{k} = \E[X_i^k]$ denote
    the $k$th moment of $X_i$. For integers $k_{1}, k_{2}, \dots, k_{l}
    \geq 1$, we have
    \begin{equation}
      \E[X_{1}^{k_{1}}X_{2}^{k_{2}}\dotsb X_{l}^{k_{l}}] =
      \nu_{k_{1}}\nu_{k_{2}}\dotsb \nu_{k_{l}},
    \end{equation}
    A necessary condition for the above expectation to be non-zero is that
    $k_i \geq 2$ for all $i=1, \dots, l$. Moreover, let $n =
    \sum_{i=1}^{l} k_{i}$, then $n \geq 2l$.
  \end{lemma}

  \paragraph*{Contribution of $N_1$}
  With the same argument as in
  Section~\ref{sec:app-all-observation-noise-notation-bias}, we have
  \begin{equation}
    \E[Abcd | \calF] = bcd\E[A | \calF] = 0.
  \end{equation}
  Similarly,
  \begin{equation}
    \E[aBcd | \calF] = \E[abCd | \calF] = \E[abcD | \calF] = 0.
  \end{equation}
  Therefore, $N_1$ contributes nothing to the variance due to noise.

  \paragraph*{Contribution of $N_2$}
  The following lemma concerns the noise-related terms in $N_2$.

  \begin{lemma}\label{lem:app-all-observation-noise-N2}
    For $P, Q \in \{A, B, C, D\}$, $P \neq Q$, we have
    \begin{equation}
      \begin{aligned}
        \E[PQ | \calF] &= \nu_2 \langle\theta_P, \theta_Q\rangle
        +\nu_4\langle\phi_P, \phi_Q\rangle \\
        &\quad +\nu_2^2 \left[-\langle\phi_P,
        \phi_Q\rangle+\langle\psi_P, \psi_Q\rangle\right] + \nu_3
        \left[\langle\theta_P, \phi_Q\rangle+\langle\phi_P,
        \theta_Q\rangle\right].
      \end{aligned}
    \end{equation}
  \end{lemma}
  \begin{proof}
    Suppose that the error terms are $\eps_{\alpha}$, $\eps_{\alpha}^{2}$,
    $\eps_{\alpha}\eps_{\alpha+1}$ in $P$, and $\eps_{\beta}$,
    $\eps_{\beta}^{2}$, $\eps_{\beta}\eps_{\beta+1}$ in $Q$. For all
    nine possible combination in $PQ$, by
    Lemma~\ref{lem:app-noise-prod-expect}, the non-zero terms are
    \begin{equation}
      \begin{aligned}
        \E [\eps_{\alpha} \eps_{\beta}] &= \nu_2 \delta_{\alpha\beta}, \\
        \E [\eps_{\alpha}^{2} \eps_{\beta}^{2}] &= \nu_4
        \delta_{\alpha\beta} + \nu_2^{2} (1 - \delta_{\alpha\beta}), \\
        \E [\eps_{\alpha} \eps_{\alpha+1} \eps_{\beta} \eps_{\beta+1}]
        &= \nu_2^{2} \delta_{\alpha\beta}, \\
        \E [\eps_{\alpha} \eps_{\beta}^{2}] &= \nu_3 \delta_{\alpha\beta}, \\
        \E [\eps_{\alpha}^{2} \eps_{\beta}] &= \nu_3 \delta_{\alpha\beta}.
      \end{aligned}
    \end{equation}
    Take the coefficients into considerations and note that
    \begin{equation}
      \sum_{\alpha, \beta} \phi_{P}^{\alpha} \phi_{Q}^{\beta} =
      \sum_{\alpha} \phi_{P}^{\alpha} \sum_{\beta} \phi_{Q}^{\beta} = 0,
    \end{equation}
    so we have
    \begin{equation}
      \begin{aligned}
        \E[PQ | \calF]
        &=
        \nu_2 \sum_{\alpha, \beta} \theta_{P}^{\alpha}
        \theta_{Q}^{\beta} \delta_{\alpha \beta}
        + \nu_4 \sum_{\alpha, \beta} \phi_{P}^{\alpha} \phi_{Q}^{\beta}
        \delta_{\alpha \beta} \\
        & \quad
        + \nu_2^{2} \sum_{\alpha, \beta} \left(
          \phi_{P}^{\alpha} \phi_{Q}^{\beta} (1 - \delta_{\alpha\beta})
          + \psi_{P}^{\alpha} \psi_{Q}^{\beta} \delta_{\alpha\beta}
        \right) \\
        & \quad
        + \nu_3 \sum_{\alpha, \beta} \left(
          \theta_{P}^{\alpha} \phi_{Q}^{\beta} \delta_{\alpha \beta}
          + \phi_{P}^{\alpha} \theta_{Q}^{\beta} \delta_{\alpha \beta}
        \right) \\
        &= \nu_2 \langle\theta_P, \theta_Q\rangle +\nu_4\langle\phi_P,
        \phi_Q\rangle \\
        &\quad +\nu_2^2 \left[-\langle\phi_P,
        \phi_Q\rangle+\langle\psi_P, \psi_Q\rangle\right] + \nu_3
        \left[\langle\theta_P, \phi_Q\rangle+\langle\phi_P,
        \theta_Q\rangle\right].
      \end{aligned}
    \end{equation}
    The proof is thus completed.
  \end{proof}

  Although there are six terms to consider in the contribution of $N_2$, as
  \begin{equation}
    \E [AB | \calF] = \E[CD | \calF] = 0,
  \end{equation}
  we only need to calculate
  \begin{equation}
    \frac{1}{k_n^2 \Delta_n^2}
    \sum_{k\in I} \sum_{l\in I}
    \E\Bigl(bdAC + bcAD + adBC + acBD \Big| \calF\Bigr).
  \end{equation}

  Notice that
  \begin{equation}
    \E[AC | \calF] = \nu_2 \Bigl(2I(k=l) - I(|k-l|=1)\Bigr),
  \end{equation}
  so
  \begin{equation}
    \begin{aligned}
      \sum_{k\in I}\sum_{l\in I}
      \E [bdAC | \calF]
      &=
      \nu_2 \left(
        2 \sum_{k=k_n+1}^{n-k_n-2} b_k^2
        - \sum_{k=k_n+2}^{n-k_n-2} b_k b_{k-1}
        - \sum_{k=k_n+1}^{n-k_n-3} b_k b_{k+1}
      \right) \\
      &=
      \nu_2 \left(
        \sum_{k=k_n+1}^{n-k_n-3} (\Delta b_k)^2
        + b_{k_n+1}^2 + b_{n-k_n-2}^2
      \right),
    \end{aligned}
  \end{equation}
  where
  \begin{equation}
    \begin{aligned}
      (\Delta b_k)^2
      &=
      \left(
        \left(\Delta X_{k+k_n+2}\right)^2 - \left(\Delta X_{k+2}\right)^2
        - \left(\Delta X_{k-1}\right)^2 + \left(\Delta X_{k-k_n-1}\right)^2
      \right)^2,
    \end{aligned}
  \end{equation}
  As \cite{mykland2009InferenceContinuousSemimartingales} point out,
  the drift terms typically have no impact on the final result in the
  high-frequency setting, and locally constant approximations can be
  applied to simplify the analysis of asymptotic results. Therefore,
  by noticing that
  \begin{equation}
    \begin{aligned}
      (\Delta b_k)^2
      &\approx
      \Bigl(
        \sigma_{k+k_n+2}^2 \left(\Delta W_{k+k_n+2}\right)^2 -
        \sigma_{k+2}^2 \left(\Delta W_{k+2}\right)^2 \\
        &\qquad
        - \sigma_{k-1}^2 \left(\Delta W_{k-1}\right)^2 +
        \sigma_{k-k_n-1}^2\left(\Delta W_{k-k_n-1}\right)^2
      \Bigr)^2 \\
      &\approx
      8\sigma_k^4 \Delta_n^2,
    \end{aligned}
  \end{equation}
  we have
  \begin{equation}\label{eq:all-observation-noise-N2-AC}
    \Delta_n^{-1}
    \sum_{k\in I}\sum_{l\in I}
    \E \Bigl(bdAC \Big| \calF\Bigr)
    \convp
    8\nu_2 \int_{0}^{T} \sigma_{t}^4 \ud t.
  \end{equation}

  Notice that
  \begin{equation}
    \begin{aligned}
      \E[AD | \calF]
      &= \nu_2 \langle\theta_A, \theta_D\rangle + \nu_3
      \langle\theta_A, \phi_D\rangle,
    \end{aligned}
  \end{equation}
  so
  \begin{equation}
    \begin{aligned}
      \sum_{k\in I}\sum_{l\in I}
      \E [bcAD | \calF]
      &=
      \nu_2 \sum_{k\in I}\sum_{l\in I} bc \langle\theta_A, \theta_D\rangle
      + \nu_3 \sum_{k\in I}\sum_{l\in I} bc \langle\theta_A, \phi_D\rangle.
    \end{aligned}
  \end{equation}
  When $\langle\theta_A, \theta_D\rangle$ is not zero, it can be
  written as $\sum_{r} \eta_r \Delta X_{r}$, where up to four terms are
  present in this summation. The coefficient $\eta_r$ satisfies
  $|\eta_r| \leq 3$, and $\langle\theta_A, \theta_D\rangle$ is
  obviously mean zero as well as independent of $bc$, while $bc$ only
  includes the linear and cubic terms of $(\Delta X_r)$. Therefore, we
  can ignore the impact of $\sum_{k\in I}\sum_{l\in I} bc
  \langle\theta_A, \theta_D\rangle$ as its ``standardized version"
  (divided by the order of its standard deviation) is asymptotically
  mean zero and therefore tight. As for $\langle \theta_A,
  \phi_D\rangle$, one can verify that
  \begin{equation}\label{eq:theta_A-phi_D}
    \begin{aligned}
      \langle \theta_A, \phi_D\rangle
      &=
      I\Bigl(|k-l| \in \{1,2\}\Bigr) - I\Bigl(|k-l| \in \{k_n+1, k_n+2\}\Bigr),
    \end{aligned}
  \end{equation}
  so for a given $k$, there are up to eight non-zero $\langle \theta_A,
  \phi_D\rangle$, each with constant absolute value 1. When $\langle
  \theta_A, \phi_D\rangle$ is non-zero, $bc$ vanishes. Therefore, the
  ``standardized version" of $\sum_{k\in I}\sum_{l\in I} bc
  \langle\theta_A, \phi_D\rangle$ is tight as well. The order of
  variance can be estimated as $O_p(k_n^2 \Delta_n)$ for $\sum_{k\in
  I}\sum_{l\in I} bc \langle\theta_A, \phi_D\rangle$ and even smaller
  for $\sum_{k\in I}\sum_{l\in I} bc \langle\theta_A, \theta_D\rangle$.
  It follows that
  \begin{equation}\label{eq:all-observation-noise-N2-AD}
    k_n^{-1} \Delta_n^{-1/2}
    \sum_{k\in I}\sum_{l\in I}
    \E \Bigl(bcAD \Big| \calF\Bigr)
    =
    O_p(1).
  \end{equation}
  Moreover, because
  \begin{equation}
    \sum_{k\in I}\sum_{l\in I} adBC = \sum_{k\in I}\sum_{l\in I} a_k
    b_l B_k A_l = \sum_{k\in I}\sum_{l\in I} a_l b_k B_l A_k =
    \sum_{k\in I}\sum_{l\in I} bcAD,
  \end{equation}
  we also have
  \begin{equation}\label{eq:all-observation-noise-N2-BC}
    k_n^{-1} \Delta_n^{-1/2}
    \sum_{k\in I}\sum_{l\in I}
    \E \Bigl(adBC \Big| \calF\Bigr)
    =
    O_p(1).
  \end{equation}

  Lastly, notice that
  \begin{equation}
    \begin{aligned}
      \E[BD | \calF]
      &= \nu_2 \langle\theta_B, \theta_D\rangle +\nu_4 \langle\phi_B,
      \phi_D\rangle \\
      &\quad +\nu_2^2 \left[-\langle\phi_B,
      \phi_D\rangle+\langle\psi_B, \psi_D\rangle\right] + \nu_3
      \left[\langle\theta_B, \phi_D\rangle+\langle\phi_B,
      \theta_D\rangle\right].
    \end{aligned}
  \end{equation}
  Because $ac \approx \sigma_{k}^2 \Delta_n I(k=l)$, and
  \begin{equation}
    \|\phi_B\|^2 = 8k_n - 4, \quad \|\psi_B\|^2 = 8k_n,
  \end{equation}
  we can know that, as $n\to\infty$,
  \begin{equation}
    \begin{aligned}
      &
      \sum_{k\in I}\sum_{l\in I}
      ac \Bigl(
        \nu_4 \langle\phi_B, \phi_D\rangle
        +\nu_2^2 \left[-\langle\phi_B, \phi_D\rangle + \langle\psi_B,
        \psi_D\rangle\right]
      \Bigr) \\
      &\approx
      \sum_{k\in I}
      \sigma_k^2 \Delta_n \Bigl[
        \nu_4 \|\phi_B\|^2
        +\nu_2^2 \left(-\|\phi_B\|^2 + \|\psi_B\|^2\right)
      \Bigr] \\
      &=
      \sum_{k\in I}
      \sigma_k^2 \Delta_n \Bigl[
        (8k_n - 4) (\nu_4 - \nu_2^2)
        + 8k_n \nu_2^2
      \Bigr] \\
      &\approx
      \Bigl[
        (8k_n - 4) (\nu_4 - \nu_2^2)
        + 8k_n \nu_2^2
      \Bigr]
      \int_{0}^{T} \sigma_t^2 \ud t
      \\
      &\approx
      8k_n \nu_4 \int_{0}^{T} \sigma_t^2 \ud t.
    \end{aligned}
  \end{equation}
  For the $ac \langle\theta_B, \theta_D\rangle$ term, when $k \neq l$,
  only $(a, D)$ and $(c, B)$ can overlap, so $ac \langle\theta_B,
  \theta_D\rangle$ only contains the cubic terms of $(\Delta X_r)$, and
  thus vanishes. Therefore, the corresponding summation can be roughly
  estimated as
  \begin{equation}
    \begin{aligned}
      \sum_{k\in I}\sum_{l\in I}
      ac \Bigl(
        \nu_2 \langle\theta_B, \theta_D\rangle
      \Bigr)
      &\approx
      \nu_2
      \sum_{k\in I}
      a^2 \|\theta_B\|^2
      \approx
      \nu_2
      \sum_{k\in I}
      \sigma_k^2 \Delta_n \|\theta_B\|^2 \\
      &=
      \nu_2
      \left(\int_{0}^{T} \sigma_t^2 \ud t\right)
      O_p(k_n \Delta_n)
      =
      O_p(k_n \Delta_n).
    \end{aligned}
  \end{equation}
  As for the $ac \langle\theta_B, \phi_D\rangle$ term, notice that
  \begin{equation}
    \sum_{k\in I}\sum_{l\in I}
    ac \langle\theta_B, \phi_D\rangle
    =
    \sum_{k\in I}
    a^2 \langle\theta_B, \phi_B\rangle
    +
    \sum_{k\in I}\sum_{l\in I}
    ac \langle\theta_B, \phi_D\rangle
    I(k \neq l).
  \end{equation}
  The first term vanishes because $\langle\theta_B, \phi_B\rangle$ only
  contains linear terms of $(\Delta X_r)$ and does not overlap with
  $a^2$. The second term vanishes because $a$ only contains the linear
  terms of $(\Delta X_r)$ and does not overlap with $c$ or
  $\langle\theta_B, \phi_D\rangle$. As a result, the ``standardized
  version" of $\sum_{k\in I} \sum_{l\in I} ac \langle\theta_B,
  \phi_D\rangle$ is tight, which is
  \begin{equation}
    k_n^{-1} \Delta_n^{-3/2}
    \sum_{k\in I}\sum_{l\in I}
    ac \langle\theta_B, \phi_D\rangle
    =
    O_p(1).
  \end{equation}
  The same argument can be applied to the $ac \langle\phi_B,
  \theta_D\rangle$ term. Therefore, we have
  \begin{equation}\label{eq:all-observation-noise-N2-BD}
    k_n^{-1}
    \sum_{k\in I}\sum_{l\in I}
    \E \Bigl(acBD \Big| \calF\Bigr)
    \convp
    8 \nu_4 \int_{0}^{T} \sigma_t^2 \ud t.
  \end{equation}

  Therefore, by combining
  Equation~\eqref{eq:all-observation-noise-N2-AC},
  \eqref{eq:all-observation-noise-N2-AD},
  \eqref{eq:all-observation-noise-N2-BC}, and
  \eqref{eq:all-observation-noise-N2-BD}, we can conclude that the
  contribution of the $N_2$ term to the variance due to noise is
  \begin{equation}\label{eq:all-observation-noise-N2}
    k_n^{-1}
    \sum_{k\in I} \sum_{l\in I}
    \E\Bigl(N_2(k, l) \Big| \calF\Bigr)
    \convp
    8\nu_4 \int_{0}^{T} \sigma_t^2 \ud t.
  \end{equation}

  \paragraph*{Contribution of $N_3$}
  The contribution of $N_3$ can be calculated in a similar way as
  $N_2$. We have the following lemma concerning the noise-related
  terms in $N_3$.

  \begin{lemma}\label{lem:app-all-observation-noise-N3}
    For $ABC$ and $ACD$ terms, we have
    \begin{equation}
      \E\left[ ABC | \calF \right] = \E\left[ ACD | \calF \right] = 0.
    \end{equation}
    For $ABD$ and $BCD$ terms, we have
    \begin{equation}
      \begin{aligned}
        \E\left[ ABD | \calF \right]
        = \sum_{i, \alpha, \beta} \Bigl\{
          & \nu_2^{2} \left[ \theta_{A}^{i} \theta_{B}^{\alpha}
          \psi_{D}^{\beta} \right]
          \left( \delta_{i\beta}^{(k+1)} \delta_{\alpha,
            \beta+1}^{(k+2)} + \delta_{\alpha\beta}^{(k)} \delta_{i,
          \beta+1}^{(k-1)} \right) \\
          & + \nu_2 \nu_3 \left[ \theta_{A}^{i} \phi_{B}^{\alpha}
          \psi_{D}^{\beta} \right]
          \left( \delta_{i\beta}^{(k+1)} \delta_{\alpha,
            \beta+1}^{(k+2)} + \delta_{\alpha\beta}^{(k)} \delta_{i,
          \beta+1}^{(k-1)} \right)
        \Bigr\},
      \end{aligned}
    \end{equation}
    where $\delta_{i\beta}^{(k+1)}$ is defined as $I(i = \beta = k+1)$, and
    \begin{equation}
      \begin{aligned}
        \E\left[ BCD | \calF \right]
        = \sum_{j, \alpha, \beta} \Bigl\{
          & \nu_2^{2} \left[ \theta_{C}^{j} \theta_{D}^{\beta}
          \psi_{B}^{\alpha} \right]
          \left( \delta_{j\alpha}^{(l+1)} \delta_{\beta,
            \alpha+1}^{(l+2)} + \delta_{\beta\alpha}^{l} \delta_{j,
          \alpha+1}^{(l-1)} \right) \\
          & + \nu_2 \nu_3 \left[ \theta_{C}^{j} \phi_{D}^{\beta}
          \psi_{B}^{\alpha} \right]
          \left( \delta_{j\alpha}^{(l+1)} \delta_{\beta,
            \alpha+1}^{(l+2)} + \delta_{\beta\alpha}^{l} \delta_{j,
          \alpha+1}^{(l-1)} \right)
        \Bigr\}.
      \end{aligned}
    \end{equation}
  \end{lemma}

  \begin{proof}
    Firstly, consider $\E\left[ ABC | X, \sigma \right]$. Notice that
    \begin{equation}
      ABC
      =
      \sum_{i,j,\alpha}
      \theta_A^i \eps_i
      \Bigl(\theta_B^\alpha \eps_\alpha + \phi_B^\alpha \eps_\alpha^2 +
      \psi_B^\alpha \eps_\alpha \eps_{\alpha+1}\Bigr)
      \theta_C^j \eps_j.
    \end{equation}
    There are three types of possible combinations in $ABC$, the
    expectations of which are
    \begin{equation}
      \begin{aligned}
        \E [\eps_{i} \eps_{\alpha} \eps_{j}] &= 0, \\
        \E [\eps_{i} \eps_{\alpha}^{2} \eps_{j}] &= \nu_2^{2} \delta_{ij}, \\
        \E [\eps_{i} \eps_{\alpha} \eps_{\alpha+1} \eps_{j}] &= 0,
      \end{aligned}
    \end{equation}
    according to Lemma~\ref{lem:app-noise-prod-expect}. This is because
    $i \neq \alpha$ for all terms and $i \neq \alpha+1$ additionally
    for terms with $\eps_{\alpha}\eps_{\alpha+1}$. Taking their
    coefficients into consideration, and notice that
    \begin{equation}
      \sum_{\alpha} \phi_{B}^{\alpha} = 0,
      \E\left[ ABC | X, \sigma \right]
      = \sum_{i, j, \alpha} \nu_2^{2} \left[
      \theta_{A}^{i} \phi_{B}^{\alpha} \theta_{C}^{j} \right]\left(
      \delta_{ij} \right)
      = \sum_{i} \nu_2^{2} \left[
      \theta_{A}^{i} \theta_{C}^{i} \right] \sum_{\alpha} \phi_{B}^{\alpha}
      = 0.
    \end{equation}
    For the same reason, we can conclude that
    \begin{equation}
      \E\left[ ACD | X, \sigma \right] = 0.
    \end{equation}

    Next, consider $\E\left[ ABD | X, \sigma \right]$. Notice that
    \begin{equation}
      ABD
      =
      \sum_{i,\alpha,\beta}
      \theta_A^i \eps_i
      \Bigl(\theta_B^\alpha \eps_\alpha + \phi_B^\alpha \eps_\alpha^2 +
      \psi_B^\alpha \eps_\alpha \eps_{\alpha+1}\Bigr)
      \Bigl(\theta_D^\beta \eps_\beta + \phi_D^\beta \eps_\beta^2 +
      \psi_D^\beta \eps_\beta \eps_{\beta+1}\Bigr).
    \end{equation}
    There are nine types of possible combinations, among which four
    types have non-zero expectations, and they are
    \begin{equation}
      \begin{aligned}
        \E \left[ \eps_{i} \eps_{\alpha} \eps_{\beta} \eps_{\beta+1} \right]
        &= \nu_2^{2} \left(
          \delta_{i\beta} \delta_{\alpha, \beta+1} +
          \delta_{\alpha\beta} \delta_{i, \beta+1}
        \right), \\
        \E \left[ \eps_{i} \eps_{\alpha}^{2} \eps_{\beta} \right]
        &= \nu_2^{2} \delta_{i\beta}, \\
        \E \left[ \eps_{i} \eps_{\alpha}^{2} \eps_{\beta}^{2} \right]
        &= \nu_2 \nu_3 \delta_{i\beta}, \\
        \E \left[ \eps_{i} \eps_{\alpha}^{2} \eps_{\beta} \eps_{\beta+1} \right]
        &= \nu_2 \nu_3 \left(
          \delta_{i\beta} \delta_{\alpha, \beta+1} +
          \delta_{\alpha\beta} \delta_{i, \beta+1}
        \right).
      \end{aligned}
    \end{equation}
    Therefore, we have
    \begin{equation}
      \begin{aligned}
        \E\left[ ABD | \calF \right]
        = \sum_{i, \alpha, \beta} \Bigl\{
          & \nu_2^{2} \left[ \theta_{A}^{i} \theta_{B}^{\alpha}
          \psi_{D}^{\beta} \right]
          \left( \delta_{i\beta} \delta_{\alpha, \beta+1} +
          \delta_{\alpha\beta} \delta_{i, \beta+1} \right) \\
          & + \nu_2^{2} \left[ \theta_{A}^{i} \phi_{B}^{\alpha}
          \theta_{D}^{\beta} \right] \left( \delta_{i\beta} \right)
          + \nu_2 \nu_3 \left[ \theta_{A}^{i} \phi_{B}^{\alpha}
          \phi_{D}^{\beta} \right] \left( \delta_{i\beta} \right) \\
          & + \nu_2 \nu_3 \left[ \theta_{A}^{i} \phi_{B}^{\alpha}
          \psi_{D}^{\beta} \right] \left( \delta_{i\beta} \delta_{\alpha,
        \beta+1} + \delta_{\alpha\beta} \delta_{i, \beta+1} \right) \Bigr\} \\
        = \sum_{i, \alpha, \beta} \Bigl\{
          & \nu_2^{2} \left[ \theta_{A}^{i} \theta_{B}^{\alpha}
          \psi_{D}^{\beta} \right]
          \left( \delta_{i\beta}^{(k+1)} \delta_{\alpha,
            \beta+1}^{(k+2)} + \delta_{\alpha\beta}^{(k)} \delta_{i,
          \beta+1}^{(k-1)} \right) \\
          & + \nu_2 \nu_3 \left[ \theta_{A}^{i} \phi_{B}^{\alpha}
          \psi_{D}^{\beta} \right] \left( \delta_{i\beta}^{(k+1)}
            \delta_{\alpha, \beta+1}^{(k+2)} + \delta_{\alpha\beta}^{(k)}
        \delta_{i, \beta+1}^{(k-1)} \right) \Bigr\},
      \end{aligned}
    \end{equation}
    and $\E[BCD|\calF]$ can be calculated in a similar way. Thus, the
    proof is completed.
  \end{proof}

  Lemma~\ref{lem:app-all-observation-noise-N3} implies that we need to calculate
  \begin{equation}
    \frac{1}{k_n^2 \Delta_n^2}
    \sum_{k\in I} \sum_{l\in I}
    \E\Bigl(cABD + aBCD \Big| \calF\Bigr)
    =
    \frac{2}{k_n^2 \Delta_n^2}
    \sum_{k\in I} \sum_{l\in I}
    \E\Bigl(cABD \Big| \calF\Bigr).
  \end{equation}

  First, notice that
  \begin{equation}
    \begin{aligned}
      &
      \sum_{i,\alpha,\beta}
      \nu_2^2
      \left[ \theta_{A}^{i} \theta_{B}^{\alpha} \psi_{D}^{\beta} \right]
      \left( \delta_{i\beta}^{(k+1)} \delta_{\alpha, \beta+1}^{(k+2)} +
      \delta_{\alpha\beta}^{(k)} \delta_{i, \beta+1}^{(k-1)} \right) \\
      &=
      \nu_2^2
      \Bigl\{
        \theta_A^{k+1} \theta_B^{k+2} \psi_D^{k+1}
        + \theta_A^{k} \theta_B^{k-1} \psi_D^{k-1}
      \Bigr\} \\
      &=
      \nu_2^2
      \Bigl\{
        (+1) (-2\Delta X_{k+2}) \psi_D^{k+1}
        + (-1) (-2\Delta X_{k-2}) \psi_D^{k-1}
      \Bigr\},
    \end{aligned}
  \end{equation}
  and $\psi_D^{k+1}$ as well as $\psi_D^{k-1}$ can be written as
  \begin{equation}
    \begin{aligned}
      \psi_D^{k+1}
      &=
      +2 I(k+3\leq l \leq k+k_n+2)
      -2 I(k-k_n \leq l \leq k-1), \\
      \psi_D^{k-1}
      &=
      +2 I(k+1\leq l \leq k+k_n)
      -2 I(k-k_n-2 \leq l \leq k-3),
    \end{aligned}
  \end{equation}
  so
  \begin{equation}
    \begin{aligned}
      &
      \sum_{k\in I} \sum_{l\in I}
      c
      \sum_{i,\alpha,\beta}
      \nu_2^2
      \left[ \theta_{A}^{i} \theta_{B}^{\alpha} \psi_{D}^{\beta} \right]
      \left( \delta_{i\beta}^{(k+1)} \delta_{\alpha, \beta+1}^{(k+2)} +
      \delta_{\alpha\beta}^{(k)} \delta_{i, \beta+1}^{(k-1)} \right) \\
      &=
      2\nu_2^2
      \Biggl\{
        - \sum_{k\in I} \sum_{l\in I}
        (\Delta X_l) (\Delta X_{k+2}) \psi_D^{k+1}
        + \sum_{k\in I} \sum_{l\in I}
        (\Delta X_l) (\Delta X_{k-2}) \psi_D^{k-1}
      \Biggr\} \\
      &=
      4\nu_2^2
      \Biggl\{
        - \sum_{k\in I} \sum_{l\in I}
        (\Delta X_l) (\Delta X_{k+2}) I(k+3\leq l \leq k+k_n+2) \\
        &\qquad\quad\;
        + \sum_{k\in I} \sum_{l\in I}
        (\Delta X_l) (\Delta X_{k+2}) I(k-k_n \leq l \leq k-1) \\
        &\qquad\quad\;
        + \sum_{k\in I} \sum_{l\in I}
        (\Delta X_l) (\Delta X_{k-2}) I(k+1\leq l \leq k+k_n) \\
        &\qquad\quad\;
        - \sum_{k\in I} \sum_{l\in I}
        (\Delta X_l) (\Delta X_{k-2}) I(k-k_n-2 \leq l \leq k-3)
      \Biggr\} \\
      &=
      4\nu_2^2
      \sum_{k=k_n+2}^{n-k_n-3}
      \Bigl\{
        - (\Delta X_{k+2}) (X_{k+k_n+3} - X_{k+3})
        + (\Delta X_{k+2}) (X_{k} - X_{k-k_n}) \\
        &\qquad\qquad\qquad\,
        + (\Delta X_{k-2}) (X_{k+k_n+1} - X_{k+1})
        - (\Delta X_{k-2}) (X_{k-2} - X_{k-k_n-2})
      \Bigr\}\\
      &\qquad
      + \text{(edge terms)} \\
      &\approx
      4\nu_2^2
      \sum_{k=k_n+2}^{n-k_n-3}
      (\Delta X_k)
      \Bigl(
        \Delta_2 X_{k+k_n+1} - \Delta_2 X_{k+1} + \Delta_2 X_{k-2} -
        \Delta_2 X_{k-k_n-2}
      \Bigr).
    \end{aligned}
  \end{equation}
  It is clear that the expectation of this term vanishes, for $(\Delta
  X_k)$ does not overlap with $(\Delta_2 X_{k+k_n+1} - \Delta_2 X_{k+1}
  + \Delta_2 X_{k-2} - \Delta_2 X_{k-k_n-2})$. In fact, a CLT for
  $\eta_n = \sum_{k=k_n+2}^{n-k_n-3} (\Delta X_k)(\Delta_2 X_{k+k_n+1}
  - \Delta_2 X_{k+1} + \Delta_2 X_{k-2} - \Delta_2 X_{k-k_n-2})$ is
  \begin{equation}
    \Delta_n^{1/2} \eta_n
    \convst
    \int_0^T 4\sigma_t^2 \ud W_t^{(\eta)}, \quad \text{as }n \to \infty,
  \end{equation}
  where $(W_{t}^{(\eta)})_{t\geq 0}$ is a Wiener process independent of
  $(W_t)_{t\geq 0}$, and the convergence mode is stable convergence.
  As a result,
  \begin{equation}\label{eq:all-observation-noise-N3-ABD-1}
    \Delta_n^{1/2}
    \sum_{k\in I} \sum_{l\in I}
    c
    \sum_{i,\alpha,\beta}
    \nu_2^2
    \left[ \theta_{A}^{i} \theta_{B}^{\alpha} \psi_{D}^{\beta} \right]
    \left( \delta_{i\beta}^{(k+1)} \delta_{\alpha, \beta+1}^{(k+2)} +
    \delta_{\alpha\beta}^{(k)} \delta_{i, \beta+1}^{(k-1)} \right)
    =
    O_p(1).
  \end{equation}

  Next, notice that
  \begin{equation}
    \begin{aligned}
      &
      \sum_{i, \alpha, \beta}
      \nu_2 \nu_3 \left[ \theta_{A}^{i} \phi_{B}^{\alpha}
      \psi_{D}^{\beta} \right]
      \left( \delta_{i\beta}^{(k+1)} \delta_{\alpha, \beta+1}^{(k+2)} +
      \delta_{\alpha\beta}^{(k)} \delta_{i, \beta+1}^{(k-1)} \right) \\
      &=
      \nu_2\nu_3
      \Bigl\{
        \theta_A^{k+1} \phi_B^{k+2} \psi_D^{k+1}
        + \theta_A^k \phi_B^{k-1} \psi_D^{k-1}
      \Bigr\} \\
      &=
      \nu_2\nu_3
      \Bigl\{
        (+1) (+1) \psi_D^{k+1}
        + (-1) (-1) \psi_D^{k-1}
      \Bigr\} \\
      &=
      2\nu_2\nu_3
      \Bigl\{
        I(k+3\leq l \leq k+k_n+2) - I(k-k_n \leq l \leq k-1) \\
        &\qquad\qquad\,
        + I(k+1\leq l \leq k+k_n) - I(k-k_n-2 \leq l \leq k-3)
      \Bigr\},
    \end{aligned}
  \end{equation}
  so
  \begin{equation}
    \begin{aligned}
      &
      \sum_{k\in I} \sum_{l\in I}
      c
      \sum_{i, \alpha, \beta}
      \nu_2 \nu_3 \left[ \theta_{A}^{i} \phi_{B}^{\alpha}
      \psi_{D}^{\beta} \right]
      \left( \delta_{i\beta}^{(k+1)} \delta_{\alpha, \beta+1}^{(k+2)} +
      \delta_{\alpha\beta}^{(k)} \delta_{i, \beta+1}^{(k-1)} \right) \\
      &=
      2\nu_2\nu_3
      \Biggl\{
        \sum_{k\in I} \sum_{l\in I}
        (\Delta X_l)
        \Bigl[
          I(k+3\leq l \leq k+k_n+2) - I(k-k_n \leq l \leq k-1) \\
          &\qquad\qquad\qquad\qquad\qquad
          + I(k+1\leq l \leq k+k_n) - I(k-k_n-2 \leq l \leq k-3)
        \Bigr]
      \Biggr\} \\
      &=
      2\nu_2\nu_3
      \Biggl\{
        \sum_{l\in I}
        (\Delta X_k)
        \sum_{k\in I}
        \Bigl[
          I(k+3\leq l \leq k+k_n+2) - I(k-k_n \leq l \leq k-1) \\
          &\qquad\qquad\qquad\qquad\qquad
          + I(k+1\leq l \leq k+k_n) - I(k-k_n-2 \leq l \leq k-3)
        \Bigr]
      \Biggr\}.
    \end{aligned}
  \end{equation}
  Different from some other summations in this section, in this
  summation, the center terms vanishes, while the edge terms dominate
  the result. After some careful calculation, we have
  \begin{equation}
    \begin{aligned}
      &
      \sum_{k\in I} \sum_{l\in I}
      c
      \sum_{i, \alpha, \beta}
      \nu_2 \nu_3 \left[ \theta_{A}^{i} \phi_{B}^{\alpha}
      \psi_{D}^{\beta} \right]
      \left( \delta_{i\beta}^{(k+1)} \delta_{\alpha, \beta+1}^{(k+2)} +
      \delta_{\alpha\beta}^{(k)} \delta_{i, \beta+1}^{(k-1)} \right) \\
      &=
      2\nu_2\nu_3
      \Biggl\{
        \Biggl(
          (-2k_n) \Delta X_{k_n+1}
          +(-2k_n+1) \Delta X_{k_n+2}
          +\sum_{q=1}^{k_n-1} (-2k_n+2q) \Delta X_{k_n+2+q}
          \\ &\qquad\qquad\qquad
          +(-1) \Delta X_{2k_n+2}
        \Biggr)
        +
        \Biggl(
          \Delta X_{n-2k_n-3}
          +\sum_{q=1}^{k_n-1} 2q \Delta X_{n-2k_n-3-q}
          \\ &\qquad\qquad\qquad
          +(2k_n-1) \Delta X_{n-k_n-3}
          +2k_n \Delta X_{n-k_n-2}
        \Biggr)
      \Biggr\}.
    \end{aligned}
  \end{equation}
  Of course, the expression vanishes if we omit the drift terms.
  However, in this case, it is safer to check it more carefully because
  of the special structure, where the drift terms at the start and end
  of the day are compared. Let the term inside the large bracket be
  denoted as $\xi_n$. The asymptotic mean of $\xi_n$ is given by
  $(\mu_T - \mu_0) k_n^2 \Delta_n$, and its asymptotic stochastic
  variance is $4 (\sigma_0^2 + \sigma_T^2) k_n^3 \Delta_n / 3$. Thus,
  we can conclude that the ``standardized version" of $\xi_n$ is tight.
  Specifically, we have
  \begin{equation}\label{eq:all-observation-noise-N3-ABD-2}
    k_n^{-3/2} \Delta_n^{-1/2}
    \sum_{k\in I} \sum_{l\in I}
    c
    \sum_{i, \alpha, \beta}
    \nu_2 \nu_3 \left[ \theta_{A}^{i} \phi_{B}^{\alpha} \psi_{D}^{\beta} \right]
    \left( \delta_{i\beta}^{(k+1)} \delta_{\alpha, \beta+1}^{(k+2)} +
    \delta_{\alpha\beta}^{(k)} \delta_{i, \beta+1}^{(k-1)} \right)
    =
    O_p(1).
  \end{equation}

  Therefore, by combining
  Equation~\eqref{eq:all-observation-noise-N3-ABD-1} and
  \eqref{eq:all-observation-noise-N3-ABD-2}, we can conclude that the
  contribution of the $N_3$ term to the variance due to noise satisfies
  \begin{equation}\label{eq:all-observation-noise-N3}
    k_n^{-3/2} \Delta_n^{-1/2}
    \sum_{k\in I} \sum_{l\in I}
    \E\Bigl(N_3(k, l) \Big| \calF\Bigr)
    =
    O_p(1).
  \end{equation}

  \paragraph*{Contribution of $N_4$}
  Similarly, we have the following lemma.

  \begin{lemma}\label{lem:app-all-observation-noise-N4}
    For $\E\left[ ABCD | \calF \right]$, we have
    \begin{equation}\label{eq:ABCD}
      \begin{aligned}
        \E\left[ ABCD | \calF \right]
        &=
        \sum_{i, j, \alpha, \beta} \Bigl\{
          \nu_2^{2} \left[
            \theta_{A}^{i} \theta_{B}^{\alpha} \theta_{C}^{j} \theta_{D}^{\beta}
          \right] \left(
            \delta_{ij} \delta_{\alpha\beta} + \delta_{i\beta} \delta_{\alpha j}
          \right)
          \\ &\qquad\quad\;\;
          + \nu_2 \nu_3 \left[
            \theta_{A}^{i} \phi_{B}^{\alpha} \theta_{C}^{j} \theta_{D}^{\beta}
            +
            \theta_{A}^{i} \theta_{B}^{\alpha} \theta_{C}^{j} \phi_{D}^{\beta}
          \right] \left(
            \delta_{ij} \delta_{\alpha\beta} + \delta_{i\beta} \delta_{\alpha j}
          \right)
          \\ &\qquad\quad\;\;
          + \nu_2 \nu_4 \left[
            \theta_{A}^{i} \phi_{B}^{\alpha} \theta_{C}^{j} \phi_{D}^{\beta}
          \right] \left(
            \delta_{ij} \delta_{\alpha\beta}
          \right)
          \\ &\qquad\quad\;\;
          + \nu_2^{3} \left[
            \theta_{A}^{i} \phi_{B}^{\alpha} \theta_{C}^{j} \phi_{D}^{\beta}
          \right] \left(
            \delta_{ij} \left(1 - \delta_{\alpha\beta}\right)
          \right)
          \\ &\qquad\quad\;\;
          + \nu_3^{2} \left[
            \theta_{A}^{i} \phi_{B}^{\alpha} \theta_{C}^{j} \phi_{D}^{\beta}
          \right] \left(
            \delta_{i\beta} \delta_{\alpha j}
          \right)
          \\ &\qquad\quad\;\;
          + \nu_2^{3}  \left[
            \theta_{A}^{i} \psi_{B}^{\alpha} \theta_{C}^{j} \psi_{D}^{\beta}
          \right] \Bigl(
            \delta_{ij} \delta_{\alpha\beta}
            + \delta_{i\beta}^{(k+1)} \delta_{\alpha, \beta+1}^{(k+2)}
            \delta_{\alpha+1, j}^{(k+3)}I(l=k+3)
            \\ &\qquad\qquad\qquad\qquad\qquad\qquad
            + \delta_{\alpha j}^{(l+1)} \delta_{\alpha+1,
            \beta}^{(l+2)} \delta_{i, \beta+1}^{(l+3)}I(l=k-3)
          \Bigr)
        \Bigr\}.
      \end{aligned}
    \end{equation}
  \end{lemma}

  \begin{proof}
    Notice that
    \begin{equation}
      ABCD
      =
      \sum_{i,j,\alpha,\beta}
      \theta_A^i \eps_i
      \Bigl(\theta_B^\alpha \eps_\alpha + \phi_B^\alpha \eps_\alpha^2 +
      \psi_B^\alpha \eps_\alpha \eps_{\alpha+1}\Bigr)
      \theta_C^j \eps_j
      \Bigl(\theta_D^\beta \eps_\beta + \phi_D^\beta \eps_\beta^2 +
      \psi_D^\beta \eps_\beta \eps_{\beta+1}\Bigr).
    \end{equation}
    There are nine types of possible combinations of noise terms, among
    which five types have non-zero expectations, and they are
    \begin{equation}
      \begin{aligned}
        \E[ \eps_{i} \eps_{\alpha} \eps_{j} \eps_{\beta} ]
        &= \nu_2^{2} \left(
          \delta_{ij} \delta_{\alpha\beta} + \delta_{i\beta} \delta_{\alpha j}
        \right), \\
        \E[ \eps_{i} \eps_{\alpha}^{2} \eps_{j} \eps_{\beta} ]
        &= \nu_2 \nu_3 \left(
          \delta_{ij} \delta_{\alpha\beta} + \delta_{i\beta} \delta_{\alpha j}
        \right), \\
        \E[ \eps_{i} \eps_{\alpha} \eps_{j} \eps_{\beta}^{2} ]
        &= \nu_2 \nu_3 \left(
          \delta_{ij} \delta_{\alpha\beta} + \delta_{i\beta} \delta_{\alpha j}
        \right), \\
        \E[ \eps_{i} \eps_{\alpha}^{2} \eps_{j} \eps_{\beta}^{2} ]
        &= \nu_2 \nu_4 \left(
          \delta_{ij} \delta_{\alpha\beta}
        \right) + \nu_2^{3} \left(
          \delta_{ij} \left( 1 - \delta_{\alpha\beta} \right)
        \right) + \nu_3^{2} \left(
          \delta_{i\beta} \delta_{\alpha j}
        \right), \\
        \E[ \eps_{i} \eps_{\alpha} \eps_{\alpha+1} \eps_{j}
        \eps_{\beta} \eps_{\beta+1} ]
        &= \nu_2^{3} \left(
          \delta_{ij} \delta_{\alpha\beta} + \delta_{i\beta}
          \delta_{\alpha, \beta+1} \delta_{\alpha+1, j} +
          \delta_{\alpha j} \delta_{\alpha+1, \beta} \delta_{i, \beta+1}
        \right).
      \end{aligned}
    \end{equation}
    The proof is completed by taking the coefficients into consideration.
  \end{proof}

  According to Lemma~\ref{lem:app-all-observation-noise-N4}, we need
  to calculate
  \begin{equation}
    \frac{1}{k_n^2 \Delta_n^2}
    \sum_{k\in I} \sum_{l\in I}
    \E\Bigl(ABCD \Big| \calF\Bigr).
  \end{equation}
  Before we proceed, it would be helpful to notice that there are two
  frequently present patterns of coefficients and Kronecker deltas,
  which imply some conditions on $k$ and $l$, and can be used to help
  simplify the calculation. Notice that
  \begin{equation}\label{eq:all-observation-tAi-tCi-dij}
    \begin{aligned}
      \sum_{i,j}
      \theta_A^i \theta_C^j \delta_{ij}
      &=
      \sum_{i=k}^{k+1} \sum_{j=l}^{l+1}
      \theta_A^i \theta_C^j \delta_{ij} I(|k-l|\leq 1)
      \\ &=
      (\theta_A^k \theta_C^k + \theta_A^{k+1} \theta_C^{k+1}) I(k=l)
      \\ & \quad
      + \theta_A^k \theta_C^k I(l=k-1)
      + \theta_A^{k+1} \theta_C^{k+1} I(l=k+1)
      \\ &=
      2I(k=l) - I(k=l-1) - I(k=l+1),
    \end{aligned}
  \end{equation}
  so for any coefficient vectors $\xi$ and $\eta$, we have the first
  typical pattern:
  \begin{equation}
    \begin{aligned}
      &
      \sum_{k\in I} \sum_{l\in I}
      \sum_{i,j,\alpha,\beta}
      \left[\theta_A^i \xi_B^\alpha \theta_C^j \eta_D^\beta\right]
      \left(\delta_{ij} \delta_{\alpha\beta}\right)
      \\ &=
      \sum_{k\in I} \sum_{l\in I}
      \Bigl\{
        \sum_\alpha \xi_B^\alpha \eta_D^\alpha
        \sum_{i,j}
        \theta_A^i \theta_C^j \delta_{ij}
      \Bigr\}
      \\ &=
      \sum_{k\in I}
      \Bigl\{
        2 \langle \xi_B, \eta_B \rangle
        - \langle \xi_{B_k}, \eta_{B_{k-1}} \rangle I\bigl((k-1)\in I\bigr)
        - \langle \xi_{B_k}, \eta_{B_{k+1}} \rangle I\bigl((k+1)\in I\bigr)
      \Bigr\},
    \end{aligned}
  \end{equation}
  The second typical pattern is
  \begin{equation}
    \begin{aligned}
      \sum_{k\in I} \sum_{l\in I}
      \sum_{i,j,\alpha,\beta}
      \left[\theta_A^i \xi_B^\alpha \theta_C^j \eta_D^\beta\right]
      \left(\delta_{i\beta} \delta_{\alpha j}\right)
      &=
      \sum_{k\in I} \sum_{l\in I}
      \Bigl\{
        \sum_i \theta_A^i \eta_D^i
        \sum_j \theta_C^i \xi_B^i
      \Bigr\}
      \\ &=
      \sum_{k\in I} \sum_{l\in I}
      \langle \theta_A, \eta_D \rangle \langle \theta_C, \xi_B \rangle.
    \end{aligned}
  \end{equation}
  These two patterns can be used to simplify the calculation.

  First, we analyze the terms without coefficients $\theta_B$ and
  $\theta_D$, as these terms only include constant coefficients, and
  thus are easier to handle and there is no need to take the
  expectation with respect to processes. They are also likely to have
  larger order than the terms with $\theta_B$ and $\theta_D$. Notice that
  \begin{align}
    \|\phi_B\|^2 &= 8k_n - 4, &\text{for all } &k \in I, \\
    \quad \langle \phi_{B_k}, \phi_{B_{k+1}} \rangle &= 8k_n - 8,
    &\text{for all } &k, k+1 \in I,
  \end{align}
  so we have
  \begin{equation}\label{eq:all-observation-noise-N4-1}
    \begin{aligned}
      &
      \sum_{k\in I} \sum_{l\in I}
      \sum_{i,j,\alpha,\beta}
      \nu_2 \nu_4 \left[
        \theta_{A}^{i} \phi_{B}^{\alpha} \theta_{C}^{j} \phi_{D}^{\beta}
      \right] \left(
        \delta_{ij} \delta_{\alpha\beta}
      \right)
      \\ &=
      \nu_2 \nu_4
      \sum_{k\in I}
      \Bigl\{
        2\|\phi_B\|^2
        - \langle \phi_{B_k}, \phi_{B_{k-1}} \rangle I\bigl((k-1)\in I\bigr)
        - \langle \phi_{B_k}, \phi_{B_{k+1}} \rangle I\bigl((k+1)\in I\bigr)
      \Bigr\}
      \\ &=
      \nu_2 \nu_4
      \sum_{k\in I}
      \Bigl\{
        2(8k_n - 4)
        - (8k_n - 8)
        - (8k_n - 8)
      \Bigr\}
      \\ &=
      \nu_2 \nu_4
      \Bigl(
        (n-2k_n-4)\times 8
        + 2\times 8k_n
      \Bigr)
      \\ &=
      \nu_2 \nu_4
      (8n - 32).
    \end{aligned}
  \end{equation}
  Recall that $\sum_\alpha \phi_B^\alpha = 0$, so we have
  \begin{equation}\label{eq:all-observation-noise-N4-2}
    \begin{aligned}
      &
      \sum_{k\in I} \sum_{l\in I}
      \sum_{i,j,\alpha,\beta}
      \nu_2^{3} \left[
        \theta_{A}^{i} \phi_{B}^{\alpha} \theta_{C}^{j} \phi_{D}^{\beta}
      \right] \left(
        \delta_{ij} \left(1 - \delta_{\alpha\beta}\right)
      \right)
      \\ &=
      \nu_2^{3}
      \sum_{k\in I} \sum_{l\in I}
      \Bigl\{
        \Bigl(\sum_\alpha \phi_B^\alpha\Bigr)
        \Bigl(\sum_\beta \phi_D^\beta\Bigr)
        \sum_{i,j} \theta_A^i \theta_C^j \delta_{ij}
        -
        \sum_{i,j,\alpha,\beta}
        \left[
          \theta_{A}^{i} \phi_{B}^{\alpha} \theta_{C}^{j} \phi_{D}^{\beta}
        \right] \left(
          \delta_{ij} \delta_{\alpha\beta}
        \right)
      \Bigr\}
      \\ &=
      -\nu_2^{3}
      \sum_{k\in I} \sum_{l\in I}
      \sum_{i,j,\alpha,\beta}
      \left[
        \theta_{A}^{i} \phi_{B}^{\alpha} \theta_{C}^{j} \phi_{D}^{\beta}
      \right] \left(
        \delta_{ij} \delta_{\alpha\beta}
      \right)
      \\ &=
      -\nu_2^{3}
      (8n - 32).
    \end{aligned}
  \end{equation}
  Recall Equation~\eqref{eq:theta_A-phi_D}. Because
  \begin{equation}
    \langle \theta_A, \phi_D\rangle
    =
    \langle \theta_C, \phi_B\rangle
    =
    I\Bigl(|k-l| \in \{1,2\}\Bigr) - I\Bigl(|k-l| \in \{k_n+1, k_n+2\}\Bigr),
  \end{equation}
  we can know that
  \begin{equation}
    \langle \theta_A, \phi_D\rangle \langle \theta_C, \phi_B\rangle
    =
    I\Bigl(|k-l| \in \{1,2,k_n+1, k_n+2\}\Bigr).
  \end{equation}
  Therefore, we have
  \begin{equation}\label{eq:all-observation-noise-N4-3}
    \begin{aligned}
      &
      \sum_{k\in I} \sum_{l\in I}
      \sum_{i,j,\alpha,\beta}
      \nu_3^{2} \left[
        \theta_{A}^{i} \phi_{B}^{\alpha} \theta_{C}^{j} \phi_{D}^{\beta}
      \right] \left(
        \delta_{i\beta} \delta_{\alpha j}
      \right)
      \\ &=
      \nu_3^{2}
      \sum_{k\in I} \sum_{l\in I}
      \langle \theta_A, \phi_D\rangle \langle \theta_C, \phi_B\rangle
      \\ &=
      \nu_3^{2}
      \sum_{k\in I} \sum_{l\in I}
      I\Bigl(|k-l| \in \{1,2,k_n+1, k_n+2\}\Bigr)
      \\ &=
      2 \nu_3^{2}
      \sum_{q\in \{1,2,k_n+1, k_n+2\}} (n-2k_n-2-q)
      \\ &=
      \nu_3^{2}
      (8n - 20k_n - 28).
    \end{aligned}
  \end{equation}
  Moreover, notice that
  \begin{align}
    \|\psi_B\|^2 &= 8k_n, &\text{for all } &k \in I, \\
    \quad \langle \psi_{B_k}, \psi_{B_{k+1}} \rangle &= 8k_n - 8,
    &\text{for all } &k, k+1 \in I,
  \end{align}
  so we have
  \begin{equation}\label{eq:all-observation-noise-N4-4}
    \begin{aligned}
      &
      \sum_{k\in I} \sum_{l\in I}
      \sum_{i,j,\alpha,\beta}
      \nu_2^{3}  \left[
        \theta_{A}^{i} \psi_{B}^{\alpha} \theta_{C}^{j} \psi_{D}^{\beta}
      \right] \left(
        \delta_{ij} \delta_{\alpha\beta}
      \right)
      \\ &=
      \nu_2^{3}
      \sum_{k\in I}
      \Bigl\{
        2\|\psi_B\|^2
        - \langle \psi_{B_k}, \psi_{B_{k-1}} \rangle I\bigl((k-1)\in I\bigr)
        - \langle \psi_{B_k}, \psi_{B_{k+1}} \rangle I\bigl((k+1)\in I\bigr)
      \Bigr\}
      \\ &=
      \nu_2^{3}
      \sum_{k\in I}
      \Bigl\{
        2\times 8k_n
        - (8k_n - 8)
        - (8k_n - 8)
      \Bigr\}
      \\ &=
      \nu_2^{3}
      \Bigl(
        (n-2k_n-4)\times 16
        + 2\times (8k_n+8)
      \Bigr)
      \\ &=
      \nu_2^{3}
      (16n - 16k_n - 48).
    \end{aligned}
  \end{equation}
  The last term without coefficients $\theta_B$ and $\theta_D$ is
  \begin{equation}\label{eq:all-observation-noise-N4-5}
    \begin{aligned}
      &
      \sum_{k\in I} \sum_{l\in I}
      \sum_{i,j,\alpha,\beta}
      \nu_2^{3}  \left[
        \theta_{A}^{i} \psi_{B}^{\alpha} \theta_{C}^{j} \psi_{D}^{\beta}
      \right] \Bigl(
        \delta_{i\beta}^{(k+1)} \delta_{\alpha, \beta+1}^{(k+2)}
        \delta_{\alpha+1, j}^{(k+3)}I(l=k+3)
        \\ &\qquad\qquad\qquad\qquad\qquad\qquad\qquad
        + \delta_{\alpha j}^{(l+1)} \delta_{\alpha+1, \beta}^{(l+2)}
        \delta_{i, \beta+1}^{(l+3)}I(l=k-3)
      \Bigr)
      \\ &=
      \nu_2^{3}
      \sum_{k\in I} \sum_{l\in I}
      \Bigl\{
        \theta_A^{k+1} \psi_B^{k+2} \theta_C^{l} \psi_D^{l-2} I(l=k+3)
        + \theta_A^{k} \psi_B^{k-2} \theta_C^{l+1} \psi_D^{l+2} I(l=k-3)
      \Bigr\}
      \\ &=
      4 \nu_2^{3}
      \sum_{k\in I} \sum_{l\in I}
      I(|k-l| = 3)
      =
      \nu_2^{3}
      (8n - 16k_n - 40).
    \end{aligned}
  \end{equation}

  Next, we consider terms with coefficients $\theta_B$ and
  $\theta_D$. Notice that
  \begin{align}
    \|\theta_B\|^2 &\approx 16\sigma_k^2 k_n\Delta_n, &\text{for all }
    &k \in I, \\
    \langle \theta_{B_k}, \theta_{B_{k+1}} \rangle &\approx
    16\sigma_k^2 (k_n-1)\Delta_n, &\text{for all } &k, k+1 \in I,
  \end{align}
  so
  \begin{equation}\label{eq:all-observation-noise-N4-6}
    \begin{aligned}
      &
      \sum_{k\in I} \sum_{l\in I}
      \sum_{i,j,\alpha,\beta}
      \nu_2^{2} \left[
        \theta_{A}^{i} \theta_{B}^{\alpha} \theta_{C}^{j} \theta_{D}^{\beta}
      \right] \left(
        \delta_{ij} \delta_{\alpha\beta}
      \right)
      \\ &=
      \nu_2^{2}
      \sum_{k\in I}
      \Bigl\{
        2 \|\theta_B\|^2
        - \langle \theta_{B_k}, \theta_{B_{k-1}} \rangle I\bigl((k-1)\in I\bigr)
        - \langle \theta_{B_k}, \theta_{B_{k+1}} \rangle I\bigl((k+1)\in I\bigr)
      \Bigr\}
      \\ &\approx
      \nu_2^{2}
      \sum_{k\in I}
      \Bigl\{
        2 \times 16\sigma_k^2 k_n\Delta_n
        - 16\sigma_k^2 (k_n-1)\Delta_n I\bigl((k-1)\in I\bigr)
        \\ &\qquad\qquad\qquad\qquad\qquad\qquad
        - 16\sigma_k^2 (k_n-1)\Delta_n I\bigl((k+1)\in I\bigr)
      \Bigr\}
      \\ &\approx
      \nu_2^{2}
      \sum_{k\in I}
      16\sigma_k^2 k_n\Delta_n
      \approx
      16 k_n \nu_2^{2} \int_0^T \sigma_t^2 \ud t.
    \end{aligned}
  \end{equation}
  Notice that any of $\langle \theta_A, \theta_D \rangle$ and $\langle
  \theta_B, \theta_C \rangle$ is a linear combination of $(\Delta
  X_r)$'s, and the $(\Delta X_r)$'s in two terms never overlap, so the
  ``standardized version'' vanishes:
  \begin{equation}\label{eq:all-observation-noise-N4-7}
    k_n^{-1/2} \Delta_n^{-1/2}
    \sum_{k\in I} \sum_{l\in I}
    \sum_{i,j,\alpha,\beta}
    \nu_2^{2} \left[
      \theta_{A}^{i} \theta_{B}^{\alpha} \theta_{C}^{j} \theta_{D}^{\beta}
    \right] \left(
      \delta_{i\beta} \delta_{\alpha j}
    \right)
    \convp
    0.
  \end{equation}
  For similar reason, because $\langle \phi_B, \theta_D \rangle$,
  $\langle \theta_B, \phi_D \rangle$, $\langle \theta_A, \phi_D
  \rangle$ and $\langle \theta_C, \phi_B \rangle$ only contain linear
  combinations $(\Delta X_r)$'s, we have
  \begin{equation}\label{eq:all-observation-noise-N4-8}
    k_n^{-1/2}
    \sum_{k\in I} \sum_{l\in I}
    \nu_2 \nu_3 \left[
      \theta_{A}^{i} \phi_{B}^{\alpha} \theta_{C}^{j} \theta_{D}^{\beta}
      +
      \theta_{A}^{i} \theta_{B}^{\alpha} \theta_{C}^{j} \phi_{D}^{\beta}
    \right] \left(
      \delta_{ij} \delta_{\alpha\beta} + \delta_{i\beta} \delta_{\alpha j}
    \right)
    \convp
    0.
  \end{equation}

  Therefore, by combing Equation~\eqref{eq:all-observation-noise-N4-1},
  \eqref{eq:all-observation-noise-N4-2},
  \eqref{eq:all-observation-noise-N4-3},
  \eqref{eq:all-observation-noise-N4-4},
  \eqref{eq:all-observation-noise-N4-5},
  \eqref{eq:all-observation-noise-N4-6},
  \eqref{eq:all-observation-noise-N4-7} and
  \eqref{eq:all-observation-noise-N4-8}, we can conclude that the
  contribution of the $N_4$ term to the variance due to noise is
  \begin{equation}\label{eq:all-observation-noise-N4}
    n^{-1}
    \sum_{k\in I} \sum_{l\in I}
    \E\Bigl(N_4(k, l) \Big| \calF\Bigr)
    \convp
    8\nu_2\nu_4 + 16\nu_2^3 + 8\nu_3^2.
  \end{equation}

  \paragraph*{Total Variance Due to Noise}
  Finally, by combining Equation~\eqref{eq:all-observation-noise-N2},
  \eqref{eq:all-observation-noise-N3} and
  \eqref{eq:all-observation-noise-N4}, we can see that $N_4$ dominates
  the variance due to noise, so
  \begin{equation}
    n^{-1}
    \sum_{m=1}^{4}
    \sum_{k \in I} \sum_{l \in I}
    \E \Bigl(
      N_{m}(k, l)
    \Big| \calF \Bigr)
    \convp
    8\nu_2\nu_4 + 16\nu_2^3 + 8\nu_3^2.
  \end{equation}
  Therefore, according to
  Equation~\eqref{all-observation-noise-variance-four-terms}, the
  variance due to noise satisfies that
  \begin{equation}
    \Delta_n^3 k_n^2
    \var\bigl( \esty^{\rm (all)}_T \big| \calF \bigr)
    \convp
    (8 \nu_2 \nu_4 + 16 \nu_2^3 + 8 \nu_3^2) T,
  \end{equation}
  as $n \to \infty$.

  \paragraph*{A Corrected Version}
  A corrected version of variance due to noise under smaller noise can
  be obtained as a by-product of the analysis above. Combining
  Equation~\eqref{eq:all-observation-noise-N2-AC},
  Equation~\eqref{eq:all-observation-noise-N2-BD} and
  Equation~\eqref{eq:all-observation-noise-N4-6}, we have
  \begin{equation}
    \Delta_n^3 k_n^2
    \var\bigl( \esty^{\rm (all)}_T \big| \calF \bigr)
    \convp
    \biggl[
      (8\nu_2\nu_4 + 16\nu_2^3 + 8\nu_3^2)
      + \frac{k_n}{n^2} \cdot (8\nu_4 + 16\nu_2^2) \int_0^T \sigma_t^2 \ud t
      + \frac{1}{n^2} \cdot 8T \nu_2 \int_0^T \sigma_t^4 \ud t
    \biggr]
    T.
  \end{equation}

  \subsubsection{A Central Limit Theorem}

  A CLT for $\bigl( \esty_T^{\rm (all)} - \estx_T^{\rm (all)} \bigr)$
  is established for subsequent analysis. Let
  $\underline{\true}_T^{\rm (all)}$ denote the pure noise version of
  the all-observation estimator (by applying the estimator to the
  noise sequence $\{\eps_i\}_{i=0}^n$). Note that the increments of
  $\underline{\true}_T$ are stationary and have finite variances.
  According to the previous analysis and the CLT for $m$-dependent
  sequence of random variables \citep[see Theorem~2.8.1
  in][]{lehmann1999ElementsLargesampleTheory}, one can establish that
  as $n\to\infty$,
  \begin{align}
    \Delta_n^{3/2} k_n \underline{\true}_T^{\rm (all)}
    \convd
    \normal(0, (8 \nu_2 \nu_4 + 16 \nu_2^3 + 8 \nu_3^2)T).
  \end{align}
  Moreover, since
  \begin{align}
    \Delta_n^{3/2} k_n
    \Bigl(
      \esty_T^{\rm (all)}
      - \estx_T^{\rm (all)}
      - \underline{\true}_T^{\rm (all)}
    \Bigr)
    =
    O_p \left(\frac{k_n^{1/2}}{n}\right),
  \end{align}
  and therefore converges to 0 in probability. By Slutsky's Lemma, it
  follows that, as $n\to\infty$,
  \begin{align}
    \Delta_n^{3/2} k_n
    \bigl(
      \esty_T^{\rm (all)} - \estx_T^{\rm (all)}
    \bigr)
    \convd
    \normal(0, (8 \nu_2 \nu_4 + 16 \nu_2^3 + 8 \nu_3^2)T).
  \end{align}

  The above convergence in law can be easily strengthened as an
  convergence $\calF$-stably in law. Simply ``realize'' the limit as a
  random variable on the same probability space of noise random
  variables, which is independent of $(\Omega, \calF, \pr)$, and the
  $\calF$-stable convergence follows.

  \subsection{Proof of Proposition~\ref{prop:SALE-noise}}\label{app:SALE-noise}

  Consider the covariance between different subsample estimators:
  \begin{equation}
    \var\left( \esty^{(H_n)}_T \middle| \calF \right)
    =
    H_n^{-2} \sum_{h=1}^{H_n} \sum_{g=1}^{H_n}
    \cov\left(\left.
      \esty^{(H_n, h)}_T,
      \esty^{(H_n, g)}_T
    \right| \calF \right).
  \end{equation}
  Under Assumption~\ref{ass:noise}\ref{ass:noise-iid}, by
  Proposition~\ref{prop:all-observation-noise}, it follows that, as
  $n\to\infty$,
  \begin{equation}
    \Delta_n^3 H_n^4 k_n^2
    \var\left( \esty^{(H_n)}_T \middle| \calF \right)
    \convp
    (8\nu_2\nu_4 + 16\nu_2^3 + 8\nu_3^2) T.
  \end{equation}
  Under Assumption~\ref{ass:noise}\ref{ass:noise-dep}, given a scale
  $H_n > 2q$, consider two subsamples $(H_n, h)$ and $(H_n, g)$ with
  $|g-h| < q$. Let $I_{H_n} = \cap_{h=1}^{H_n} I_{H_n, h} = \{k_n+1,
  \dotsc, n_{H_n}-k_n-2\}$, where $n_{H_n} = \lfloor (n+1)/H_n \rfloor
  - 1$. Rewrite the subsample estimators as (edge effects ignored)
  \begin{align}
    \esty^{(H_n, h)}_T
    &=
    \sum_{k\in I_{H_n}}
    (\Delta_{H_n} Y_{k_{H_n, h}})
    \: \widehat{\delta}(k_{H_n, h}, H_n, k_n)
    =
    \frac{1}{k_n H_n \Delta_n}
    \sum_{k\in I_{H_n}} u_{k, h}, \\
    \esty^{(H_n, g)}_T
    &=
    \sum_{l\in I_{H_n}}
    (\Delta_{H_n} Y_{l_{H_n, h}})
    \: \widehat{\delta}(l_{H_n, h}, H_n, k_n)
    =
    \frac{1}{k_n H_n \Delta_n}
    \sum_{l\in I_{H_n}} u_{l, g},
  \end{align}
  The increment are decomposed similarly as
  \begin{align}
    u_{k, h}
    &=
    \left( a_{k, h} + A_{k, h} \right) \left( b_{k, h} + B_{k, h} \right),
  \end{align}
  where
  \begin{align}
    \textcolor{blue}{a_{k, h}} + \textcolor{red}{A_{k, h}}
    &=
    (\Delta_{H_n} Y_{k_{H_n, h}})
    =
    \textcolor{blue}{
      (\Delta_{H_n} X_{k_{H_n, h}})
    }
    +
    \textcolor{red}{\langle \theta_{A_{k, h}}, \eps^{(h)} \rangle},
    \\
    \textcolor{blue}{b_{k, h}} + \textcolor{red}{B_{k, h}}
    &=
    \widehat{\delta}(k_{H_n, h}, H_n, k_n)
    \notag
    \\
    &=
    \textcolor{blue}{
      \overline{\delta}(k_{H_n, h}, H_n, k_n)
    }
    +
    \textcolor{red}{
      \langle \theta_{B_{k, h}}, \eps^{(h)} \rangle
      + \langle \phi_{B_{k, h}}, \eps^{(h)2} \rangle
      + \langle \psi_{B_{k, h}}, \eps\eps^{(h)}_{+} \rangle
    }.
  \end{align}
  Here, the noise vector is defined by
  \begin{equation}
    (\eps^{(h)})_j = \eps_{j_h}, \quad
    (\eps^{(h)2})_j = (\eps_{j_h})^2, \quad
    (\eps\eps^{(h)}_{+})_j = \eps_{j_h} \eps_{(j+1)_h}I(j\neq n_{H_n, h}),
  \end{equation}
  and the non-zero values of $\theta_{A_{i, h}}$, $\theta_{B_{i, h}}$,
  $\phi_{B_{i, h}}$ and $\psi_{B_{i, h}}$ are listed in
  Table~\ref{tab:coeff-subgrid}, where the column represents the
  observation indices in subsample $(H_n, h)$.

  \begin{table}[!ht]
    \centering
    \caption{Values of coefficient vectors in the subsample estimator}
    \label{tab:coeff-subgrid}
    \resizebox{\textwidth}{!}{
      \begin{tabular}{ccccccccccccc}
        \toprule
        ~ & $i-k_n-1$ & $i-k_n$ & $\cdots$ & $i-2$ & $i-1$ & $i$ &
        $i+1$ & $i+2$ & $i+3$ & $\cdots$ & $i+k_n+1$ & $i+k_n+2$ \\
        \midrule
        $\theta_{A_{i, h}}$ & \textcolor{blue}{$0$} &
        \textcolor{blue}{$0$} & $\cdots$ & \textcolor{blue}{$0$} &
        \textcolor{blue}{$0$} & $-1$ & $+1$ & \textcolor{blue}{$0$} &
        \textcolor{blue}{$0$} & $\cdots$ & \textcolor{blue}{$0$} &
        \textcolor{blue}{$0$} \\
        $\theta_{B_{i, h}}$ & $+2\Delta_{H_n} X_{i-k_n H_n-H_n}$ &
        $+2\Delta_{H_n}^2 X_{i-k_n H_n-H_n}$ & $\cdots$ &
        $+2\Delta_{H_n}^2 X_{i-3 H_n}$ & $-2\Delta_{H_n} X_{i-2H_{n}}$
        & \textcolor{blue}{$0$} & \textcolor{blue}{$0$} &
        $-2\Delta_{H_n} X_{i+2H_{n}}$ &
        $-2\Delta_{H_n}^2 X_{i+2 H_n}$ & $\cdots$ &
        $-2\Delta_{H_n}^2 X_{i+k_n H_n}$ & $+2\Delta_{H_n} X_{i+k_n H_n+H_n}$ \\
        $\phi_{B_{i, h}}$ & $-1$ & $-2$ & $\cdots$ & $-2$ & $-1$ &
        \textcolor{blue}{$0$} & \textcolor{blue}{$0$} & $+1$ & $+2$ &
        $\cdots$ & $+2$ & $+1$ \\
        $\psi_{B_{i, h}}$ & $+2$ & $+2$ & $\cdots$ & $+2$ &
        \textcolor{blue}{$0$} & \textcolor{blue}{$0$} &
        \textcolor{blue}{$0$} & $-2$ & $-2$ & $\cdots$ & $-2$ &
        \textcolor{blue}{$0$} \\
        \bottomrule
      \end{tabular}
    }
  \end{table}

  The proof differs from Section~\ref{sec:app-all-observation-noise} in
  two ways: (i) the form of coefficient vector $\theta_{B_{i, h}}$ is
  different (but is finally negligible); and (ii) the noise random
  variables in different subsamples satisfy that
  \begin{align}
    \eps_{i_h} \perp \eps_{j_l}, \quad \text{if } i \neq j,
  \end{align}
  and specifically, we have
  \begin{align}
    \E [ \eps_{i_h} \eps_{j_g} ]
    &=
    \nu_2 \rho_2(g-h) \delta_{ij}, \\
    \E [ \eps_{i_h} \eps_{j_g}^{2} ]
    &=
    \sqrt{\nu_2 \left(\nu_4 - \nu_2^2\right)}
    \rho_3(g-h) \delta_{ij}, \\
    \E [ \eps_{i_h}^{2} \eps_{j_g}^{2} ]
    &=
    \nu_2^2 + \left(\nu_4 - \nu_2^2\right) \rho_4(g-h) \delta_{ij}.
  \end{align}
  With these differences, the covariance can be obtained: as $n\to\infty$,
  \begin{align}
    &
    \Delta_n^3 H_n^3 k_n^2
    \cov\left(
      \esty^{(H_n, h)}_T,
      \esty^{(H_n, g)}_T
    \middle| \calF \right)
    \notag
    \\
    & \convp
    \Bigl(
      8 \nu_2 \bigl(\nu_4 - \nu_2^2\bigr) \rho_2(g-h)\rho_4(g-h)
      + 24 \nu_2^3 \rho_2^3(g-h)
      \notag
      \\
      & \qquad
      + 8 \nu_2 \bigl(\nu_4 - \nu_2^2\bigr) \rho_3(g-h) \rho_3(h-g)
    \Bigr)
    T.
  \end{align}
  Consequently, as $n\to\infty$,
  \begin{equation}\label{eq:E-Var-tilde}
    \Delta_n^3 H_n^4 k_n^2
    \var \left( \esty^{(H_n)}_T \middle| \calF \right)
    \\
    \convp
    \Phi T,
  \end{equation}
  where
  \begin{equation}
    \Phi = 8 \nu_2 \left(\nu_4 - \nu_2^2\right)
    \sum_{l=-q}^{q} \Bigl( \rho_2(l)\rho_4(l) + \rho_3(l) \rho_3(-l) \Bigr)
    + 24 \nu_2^3 \sum_{l=-q}^{q} \rho_2^3(l).
  \end{equation}
  This completes the proof.

  Similarly, a more precise version is
  \begin{equation}\label{eq:SALE-noise-variance-corrected}
    \Delta_n^3 H_n^4 k_n^2
    \var \left( \esty^{(H_n)}_T \middle| \calF \right)
    \\
    \convp
    \Phi' T,
  \end{equation}
  where
  \begin{equation}
    \begin{aligned}
      \Phi'
      &=
      8 \nu_2 \left(\nu_4 - \nu_2^2\right)
      \sum_{l=-q}^{q} \Bigl( \rho_2(l)\rho_4(l) + \rho_3(l) \rho_3(-l) \Bigr)
      + 24 \nu_2^3 \sum_{l=-q}^{q} \rho_2^3(l)
      \\ & \quad
      + \frac{k_nH_n}{n} \cdot \Biggl[
        8 (\nu_4 - \nu_2^2) \int_0^T \sigma_t^2 \ud t
        \sum_{l=-q}^{q} \left(1-\frac{|d|}{H_n}\right)
        \left(
          \rho_4(d) + \frac{\nu_2^2}{\nu_4-\nu_2^2} \Bigl(1-\rho_2^2(d)\Bigr)
        \right)
      \Biggr]
      \\ & \quad
      + \frac{k_nH_n}{n} \cdot \Biggl[
        24 \nu_2^2 \int_0^T \sigma_t^2 \ud t
        \sum_{l=-q}^{q} \left(1-\frac{|d|}{H_n}\right) \rho_2^2(d)
      \Biggr]
      \\ & \quad
      + \frac{H_n^2}{n^2} \cdot \Biggl[
        8T \nu_2 \int_0^T \sigma_t^4 \ud t
        \sum_{l=-q}^{q} \left(1-\frac{|d|}{H_n}\right)^2 \rho_2(d)
      \Biggr].
    \end{aligned}
  \end{equation}
  A corresponding $\calF$-stable convergence CLT can also be established.

  \subsection{Proof of Proposition~\ref{prop:MSLE-noise}}
  \label{app:MSLE-noise}

  \paragraph*{The Independent Noise Case}

  Under Assumption~\ref{ass:noise}\ref{ass:noise-iid}, the correlation
  due to noise between SALE at different scales are negligible except
  for a special case, where one scale is exactly double the other. This
  can be established with the same way of calculating the contribution
  of noise terms as above. The calculation is tedious, so we will only
  explain the intuition instead of expanding the details. For example,
  consider the increment from two scales $H_p \geq H_q$.
  \begin{align}
    \esty^{(H_p)}_T
    &=
    \frac{1}{k_p H_p^2 \Delta_n} \sum_{k=(k_p+1)H_p}^{n-(k_p+2)H_p}
    u_k^{(H_p)}, \\
    \esty^{(H_q)}_T
    &=
    \frac{1}{k_q H_q^2 \Delta_n} \sum_{l=(k_q+1)H_q}^{n-(k_q+2)H_q} u_l^{(H_q)},
  \end{align}
  where the increment can be decomposed similarly, for example,
  $u_k^{(H_p)}$ as $(a_k+A_k)(b_k+B_k)$, and $u_l^{(H_q)}$ as
  $(c_l+C_l)(d_l+D_l)$. The dominant terms can be established similar
  as in Lemma~\ref{lem:app-all-observation-noise-N4}, where we only
  consider terms without $\theta_B$ and $\theta_D$ as they become
  negligible as $n \to \infty$. However, for $H_p > H_q$, the only
  terms that remain after summation over $l$ is
  \begin{equation}
    \sum_{l=(k_q+1)H_q}^{n-(k_q+2)H_q}
    \sum_{i,j,\alpha,\beta}
    \nu_2(\nu_4-\nu_2^2) \left[\theta_A^i \phi_B^\alpha \theta_C^j
    \phi_B^\beta\right]
    (\delta_{ij} \delta_{\alpha\beta})
    =
    2 \nu_2(\nu_4-\nu_2^2) I(H_p = 2H_q),
  \end{equation}
  so $H_p = 2H_q$ is a very special situation. The contribution of
  elements from $\phi_B$ and $\phi_D$ at $l-2,l-1,l+2$, and $l+3$ adds
  up to 2 in total, while the contribution at $l-m$ and $l+m+1$ for any
  $m\geq 3$ cancels out. The former distinguishes $H_p = 2H_q$ from any
  other cases, and the variance due to noise of the MSLE estimator
  follows. This completes the proof.

  For the same reason, shifting the spot volatility estimation window
  outward by one time interval (shown in
  Figure~\ref{fig:base-estimators-3}) can completely eliminate this
  distribution, so the correlation between different scales become
  totally negligible.

  \paragraph*{The Dependent Noise Case}

  Although the closed-form expression of the covariance between
  different scales is difficult to derive, a multi-scale alternative of
  Lemma~\ref{lem:app-all-observation-noise-N4} (dropping terms with
  $\theta_B$ and $\theta_D$) can be used as an algorithm to explicitly
  calculate the covariance due to noise when the noise is dependent.
  The idea is straightforward:
  \begin{enumerate}
    \item For each given pair of $k$ and $l$, the contribution can be
      calculated with this multi-scale and dependent alternative of
      Lemma~\ref{lem:app-all-observation-noise-N4}.
    \item Given a fixed $k$, for the summation over $l$, we need to
      consider all $l$'s that may have potential contribution to the
      final result. Specifically, if the dependent level of noise is
      $\tilde{q}$, the following $l$'s are considered:
      \begin{equation}
        l \in \{
          k-(k_p+1)H_p-\tilde{q}-(k_q+2)H_q,
          \dotsc,
          k+(k_p+2)H_p+\tilde{q}+(k_q+1)H_p
        \}.
      \end{equation}
    \item The summation over $k$ is substituted by simply multiplying $n$.
    \item The result is normalized by $k_p H_p^2 \Delta_n$ and $k_q
      H_q^2 \Delta_n$.
  \end{enumerate}
  The code has been implemented and validated in the simulation study.

  \subsection{Proof of Theorem~\ref{thm:SALE-clt-noisefree}}
  \label{app:SALE-clt-noisefree}

  Let $\overline{\delta} (\cdot)$ and $\overline{\sigma}_\pm^2 (\cdot)$
  denote the noise-free versions of $\widehat{\delta} (\cdot)$ and
  $\widehat{\sigma}_\pm^2 (\cdot)$, respectively. Recall that the
  noise-free SALE estimator can be written as
  \begin{align}
    \estx_T^{(H_n)}
    =
    \frac{1}{H_n}
    \sum_{i=(k_n+1)H_n}^{n-(k_n+2)H_n}
    (\Delta_{H_n} X_i)
    \: \overline{\delta} (i, H_n, k_n)
    =
    \frac{1}{H_n}
    \sum_{i=(k_n+1)H_n}^{n-(k_n+2)H_n}
    V_{H_n, i},
  \end{align}
  where
  \begin{align}
    V_{H_n, i}
    =
    (\Delta_{H_n} X_i)
    \bigl(
      \overline{\sigma}_+^2 (i, H_n, k_n)
      -
      \overline{\sigma}_-^2 (i, H_n, k_n)
    \bigr)
  \end{align}

  Jacod's stable central limit theorem for partial sums of triangular
  arrays \citep{jacod1997ContinuousConditionalGaussian} \citep[also
  refer to Theorem~2.6 in][]{podolskij2010UnderstandingLimitTheorems}
  is used in the proof, and an essential step is to calculate the
  predictable quadratic variation of the error process $u_n (\estx -
  \true_t)$, where the normalizing sequence is taken as $u_n =
  \sqrt{k_n \land (k_n H_n \Delta_n)^{-1}}$. However, since the SALE
  estimator contains an increasing and overlapping window for spot
  volatility estimation, there is a lack of conditional independence
  of the summands. To overcome this challenge, we adopt the ``big
  block'' technique to establish some conditional independence
  \citep[for example, see Section~12.2
  in][]{jacod2012DiscretizationProcesses}. Specifically, we divide
  the time interval $[0, T]$ into big blocks of size $L =
  (K+2k_{n}+3)H_{n}\Delta_{n}$, where $K = K_n$ is an appropriately
  selected sequence. As illustrated in Figure~\ref{fig:block}, the
  time span of the $\nu$th block is $[L\nu, L\nu + L)$, where $\nu =
  0, 1, \dotsc, \lfloor T/L \rfloor - 1$.

  \begin{figure}[!ht]
    \begin{center}
      \begin{tikzpicture}
        \def\xStart{-5}
        \def\xEnd{11}
        \draw[-latex] (\xStart,0) -- (\xEnd,0);
        \node[above] at (\xEnd-0.1, 0) {\scriptsize{$t$}};
        \foreach \x in {-4, 7} {
          \draw (\x, -1) -- (\x, 1);
        }
        \node[below] at (-4, -1) {\footnotesize{$\calF_{L\nu}$}};
        \node[below] at (7, -1) {\footnotesize{$\calF_{L(\nu+1)}$}};
        \foreach \x in {-2, -1.5, -1, -0.5, 3.5, 4, 4.5, 5, 9, 9.5, 10} {
          \draw (\x, 0) -- (\x, 0.1);
        }
        \node at (1.5, 0.1) {\scriptsize{$\cdots\cdots$}};
        \foreach \x/\y in {-4/-2, 7/9} {
          \draw[decorate,decoration={brace,amplitude=5pt,mirror}]
          (\x+0.05,-0.1) -- (\y-0.05,-0.1);
          \node[below] at (\x+1, -0.2)
          {\scriptsize{$(k_{n}+1)H_{n}\Delta_{n}$}};
        }
        \draw[decorate,decoration={brace,amplitude=5pt,mirror}]
        (4.5+0.05,-0.1) -- (7-0.05,-0.1);
        \node[below] at (5.75, -0.2) {\scriptsize{$(k_{n}+2)H_{n}\Delta_{n}$}};
        \draw[decorate,decoration={brace,amplitude=5pt,mirror}]
        (-2+0.05,-0.1) -- (4.5-0.05,-0.1);
        \node[below] at (1.5, -0.2) {\scriptsize{$KH_{n}\Delta_{n}$}};
        \draw[-latex] (-0.5, 0.5) -- (-4, 0.5);
        \draw[-latex] (3.5, 0.5) -- (7, 0.5);
        \node at (1.5, 0.5) {\footnotesize{$L = (K+2k_{n}+3)H_{n}\Delta_{n}$}};
        \foreach \xStart/\xEnd in {-2/4.5, 9/10.5} {
          \fill [fill=blue, opacity=0.2] (\xStart,0.05) -- (\xStart,
          -0.05) -- (\xEnd, -0.05) -- (\xEnd, 0.05) -- cycle;
        }
        \foreach \xStart/\xEnd in {-4.5/-2, 4.5/9} {
          \fill [fill=red, opacity=0.2] (\xStart,0.05) -- (\xStart,
          -0.05) -- (\xEnd, -0.05) -- (\xEnd, 0.05) -- cycle;
        }
      \end{tikzpicture}
    \end{center}
    \caption{The partition within the $\nu$th big block.}
    \label{fig:block}
  \end{figure}

  In the $\nu$th big block, the blue interval in the middle corresponds
  to the time span $[L\nu+(k_n+1)H_n\Delta_n,
  L\nu+(K+k_n+1)H_n\Delta_n)$, with a length of $KH_n\Delta_n$. The red
  intervals on both sides are $(k_n+1)H_n\Delta_n$ and
  $(k_n+2)H_n\Delta_n$ long, respectively. Within this big block, we
  only consider the $KH_n$ increment terms $V_i$'s of which the time
  point $i\Delta_n$ lies within the left-closed, right-open blue
  interval, and their indices are $J(\nu) = \{L\nu/\Delta_n+(k_n+1)H_n,
  \dotsc, L\nu/\Delta_n+(K+k_n+1)H_n-1\}$. Note that all observations
  contained in these $V_i$'s are within this big block, and thus
  $\calF_{L(\nu+1)}$-measurable. An appropriate choice of $K$ satisfies
  two conditions: (i) $K/k_n \to \infty$ as $n \to \infty$, so that the
  $(2k_n+3)H_n$ terms $V_i$'s in the red small blocks are negligible;
  and (ii) $K H_n \Delta_n \to 0$ as $n \to \infty$, so that the time
  span of big blocks is short enough for a local constant approximation.

  With these notations, the SALE estimator is further rewritten as
  \begin{equation}
    \estx^{(H_n)}_T
    =
    \frac{1}{H_n}
    \sum_{\nu=0}^{\lfloor T/L \rfloor - 1}
    \sum_{i\in J(\nu)} V_{H_n, i}
    + \text{(edge terms)}.
  \end{equation}
  The edge terms include the increments in the red small blocks, and
  the end-of-the-day increments beyond the last big block (if any).
  Obviously, in both cases, they are negligible. Ignoring them, and
  rewrite the error process as
  \begin{equation}
    u_n \Bigl( \estx^{(H_n)}_T - \true_T \Bigr)
    =
    \sum_{\nu=0}^{\lfloor T/L \rfloor - 1} U_\nu,
  \end{equation}
  where
  \begin{align}
    U_\nu
    &=
    u_n
    \left(
      \frac{1}{H_n} \sum_{i\in J(\nu)} V_{H_n, i}
      - \int_{L\nu}^{L(\nu+1)} \ud \true_t
    \right)
    - \text{(edge terms)}
    \notag
    \\ &=
    \frac{u_n}{H_n}
    \sum_{i\in J(\nu)}
    \left(
      V_{H_n, i}
      - \int_{i\Delta_n}^{(i+H_n)\Delta_n} \ud \true_t
    \right)
    \notag
    \\ &=
    U_{\nu, 1} + U_{\nu, 2}.
  \end{align}
  and
  \begin{align}
    U_{\nu, 1} &= \frac{u_n}{H_n} \sum_{i\in J(\nu)} \Bigl( V_{H_n, i}
    - (\Delta_{H_n} X_i) (\Delta_{H_n} \sigma_i^2) \Bigr), \\
    U_{\nu, 2} &= \frac{u_n}{H_n} \sum_{i\in J(\nu)} \left(
      (\Delta_{H_n} X_i) (\Delta_{H_n} \sigma_i^2) -
    \int_{i\Delta_n}^{(i+H_n)\Delta_n} \ud \true_t \right).
  \end{align}
  Here, $U_{\nu, 1}$ is the error due to spot volatility estimation,
  whereas $U_{\nu, 2}$ is the discretization error of ``oracle''
  leverage estimation. Note that
  \begin{align}
    U_{\nu, 1}
    =
    \frac{u_n}{H_n} \sum_{i\in J(\nu)}(\Delta_{H_n} X_i) (e_{H_n, i+} -
    e_{H_n, i-}).
  \end{align}
  where
  \begin{align}
    e_{H_n, i+}
    &=
    \overline{\sigma}_+^2 (i, H_n, k_n) - \sigma_{i+H_n}^2,
    \\
    e_{H_n, i-}
    &=
    \overline{\sigma}_-^2 (i, H_n, k_n) - \sigma_{i}^2.
  \end{align}
  The following lemma concerns these $e_{H_n, i\pm}$ terms.

  \begin{lemma}\label{lem:err-spot-volatility}
    Suppose that $i, j \in J(\nu)$ and $|i-j| < H_n$. As $n\to\infty$,
    \begin{align}
      \E \Bigl[
        e_{H_n, i\pm} e_{H_n, j\pm}
      \Big| \calF_{L\nu} \Bigr]
      &=
      \left(\frac{1}{k_n}\right) \cdot 2 \left[ \left(
        \frac{|i-j|}{H_n} \right)^2 + \left( 1 - \frac{|i-j|}{H_n}
      \right)^2 \right] \sigma_\nu^4
      \notag
      \\ &\qquad
      + (k_n H_n \Delta_n) \cdot \frac{1}{3}
      \left.\frac{\ud\langle\sigma^2, \sigma^2\rangle_t}{\ud
      t}\right|_{t=L\nu} + o_p(k_n^{-1} + k_n H_n \Delta_n),
      \\
      \E \Bigl[
        e_{H_n, i\pm} e_{H_n, j\mp}
      \Big| \calF_{L\nu} \Bigr]
      &=
      o_p(k_n^{-1} + k_n H_n \Delta_n).
      \label{eq:SALE-spot-error-cross}
    \end{align}
    Moreover, $u_n (e_{H_n, i+}, e_{H_n, i+})$ converges stably in
    law to a pair of mean-zero, independent normal random variables
    that are independent of $\calF$.
  \end{lemma}

  \begin{proof}
    To begin with, let
    \begin{align}
      e_{H_n, i+}
      & = (\overline{\sigma}_+^2(i, H_n, k_n) - \sigma_{i+2H_n}^2) +
      (\sigma_{i+2H_n}^2 - \sigma_{i+H_n}^2)
      = \tilde{e}_{H_n, i+} + \Delta_{H_n} \sigma_{i+H_n}^2, \\
      e_{H_n, i-}
      & = (\overline{\sigma}_-^2(i, H_n, k_n) - \sigma_{i-H_n}^2) +
      (\sigma_{i-H_n}^2 - \sigma_{i}^2)
      = \tilde{e}_{H_n, i-} - \Delta_{H_n} \sigma_{i-H_n}^2,
    \end{align}
    where
    \begin{align}
      &\quad
      (k_n H_n \Delta_n) \tilde{e}_{H_n, i+}
      \notag
      \\
      &=
      \sum_{k=2}^{k_n+1}
      \int_{(i+kH_n)\Delta_n}^{(i+(k+1)H_n)\Delta_n}
      2(X_r - X_{i+kH_n}) \ud X_r
      + (\sigma_r^2 - \sigma_{i+2H_n}^2) \ud r
      \notag
      \\ &=
      \Biggl(
        \sum_{k=2}^{k_n+1}
        \int_{(i+kH_n)\Delta_n}^{(i+(k+1)H_n)\Delta_n} 2(X_r -
        X_{i+kH_n}) \ud X_r
      \Biggr)
      +
      \int_{(i+2H_n)\Delta_n}^{(i+(k_n+2)H_n)\Delta_n}
      (\sigma_r^2 - \sigma_{i+2H_n}^2) \ud r
      \notag
      \\ &=
      \Biggl(
        \sum_{k=2}^{k_n+1} \xi_{i, k}
      \Biggr)
      + \eta_{i}.
    \end{align}
    Without loss of generality, take $i \leq j < i + H_n$. By
    Lemma~\ref{lem:overlap-expectation}, it follows that
    \begin{align}
      \E \Bigl[
        \xi_{i, k} \xi_{j, k}
      \Big| \calF_{(i+kH_n)\Delta_n} \Bigr]
      &=
      2\sigma_{(i+kH_n)\Delta_n}^4 \left(\left(1 -
      \frac{j-i}{H_n}\right) H_n \Delta_n\right)^2
      +
      o_p(H_n^2 \Delta_n^2), \\
      \E \Bigl[
        \xi_{i, k+1} \xi_{j, k}
      \Big| \calF_{(j+kH_n)\Delta_n} \Bigr]
      &=
      2\sigma_{(j+kH_n)\Delta_n}^4 \left(\left(\frac{j-i}{H_n}\right)
      H_n \Delta_n\right)^2
      +
      o_p(H_n^2 \Delta_n^2), \\
      \E \Bigl[
        \eta_{i} \eta_{j}
      \Big| \calF_{(i+2H_n)\Delta_n} \Bigr]
      &=
      \frac{1}{3} \left.\frac{\ud\langle\sigma^2,
      \sigma^2\rangle_t}{\ud t}\right|_{t=(i+2H_n)\Delta_n} (k_n H_n
      \Delta_n)^3 + o_p(k_n^3 H_n^3 \Delta_n^3).
    \end{align}
    The negligibility of cross terms such as $\E[\xi_{j, k} \eta_i |
    \calF_{(i+2H_n)\Delta_n}]$ follows from
    \begin{align}
      &
      \E\Bigl[\xi_{j, k} \eta_i \Big| \calF_{(i+2H_n)\Delta_n}\Bigr]
      \notag
      \\ &=
      \E \Biggl[
        \Biggl(
          \int_{(j+kH_n)\Delta_n}^{(j+(k+1)H_n)\Delta_n} 2(X_r -
          X_{j+kH_n}) \ud X_r
        \Biggr)
        \notag
        \\ &\qquad \qquad
        \Biggl(
          \int_{(i+2H_n)\Delta_n}^{(i+(k_n+2)H_n)\Delta_n} \Bigl[
          (i+(k_n+2)H_n)\Delta_n - r \Bigr] \ud \sigma_r^2
        \Biggr)
      \Bigg| \calF_{(i+2H_n)\Delta_n} \Biggr]
      \notag
      \\ &\approx
      \Bigl( (i+(k_n+2)H_n)\Delta_n - (j+kH_n)\Delta_n \Bigr)
      \left.\frac{\ud \langle X, \sigma^2 \rangle_t}{\ud
      t}\right|_{t=(i+2H_n)\Delta_n}
      \notag
      \\ &\qquad \qquad
      \E \Biggl[
        \int_{(j+kH_n)\Delta_n}^{(j+(k+1)H_n)\Delta_n \land
        (i+(k_n+2)H_n)\Delta_n}
        2(X_r - X_{j+kH_n}) \ud r
      \Bigg| \calF_{(i+2H_n)\Delta_n} \Biggr]
      \notag
      \\ &=
      O_p(k_n H_n \Delta_n \cdot k_n H_n \Delta_n \cdot H_n \Delta_n)
      \notag
      \\ &=
      O_p(k_n^2 H_n^3 \Delta_n^3).
    \end{align}
    Consequently, we have
    \begin{align}
      \E \Bigl[
        \tilde{e}_{H_n, i+} \tilde{e}_{H_n, j+}
      \Big| \calF_{L\nu} \Bigr]
      &=
      \frac{1}{k_n^2 H_n^2 \Delta_n^2}
      \E \Biggl[
        \Biggl(\sum_{k=2}^{k_n+1} \xi_{i, k} + \eta_i\Biggr)
        \Biggl(\sum_{l=2}^{k_n+1} \xi_{j, l} + \eta_j\Biggr)
      \Bigg| \calF_{L\nu} \Biggr]
      \notag
      \\ &=
      \left(\frac{1}{k_n}\right) \cdot 2 \left[ \left(
        \frac{j-i}{H_n} \right)^2 + \left( 1 - \frac{j-i}{H_n}
      \right)^2 \right] \sigma_\nu^4
      \notag
      \\ &\qquad
      + (k_n H_n \Delta_n) \cdot \frac{1}{3}
      \left.\frac{\ud\langle\sigma^2, \sigma^2\rangle_t}{\ud
      t}\right|_{t=L\nu} + o_p(k_n^{-1} + k_n H_n \Delta_n).
    \end{align}
    On the other hand, since $\E[(\Delta_{H_n} \sigma_{i+H_n}^2)^2 |
    \calF_{(i+H_n)\Delta_n}] = O_p(H_n \Delta_n)$, we have
    \begin{align}
      \E \Bigl[
        e_{H_n, i+} e_{H_n, j+}
      \Big| \calF_{L\nu} \Bigr]
      &=
      \left(\frac{1}{k_n}\right) \cdot 2 \left[ \left(
        \frac{j-i}{H_n} \right)^2 + \left( 1 - \frac{j-i}{H_n}
      \right)^2 \right] \sigma_\nu^4
      \notag
      \\ &\qquad
      + (k_n H_n \Delta_n) \cdot \frac{1}{3}
      \left.\frac{\ud\langle\sigma^2, \sigma^2\rangle_t}{\ud
      t}\right|_{t=L\nu} + o_p(k_n^{-1} + k_n H_n \Delta_n).
    \end{align}
    Similar results can be proved for $e_{H_n, i-}$ related terms and
    Equation~\eqref{eq:SALE-spot-error-cross}. The stable convergence
    of $u_n (e_{H_n, i+}, e_{H_n, i+})$ is similar to established
    results on integrated volatility~\citep[also see Lemma~1 in
      Supplementary Material
    of][]{aitsahalia2017EstimationContinuousDiscontinuous}.
  \end{proof}

  The conditions of Jacod's stable central limit theorem are verified below.

  \paragraph*{Ucp Convergence}
  First, we need to prove that
  \begin{equation}\label{eq:SALE-clt-noisefree-ucp}
    \sum_{\nu=0}^{\lfloor t/L \rfloor - 1} \E[U_\nu | \calF_{L\nu}] \convucp 0.
  \end{equation}
  This uniform convergence on compacts in probability is defined as
  \begin{equation}
    \sup_{t\in [0,T]}
    \left| \sum_{\nu=0}^{\lfloor t/L \rfloor - 1} \E[U_\nu |
    \calF_{L\nu}] \right|
    \convp 0.
  \end{equation}
  The left-hand side is bounded by
  \begin{equation}
    \sup_{t\in [0,T]}
    \left| \sum_{\nu=0}^{\lfloor t/L \rfloor - 1} \E[U_\nu |
    \calF_{L\nu}] \right|
    \leq
    \sum_{\nu=0}^{\lfloor t/L \rfloor - 1} \Bigl| \E[ U_\nu |
    \calF_{L\nu} ] \Bigr|
    \leq
    \sum_{\nu=0}^{\lfloor t/L \rfloor - 1}
    \Bigl(
      \Bigl| \E[ U_{\nu, 1} | \calF_{L\nu} ] \Bigr|
      + \Bigl| \E[ U_{\nu, 2} | \calF_{L\nu} ] \Bigr|
    \Bigr).
  \end{equation}
  By Lemma~\ref{lem:err-spot-volatility}, we have
  \begin{align}
    \E [ U_{\nu, 1} | \calF_{L\nu} ]
    &=
    \frac{u_n}{H_n}
    \sum_{i\in J(\nu)}
    \E \Bigl[ (\Delta_{H_n} X_i) (e_{H_n, i+} - e_{H_n, i-}) \Big|
    \calF_{L\nu} \Bigr]
    \notag
    \\ &=
    \frac{1}{H_n}
    \sum_{i\in J(\nu)}
    \E\Bigl[\Delta_{H_n} X_i \Big| \calF_{L\nu}\Bigr]
    \E\Bigl[u_n (e_{H_n, i+} - e_{H_n, i-}) \Big| \calF_{L\nu}\Bigr]
    \notag
    \\ &=
    \frac{1}{H_n}
    \sum_{i\in J(\nu)}
    O_p(H_n \Delta_n) o_p(1)
    =
    o_p(K H_n \Delta_n).
  \end{align}
  On the other hand, by Lemma~\ref{lem:err-oracle-leverage}, we have
  \begin{align}
    \Bigl| \E [ U_{\nu, 2} | \calF_{L\nu} ] \Bigr|
    &\leq
    \frac{u_n}{H_n}
    \sum_{i\in J(\nu)}
    \E \Biggl[
      \Bigl|
      (\Delta_{H_n} X_i) (\Delta_{H_n} \sigma_i^2)
      - \int_{i\Delta_n}^{(i+H_n)\Delta_n} \ud \true_t
      \Bigr|
      \Bigg| \calF_{L\nu}
    \Biggr]
    \notag
    \\ &\leq
    \frac{u_n}{H_n}
    \sum_{i\in J(\nu)}
    \sqrt{
      \E \Biggl[
        \Bigl( (\Delta_{H_n} X_i) (\Delta_{H_n} \sigma_i^2)
        - \int_{i\Delta_n}^{(i+H_n)\Delta_n} \ud \true_t \Bigr)^2
        \Bigg| \calF_{L\nu}
      \Biggr]
    }
    \notag
    \\ &=
    \frac{u_n}{H_n}
    \sum_{i\in J(\nu)}
    O_p(H_n \Delta_n)
    =
    o_p(u_n K H_n \Delta_n).
  \end{align}
  It follows that
  \begin{equation}
    \sum_{\nu=0}^{\lfloor t/L \rfloor - 1}
    \Bigl(
      \Bigl| \E[ U_{\nu, 1} | \calF_{L\nu} ] \Bigr|
      + \Bigl| \E[ U_{\nu, 2} | \calF_{L\nu} ] \Bigr|
    \Bigr)
    =
    \sum_{\nu=0}^{\lfloor t/L \rfloor - 1}
    o_p(K H_n \Delta_n)
    + o_p(u_n K H_n \Delta_n)
    =
    o_p(t).
  \end{equation}

  \paragraph*{Predictable Quadratic Variation}
  We need to prove that
  \begin{equation}\label{eq:SALE-clt-noisefree-qv}
    \sum_{\nu=0}^{\lfloor t/L \rfloor - 1}
    \Bigl(
      \E[U_\nu^2 | \calF_{L\nu}]
      - \E[U_\nu | \calF_{L\nu}]^2
    \Bigr)
    \convp
    \int_0^t \zeta_s^2 \ud s,
  \end{equation}
  for some predictable process $(\zeta_s)_{s\geq 0}$. According to
  previous analysis, the only non-trivial term in the summand is
  $\E[U_{\nu, 1}^2 | \calF_{L\nu}]$. By
  Lemma~\ref{lem:err-spot-volatility}, up to the highest order, we have
  \begin{align}
    &
    \E[U_{\nu, 1}^2 | \calF_{L\nu}]
    \notag
    \\ &=
    \frac{u_n^2}{H_n^2}
    \sum_{i\in J(\nu)}
    \sum_{j\in J(\nu)}
    \E\Bigl[
      (\Delta_{H_n} X_i)
      (\Delta_{H_n} X_j)
    \Big| \calF_{L\nu} \Bigr]
    \E\Bigl[
      (e_{H_n, i+} - e_{H_n, i-})
      (e_{H_n, j+} - e_{H_n, j-})
    \Big| \calF_{L\nu} \Bigr]
    \notag
    \\ &=
    \frac{u_n^2}{H_n^2}
    \sum_{i\in J(\nu)}
    \sum_{j\in J(\nu), |i-j|<H_n}
    \sigma_\nu^2 \left( 1 - \frac{|i-j|}{H_n} \right) H_n \Delta_n
    \E\Bigl[
      e_{H_n, i+} e_{H_n, j+} + e_{H_n, i-} e_{H_n, j-}
    \Big| \calF_{L\nu} \Bigr]
    \notag
    \\ &=
    \frac{u_n^2}{H_n^2}
    \sum_{i\in J(\nu)}
    \sum_{j\in J(\nu), |i-j|<H_n}
    \sigma_\nu^2 \left( 1 - \frac{|i-j|}{H_n} \right) H_n \Delta_n
    \notag
    \\ &\qquad
    2 \left\{
      \frac{1}{k_n} \cdot 2 \left[ \left( \frac{|i-j|}{H_n} \right)^2
      + \left( 1 - \frac{|i-j|}{H_n} \right)^2 \right] \sigma_\nu^4
      + k_n H_n \Delta_n \cdot \frac{1}{3}
      \left.\frac{\ud\langle\sigma^2, \sigma^2\rangle_t}{\ud t}\right|_{t=L\nu}
    \right\}
    \notag
    \\ &=
    u_n^2
    \left(
      \frac{1}{k_n} \cdot \left(\frac{8}{3}+\frac{4}{3H_n^2}\right)
      \sigma_\nu^6 \cdot K H_n \Delta_n
      + k_n H_n \Delta_n \cdot \frac{2}{3} \sigma_\nu^2
      \left.\frac{\ud\langle\sigma^2, \sigma^2\rangle_t}{\ud
      t}\right|_{t=L\nu} \cdot K H_n \Delta_n
    \right).
    \label{eq:SALE-clt-noisefree-qv-Uv1squared}
  \end{align}
  It follows that
  \begin{align}
    &
    \sum_{\nu=0}^{\lfloor t/L \rfloor - 1}
    \Bigl(
      \E[U_\nu^2 | \calF_{L\nu}]
      - \E[U_\nu | \calF_{L\nu}]^2
    \Bigr)
    \notag
    \\ &=
    \frac{u_n^2}{k_n}
    \sum_{\nu=0}^{\lfloor t/L \rfloor - 1}
    \left(\frac{8}{3}+\frac{4}{3H_n^2}\right) \sigma_\nu^6 \cdot K H_n \Delta_n
    +
    u_n^2 k_n H_n \Delta_n
    \sum_{\nu=0}^{\lfloor t/L \rfloor - 1}
    \frac{2}{3} \sigma_\nu^2 \left.\frac{\ud\langle\sigma^2,
    \sigma^2\rangle_t}{\ud t}\right|_{t=L\nu} \cdot K H_n \Delta_n
    \notag
    \\ &\convp
    \frac{u_n^2}{k_n} \int_0^t
    \left(\frac{8}{3}+\frac{4}{3H_n^2}\right) \sigma_s^6 \ud s
    + u_n^2 k_n H_n \Delta_n \int_0^t \frac{2}{3} \sigma_s^2
    \ud\langle\sigma^2, \sigma^2\rangle_s.
  \end{align}

  \paragraph*{Predictable Quadratic Covariation}
  We need to prove that
  \begin{equation}\label{eq:SALE-clt-noisefree-qcv}
    \sum_{\nu=0}^{\lfloor t/L \rfloor - 1} \E\Bigl[U_\nu (M_{L(\nu+1)}
    - M_{L\nu}) \Big| \calF_{L\nu}\Bigr] \convp 0,
    \quad \forall M \in \{W, B, N\},
  \end{equation}
  where $(N_s)_{s\geq 0}$ is an arbitrary bounded $\calF_t$-martingale
  with $N_0 = 0$ and $\langle W, N\rangle = \langle B, N\rangle = 0$.
  As before, omit $U_{\nu, 2}$ for its minor contribution, and we have
  \begin{align}
    &
    \sum_{\nu=0}^{\lfloor t/L \rfloor - 1}
    \E\Bigl[U_{\nu, 1} (M_{L(\nu+1)} - M_{L\nu}) \Big| \calF_{L\nu}\Bigr]
    \notag
    \\ &=
    \sum_{\nu=0}^{\lfloor t/L \rfloor - 1}
    \frac{1}{H_n}
    \sum_{i \in J(\nu)}
    \E\Bigl[
      (\Delta_{H_n} X_i) (M_{L(\nu+1)} - M_{L\nu})
    \Big| \calF_{L\nu} \Bigr]
    \E\Bigl[
      u_n (e_{H_n, i+} - e_{H_n, i-})
    \Big| \calF_{L\nu} \Bigr]
    \notag
    \\ &=
    \sum_{\nu=0}^{\lfloor t/L \rfloor - 1}
    \frac{1}{H_n}
    \sum_{i \in J(\nu)}
    O_p(H_n \Delta_n) o_p(1)
    =
    o_p\left(\frac{t}{KH_n\Delta_n} \cdot \frac{1}{H_n} \cdot KH_n
    \cdot H_n\Delta_n \cdot 1\right)
    \notag
    \\ &=
    o_p(t).
  \end{align}
  This argument holds for any bounded $\calF_t$-martingale.

  \paragraph*{Lindeberg Condition}
  We need to prove that
  \begin{equation}\label{eq:SALE-clt-noisefree-lindeberg}
    \sum_{\nu=0}^{\lfloor t/L \rfloor - 1} \E\Bigl[U_\nu^2 1_{\{|U_\nu|
    > \epsilon\}} \Big| \calF_{L\nu}\Bigr] \convp 0,
    \quad \forall \epsilon > 0.
  \end{equation}
  It is well-known that a sufficient condition for
  Equation~\eqref{eq:SALE-clt-noisefree-lindeberg} is
  \begin{equation}
    \sum_{\nu=0}^{\lfloor t/L \rfloor - 1} \E\Bigl[|U_\nu|^{2+\delta}
    \Big| \calF_{L\nu}\Bigr] \convp 0,
  \end{equation}
  for any $\delta > 0$. We take $\delta=2$ for simplicity, and omit
  $U_{\nu, 2}$. The resulting expression is
  \begin{align}
    &
    \sum_{\nu=0}^{\lfloor t/L \rfloor - 1} \E[U_{\nu, 1}^{4} | \calF_{L\nu}]
    \notag
    \\ &=
    \sum_{\nu=0}^{\lfloor t/L \rfloor - 1}
    \frac{1}{H_n^4}
    \E\Biggl[
      \prod_{k=1}^4 \sum_{i_k \in J(\nu)}
      (\Delta_{H_n} X_{i_k}) u_n (e_{H_n, i_k+} - e_{H_n, i_k-})
    \Bigg| \calF_{L\nu} \Biggr]
    \notag
    \\ &=
    \sum_{\nu=0}^{\lfloor t/L \rfloor - 1}
    \frac{1}{H_n^4}
    \sum_{i_1, i_2, i_3, i_4 \in J(\nu)}
    \E\Biggl[
      \prod_{k=1}^4
      (\Delta_{H_n} X_{i_k})
    \Bigg| \calF_{L\nu} \Biggr]
    \E\Biggl[
      \prod_{k=1}^4
      u_n (e_{H_n, i_k+} - e_{H_n, i_k-})
    \Bigg| \calF_{L\nu} \Biggr]
    \notag
    \\ &=
    \sum_{\nu=0}^{\lfloor t/L \rfloor - 1}
    \frac{1}{H_n^4}
    O_p((KH_n)^2 H_n^2 \cdot H_n^2 \Delta_n^2 + (KH_n)^4 \cdot H_n^4 \Delta_n^4)
    \notag
    \\ &=
    O_p(t KH_n\Delta_n)
    = o_p(t).
  \end{align}

  \paragraph*{Result}
  By Jacod's stable CLT, there exist a standard Brownian motion
  $(W_{1,t})_{t\geq 0}$ independent of $\calF$ and a predictable
  process $(\zeta_{1,t})_{t\geq 0}$ such that, as $n\to\infty$,
  \begin{gather}
    u_n \bigl( \estx_T^{(H_n)} - \true_T \bigr)
    \convst
    \int_0^T \zeta_{1,t} \ud W_{1,t},
    \\
    \int_0^T \zeta_{1,t}^2 \ud t =
    \frac{u_n^2}{k_n} \left(\frac{8}{3} + \frac{4}{3H_n^2}\right)
    \int_0^T \sigma_s^6 \ud t +
    u_n^2 k_n H_n \Delta_n \frac{2}{3} \int_0^T \sigma_t^2 \ud
    \langle \sigma^2, \sigma^2 \rangle_t.
  \end{gather}
  Taking $H_n$ and $k_n$ as specified in
  Assumption~\ref{ass:SALE-para}\ref{ass:SALE-para-finite} and
  \ref{ass:SALE-para}\ref{ass:SALE-para-asym} respectively completes the proof.

  \subsection{Proof of Theorem~\ref{thm:SALE-clt-noisy}}
  \label{app:SALE-clt-noisy}

  By Theorem~\ref{thm:SALE-clt-noisefree} and
  Section~\ref{app:SALE-noise}, as $n\to\infty$, we have
  \begin{align}
    \sqrt{k_n \land (k_n H_n \Delta_n)^{-1}}
    \bigl( \estx_T^{(H_n)} - \true_T \bigr)
    & \convst
    \int_0^T \zeta_{1,t} \ud W_{1,t},
    \\
    \Delta_n^{3/2} H_n^2 k_n
    \bigl( \esty_T^{(H_n)} - \estx_T^{(H_n)} \bigr)
    & \convst
    \sqrt{\Phi T} Z,
  \end{align}
  where $Z$ is an $\normal(0, 1)$ random variable independent of
  $\calF$. Since the two sequences are asymptotically independent and
  each converges $\calF$-stably in law, they jointly converge
  $\calF$-stably in law to independent limits. Thus, applying the
  continuous mapping theorem completes the proof.

  \subsection{Proof of Proposition~\ref{prop:SALE-acov-noisefree}}
  \label{app:SALE-acov-noisefree}

  To start with, we define the adjustment factors $v_{p,q}^{(1)}$ and
  $v_{p,q}^{(2)}$.

  \begin{definition}\label{def:adjustment-factors}
    Let $p,q,H_p,H_q,k_p,k_q$ be positive integers such that $p \geq q$
    and $H_p \geq H_q$. Define the kernel function $\mathcal{K}_{p,q}:
    \Z \to \R$ that satisfies $\sum_{d\in\Z}\mathcal{K}_{p,q}(d) = 1$ by
    \begin{align}\label{eq:kernel-Kpq}
      \mathcal{K}_{p,q}(d) =
      \begin{cases}
        \frac{d+H_q}{H_pH_q}, & \text{if } -H_q < d < 0, \\
        \frac{1}{H_p}, & \text{if } 0 \leq d \leq H_p-H_q, \\
        \frac{-d+H_p}{H_pH_q}, & \text{if } H_p-H_q < d < H_p, \\
        0, & \text{otherwise},
      \end{cases}
    \end{align}
    For any $i,j\in \Z$, define the following sets and quantities:
    \begin{align}
      G_{p}(i) &= \{i+kH_p: k\in \Z\},
      \\
      G_{q}(j) &= \{j+kH_q: k\in \Z\},
      \\
      L_{p,q}(i,j) &= (i+2H_p) \lor (j+2H_q),
      \\
      R_{p,q}(i,j) &= (i+(k_p+2)H_p) \land (j+(k_q+2)H_q),
      \\
      P_{p,q}(i,j) &= (G_{p}(i) \cup G_{q}(j)) \cap [L_{p,q}(i,j),
      R_{p,q}(i,j)].
    \end{align}
    Let $x_1 < x_2 < \dotsb < x_m$ be the elements of the set
    $P_{p,q}(i,j)$ sorted in ascending order. Define the \emph{sum of
    squares of overlapping segments} (SSOS), the \emph{cubed
    overlapping length} (COL), and their normalized versions by
    \begin{align}
      \mathrm{SSOS}_{p,q}(i,j) &= \sum_{k=1}^{m-1} (x_{k+1} - x_k)^2, &
      \widetilde{\mathrm{SSOS}}_{p,q}(i,j) &=
      \frac{\mathrm{SSOS}_{p,q}(i,j)}{k_qH_q^2}.
      \label{eq:SSOS}
      \\
      \mathrm{COL}_{p,q}(i,j) &= \bigl[(R_{p,q}(i,j) -
      L_{p,q}(i,j))_+\bigr]^3, &
      \widetilde{\mathrm{COL}}_{p,q}(i,j) &=
      \frac{\mathrm{COL}_{p,q}(i,j)}{k_q^3H_q^3},
      \label{eq:COL}
    \end{align}
    Note that $\mathrm{SSOS}_{p,q}(i,j)$ and $\mathrm{COL}_{p,q}(i,j)$
    depend on $i,j$ only through the difference $j-i$. Finally, the
    \emph{adjustment factors} $v_{p,q}^{(1)}$ and $v_{p,q}^{(2)}$ are defined as
    \begin{align}
      v_{p,q}^{(1)} &= \sum_{d\in\Z} \mathcal{K}_{p,q}(d)
      \widetilde{\mathrm{SSOS}}_{p,q}(0,d), \\
      v_{p,q}^{(2)} &= \sum_{d\in\Z} \mathcal{K}_{p,q}(d)
      \widetilde{\mathrm{COL}}_{p,q}(0,d).
    \end{align}
  \end{definition}

  Proposition~\ref{prop:SALE-acov-noisefree} is a two-scale extension
  of Theorem~\ref{thm:SALE-clt-noisefree}. Similar as
  Section~\ref{app:SALE-clt-noisefree}, the big block technique is used to
  establish the conditional independence in the proof. Let $H_n^* \in
  \Z$ satisfies that $H_n^* \geq H_p \geq H_q$. Consier big blocks of
  size $L$, where $L / (KH_n^*\Delta_n) \to 1$ as $n\to\infty$, where
  $K=K_n$ is a sequence satisfying $K\to\infty$ and $KH_n^*\Delta_n \to
  0$ as $n\to\infty$. Similarly, the time span of the $\nu$th big block
  is denoted as $[L\nu, L\nu+L)$ with $\nu=0, 1, \dots, \lfloor T/L\rfloor - 1$.

  Recall that, with these notations, rewrite the error of the
  $H_p$-scaled SALE estimator as
  \begin{align}
    u_n \Bigl( \estx^{(H_p)}_T - \true_T \Bigr)
    =
    \sum_{\nu=0}^{\lfloor T/L \rfloor - 1} U_{H_p, \nu}
    + \text{(edge terms)},
  \end{align}
  where
  \begin{align}
    \label{eq:SALE-err-U}
    U_{H_p, \nu}
    &=
    U_{H_p, \nu, 1} + U_{H_p, \nu, 2},
    \\
    \label{eq:SALE-err-U1}
    U_{H_p, \nu, 1}
    &=
    \frac{u_n}{H_p} \sum_{i\in J_{H_p}(\nu)}
    \Bigl(
      V_{H_p, i}
      -
      (\Delta_{H_p} X_i) (\Delta_{H_p} \sigma_i^2)
    \Bigr), \\
    \label{eq:SALE-err-U2}
    U_{H_p, \nu, 2}
    &=
    \frac{u_n}{H_p} \sum_{i\in J_{H_p}(\nu)}
    \left(
      (\Delta_{H_p} X_i) (\Delta_{H_p} \sigma_i^2)
      -
      \int_{i\Delta_n}^{(i+H_p)\Delta_n} \ud \true_t
    \right).
  \end{align}
  Here, $J_{H_p}(\nu) = \{L\nu/\Delta_n+(k_p+1)H_p, \dotsc,
  L\nu/\Delta_n+(K+k_p+1)H_p-1\}$ is the set of indices of the main
  increments within the big block. Similar notations for $H_q$ can be
  established. Set $u_n = \sqrt{k_p \land (k_p H_p \Delta_n)^{-1}}$. We
  have the following lemma.

  \begin{lemma}\label{lem:SALE-increment-block-cov}
    As $n\to\infty$, we have
    \begin{align}
      &\E \Bigl(U_{H_p, \nu, 1} U_{H_q, \nu, 1} \Big| \calF_{L\nu} \Bigr)
      \notag
      \\ &=
      \frac{u_n^2}{k_p} \cdot 4 v_{p,q}^{(1)} \frac{H_q}{H_p} \cdot
      \sigma_\nu^6 L
      + u_n^2 k_pH_p\Delta_n \cdot \frac{2}{3} v_{p,q}^{(2)}
      \left(\frac{k_qH_q}{k_pH_p}\right)^2 \cdot \sigma_\nu^2
      \left.\frac{\ud\langle\sigma^2, \sigma^2\rangle_t}{\ud
      t}\right|_{t=L\nu} L
      \notag
      \\ &\qquad
      + o_p(u_n^2 k_p^{-1} L + u_n^2 k_pH_p\Delta_n L).
    \end{align}
  \end{lemma}

  \begin{proof}
    Note that
    \begin{align}
      U_{H_p, \nu, 1}
      &=
      \frac{u_n}{H_p} \sum_{i\in J_{H_p}(\nu)}
      (\Delta_{H_p} X_i) (e_{H_p, i+} - e_{H_p, i-}),
      \\
      e_{H_p, i+}
      &=
      \overline{\sigma}_+^2 (i, H_p, k_p) - \sigma_{i+H_p}^2,
      \\
      e_{H_p, i-}
      &=
      \overline{\sigma}_-^2 (i, H_p, k_p) - \sigma_{i}^2.
    \end{align}
    Similar expression follows for $U_{H_q, \nu, 2}$. We have
    \begin{align}
      &
      \E \Bigl(U_{H_p, \nu, 1} U_{H_q, \nu, 1} \Big| \calF_{L\nu} \Bigr)
      \notag
      \\
      &=
      \frac{u_n^2}{H_p H_q} \sum_{i \in J_{H_p}(\nu)}
      \sum_{j \in J_{H_q}(\nu)}
      \E \Bigl[
        (\Delta_{H_p} X_i) (\Delta_{H_q} X_j)
      \Big| \calF_{L\nu} \Bigr]
      \E \Bigl[
        (e_{H_p, i+} - e_{H_p, i-}) (e_{H_q, j+} - e_{H_q, j-})
      \Big| \calF_{L\nu} \Bigr].
    \end{align}
    For price-related terms, note that
    \begin{equation}
      \E \Bigl[
        (\Delta_{H_p} X_i) (\Delta_{H_q} X_j)
      \Big| \calF_{L\nu} \Bigr]
      =
      \sigma_\nu^2 s(i, j) + o_p(H_q \Delta_n),
    \end{equation}
    where $s(i,j) = H_pH_q\Delta_n \mathcal{K}_{p,q}(j-i)$. For
    spot-volatility-related terms, note that
    \begin{align}
      e_{H_p, i+}
      &=
      \tilde{e}_{H_p, i+} + \Delta_{H_p} \sigma_{i+H_p}^2,
      \\
      \tilde{e}_{H_p, i+}
      &=
      \frac{1}{k_p H_p \Delta_n}
      \left[
        \left(\sum_{k \in S_{H_p}^+} \xi_{H_p, i, k} \right)
        + \eta_{H_p, i}^+
      \right],
      \\
      \xi_{H_p, i, k}
      &=
      \int_{(i+kH_p)\Delta_n}^{(i+(k+1)H_p)\Delta_n} 2(X_r -
      X_{i+kH_p}) \ud X_r
      \quad \text{for} \quad
      k = 2, \dots, k_p+1,
      \\
      \eta_{H_p, i}^+
      &=
      \int_{(i+2H_p)\Delta_n}^{(i+(k_p+2)H_p)\Delta_n} (\sigma_r^2 -
      \sigma_{i+2H_p}^2) \ud r,
    \end{align}
    and by Lemma~\ref{lem:overlap-expectation}, we have
    \begin{gather}
      \E \Biggl[
        \Biggl(\sum_{k \in S_{H_p}^+} \xi_{H_p, i, k} \Biggr)
        \Biggl(\sum_{l \in S_{H_q}^+} \xi_{H_q, j, l} \Biggr)
      \Bigg| \calF_{L\nu} \Biggr]
      =
      2 \sigma_\nu^4 \mathrm{SSOS}_{p,q}(i,j) \Delta_n^2
      + o_p(k_q H_q^2 \Delta_n^2),
      \\
      \E \Bigl[
        \eta_{H_p, i}^+ \eta_{H_q, j}^+
      \Big| \calF_{L\nu} \Bigr]
      =
      \frac{1}{3}
      \left.\frac{\ud \langle \sigma^2, \sigma^2 \rangle_t}{\ud
      t}\right|_{t=L\nu}
      \mathrm{COL}_{p,q}(i,j) \Delta_n^3
      + o_p(k_q^3 H_q^3 \Delta_n^3).
    \end{gather}
    Other terms are negligible, similar as in
    Section~\ref{app:SALE-clt-noisefree}. It follows that
    \begin{align}
      &\E \Bigl(U_{H_p, \nu, 1} U_{H_q, \nu, 1} \Big| \calF_{L\nu} \Bigr)
      \notag
      \\ &\approx
      \frac{u_n^2}{H_p H_q} \frac{1}{k_pH_p\Delta_n}
      \frac{1}{k_qH_q\Delta_n} \sum_{i \in J_{H_p}(\nu)}
      \notag
      \\ & \quad
      \sum_{j \in J_{H_q}(\nu)}
      \E \Bigl[
        (\Delta_{H_p} X_i) (\Delta_{H_q} X_j)
      \Big| \calF_{L\nu} \Bigr]
      \E \Biggl[
        \Biggl(\sum_{k \in S_{H_p}^+} \xi_{H_p, i, k} \Biggr)
        \Biggl(\sum_{l \in S_{H_q}^+} \xi_{H_q, j, l} \Biggr)
        +
        \eta_{H_p, i}^+ \eta_{H_q, j}^+
      \Bigg| \calF_{L\nu} \Biggr]
      \times 2
      \notag
      \\ &\approx
      \frac{u_n^2}{k_pk_qH_p^2H_q^2\Delta_n^2}
      \sum_{i \in J_{H_p}(\nu)}
      k_q H_p H_q^3 \Delta_n^3 \cdot 4\sigma_\nu^6 v_{p,q}^{(1)} +
      k_q^3 H_p H_q^4 \Delta_n^4
      \cdot
      \frac{2}{3} \sigma_\nu^2 \left.\frac{\ud \langle \sigma^2,
      \sigma^2 \rangle_t}{\ud t}\right|_{t=L\nu} v_{p,q}^{(2)}
      \notag
      \\ &\approx
      \frac{u_n^2}{k_p} \cdot 4 v_{p,q}^{(1)} \frac{H_q}{H_p} \cdot
      \sigma_\nu^6 L
      + u_n^2 k_pH_p\Delta_n \cdot \frac{2}{3} v_{p,q}^{(2)}
      \left(\frac{k_qH_q}{k_pH_p}\right)^2 \cdot \sigma_\nu^2
      \left.\frac{\ud\langle\sigma^2, \sigma^2\rangle_t}{\ud
      t}\right|_{t=L\nu} L.
    \end{align}
    This completes the proof of Lemma~\ref{lem:SALE-increment-block-cov}.
  \end{proof}

  The rest of the proof is similar with
  Section~\ref{app:SALE-clt-noisefree}, except that a
  multi-dimensional stable CLT is used \citep[for example,
  Theorem~7.19 in][]{jacod2003LimitTheoremsStochastic}.

  \subsection{Proof of Proposition~\ref{prop:adj-factor-asym}}

  For $v_{p,q}^{(2)}$, note that a necessary condition of
  $\mathcal{K}_{p,q}(d) > 0$ is that $-H_q < d < H_p$, which implies
  that $\widetilde{\mathrm{COL}}_{p,q}(0,d) \to 1$ as $n\to\infty$, and
  thus
  \begin{align}
    v_{p,q}^{(2)} \to 1.
  \end{align}

  For $v_{p,q}^{(1)}$, note that, for a given $i$ and all $j$ that
  satisfies $i-H_q < j < i+H_p$, as $n\to\infty$, the SSOS is
  asymptotically equivalent to the sum of $N_{p,q}$ SSOS's within
  consecutive $H_p$-intervals:
  \begin{gather}
    \mathrm{SSOS}_{p,q}(i,j) \asymp \sum_{l=1}^{N_{p,q}} g_{p,q}(l;i,j),
    \\
    N_{p,q} \asymp \frac{k_qH_q}{H_p} \asymp \beta c^{1-b} (n/H_p)^b \to \infty.
  \end{gather}
  Here, $g_{p,q}(l;i,j)$ is the SSOS within the group $l$. Thus, we have
  \begin{align}
    v_{p,q}^{(1)}
    \asymp
    \frac{1}{k_q H_q^2}
    \sum_{d\in\Z} \mathcal{K}_{p,q}(d) \sum_{l=1}^{N_{p,q}} g_{p,q}(l;0,d)
    =
    \frac{1}{k_q H_q^2} \sum_{l=1}^{N_{p,q}}
    \sum_{d\in\Z} \mathcal{K}_{p,q}(d) g_{p,q}(l;0,d).
  \end{align}
  Note that $g_{p,q}(l;0,d) = g_{p,q}(l;0,d+H_q)$ for any valid $d$, as
  the pattern of grid points within the $l$th $H_p$-interval repeats
  itself when $j$ shifts $H_q$. Another important observation is that,
  through direct calculation, we have
  \begin{align}
    \sum_{d\in\Z}
    \mathcal{K}_{p,q}(d)
    1_{\{d\equiv d_0\text{ mod }H_q\}}
    =
    \frac{1}{H_q},
    \quad \text{for all }d_0 = 0, 1, \dotsc, H_q-1.
  \end{align}
  Therefore, we have
  \begin{align}
    v_{p,q}^{(1)}
    &\asymp
    \frac{1}{k_q H_q^2} \sum_{l=1}^{N_{p,q}}
    \sum_{d\in\Z} \mathcal{K}_{p,q}(d) g_{p,q}(l;0,d)
    \notag
    \\ &=
    \frac{1}{k_q H_q^2} \sum_{l=1}^{N_{p,q}}
    \sum_{d\in\Z} \mathcal{K}_{p,q}(d) g_{p,q}(l;0,d)
    \sum_{d_0=0}^{H_q-1} 1_{\{d\equiv d_0\text{ mod }H_q\}}
    \notag
    \\ &=
    \frac{1}{k_q H_q^2} \sum_{l=1}^{N_{p,q}}
    \sum_{d_0=0}^{H_q-1} g_{p,q}(l;0,d_0)
    \sum_{d\in\Z} \mathcal{K}_{p,q}(d) 1_{\{d\equiv d_0\text{ mod }H_q\}}
    \notag
    \\ &=
    \frac{1}{k_q H_q^2} \sum_{l=1}^{N_{p,q}}
    \frac{1}{H_q} \sum_{d_0=0}^{H_q-1} g_{p,q}(l;0,d_0).
  \end{align}
  On the other hand, by Lemma~\ref{lem:grid-integral}, as $H_q \to
  \infty$, we have
  \begin{align}
    \frac{1}{H_q} \sum_{d_0=0}^{H_q-1} g_{p,q}(l;0,d_0)
    &=
    H_p^2 \E_{\alpha_1\sim \unif\{1/H_q,\dotsc, H_q/H_q\}}
    \Bigl[S_2(\alpha_1, H_q/H_p)\Bigr]
    \notag
    \\ &\to
    H_p^2 \E_{\alpha_1\sim \unif(0,1]} \Bigl[S_2(\alpha_1, H_q/H_p)\Bigr]
    \\ &=
    H_pH_q \left(1 - \frac{H_q}{3H_p}\right),
  \end{align}
  It follows that, as $n\to\infty$, we have
  \begin{align}
    v_{p,q}^{(1)}
    \to
    \frac{1}{k_q H_q^2} \sum_{l=1}^{N_{p,q}} H_pH_q \left(1 -
    \frac{H_q}{3H_p}\right)
    \asymp
    \frac{1}{k_q H_q^2} \frac{k_qH_q}{H_p} H_p H_q \left(1 -
    \frac{H_q}{3H_p}\right)
    = 1 - \frac{\rho}{3}.
  \end{align}

  This completes the proof.

  \subsection{
    Proofs of Theorem~\ref{thm:MSLE-clt-noisefree} and \ref{thm:MSLE-clt-noisy}
  }

  We use the same block partition as
  Section~\ref{app:SALE-acov-noisefree}, where $H_n^*$ is specificed by
  Assumption~\ref{ass:MSLE-para-detail}\ref{ass:MSLE-para-detail-scale}.
  Rewrite the error of the MSLE estimator as
  \begin{align}
    u_n \Bigl( \estx^{\rm (MS)}_T - \true_T \Bigr)
    =
    \sum_{\nu=0}^{\lfloor T/L \rfloor - 1}
    \sum_{p=1}^{M_n}
    w_p U_{H_p, \nu}
    + \text{(edge terms)},
  \end{align}
  where $U_{H_p, \nu}$ and its decomposition are given by
  Equation~\eqref{eq:SALE-err-U}, \eqref{eq:SALE-err-U1} and
  \eqref{eq:SALE-err-U2}. Let $x_p = H_p / H_n^*$ for all $p = 1,
  \dots, M_n$. By Lemma~\ref{lem:SALE-increment-block-cov} and
  Proposition~\ref{prop:adj-factor-asym}, for any $1 \leq q \leq p \leq
  M_n$, as $n\to\infty$, we have
  \begin{align}
    &
    \E \Bigl(U_{H_p, \nu, 1} U_{H_q, \nu, 1} \Big| \calF_{L\nu} \Bigr)
    \notag
    \\
    &=
    u_n^2
    n^{-(1-a)b} \cdot
    \frac{4\alpha^b}{\beta} \cdot
    x_p^b \left(
      \frac{x_q}{x_p} - \frac{x_q^2}{3x_p^2}
    \right) \cdot
    \sigma_\nu^6 L
    \notag
    \\
    &\qquad
    +
    u_n^2
    n^{-(1-a)(1-b)} \cdot
    \frac{2\alpha^{1-b} \beta T}{3} \cdot
    \frac{x_q^{2(1-b)}}{x_p^{1-b}} \cdot
    \sigma_\nu^2 \left.
    \frac{\ud\langle\sigma^2, \sigma^2\rangle_t}{\ud t}
    \right|_{t=L\nu} L
    \notag
    \\
    &\qquad
    +
    o_p(u_n^2 n^{-(1-a)(b \land (1-b))}).
  \end{align}
  Note that
  Assumption~\ref{ass:MSLE-para-detail}\ref{ass:MSLE-para-detail-weight}
  implies that $\sum_{p=1}^{M_n} |w_p|$ is bounded as $n\to\infty$,
  which is enough for the rest of the proof to be completed similarly
  as in Section~\ref{app:SALE-clt-noisefree} and \ref{app:SALE-clt-noisy}.

  \section{Feasible Central Limit Theorems}

  To establish feasible central limit theorems for SALE and MSLE
  estimators, consistent estimators for the quantities in asymptotic
  variances are needed, including $\int_0^T \sigma_t^6 \ud t$,
  $\int_0^T \sigma_t^2 \ud \langle \sigma^2, \sigma^2 \rangle_t$,
  $\Phi$, and $\bm{F}$.

  In the noise-free case, the following estimators have been
  established in the literature:
  \begin{align}
    G_n^{(1)}
    & =
    \frac{1}{15\Delta_n^2} \sum_{i=0}^{n-1} ( \Delta X_i )^6
    \convp
    \int_0^T \sigma_t^6 \ud t,
    \label{eq:G1}
    \\
    G_n^{(2)}
    & =
    \frac{1}{k_n\Delta_n} \sum_{i=k_n+1}^{n-k_n-2}
    ( \Delta X_i )^{2}
    \Biggl(
      \frac{3}{2}
      \bigl(
        \widehat \delta(i, 1, k_n)
      \bigr)^2
      -
      \frac{1}{k_n^2 \Delta_n^2}
      \Biggl(
        \sum_{j=-k_n-1}^{-2}
        + \sum_{j=2}^{k_n+1}
      \Biggr)
      ( \Delta X_j )^4
    \Biggr)
    \notag
    \\
    & \convp
    \int_0^T \sigma_t^2 \ud \langle \sigma^2, \sigma^2 \rangle_t. \label{eq:G2}
  \end{align}
  For $G_n^{(1)}$, refer to Theorem~3.4.1 of
  \cite{jacod2012DiscretizationProcesses}, or Example~3.2 of
  \cite{podolskij2010UnderstandingLimitTheorems}. For $G_n^{(2)}$,
  refer to Theorem~8.11 of
  \cite{aitsahalia2014HighFrequencyFinancialEconometrics}, or
  Theorem~2.6 of \cite{vetter2015EstimationIntegratedVolatility}
  \citep[also see Equation~(5.4)
  in][]{aitsahalia2017EstimationContinuousDiscontinuous}.

  In the noisy case, we employ the pre-averaging technique to construct
  corresponding estimators. Consider a sequence of non-overlapping
  windows, each of length $A_n$, satisfying that $A_n \to \infty$ and
  $A_n \Delta_n \to \infty$ as $n\to \infty$ (for example, take $A_n =
  \lfloor n^{1/2} \rfloor$). The pre-averaged pseudo-observations are defined by
  $\overline X_i = A_n^{-1} \sum_{j=iA_n}^{(i+1)A_n}$ for all $i=0, 1,
  \dots, \overline n$, where $\overline n = \lfloor (n+1) / A_n \rfloor
  - 1$. Let $\overline \Delta_n = A_n \Delta_n$, and define $\overline
  k_n$, $\widehat{\overline \delta}$ similarly. Thus, we have
  \begin{align}
    \widehat G_n^{(1)}
    & =
    \frac{9}{40 \overline\Delta_n}
    \sum_{i=0}^{\overline n - 1} ( \overline\Delta X_i )^6
    \convp
    \int_0^T \sigma_t^6 \ud t,
    \\
    \widehat G_n^{(2)}
    & =
    \frac{27}{8\overline k_n \overline \Delta_n} \sum_{i \in \overline{I}}
    ( \Delta \overline X_i )^2
    \Biggl(
      \frac{3}{2}
      \bigl(
        \widehat{\overline \delta}(i, 1, \overline k_n)
      \bigr)^2
      -
      \frac{1}{\overline k_n^2 \overline \Delta_n^2}
      \Biggl(
        \sum_{j=-\overline k_n-1}^{-2}
        + \sum_{j=2}^{\overline k_n+1}
      \Biggr)
      ( \Delta \overline X_j )^4
    \Biggr)
    \notag
    \\
    & \convp
    \int_0^T \sigma_t^2 \ud \langle \sigma^2, \sigma^2 \rangle_t.
  \end{align}

  The quantities $\Phi$ and $\bm F$ depend on the autocovariances of
  the noise variables, which can be estimated with existing methods,
  including \citet{jacod2017StatisticalPropertiesMicrostructure} and
  \citet{li2022ReMeDIMicrostructureNoise}. We employ the
  \emph{Realized moMents of Disjoint Increments} (ReMeDI)
  method~\citep{li2022ReMeDIMicrostructureNoise} to estimate the
  required moments. For any fixed integer $l$, suppose there is a
  sequence of positive numbers $k_n'$ such that $k_n' \to \infty$ and
  $k_n' \Delta_n \to 0$ as $n \to \infty$, we can establish that
  \begin{align}
    &\frac{1}{n} \sum_{i=k_n'}^{n-m-k_n'} (Y_{i+l} - Y_{i+l+k_n'})
    (Y_{i} - Y_{i-k_n'})
    \convp
    \E[\eps_{i} \eps_{i+l}], \\
    &\frac{1}{n} \sum_{i=2k_n'}^{n-m-k_n'} (Y_{i+l} - Y_{i+l+k_n'})
    (Y_{i+l} - Y_{i+l-k_n'}) (Y_{i} - Y_{i-2k_n'})
    \convp
    \E[\eps_{i} \eps_{i+l}^2], \\
    &\frac{1}{n} \sum_{i=3k_n'}^{n-m-k_n'} (Y_{i+l} - Y_{i+l+k_n'})
    (Y_{i+l} - Y_{i+l-k_n'}) (Y_{i} - Y_{i-2k_n'}) (Y_{i} - Y_{i-3k_n'})
    \convp
    \E[\eps_{i}^2 \eps_{i+l}^2].
  \end{align}
  The estimators for the noise moments $\widehat{\nu}_2$,
  $\widehat{\nu}_4$ and the generalized acfs $\widehat{\rho}_2(l)$,
  $\widehat{\rho}_3(l)$, $\widehat{\rho}_4(l)$ can be constructed
  accordingly. Thus, by plugging these estimators into
  Equation~\eqref{eq:Phi}, the following estimator for $\Phi$ can be derived:
  \begin{equation}\label{eq:Phi-hat}
    \widehat{\Phi}_n = 8 \widehat{\nu}_2 (\widehat{\nu}_4 -
    \widehat{\nu}_2^2) \sum_{l=-q}^{q} \widehat{\rho}_2(l)\widehat{\rho}_4(l)
    + 24 \widehat{\nu}_2^3 \sum_{l=-q}^{q} \widehat{\rho}_2^3(l)
    + 8 \widehat{\nu}_2 (\widehat{\nu}_4 - \widehat{\nu}_2^2)
    \sum_{l=-q}^{q} \widehat{\rho}_3(l)\widehat{\rho}_3(-l)
    \convp \Phi.
  \end{equation}
  The construction for $\bm F$ follows similarly.

  Given these consistent estimators of asymptotic variances, the
  feasible CLTs follow from the property of stable convergence in law:
  given $Y_n \convst VU$ and $V_n \convp V$, we have $Y_n / V_n \convst U$.
  For example, the feasible CLT for SALE estimators is as follows.

  \begin{theorem}\label{thm:SALE-clt-feasible}
    \begin{enumerate}[label=(\arabic*)]
      \item [] %
      \item Under Assumption~\ref{ass:process} and
        \ref{ass:SALE-para}\ref{ass:SALE-para-finite}, let $u_n =
        \sqrt{k_n \land (k_n H \Delta_n)^{-1}}$. As $n \to \infty$, we have
        \begin{gather}
          \frac{u_n}{\sqrt{V_{1,n}}} \bigl(
            \estx_T^{(H)} - \true_T
          \bigr)
          \convst
          \normal(0, 1),
          \\
          \text{where} \quad
          V_{1,n} =
          \frac{u_n^2}{k_n} \left(\frac{8}{3} + \frac{4}{3H^2}\right)
          G_n^{(1)} +
          u_n^2 k_n H \Delta_n \frac{2}{3} G_n^{(2)},
        \end{gather}
        and the limiting distribution is independent of $\calF$.
      \item Under Assumption~\ref{ass:process} and
        \ref{ass:SALE-para}\ref{ass:SALE-para-asym}, as $n\to\infty$, we have
        \begin{gather}
          \frac{n^{\frac{1}{2}(1-a)(b\land (1-b))}}{\sqrt{V_{1,n}}}
          \bigl(
            \estx^{(H_n)}_T - \true_T
          \bigr)
          \convst
          \normal(0, 1),
          \\
          \text{where} \quad
          V_{1,n} =
          \frac{8 \alpha^b}{3 \beta} G_n^{(1)} \cdot 1_{(0, 1/2]}(b)
          +
          \frac{2 \alpha^{1-b} \beta T}{3} G_n^{(2)} \cdot 1_{[1/2, 1)}(b),
        \end{gather}
        and the limiting distribution is independent of $\calF$.
      \item Under Assumption~\ref{ass:process},
        \ref{ass:noise}\ref{ass:noise-dep} and
        \ref{ass:SALE-para}\ref{ass:SALE-para-asym}, suppose that $H_n >
        2q$ and $4a + 2b - 2ab > 3$. Let $r = [(1-a)(b\land (1-b))] \land
        [4a+2b-2ab-3]$. As $n\to\infty$, we have
        \begin{align}
          & \qquad
          \frac{n^{\frac{1}{2}r}}{\sqrt{V_{2,n}}}
          \bigl(
            \esty^{(H_n)}_T - \true_T
          \bigr)
          \convst
          \normal(0, 1),
          \\
          \text{where }
          V_{2,n} &=
          \frac{8 \alpha^{b}}{3 \beta} \widehat{G}_n^{(1)} \cdot
          1_{\{(1-a)b\}}(r)
          +
          \frac{2 \alpha^{1-b} \beta T}{3} \widehat{G}_n^{(2)} \cdot
          1_{\{(1-a)(1-b)\}}(r)
          \notag
          \\ & \qquad
          + \frac{1}{\alpha^{4-2b} \beta^2 T^3} \int_0^T
          \widehat{\Phi}_n \ud t \cdot 1_{\{4a+2b-2ab-3\}}(r),
        \end{align}
        and the limiting distribution is independent of $\calF$.
    \end{enumerate}
  \end{theorem}

  A feasible CLT for MSLE estimators follows similarly.

  \section{Practical Aspects}

  \subsection{A Correction for SALE Estimators}

  A small noise correction of SALE estimators has been given in
  Equation~\eqref{eq:SALE-noise-variance-corrected}.
  Figure~\ref{fig:avar-sources-scale} compares this corrected version
  with the original (dominant) version and the empirical result.

  \begin{figure}[!htp]
    \centering
    \includegraphics[width=0.6\textwidth]{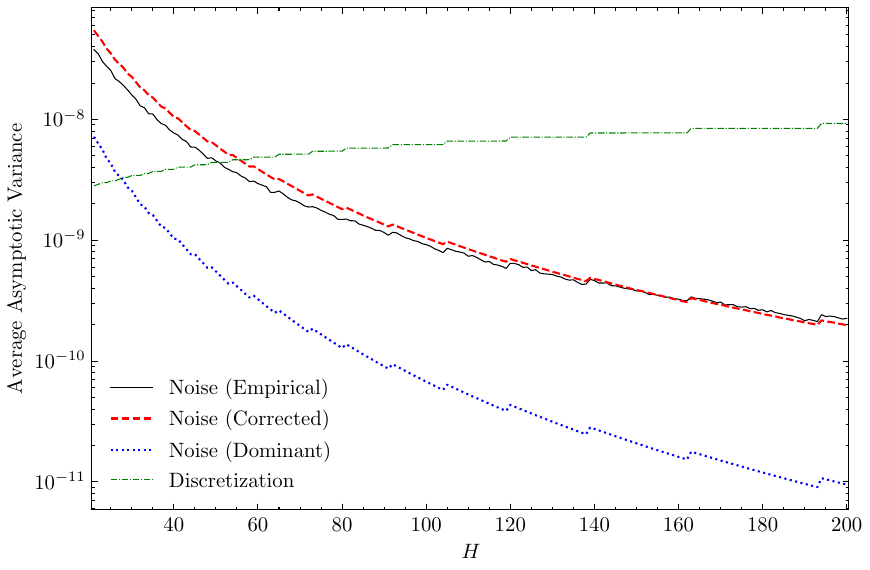}
    \caption{
      Asymptotic variance of SALE estimators. The same simulated paths
      are used as Figure~\ref{fig:avar-noise}, and the noise level is
      $\varsigma=6\times 10^{-4}$.
    }
    \label{fig:avar-sources-scale}
  \end{figure}

  Unlike pure-noise terms, these correction terms are generally not
  independent across scales. Thus, for MSLE estimators, we can use the
  following upper bound for variances due to noise:
  \begin{align}\label{eq:MSLE-noise-variance-upper-bound}
    \var\bigl(
      \sum_{p=1}^{M_n} w_p \esty_T^{(H_p)}
      \big| \calF
    \bigr)
    \leq
    \Biggl(
      \sum_{p=1}^{M_n}
      |w_p|
      \sqrt{\var\bigl( \esty_T^{(H_p)} \big| \calF \bigr)}
    \Biggr)^2.
  \end{align}
  Equation~\eqref{eq:MSLE-noise-variance-upper-bound} is used to
  construct confidence intervals in Figure~\ref{fig:signature} and the
  empirical study.

  \subsection{Derivation of Approximate Weights}

  The following lemma is useful in the derivation of approximate weights.

  \begin{lemma}\label{lem:tridiagonal-inverse}
    Let $n$ be a positive integer, and let $m \in \R$. Define the matrix
    $\bm{B} \in \R^{n\times n}$ by
    \begin{equation}
      B_{p,q}
      =
      f\bigl((m+p) \lor (m+q)\bigr)
      \frac{g\bigl((m+p) \land (m+q)\bigr)}{g\bigl((m+p) \lor
      (m+q)\bigr)}, \quad
      \text{for all } p, q \in \{1, \dotsc, n\},
    \end{equation}
    where $f, g$ are well-defined functions such that all expressions
    in this lemma are valid. Let
    \begin{equation}
      h(p, q) = f(m+p) - \frac{g^2(m+p)}{g^2(m+q)} f(m+q).
    \end{equation}
    The inverse $\bm{A} = \bm{B}^{-1}$ is tridiagonal, with
    superdiagonal and subdiagonal elements given by
    \begin{equation}
      A_{p,p+1} = A_{p+1,p}
      =
      -\frac{g(m+p)}{g(m+p+1)} \frac{1}{h(p,p+1)},
      \quad p = 1, \dotsc, n-1,
    \end{equation}
    and main diagonal elements given by
    \begin{align}
      A_{1,1} &= \frac{1}{h(1,2)}, \\
      A_{p,p} &= \frac{h(p-1, p+1)}{h(p-1, p)h(p, p+1)}, \quad p = 2,
      \dotsc, n-1, \\
      A_{n,n} &= \frac{f(m+n-1)}{f(m+n)}\frac{1}{h(n-1,n)}.
    \end{align}
  \end{lemma}

  For simplicity, denote $x_p = H_p / H_n^*$ for any $p = 1, \dots,
  M_n$. Under the conditions of Definition~\ref{def:approx-weights},
  the asymptotical covariance between scales is given by
  \begin{align}\label{eq:Sigma-disc}
    \Sigma_{p,q} = (s_1 + s_2)
    (n/H_n^*)^{-1/2} (x_p \lor x_q)^{1/2}
    \left[\frac{x_p \land x_q}{x_p \lor x_q} - \frac{s_1}{3(s_1 + s_2)}
    \left(\frac{x_p \land x_q}{x_p \lor x_q}\right)^2\right].
  \end{align}
  Note that the second term within the bracket has a coefficient of
  $s_1 / (3(s_1+s_2)) \in (0, 1/3)$, which depends on the ratio of
  $s_1$ to $s_2$, and the scale ratio term satisfies that $(x_p \land
  x_q) / (x_p \lor x_q) \leq 1$. To formulate the approximate weights
  without relying on estimates of $s_1$ and $s_2$, we only retain the
  first term in the bracket, which contributes more to the covariance.
  Thus, the approximate covariance matrix is given by
  \begin{align}
    \widetilde{\Sigma}_{p,q}
    =
    (s_1 + s_2) (n / H_n^*)^{-1/2}
    (x_p \lor x_q)^{1/2}
    \frac{x_p \land x_q}{x_p \lor x_q}
    =
    (s_1 + s_2) n^{-1/2}
    B_{p,q}.
  \end{align}
  Let $\bm\omega = \bm B^{-1} \bm 1_{M_n}$. Applying
  Lemma~\ref{lem:tridiagonal-inverse} with $f(x) \leftarrow \sqrt{x}$,
  $g(x) \leftarrow x$, $m \leftarrow m_n$, and $n \leftarrow M_n$, for
  $p = 2, \dots, M_n - 1$, we have
  \begin{align}
    \omega_1 &= \left(\frac{m_n+2}{m_n+1}\right)^{1/2}
    \frac{1}{(m_n+2)^{3/2} - (m_n+1)^{3/2}} \\
    \omega_p &= (m_n+p)^{1/2}
    \left(\frac{(m_n+p)^{1/2}-(m_n+p-1)^{1/2}}{(m_n+p)^{3/2}-(m_n+p-1)^{3/2}}
      -
    \frac{(m_n+p+1)^{1/2}-(m_n+p)^{1/2}}{(m_n+p+1)^{3/2}-(m_n+p)^{3/2}}\right),
    \\
    \omega_{M_n} &= (m_n+M_n)^{1/2}
    \frac{(m_n+M_n)^{1/2}-(m_n+M_n-1)^{1/2}}{(m_n+M_n)^{3/2}-(m_n+M_n-1)^{3/2}}.
  \end{align}
  For simplicity, taking the limit $m_n\to \infty$, for $p = 2, \dots,
  M_n - 1$, we have
  \begin{align}
    \omega_1 \to \frac{2}{3(m_n+1)^{1/2}}, \quad
    \omega_p \to \frac{1}{3(m_n+p)^{3/2}}, \quad
    \omega_{M_n} \to \frac{1}{3(m_n+M_n)^{1/2}}.
  \end{align}
  Consequently, by Equation~\eqref{eq:weight-optimization-solution},
  the approximate weights $\bm{\widetilde w}$ in
  Definition~\ref{def:approx-weights} are obtained.

  \subsection{Derivation of Equation~\eqref{eq:fredholm}}

  Recall that $\gamma = s_1 / (3(s_1 + s_2))$. With the simplifications
  in Section~\ref{sec:practical-weight-noisy}, the total asymptotic
  covariances of SALE estimators are given by
  \begin{align}
    \Sigma_{p,q}
    & =
    \Sigma_{p,q}^{(1)} + \Sigma_{p,q}^{(2)},
    \\
    \text{where} \quad
    \Sigma_{p,q}^{(1)}
    & =
    (s_1 + s_2)
    (n/H_n^*)^{-1/2}
    (x_p \lor x_q)^{1/2}
    \left[
      \frac{x_p \land x_q}{x_p \lor x_q} -
      \gamma \left(\frac{x_p \land x_q}{x_p \lor x_q}\right)^2
    \right],
    \\
    \text{and} \quad
    \Sigma_{p,q}^{(2)}
    & =
    \alpha^{9/2} s_3 n^2 (H_n^*)^{-3} x_p^{-3} \delta_{p,q}.
  \end{align}
  Thus, the asymptotic variance of the MSLE estimator is
  \begin{align}
    \sum_{p=1}^{M_n} \sum_{q=1}^{M_n}
    w_p w_q \Sigma_{p,q}
    =
    \sum_{p=1}^{M_n} \sum_{q=1}^{M_n}
    w_p w_q \Sigma_{p,q}^{(1)}
    +
    \sum_{p=1}^{M_n} \sum_{q=1}^{M_n}
    w_p w_q \Sigma_{p,q}^{(2)},
  \end{align}
  where
  \begin{align}
    &
    \sum_{p=1}^{M_n} \sum_{q=1}^{M_n}
    w_p w_q \Sigma_{p,q}^{(1)}
    \notag
    \\
    & =
    (s_1 + s_2) \left( \frac{n}{H_n^*} \right)^{-1/2}
    \sum_{p=1}^{M_n} \sum_{q=1}^{M_n}
    \left( \frac1{H_n^*} \right)^2 \phi(x_p) \phi(x_q)
    (x_p \lor x_q)^{1/2}
    \left[
      \frac{x_p \land x_q}{x_p \lor x_q} -
      \gamma \left(\frac{x_p \land x_q}{x_p \lor x_q}\right)^2
    \right]
    \notag
    \\
    & \to
    (s_1 + s_2) \left( \frac{n}{H_n^*} \right)^{-1/2}
    \int_c^1 \int_c^1
    \ud x \ud y
    \phi(x) \phi(y)
    (x \lor y)^{1/2}
    \left[
      \frac{x \land y}{x \lor y} -
      \gamma \left(\frac{x \land y}{x \lor y}\right)^2
    \right]
    \notag
    \\
    & =
    (s_1 + s_2) \alpha^{1/2} n^{-2/9}
    \int_c^1 \int_c^1
    \phi(x) \phi(y)
    (x \lor y)^{1/2}
    \left[
      \frac{x \land y}{x \lor y} -
      \gamma \left(\frac{x \land y}{x \lor y}\right)^2
    \right]
    \ud x \ud y,
  \end{align}
  and
  \begin{align}
    &
    \sum_{p=1}^{M_n} \sum_{q=1}^{M_n}
    w_p w_q \Sigma_{p,q}^{(2)}
    \notag
    \\
    & =
    \alpha^{9/2} s_3 n^2 (H_n^*)^{-3}
    \sum_{p=1}^{M_n} \sum_{q=1}^{M_n}
    \left( \frac1{H_n^*} \right)^2 \phi(x_p) \phi(x_q)
    x_p^{-3} \delta_{p,q}
    \notag
    \\
    & =
    \alpha^{9/2} s_3 n^2 (H_n^*)^{-4}
    \sum_{p=1}^{M_n}
    \left( \frac1{H_n^*} \right) \phi^2(x_p)
    x_p^{-3}
    \notag
    \\
    & \to
    \alpha^{9/2} s_3 n^2 (H_n^*)^{-4}
    \int_c^1
    \ud x
    \phi^2(x) x^{-3}
    \notag
    \\
    & =
    s_3 \alpha^{1/2} n^{-2/9}
    \int_c^1
    \phi^2(x) x^{-3}
    \ud x.
  \end{align}
  It follows that
  \begin{align}
    \label{eq:MSLE-total-variance-functional}
    \sum_{p=1}^{M_n} \sum_{q=1}^{M_n}
    w_p w_q \Sigma_{p,q}
    =
    \alpha^{1/2} n^{-2/9}
    S_0 [\phi],
  \end{align}
  where $S_0 [\phi]$ is a functional of $\phi(x)$:
  \begin{align}
    S_0 [\phi]
    & =
    (s_1 + s_1)
    \int_c^1 \int_c^1
    \phi(x) \phi(y)
    (x \lor y)^{1/2}
    \left[
      \frac{x \land y}{x \lor y} -
      \gamma \left(\frac{x \land y}{x \lor y}\right)^2
    \right]
    \ud x \ud y
    \notag
    \\
    & \qquad
    + s_3
    \int_c^1
    \phi^2(x) x^{-3}
    \ud x
    - 2k
    \left(
      \int_c^1 \phi(x) \ud x - 1
    \right).
  \end{align}
  Here, $k\in \R$ is a Lagrange multiplier that enforces the constraint
  on $\phi(x)$ in
  Assumption~\ref{ass:MSLE-para-detail}\ref{ass:MSLE-para-detail-weight}.
  Note that Equation~\eqref{eq:MSLE-total-variance-functional} again
  confirms the rate of convergence $n^{-1/9}$ of the MSLE estimator,
  as the functional $S_0 [\phi]$ has a finite value. Recall that
  $\varphi(x) = x^{-3/2} \phi(x)$ and $\lambda = -(s_1 + s_2) / s_3$.
  Scale $S_0 [\phi]$ with $s_3$ to obtain that
  \begin{align}
    S [\varphi]
    &=
    s_3^{-1} S_0 [\phi]
    \notag
    \\
    &=
    - \int_c^1 \int_c^1
    \varphi(x) \varphi(y)
    K(x, y)
    \ud x \ud y
    + \int_0^1
    \varphi(x) \bigl(
      \varphi(x) - 2kx^{3/2}
    \bigr)
    \ud x
    + \mathrm{const.},
  \end{align}
  where $K(x, y)$ is defined in Equation~\eqref{eq:fredholm-kernel}. A
  necessary condition for the extremum of $S[\varphi]$ is that,
  $\varphi(x)$ is a solution to a Fredholm integral equation of the
  second kind \citep[for example,
  see][]{lao2021FundamentalTheoriesTheir}, which is given by
  Equation~\eqref{eq:fredholm}.

  \section{Simulation Study}

  \subsection{Moments of Dependent Noise}

  Since higher order autocovariances of noise variables are used in our
  method, we further describe these quantities. Recall that, both MA(2)
  and AR(1) processes driven by Gaussian white noise are used to
  generate dependent noise. Specifically, AR(1) process does not
  satisfy Assumption~\ref{ass:noise}, but given its fast decaying speed
  and wide appearance in literature, we include the case to demonstrate
  the robustness of our methods.
  \begin{enumerate}
    \item For the MA(2) process, we set
      \begin{align}
        \eps_{i} = e_{i} + \theta_{1}e_{i-1} + \theta_{2}e_{i-2}, \quad
        e_i \overset{\text{iid}}{\sim} \normal(0,
        \varsigma^2(1+\theta_1^2+\theta_2^2)^{-1}).
      \end{align}
      The second and fourth moments of the process are
      \begin{align}
        \nu_2 = \varsigma^2, \quad
        \nu_4 =
        2\varsigma^4(1+2\theta_1^2+2\theta_2^2+\theta_1^4+\theta_2^4+2\theta_1^2\theta_2^2),
      \end{align}
      and the generalized ACFs defined in Equation~\eqref{eq:general-acf} are
      \begin{align}
        \rho_2(l) &= 1_{\{l=0\}} +
        \frac{\theta_{1}+\theta_{1}\theta_{2}}{1+\theta_{1}^{2}+\theta_{2}^{2}}1_{\{|l|=1\}}
        + \frac{\theta_{2}}{1+\theta_{1}^{2}+\theta_{2}^{2}}1_{\{|l|=2\}},
        \\
        \rho_3(l) &= 0,
        \\
        \rho_4(l) &= 1_{\{l=0\}} +
        \frac{\theta_{1}^{2}\theta_{2}^{2}+\theta_{1}^{2}+2\theta_{1}^{2}\theta_{2}}
        {1+2\theta_1^2+2\theta_2^2+\theta_1^4+\theta_2^4+2\theta_1^2\theta_2^2}
        1_{\{|l|=1\}}
        \notag
        \\
        &\qquad \qquad
        +
        \frac{\theta_2^2}{1+2\theta_1^2+2\theta_2^2+\theta_1^4+\theta_2^4+2\theta_1^2\theta_2^2}
        1_{\{|l|=2\}}.
      \end{align}
    \item For the AR(1) process, we set
      \begin{align}
        \eps_i = \phi \eps_{i-1} + e_i, \quad e_i
        \overset{\text{iid}}{\sim} \normal(0, \varsigma^2\sqrt{1-\phi^2}),
      \end{align}
      where $\phi \in (-1, 1)$ is the autoregressive coefficient. This
      is a common setting for modeling the stationary component of
      dependent noise
      \citep[see][]{jacod2017StatisticalPropertiesMicrostructure,
      li2022ReMeDIMicrostructureNoise}. The moments of the process are
      \begin{align}
        \nu_2 = \varsigma^2, \quad
        \nu_4 = 3\varsigma^4,
      \end{align}
      and the generalized ACFs are
      \begin{align}
        \rho_2(l) = \phi^{|l|}, \quad
        \rho_3(l) = 0, \quad
        \rho_4(l) = \phi^{2|l|}.
      \end{align}
  \end{enumerate}

  \subsection{Asymptotic Normality}

  The Q-Q plots of the infeasible and feasible standardized estimation
  errors are presented in
  Figures~\ref{fig:infeasible-standardized-errors} and
  \ref{fig:feasible-standardized-errors}, respectively.

  \begin{figure}[!p]
    \centering
    \subfloat[Noise-free
    LE]{\includegraphics[width=0.33\textwidth]{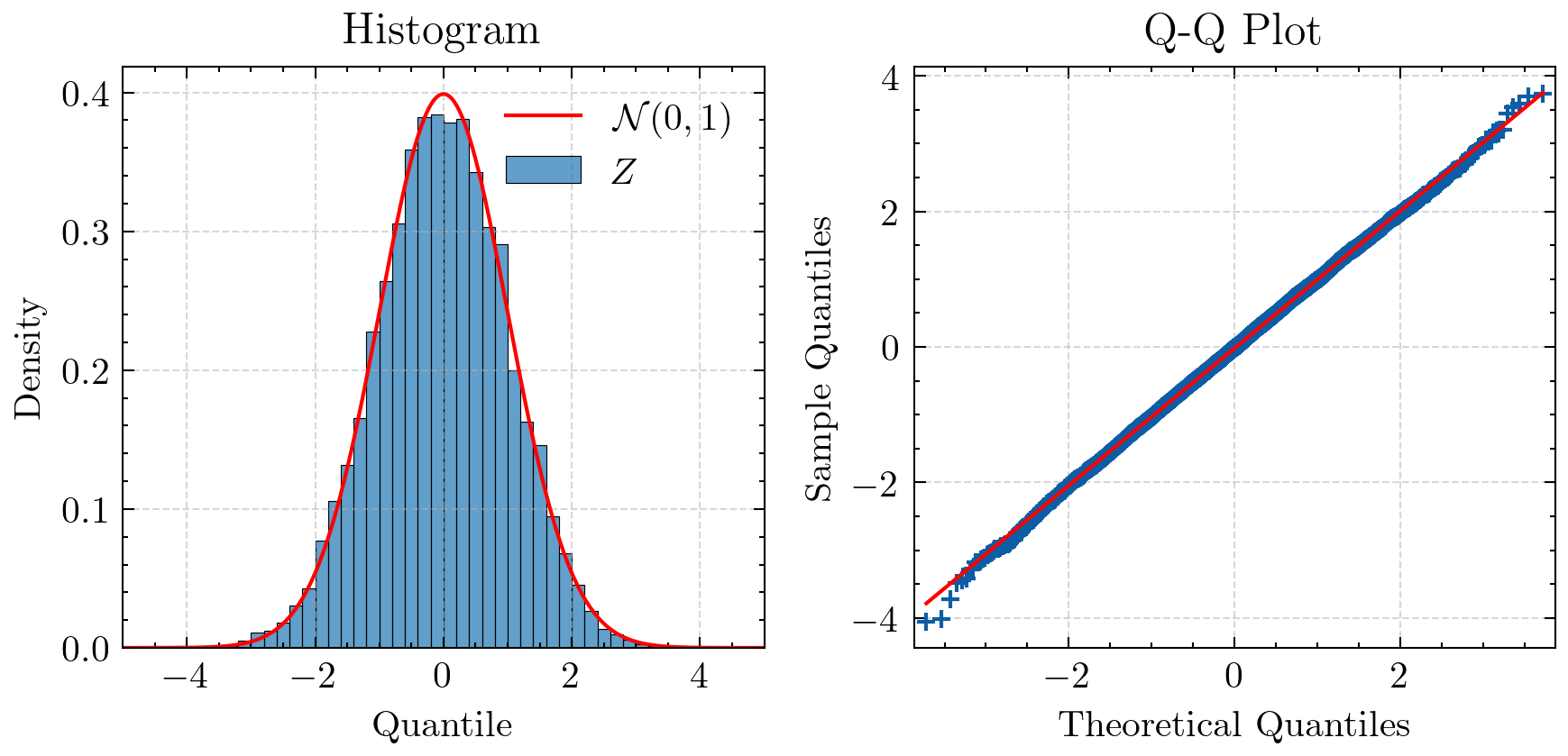}}
    \hfill
    \subfloat[Noise-free
    SALE]{\includegraphics[width=0.33\textwidth]{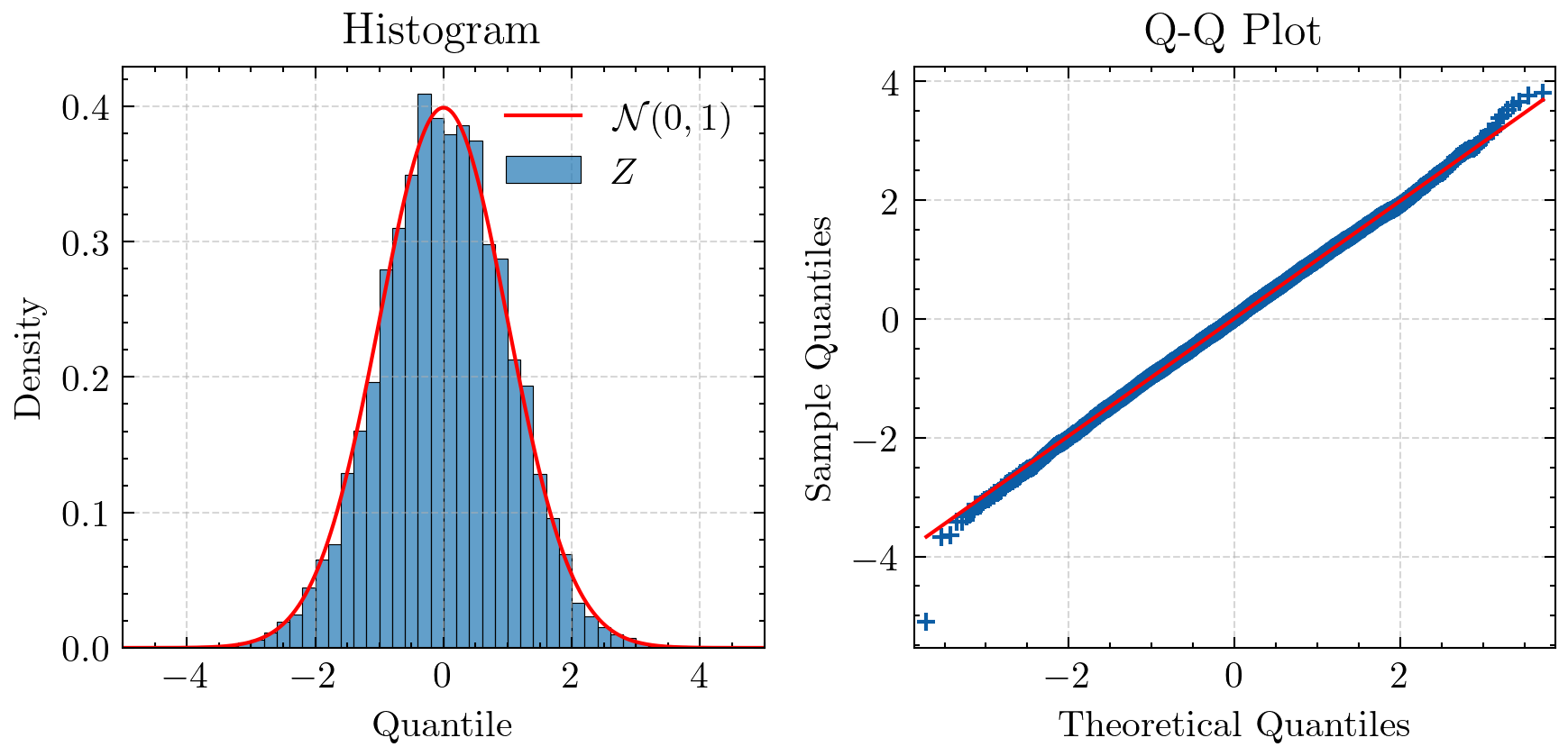}}
    \hfill
    \subfloat[Noise-free
    MSLE]{\includegraphics[width=0.33\textwidth]{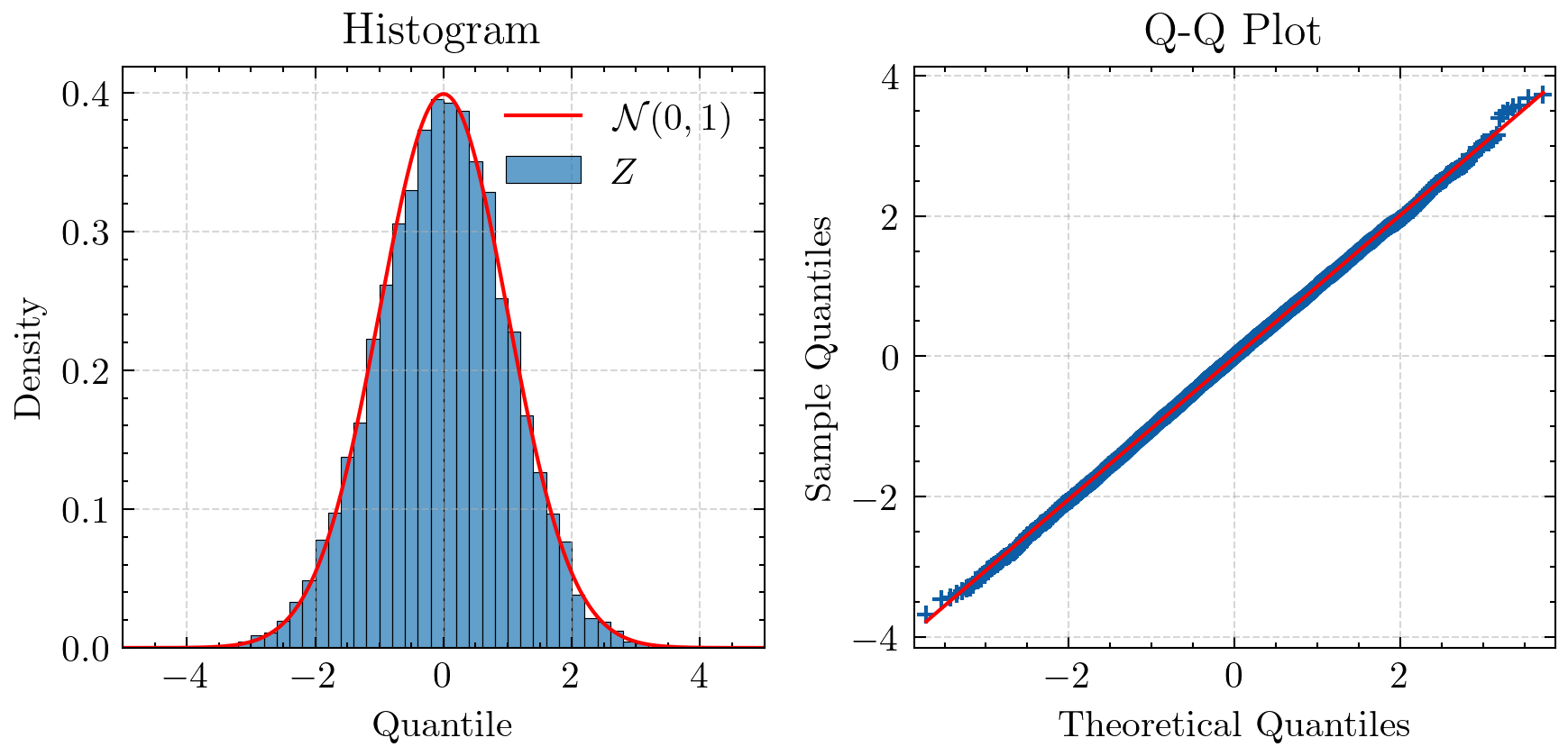}}
    \\
    \subfloat[Normal
    LE]{\includegraphics[width=0.33\textwidth]{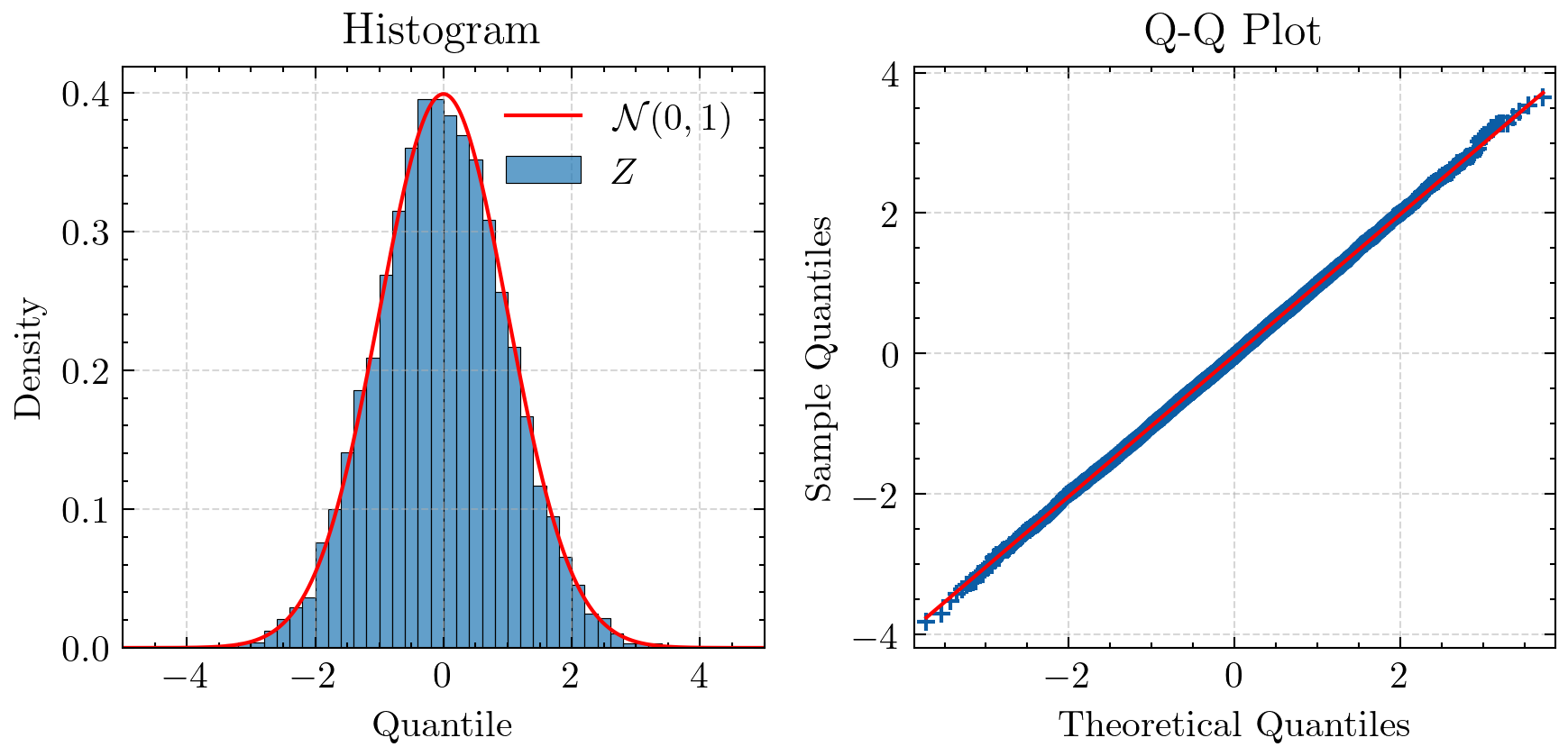}}
    \hfill
    \subfloat[Normal
    SALE]{\includegraphics[width=0.33\textwidth]{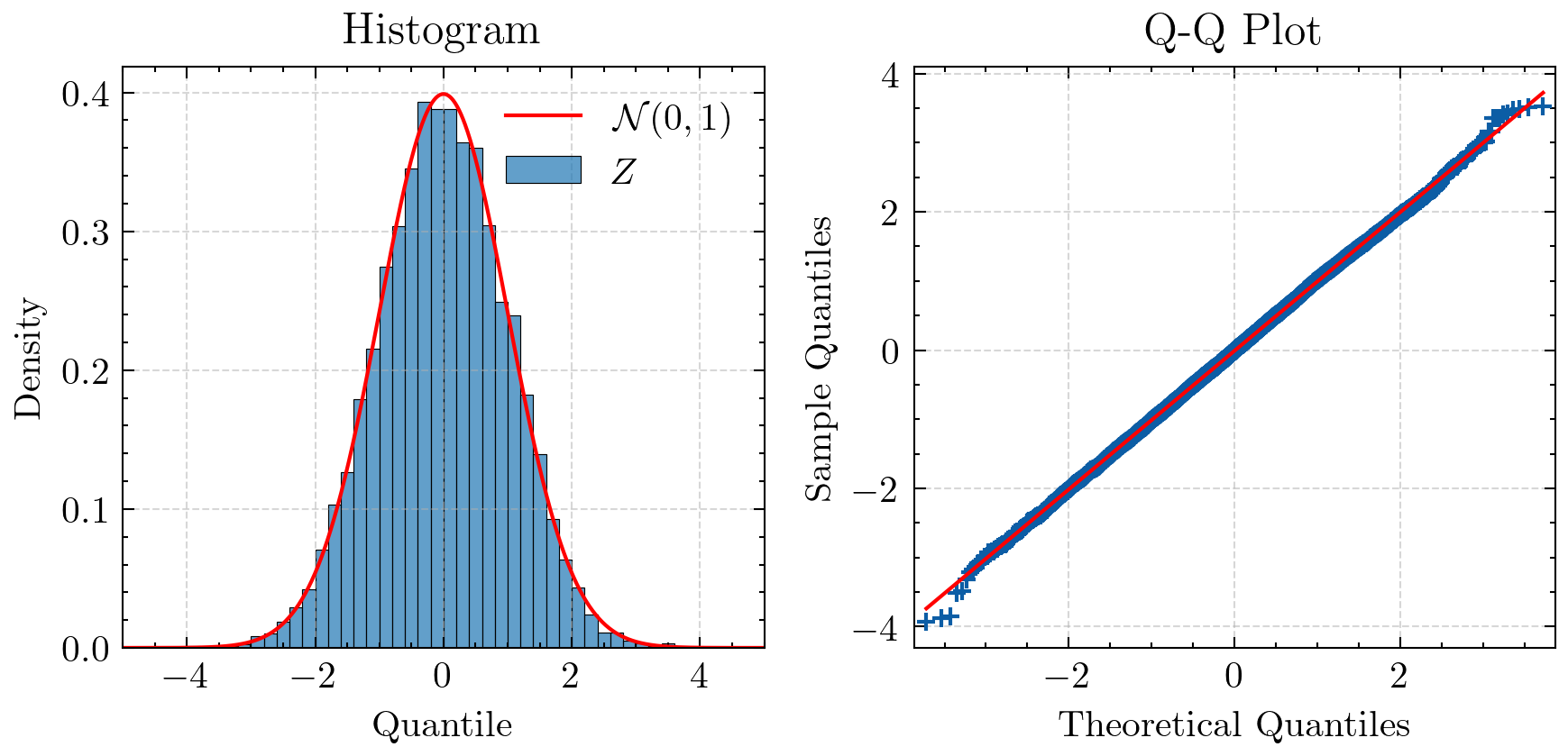}}
    \hfill
    \subfloat[Normal
    MSLE]{\includegraphics[width=0.33\textwidth]{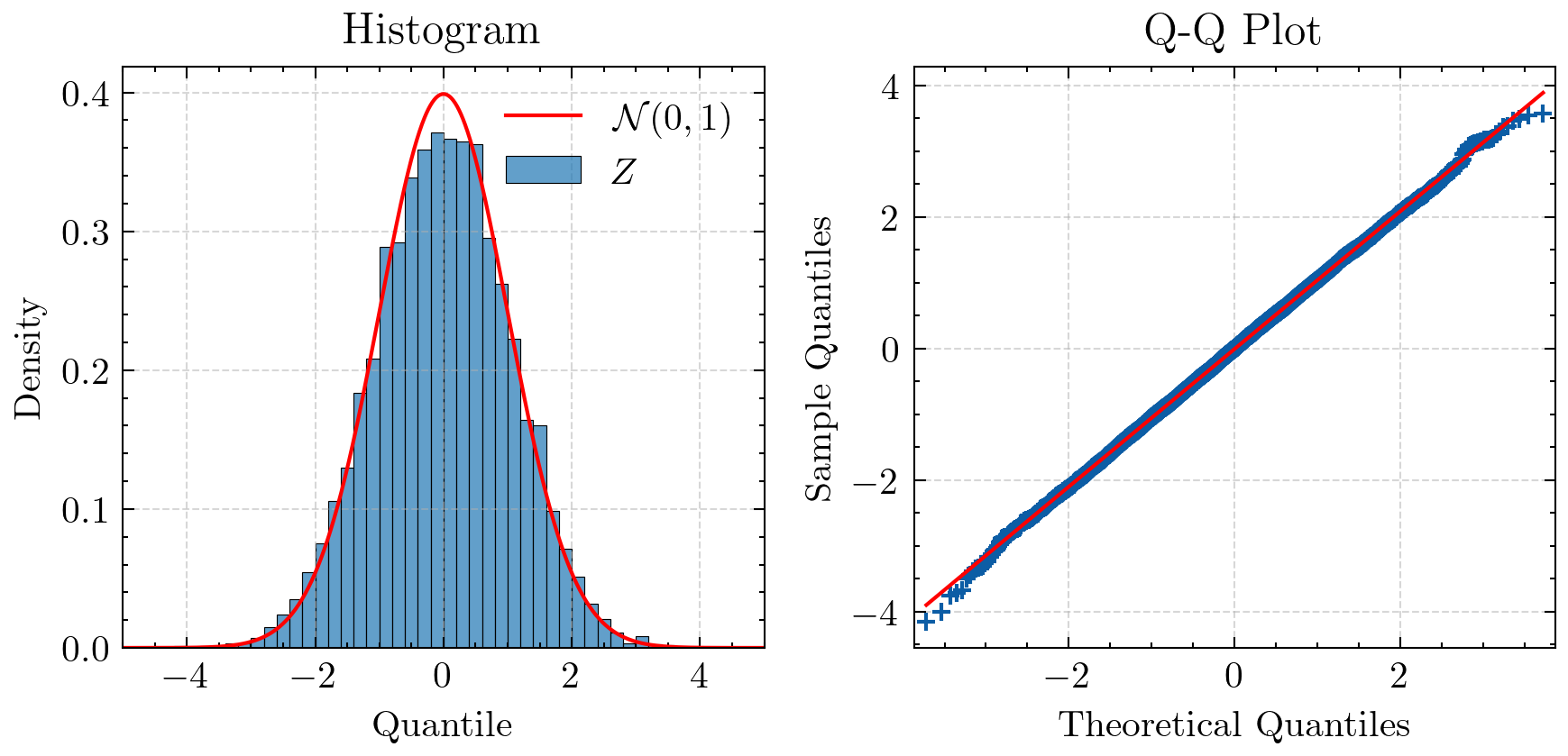}}
    \\
    \subfloat[Uniform
    LE]{\includegraphics[width=0.33\textwidth]{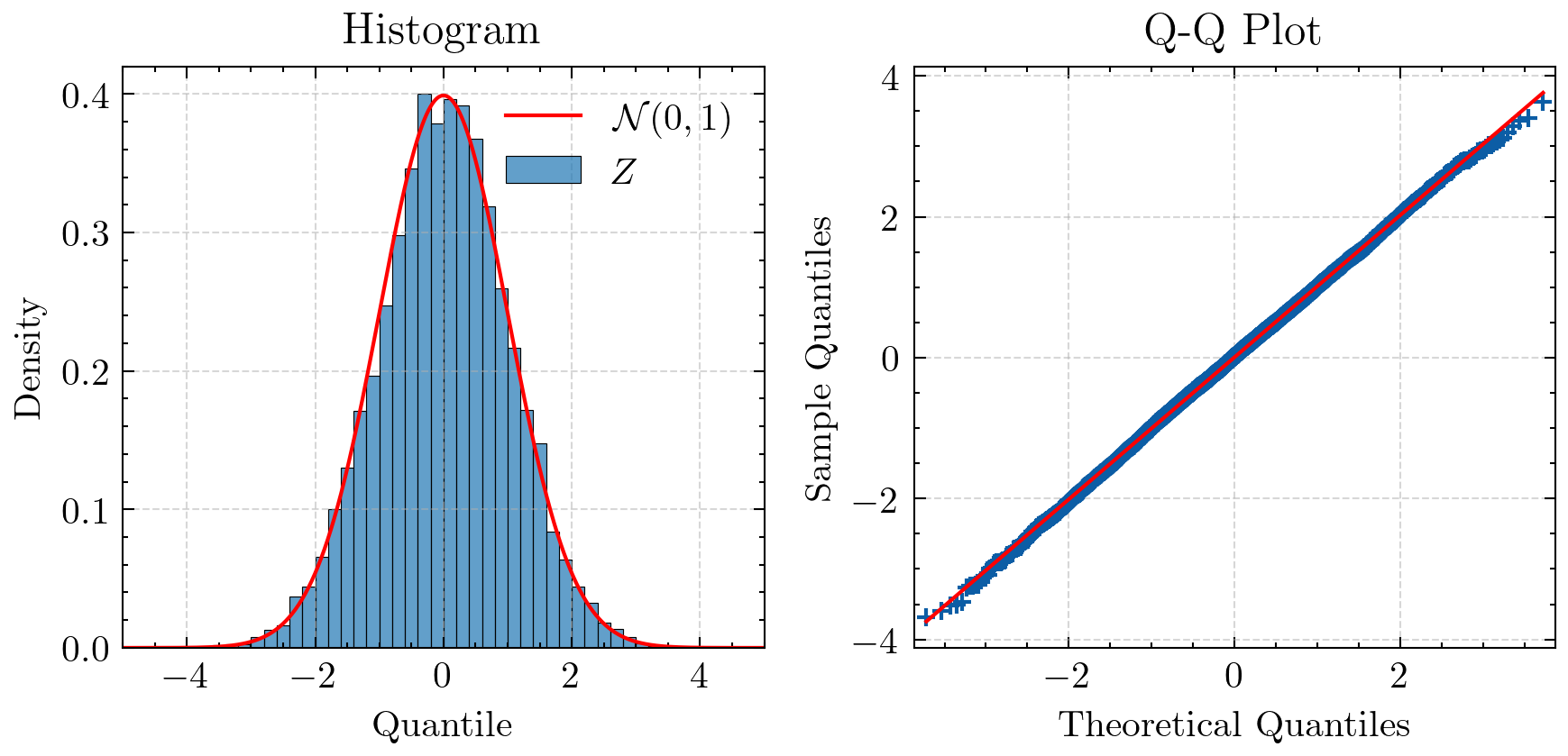}}
    \hfill
    \subfloat[Uniform
    SALE]{\includegraphics[width=0.33\textwidth]{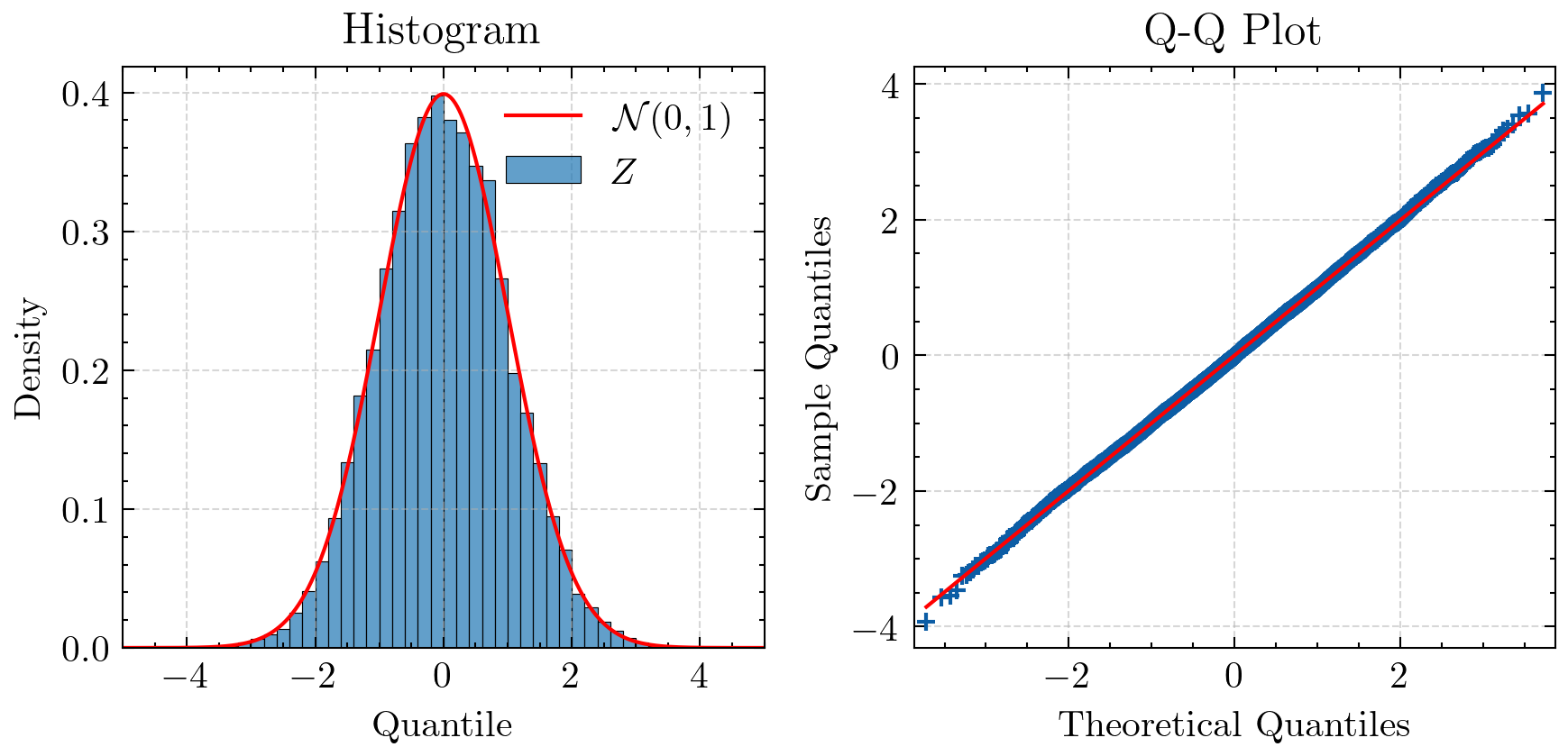}}
    \hfill
    \subfloat[Uniform
    MSLE]{\includegraphics[width=0.33\textwidth]{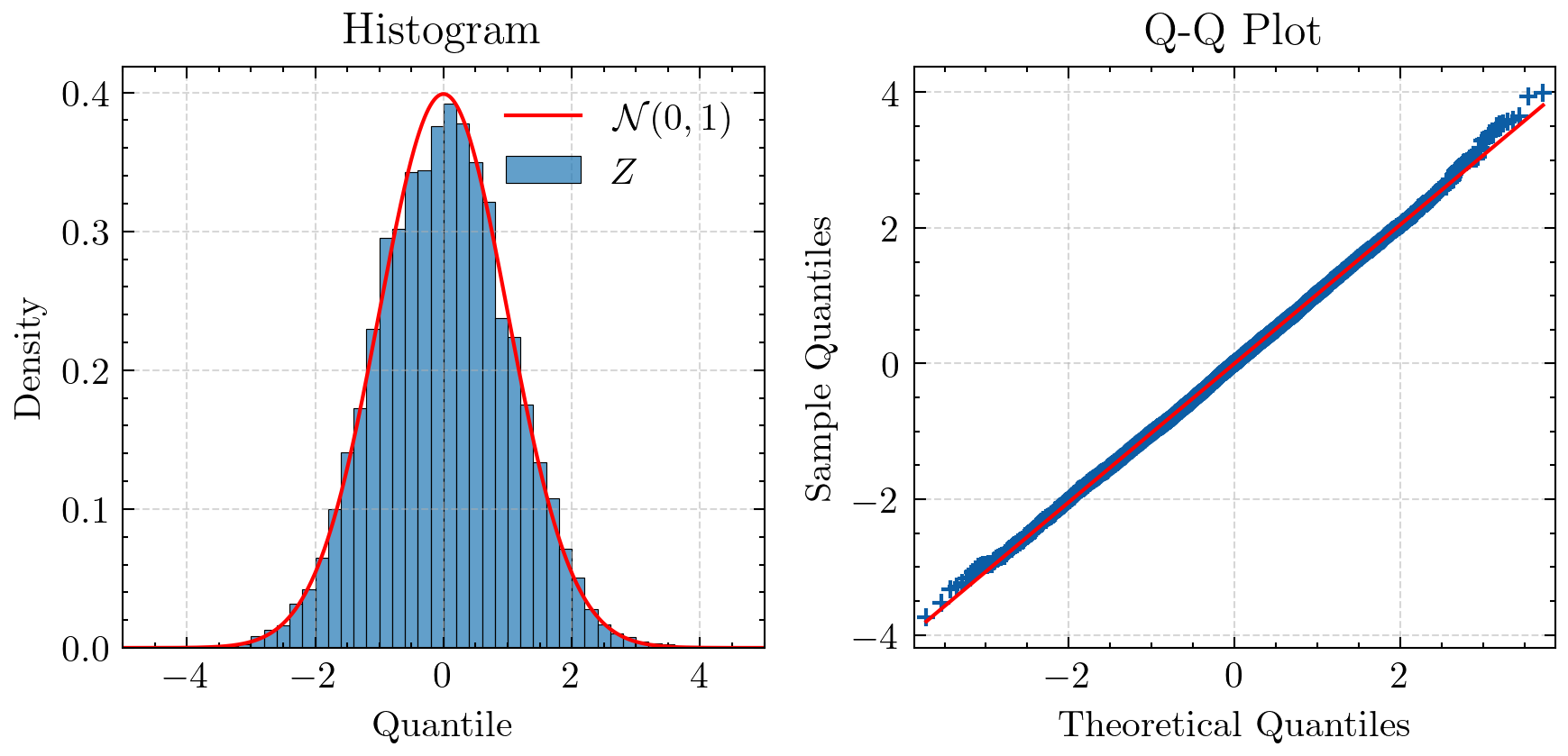}}
    \\
    \subfloat[Skew-normal
    LE]{\includegraphics[width=0.33\textwidth]{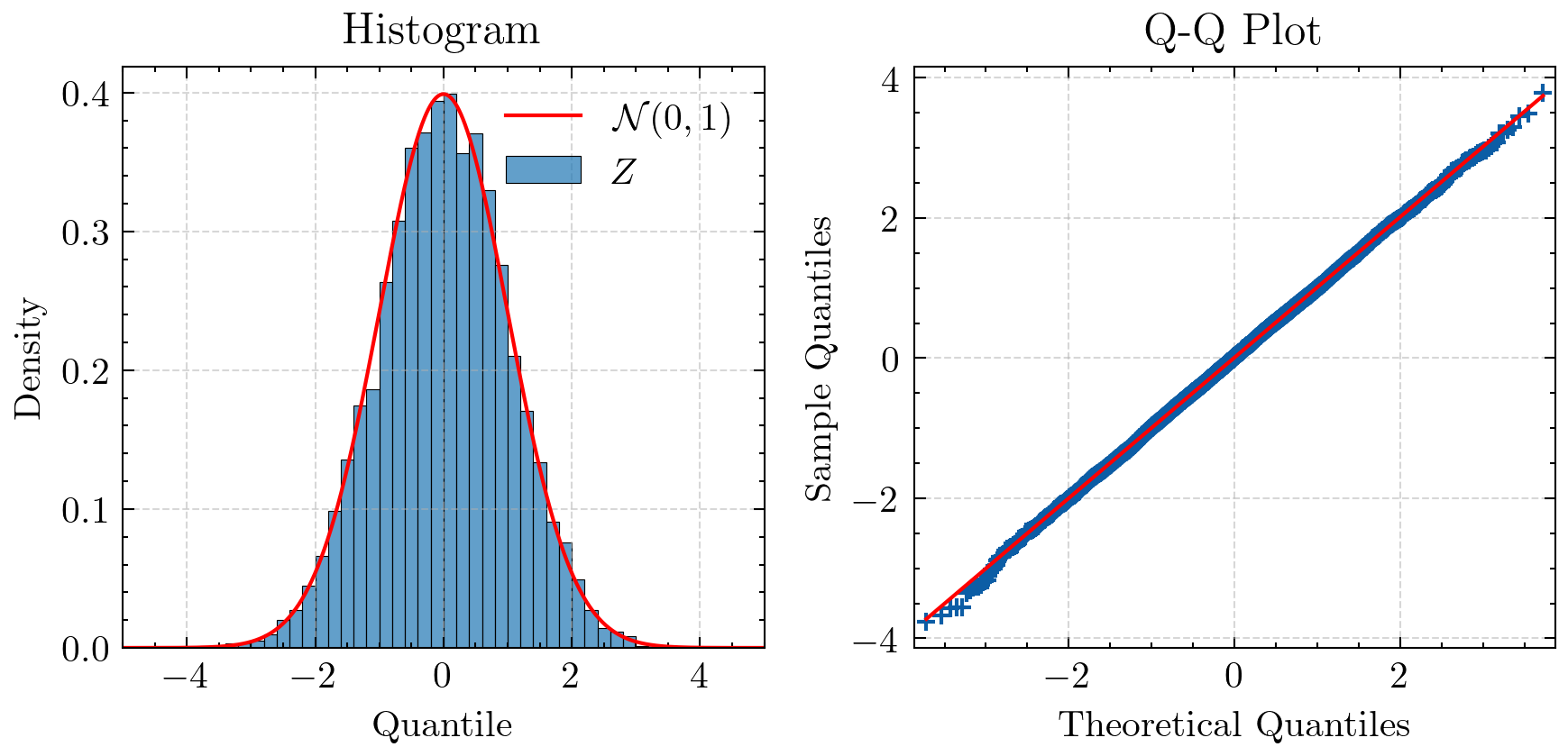}}
    \hfill
    \subfloat[Skew-normal
    SALE]{\includegraphics[width=0.33\textwidth]{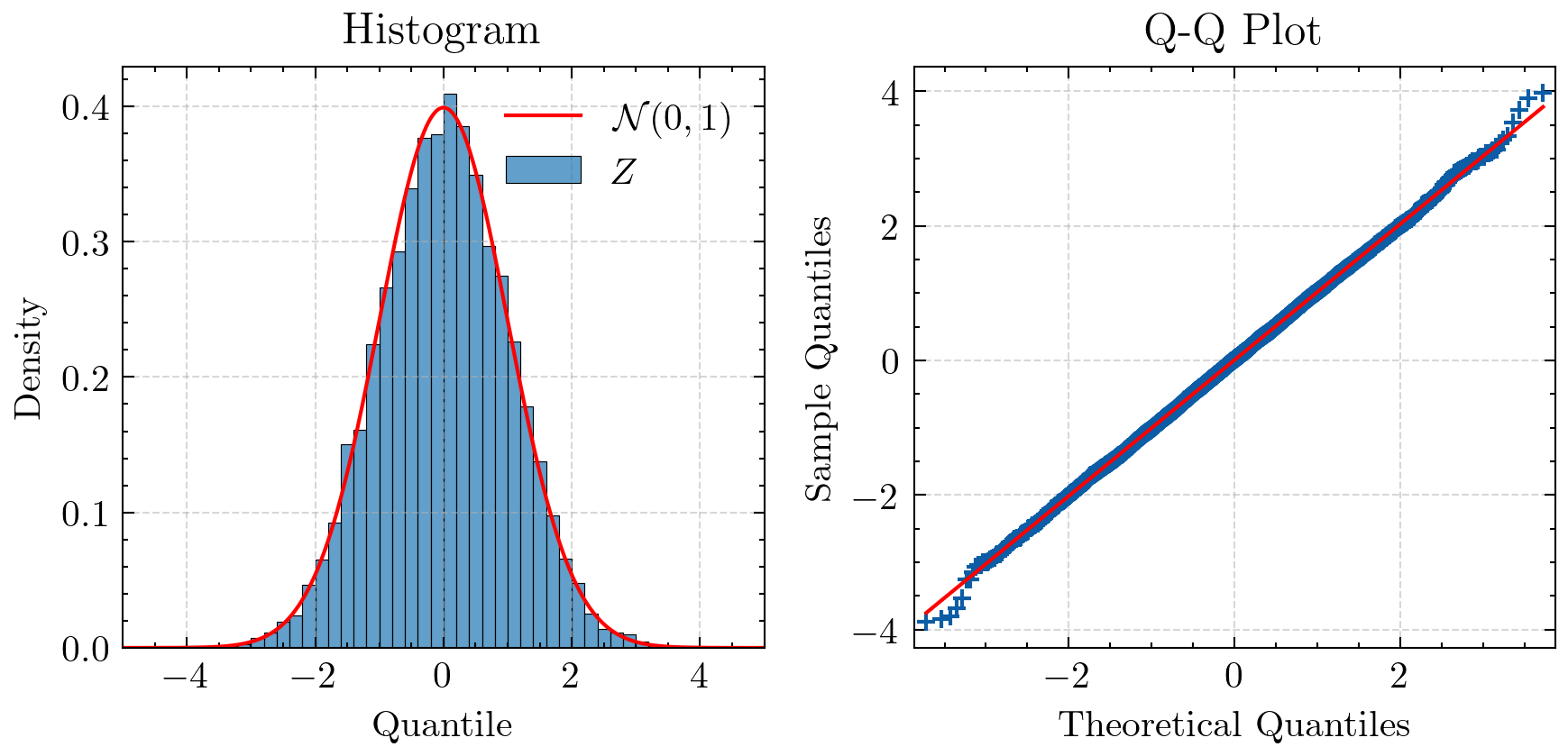}}
    \hfill
    \subfloat[Skew-normal
    MSLE]{\includegraphics[width=0.33\textwidth]{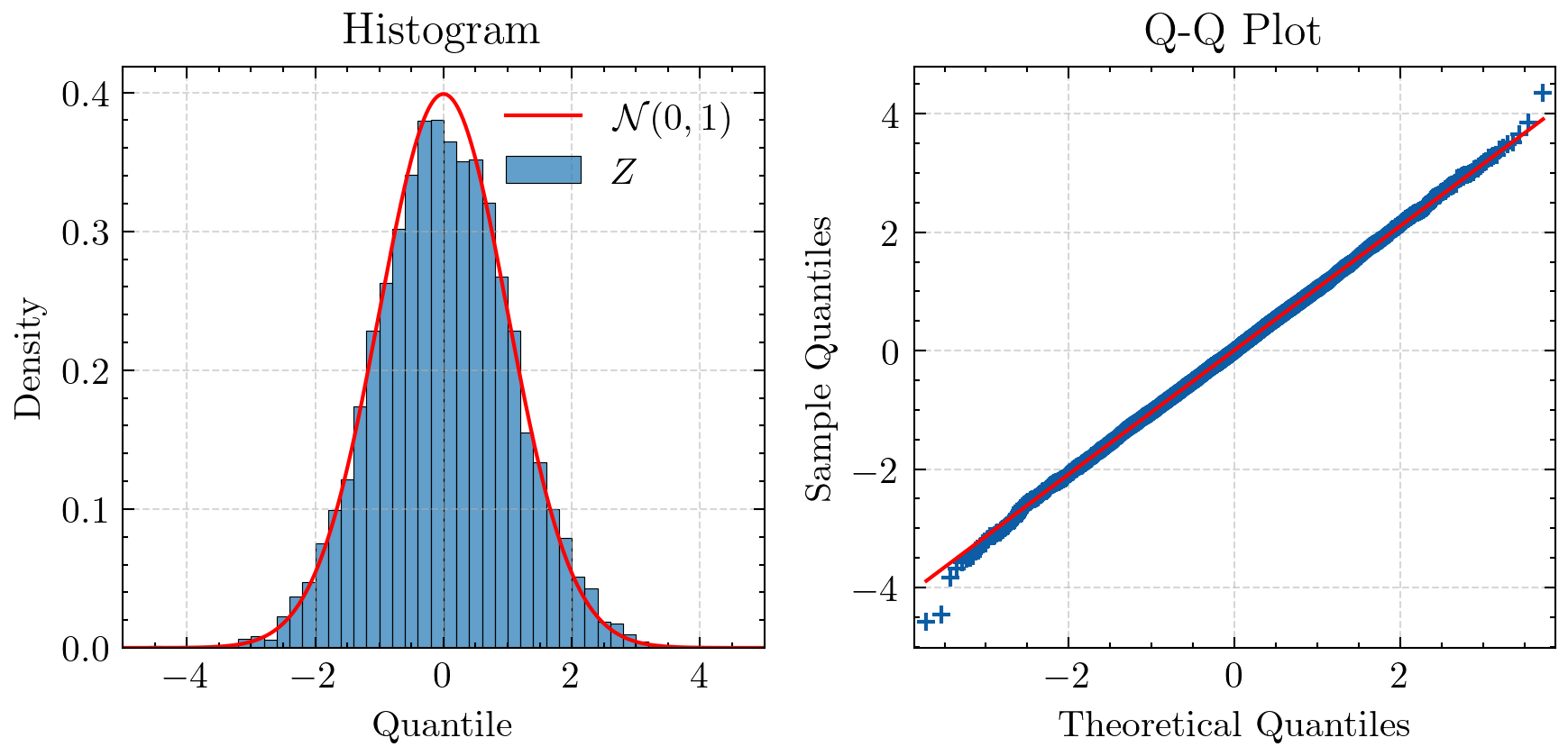}}
    \\
    \subfloat[MA(2) ($\theta_1=0.7$)
    SALE]{\includegraphics[width=0.33\textwidth]{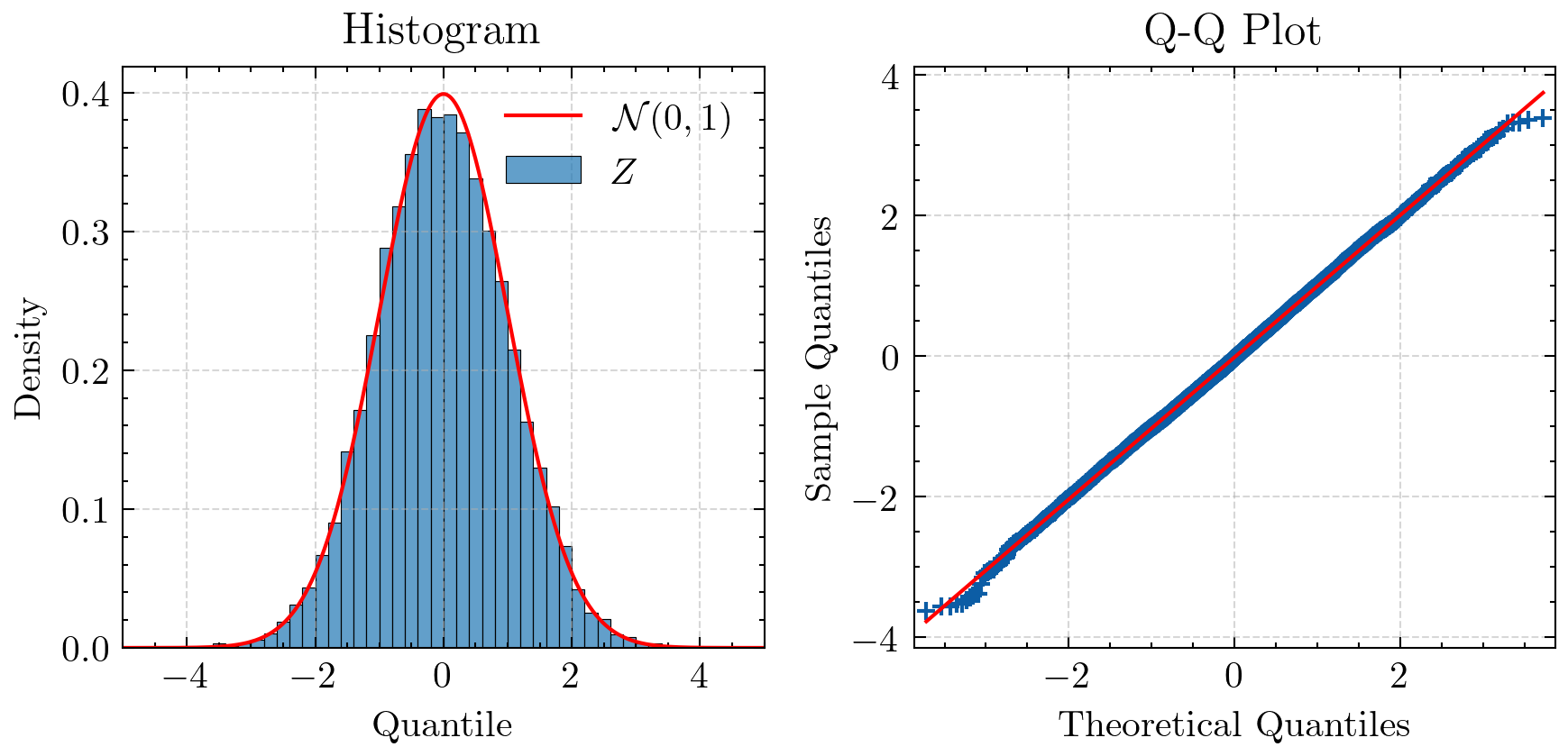}}
    \hfill
    \subfloat[MA(2) ($\theta_1=0.7$)
    MSLE]{\includegraphics[width=0.33\textwidth]{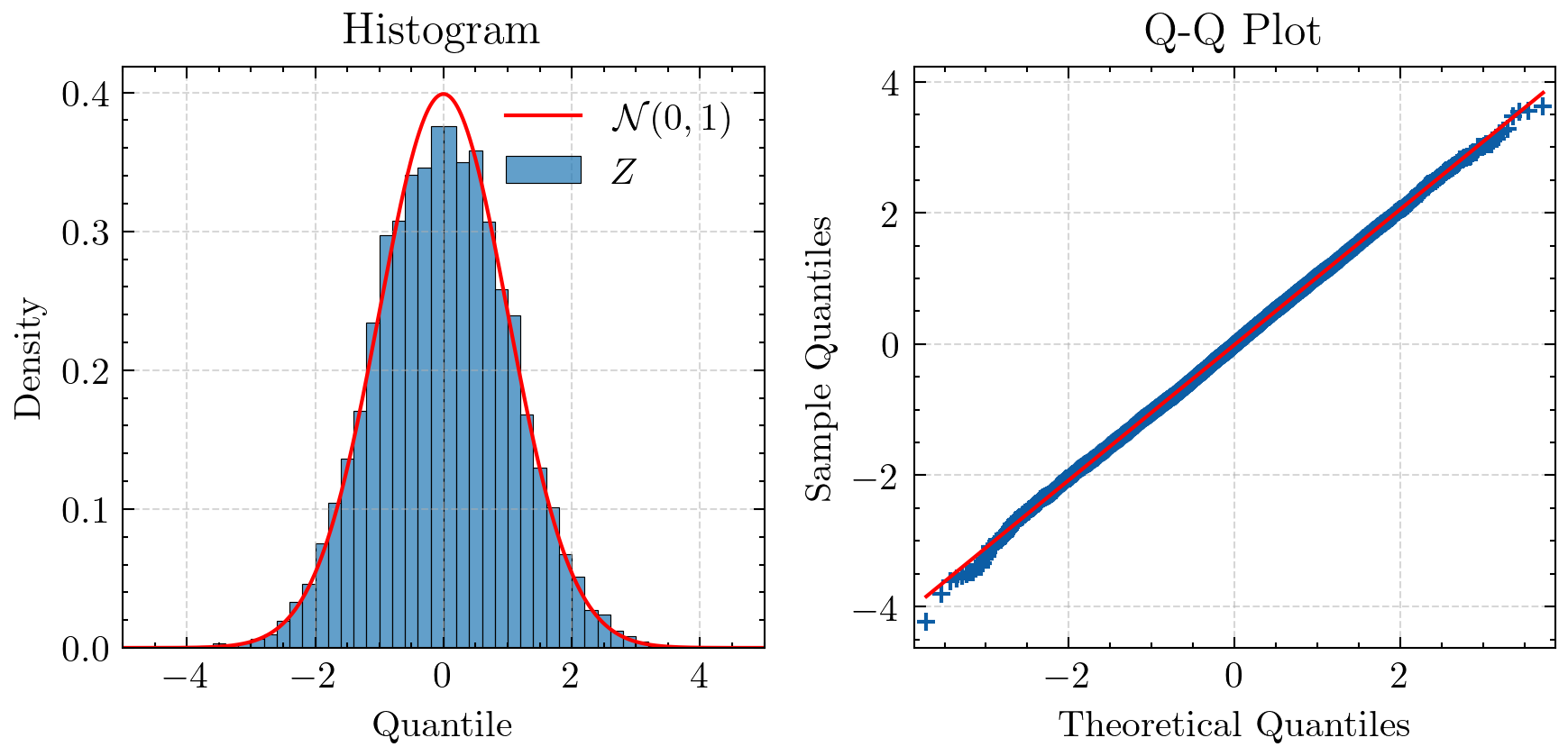}}
    \hfill
    \subfloat[MA(2) ($\theta_1=-0.7$)
    SALE]{\includegraphics[width=0.33\textwidth]{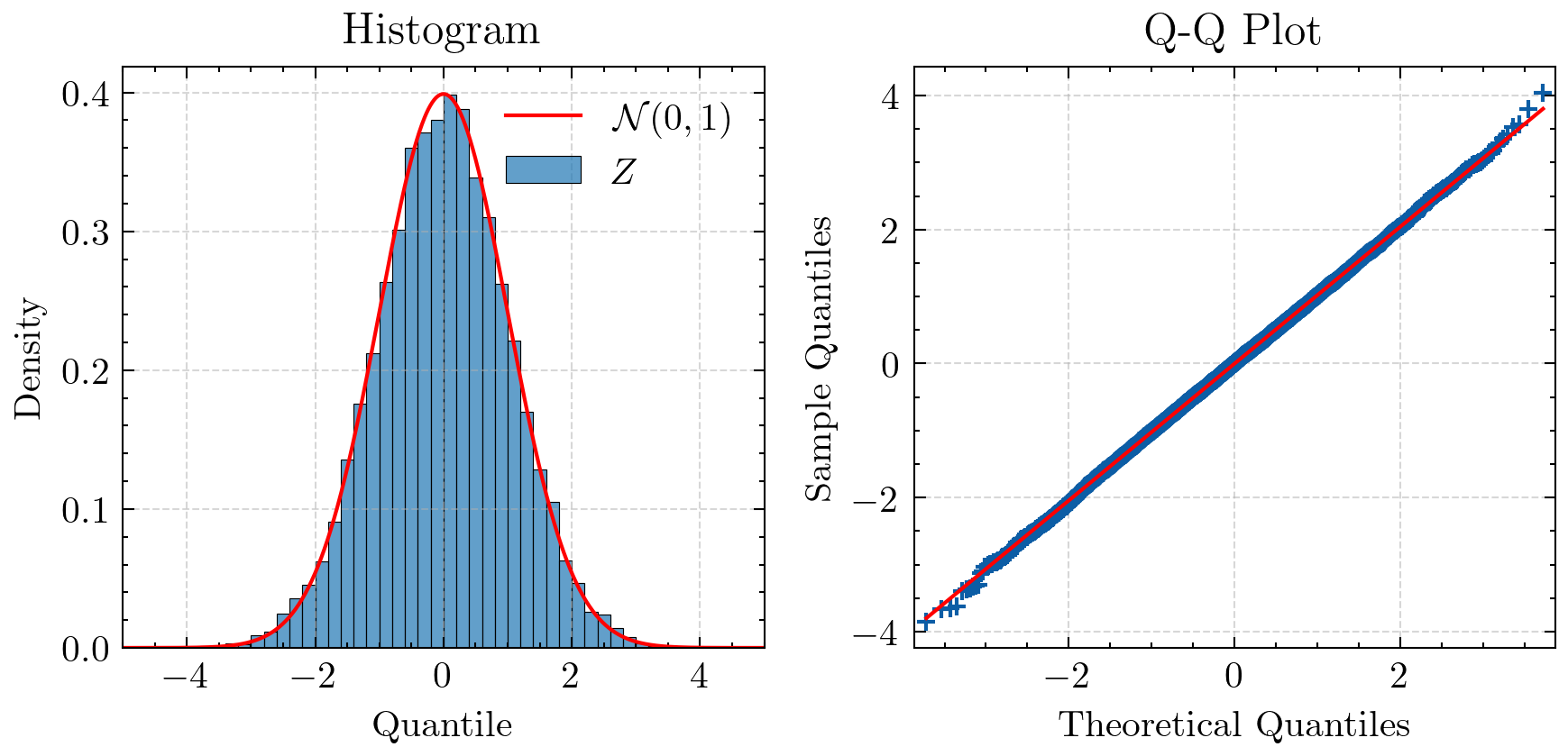}}
    \\
    \subfloat[MA(2) ($\theta_1=-0.7$)
    MSLE]{\includegraphics[width=0.33\textwidth]{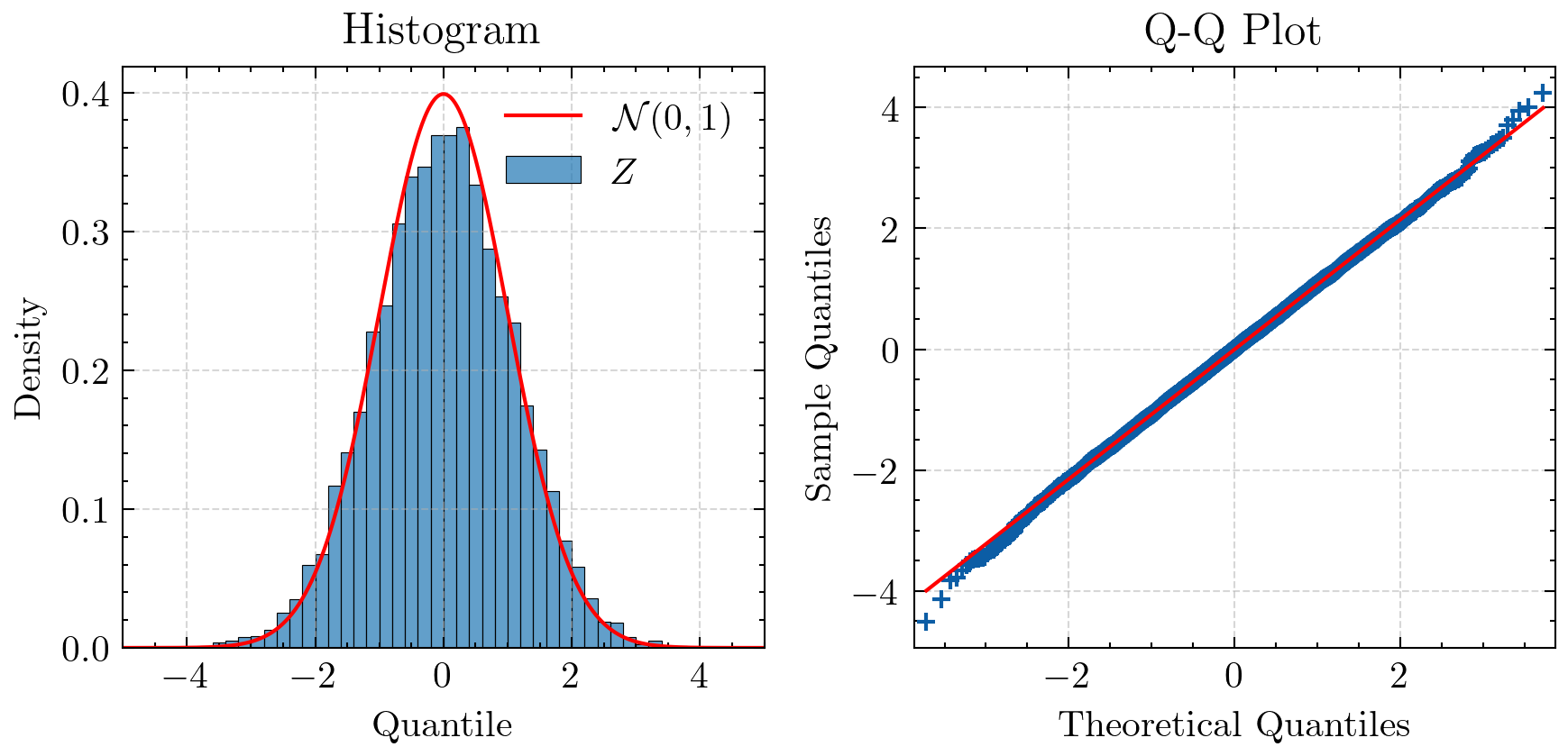}}
    \hfill
    \subfloat[AR(1)
    SALE]{\includegraphics[width=0.33\textwidth]{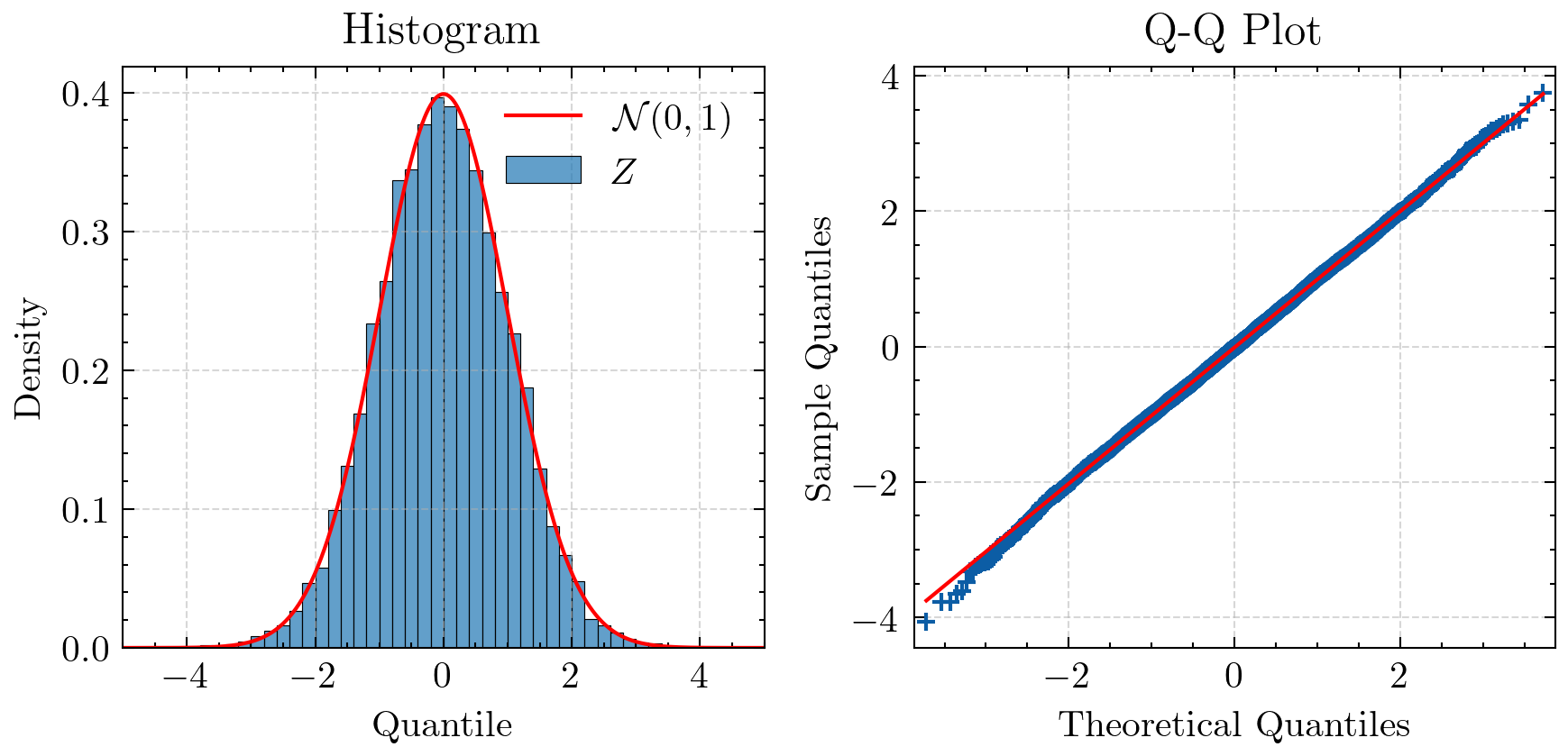}}
    \hfill
    \subfloat[AR(1)
    MSLE]{\includegraphics[width=0.33\textwidth]{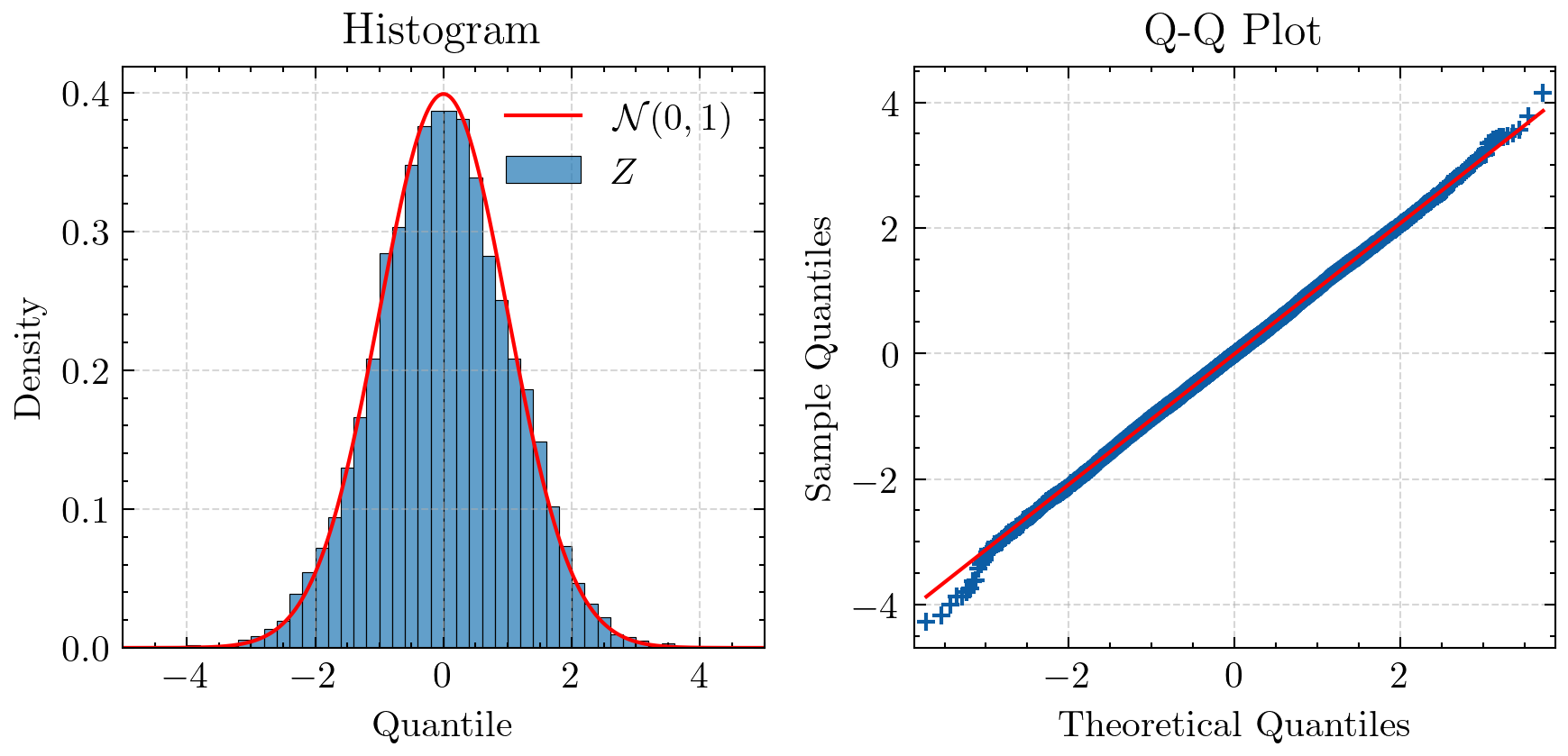}}
    \caption{Histograms and Q-Q plots of standardized estimation errors
      calculated with infeasible asymptotic variances across noise
      settings and estimators. Blue bars and points represent empirical
      standardized errors; red curves and lines represent the standard
    normal distribution.}
    \label{fig:infeasible-standardized-errors}
  \end{figure}

  \begin{figure}[!p]
    \centering
    \subfloat[Noise-free
    LE]{\includegraphics[width=0.33\textwidth]{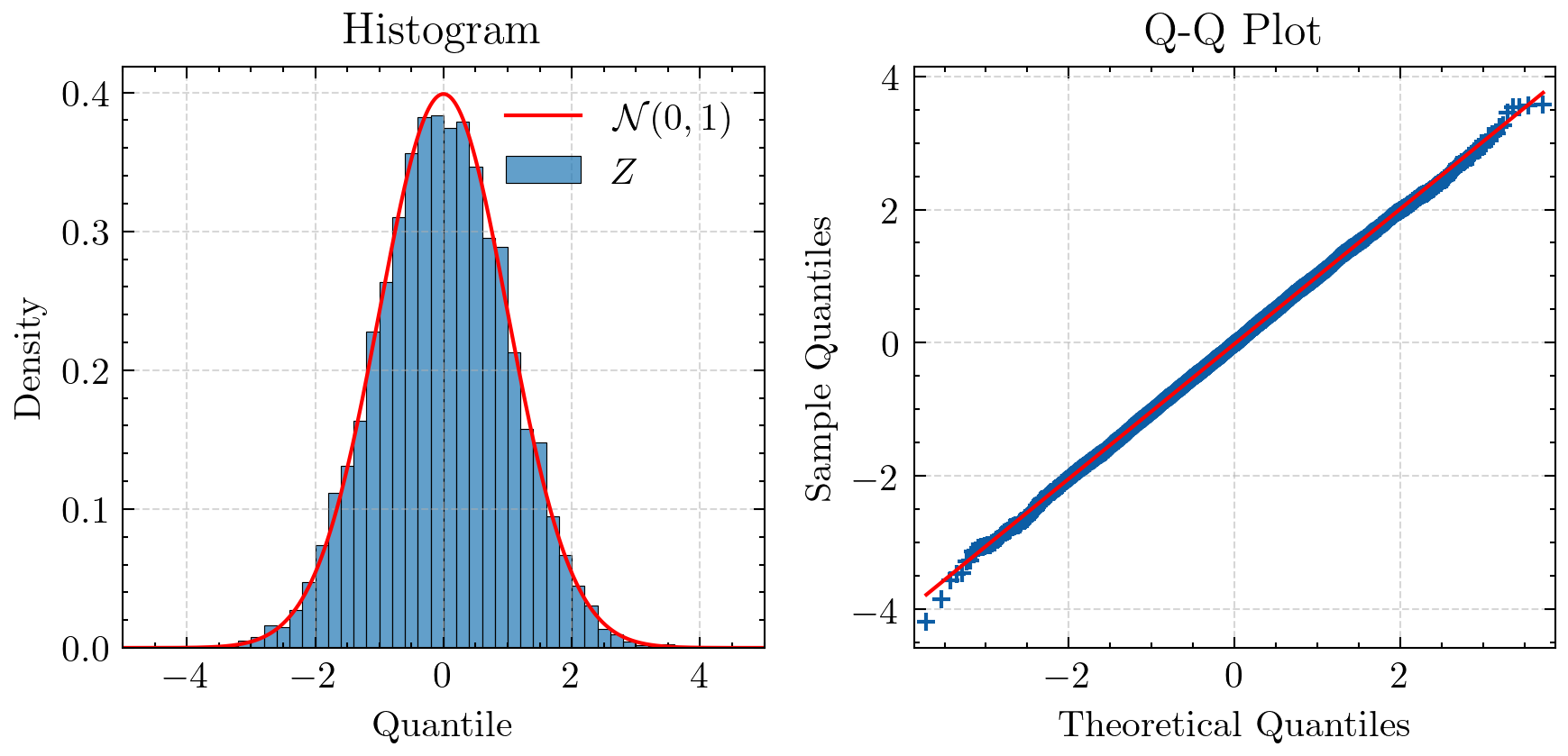}}
    \hfill
    \subfloat[Noise-free
    SALE]{\includegraphics[width=0.33\textwidth]{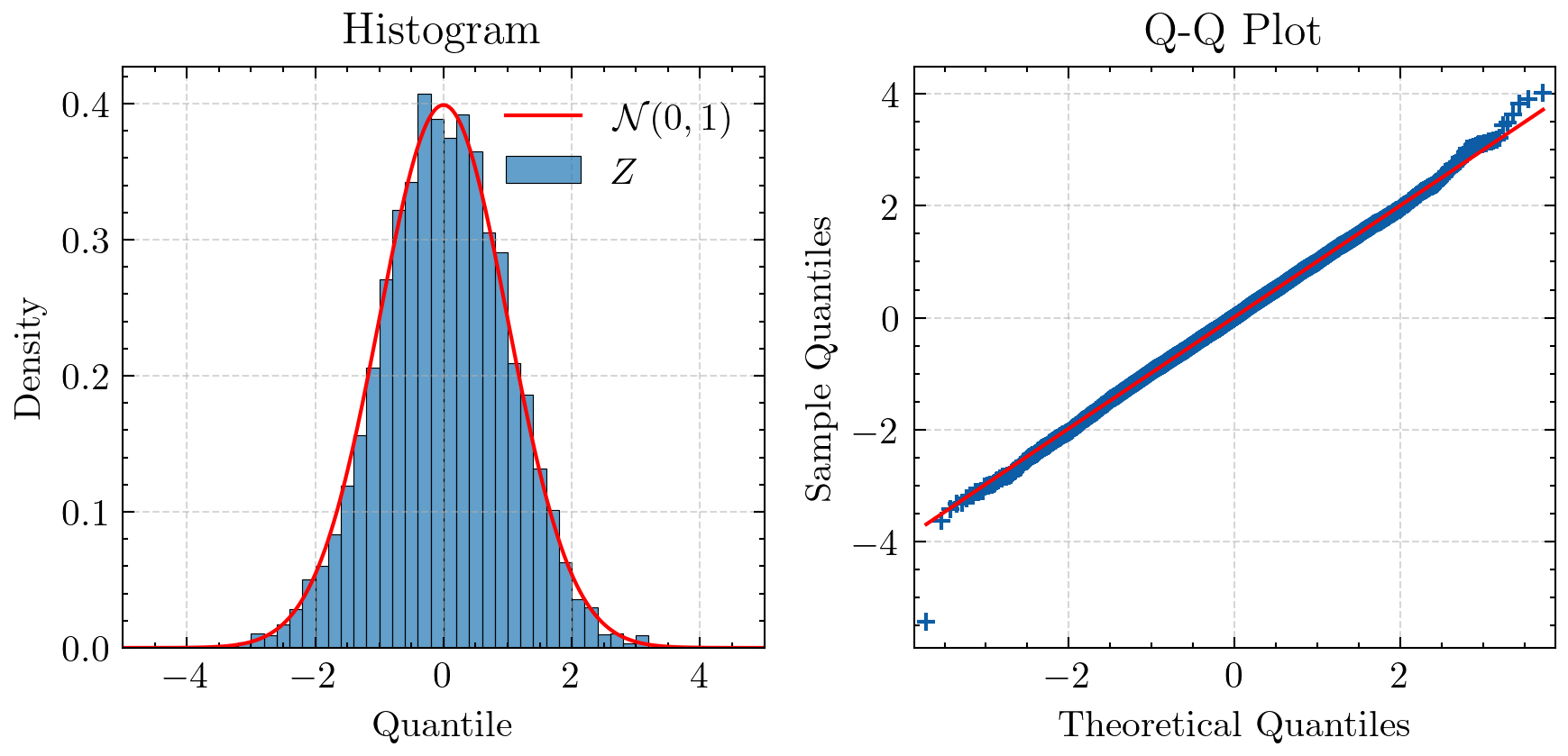}}
    \hfill
    \subfloat[Noise-free
    MSLE]{\includegraphics[width=0.33\textwidth]{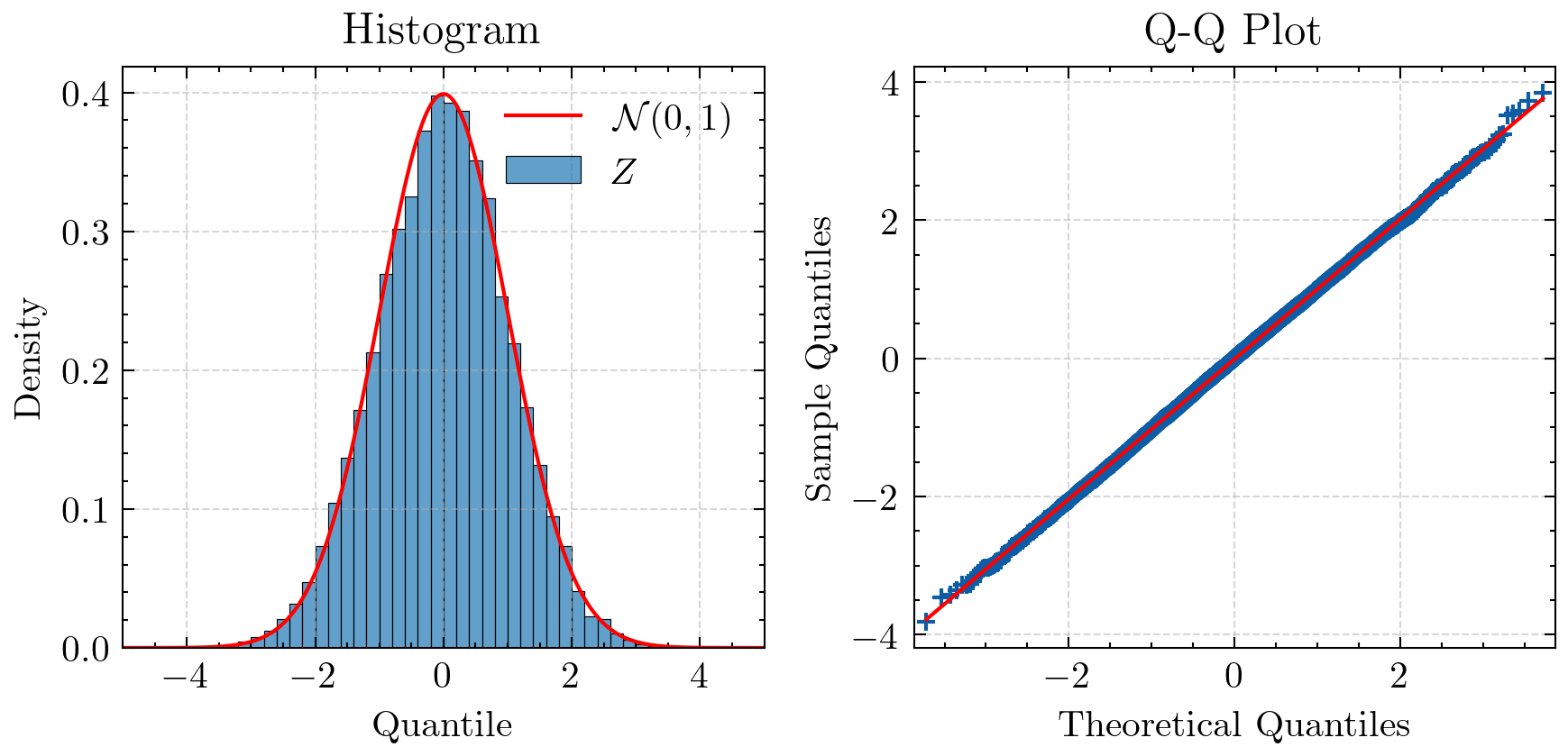}}
    \\
    \subfloat[Normal
    LE]{\includegraphics[width=0.33\textwidth]{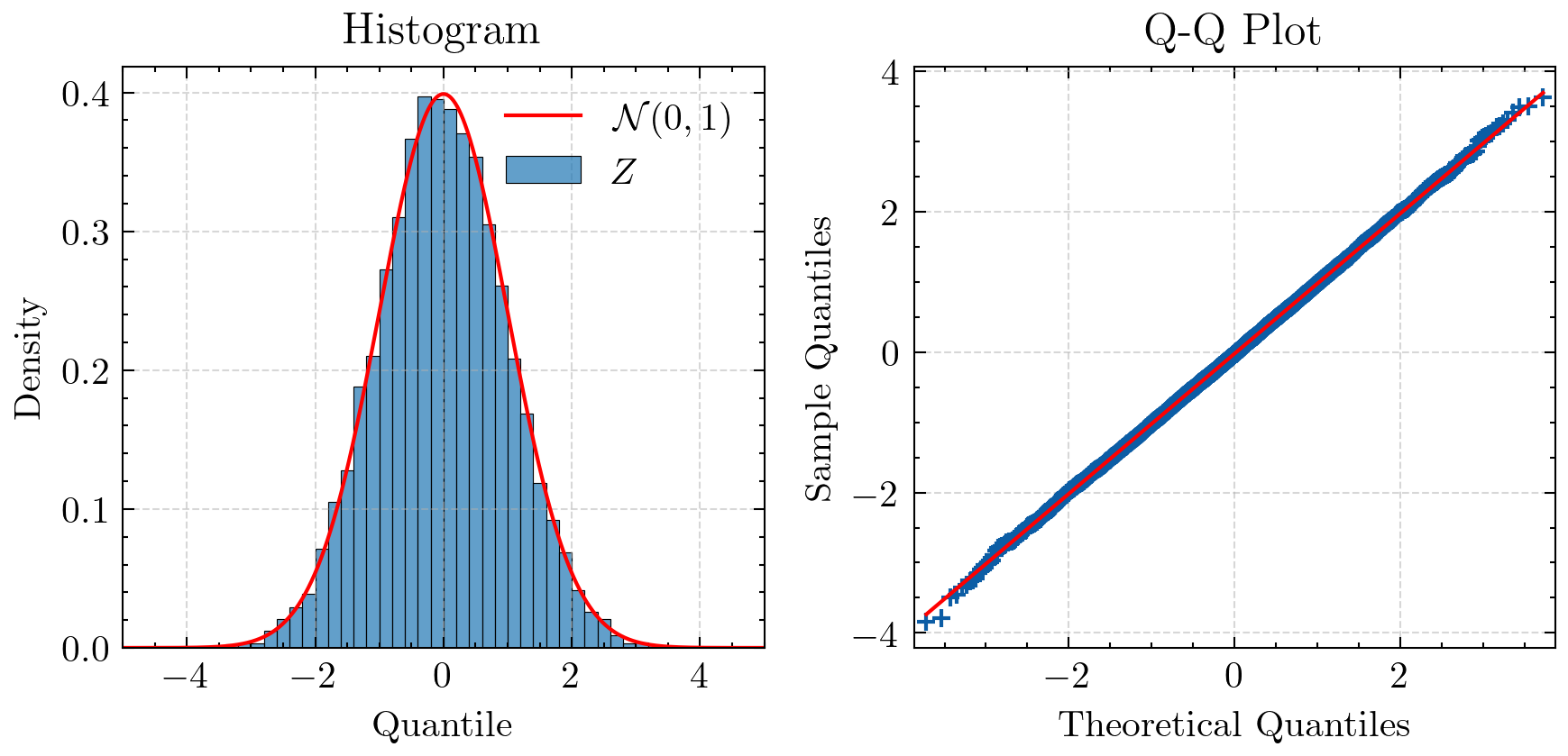}}
    \hfill
    \subfloat[Normal
    SALE]{\includegraphics[width=0.33\textwidth]{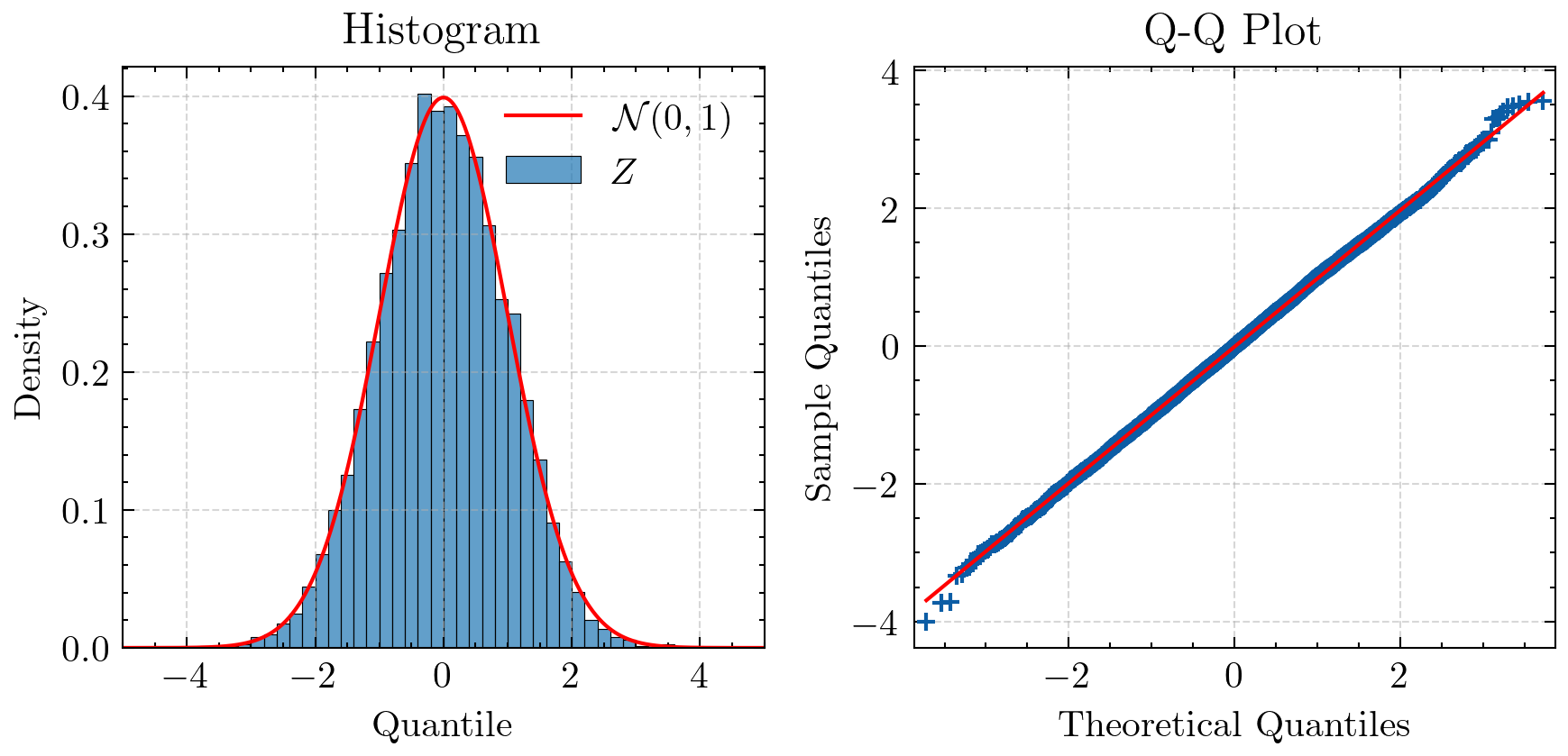}}
    \hfill
    \subfloat[Normal
    MSLE]{\includegraphics[width=0.33\textwidth]{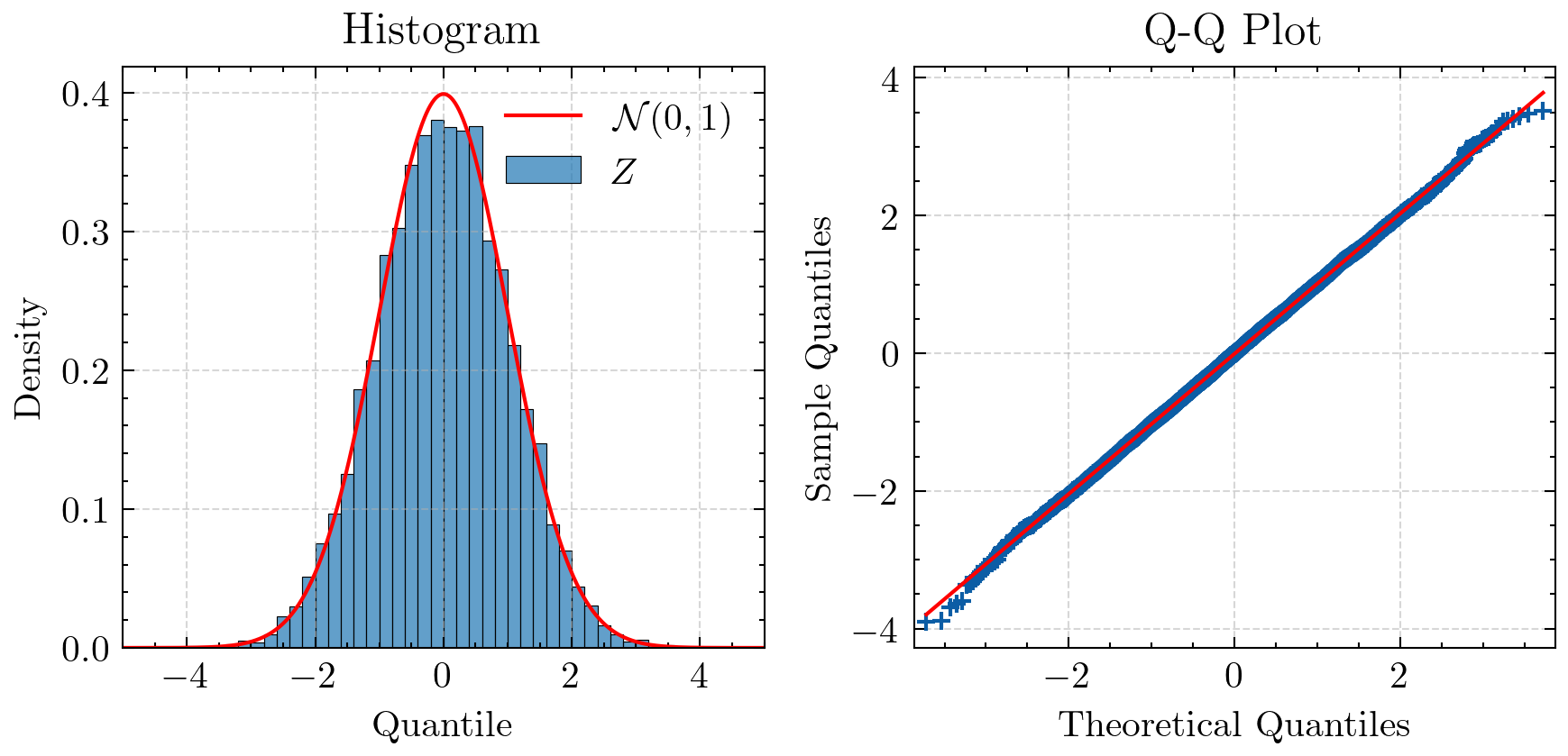}}
    \\
    \subfloat[Uniform
    LE]{\includegraphics[width=0.33\textwidth]{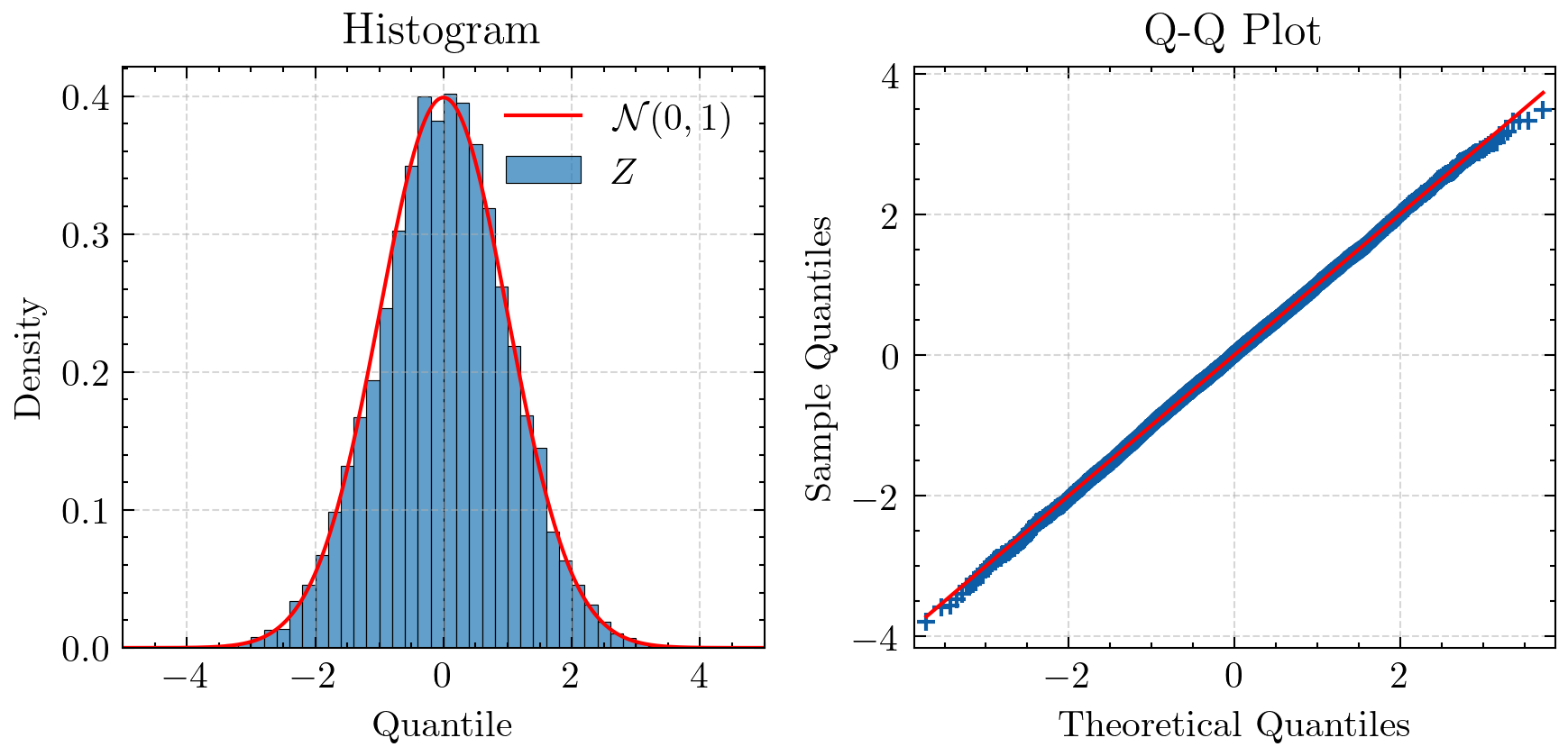}}
    \hfill
    \subfloat[Uniform
    SALE]{\includegraphics[width=0.33\textwidth]{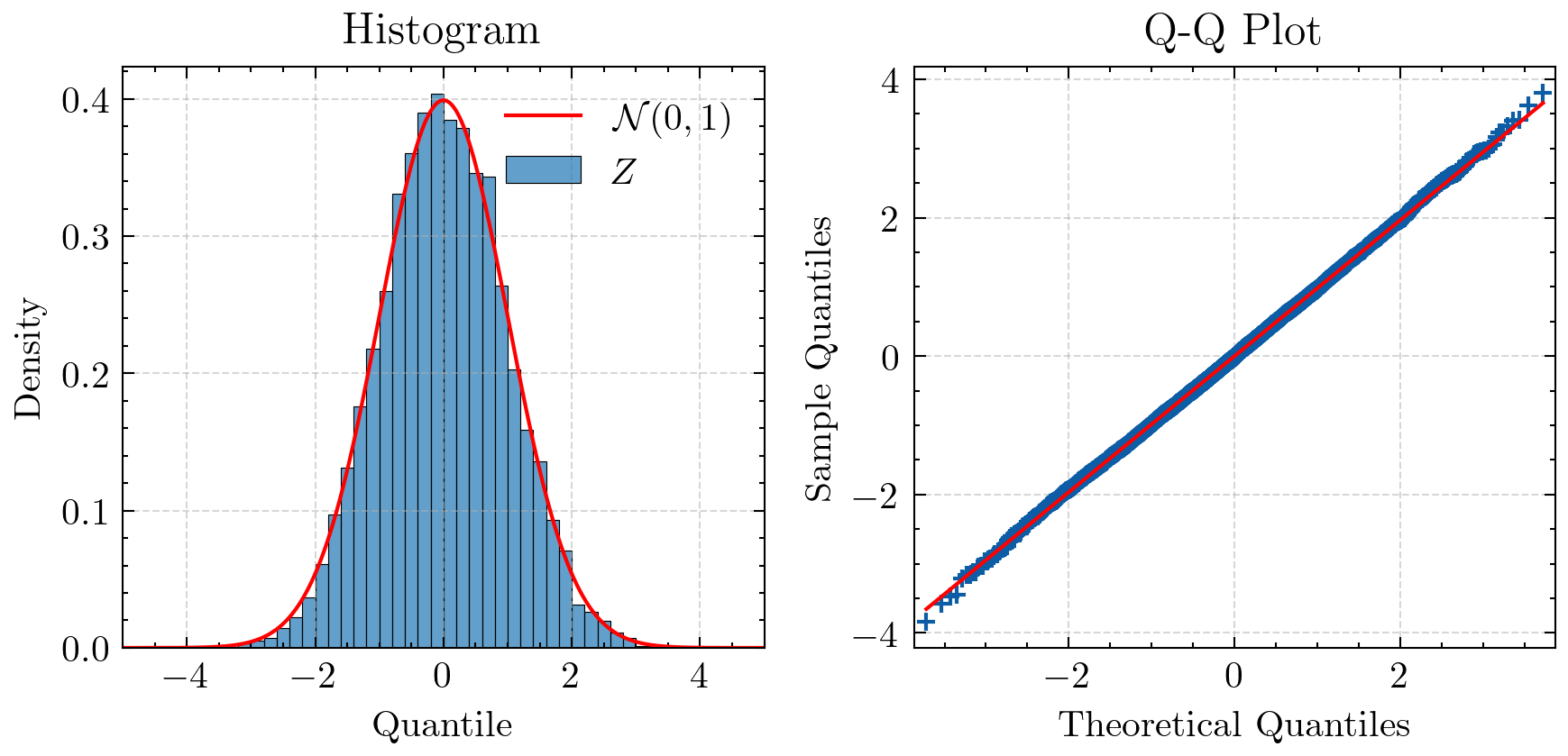}}
    \hfill
    \subfloat[Uniform
    MSLE]{\includegraphics[width=0.33\textwidth]{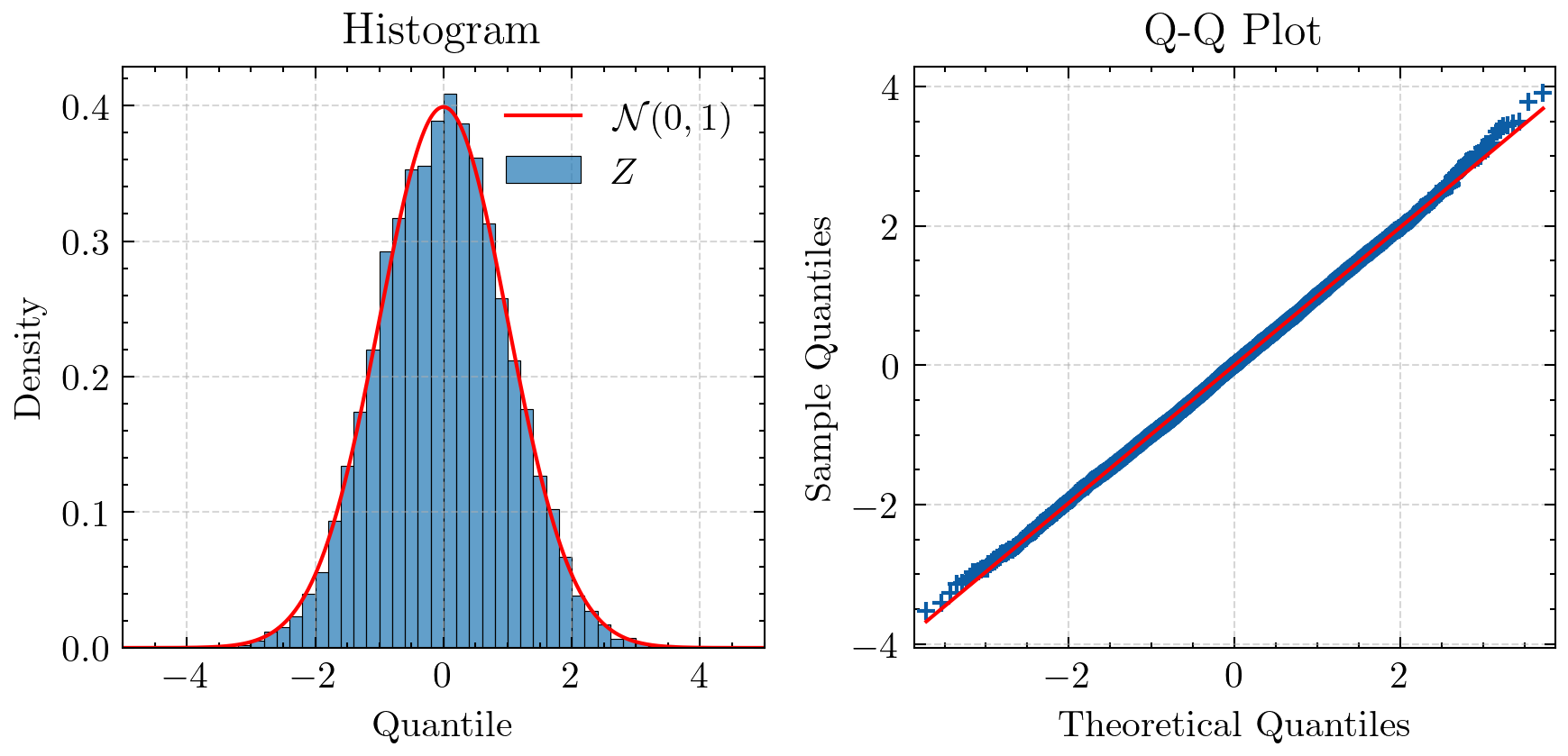}}
    \\
    \subfloat[Skew-normal
    LE]{\includegraphics[width=0.33\textwidth]{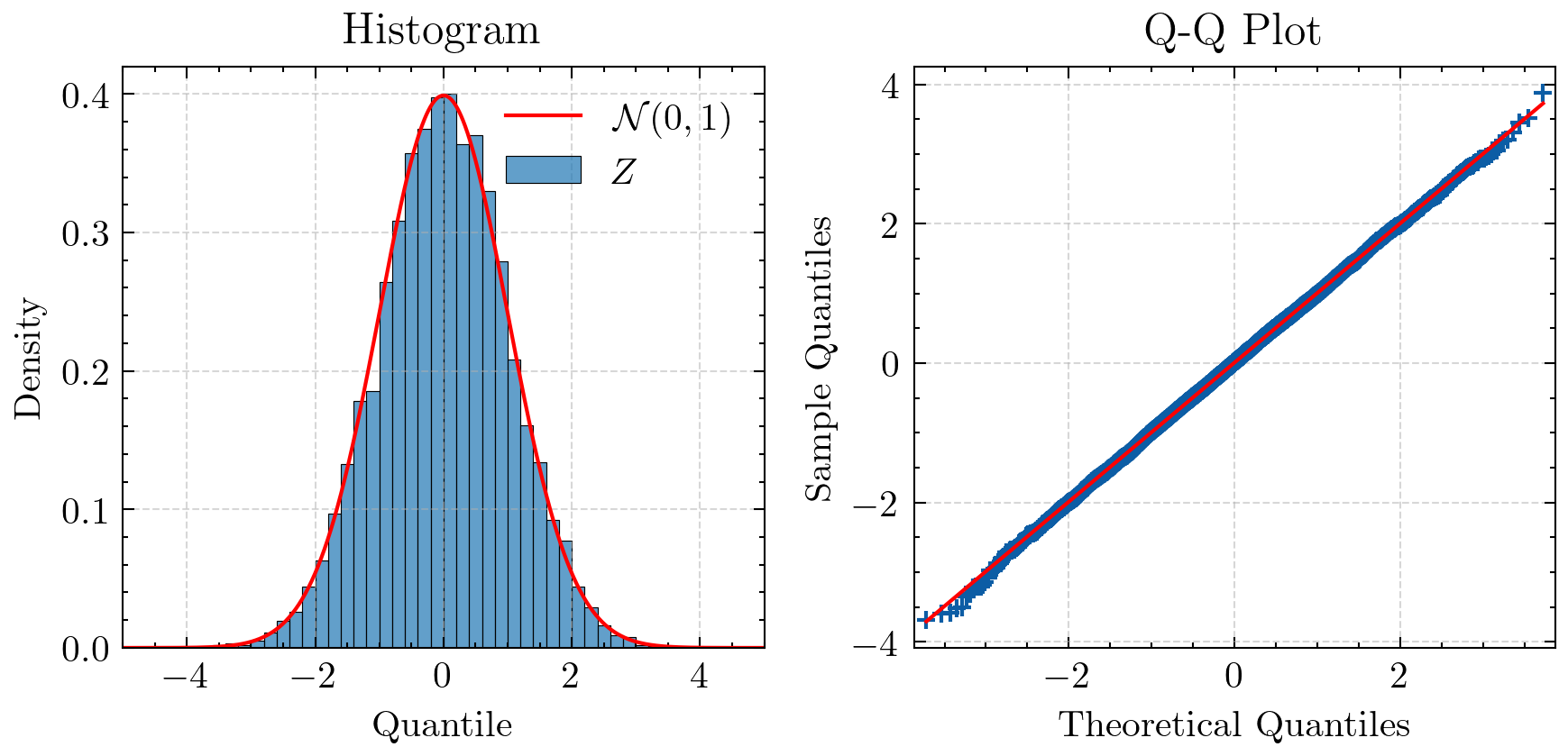}}
    \hfill
    \subfloat[Skew-normal
    SALE]{\includegraphics[width=0.33\textwidth]{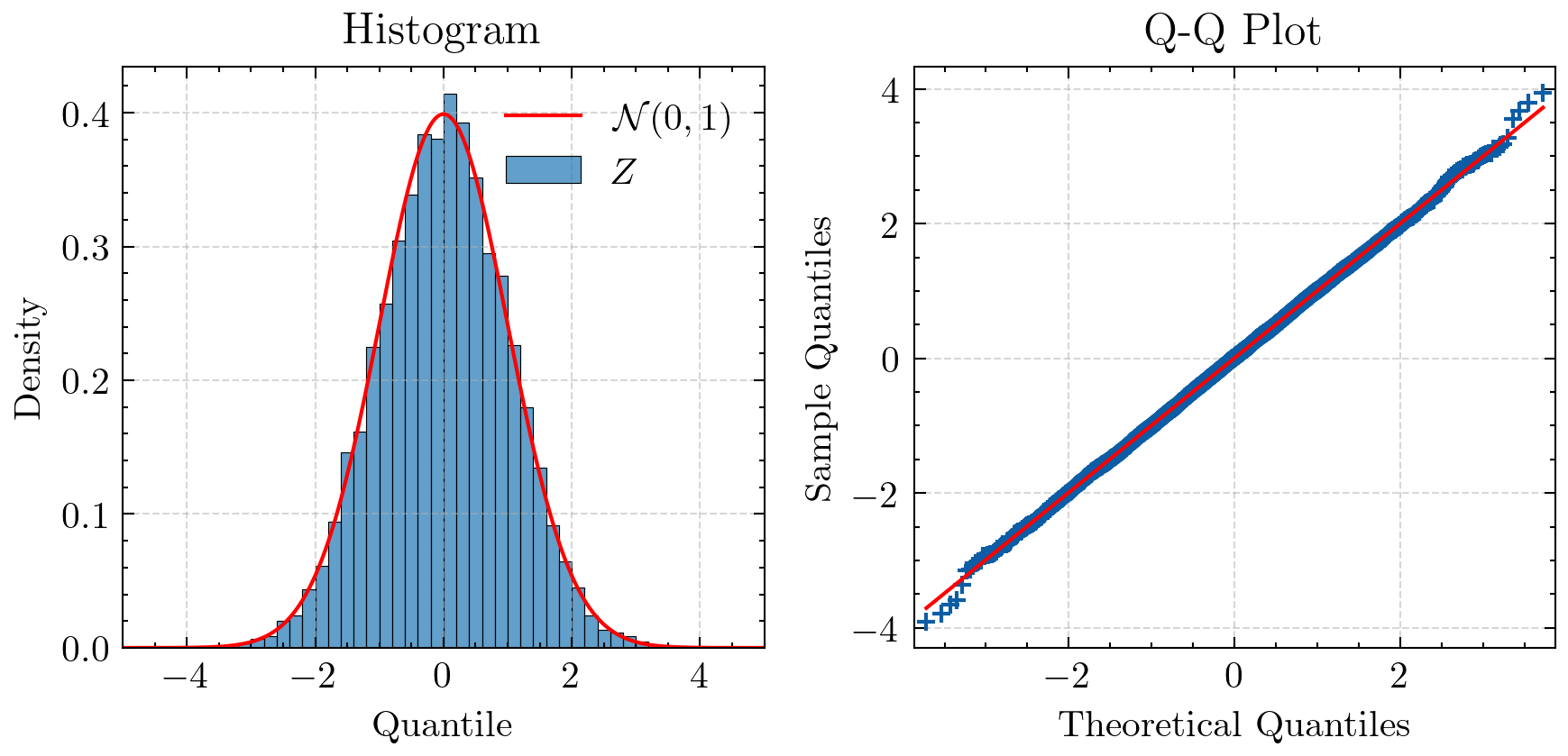}}
    \hfill
    \subfloat[Skew-normal
    MSLE]{\includegraphics[width=0.33\textwidth]{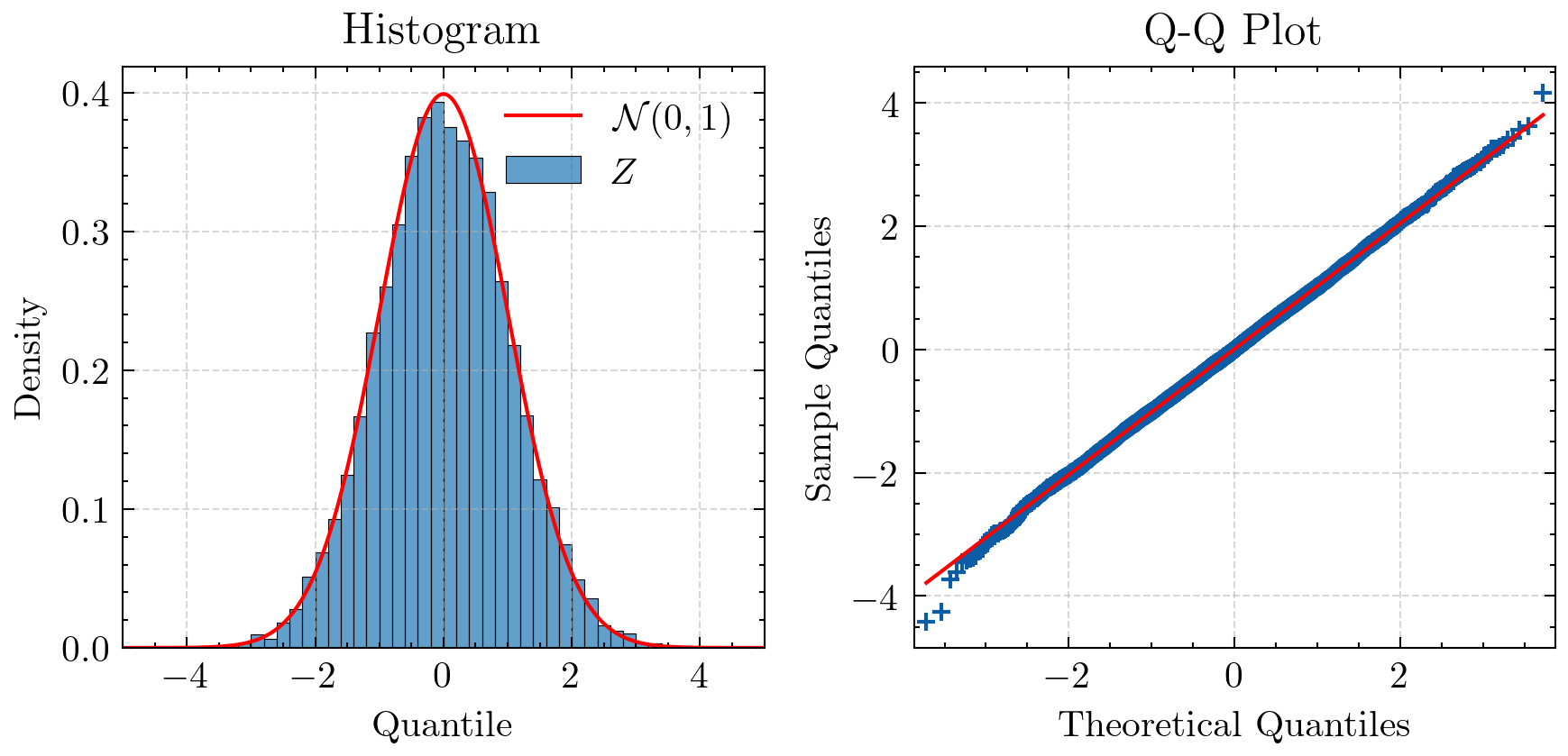}}
    \\
    \subfloat[MA(2) ($\theta_1=0.7$)
    SALE]{\includegraphics[width=0.33\textwidth]{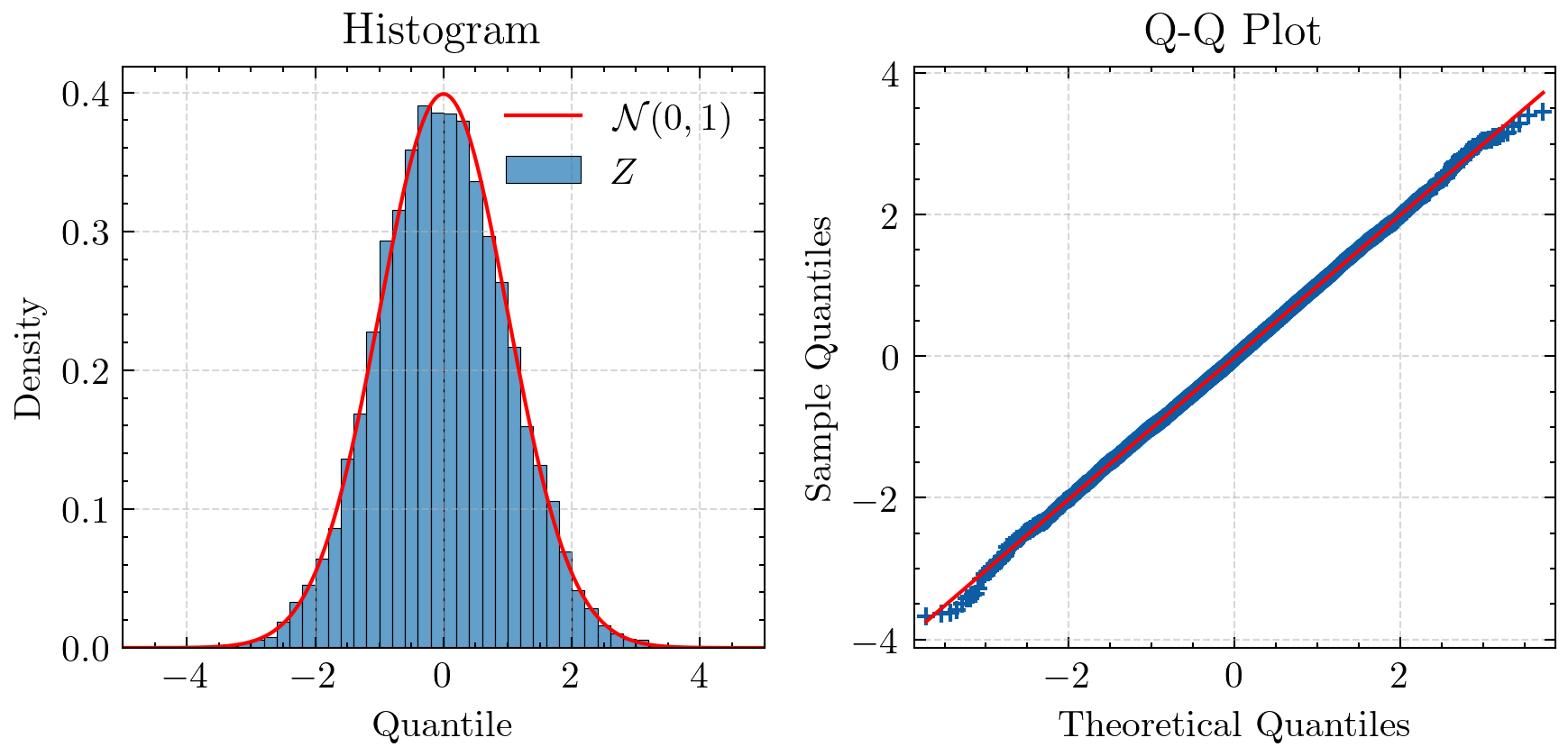}}
    \hfill
    \subfloat[MA(2) ($\theta_1=0.7$)
    MSLE]{\includegraphics[width=0.33\textwidth]{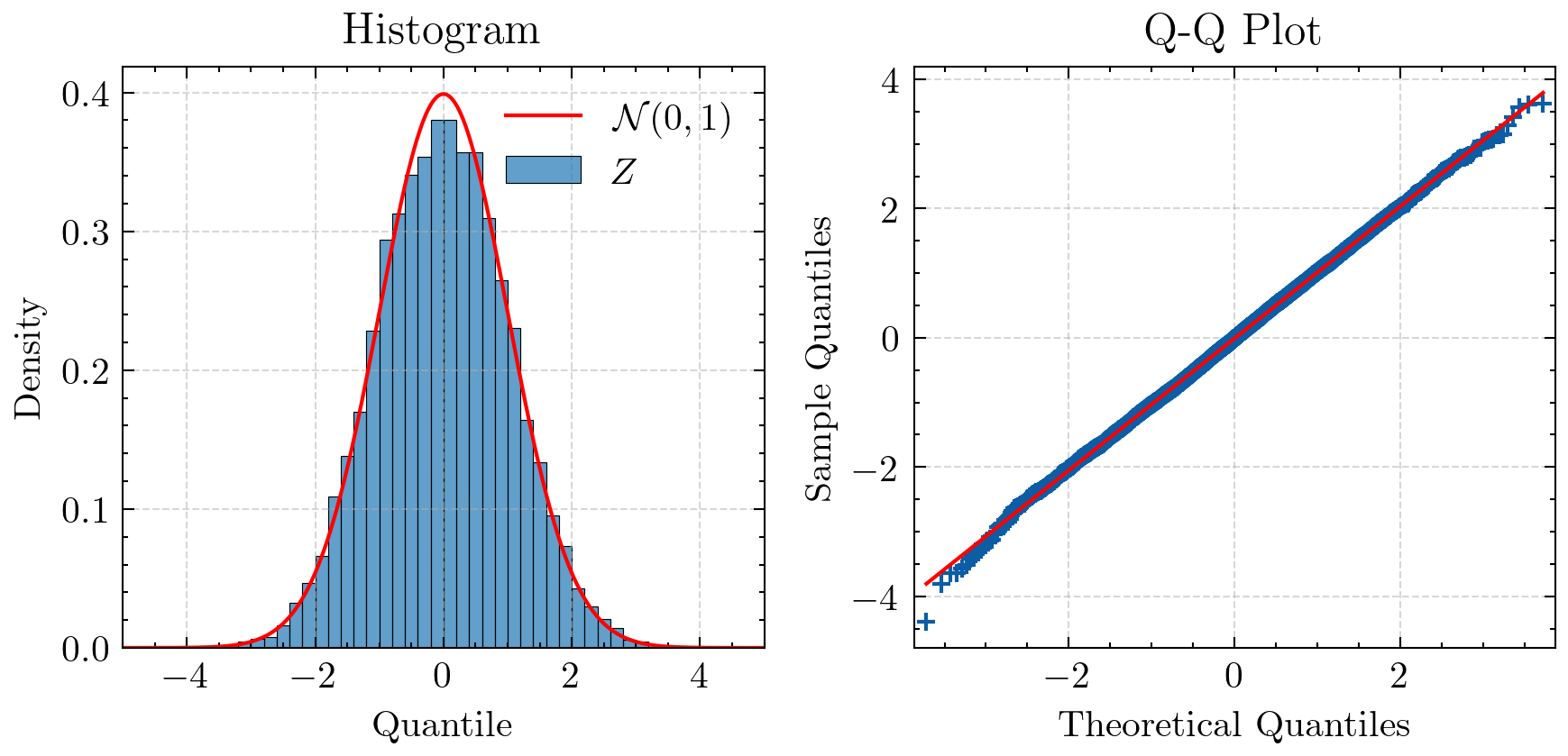}}
    \hfill
    \subfloat[MA(2) ($\theta_1=-0.7$)
    SALE]{\includegraphics[width=0.33\textwidth]{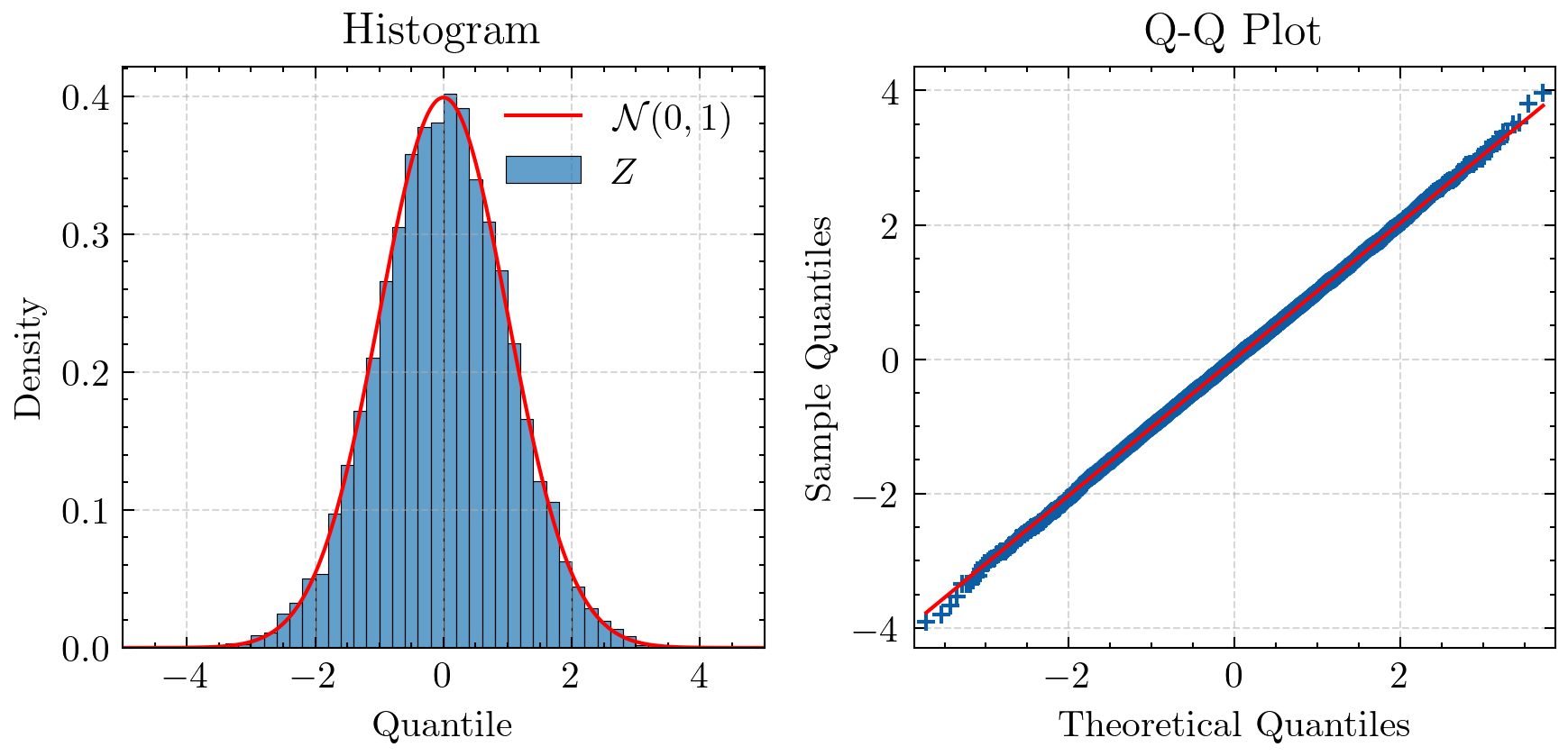}}
    \\
    \subfloat[MA(2) ($\theta_1=-0.7$)
    MSLE]{\includegraphics[width=0.33\textwidth]{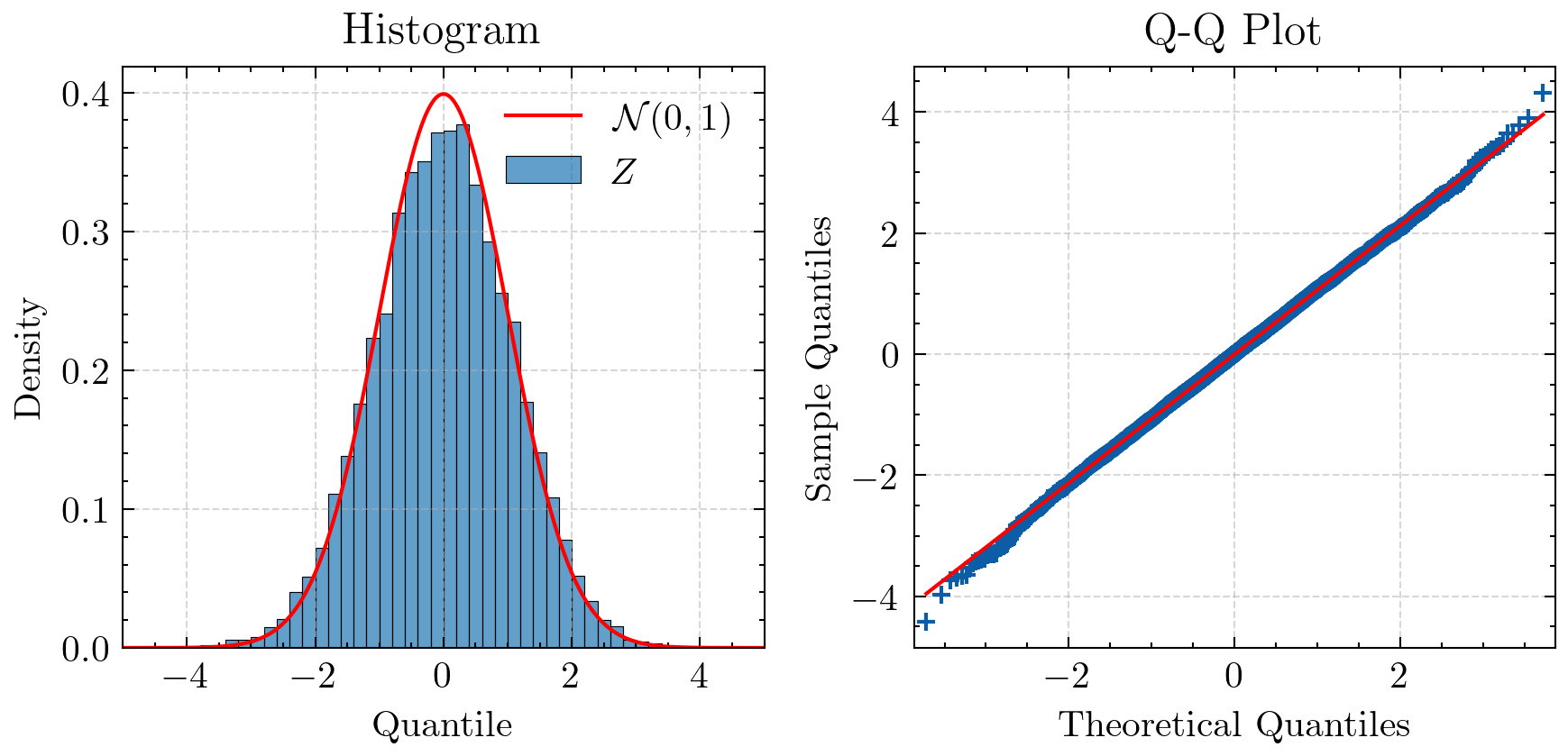}}
    \hfill
    \subfloat[AR(1)
    SALE]{\includegraphics[width=0.33\textwidth]{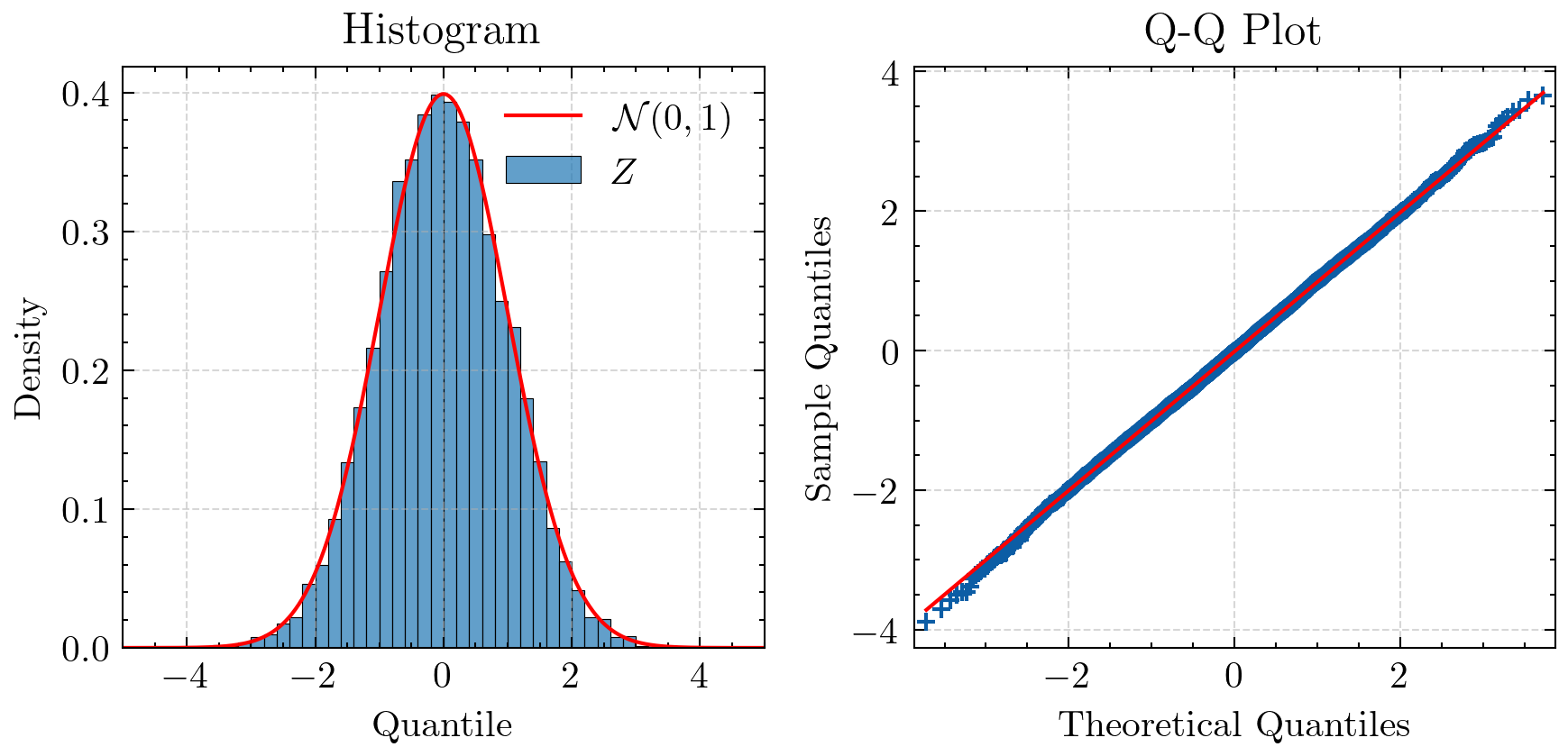}}
    \hfill
    \subfloat[AR(1)
    MSLE]{\includegraphics[width=0.33\textwidth]{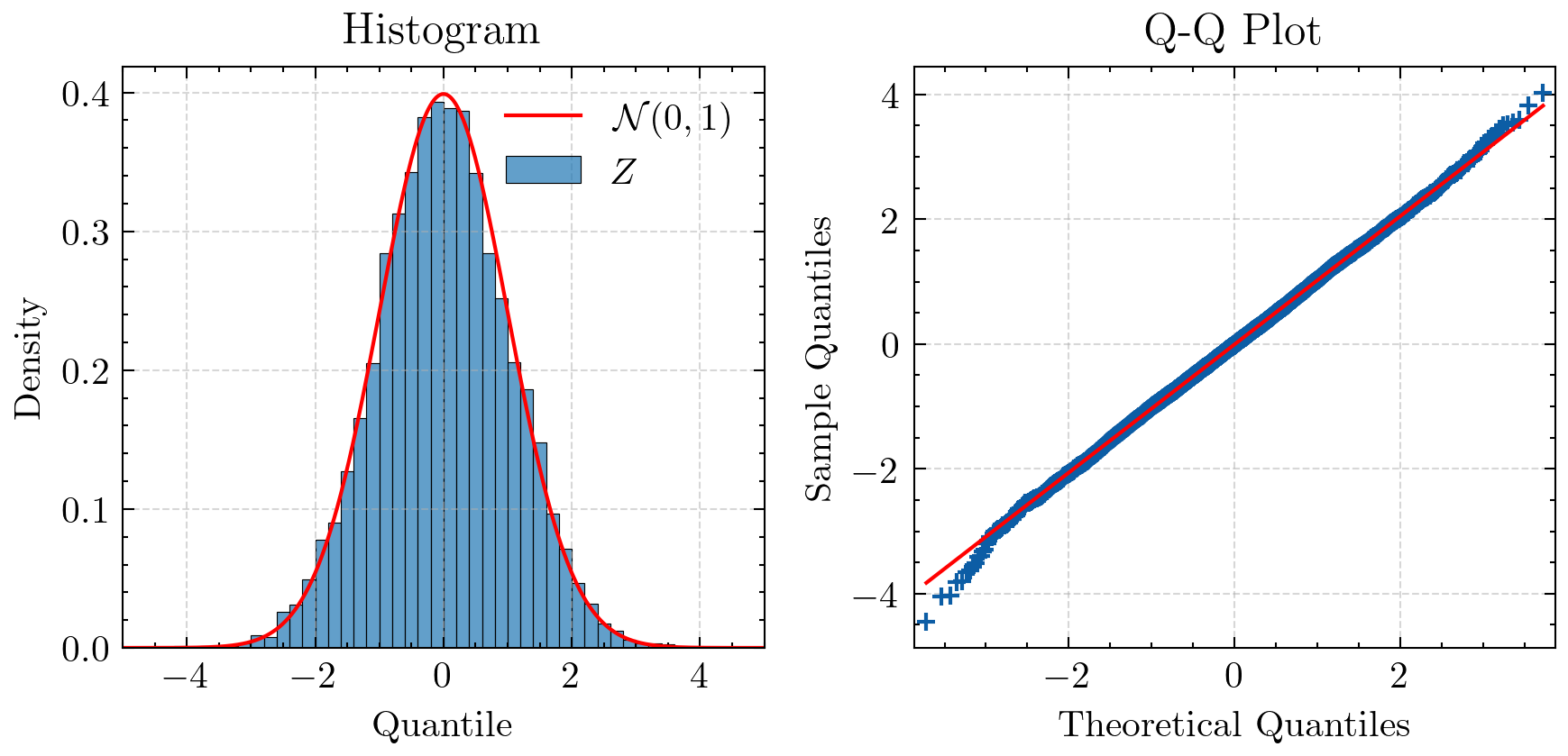}}
    \caption{Histograms and Q-Q plots of standardized estimation errors
      calculated with feasible asymptotic variances across noise settings
      and estimators. Blue bars and points represent empirical
      standardized errors; red curves and lines represent the standard
    normal distribution.}
    \label{fig:feasible-standardized-errors}
  \end{figure}

  \subsection{Finite-Sample Performance}

  \paragraph*{Parameter Tuning in the Noisy Case}

  In both the i.i.d. and dependent noise cases, the weights in MSLE
  estimators and the window sizes in pre-averaging estimators need to
  be tuned. The tuning methods are as follows:
  \begin{enumerate}
    \item The MSLE Estimator.
      \begin{enumerate}
        \item Optimal weights (for the i.i.d. case only). Calculated
          using Equation~\eqref{eq:weight-optimization-solution}, where
          the covariance due to discretization is given by
          Equation~\eqref{eq:SALE-acov-disc}, and the diagonal elements
          of covariance due to noise are given by
          Equation~\eqref{eq:SALE-noise-variance-corrected}.
        \item Approximate weights (for both cases). Derived from the
          method in Section~\ref{sec:practical-weight-noisy}, with the
          asymptotic variances of SALE estimators calculated using
          Equations~\eqref{eq:SALE-acov-disc} and
          \eqref{eq:SALE-noise-variance-corrected}. The optimal sclae
          is denoted as $\overline H_n$.
      \end{enumerate}
    \item The Pre-Averaging Estimator.
      \begin{enumerate}
        \item Optimal weights (for the i.i.d. case only). Determined
          by minimizing the asymptotic variance in Equation~6.6 of
          \citet{aitsahalia2017EstimationContinuousDiscontinuous}.
          The optimal window size is denoted as $A_n$.
        \item \emph{Ex-post} selected weights (for both cases).
      \end{enumerate}
  \end{enumerate}
  For computational efficiency, we use the same set of $\overline H_n$
  and $A_n$ for all paths within a given setting, rather than
  specifying these parameters in a path-wise manner. Specifically, the
  average quantities $\int_0^T \sigma_t^6 \ud t$ and $\int_0^T
  \sigma_t^2 \ud \langle \sigma^2, \sigma^2 \rangle_t$ over all paths
  are used to determine these paramters. Tables~\ref{tab:iid-settings}
  and \ref{tab:dependent-settings} list the selected parameters for
  each setting.

  \begin{table}[!ht]
    \centering
    \caption{Settings of MSLE estimators with approximate weights and
      pre-averaging estimators with optimal pre-averaging window in the
    i.i.d. noise case.}
    \label{tab:iid-settings}
    \begin{tabular}{l *{8}{r}}
      \toprule
      \multirow{2.5}{*}{$\bm{\varsigma}$} &
      \multicolumn{2}{c}{\textbf{1 day}} & \multicolumn{2}{c}{\textbf{5
      days}} & \multicolumn{2}{c}{\textbf{10 days}} &
      \multicolumn{2}{c}{\textbf{22 days}} \\
      \cmidrule(lr){2-3} \cmidrule(lr){4-5} \cmidrule(lr){6-7}
      \cmidrule(lr){8-9}
      & $\overline{H}_n$ & $A_n$ & $\overline{H}_n$ & $A_n$ &
      $\overline{H}_n$ & $A_n$ & $\overline{H}_n$ & $A_n$ \\
      \midrule
      $10^{-4}$ & 7 & 8 & 7 & 7 & 6 & 7 & 5 & 6 \\
      $10^{-3.5}$ & 43 & 25 & 35 & 23 & 32 & 21 & 26 & 19 \\
      $10^{-3}$ & 190 & 79 & 171 & 73 & 152 & 67 & 128 & 59 \\
      \bottomrule
    \end{tabular}
  \end{table}

  \begin{table}[!ht]
    \centering
    \caption{Settings of MSLE estimators with approximate weights in
    the dependent noise case.}
    \label{tab:dependent-settings}
    \begin{tabular}{l *{4}{r}}
      \toprule
      \multirow{2.5}{*}{$\bm{\varsigma}$} &
      \multicolumn{4}{c}{$\bm{\overline{H}_n}$} \\
      \cmidrule(lr){2-5}
      & 1 day & 5 days & 10 days & 22 days \\
      \midrule
      $10^{-4}$ & 10 & 9 & 8 & 6 \\
      $10^{-3.5}$ & 57 & 47 & 44 & 37 \\
      $10^{-3}$ & 230 & 219 & 201 & 169 \\
      \bottomrule
    \end{tabular}
  \end{table}

  \paragraph*{Results of the I.I.D. Noise Case}

  The settings of $T$, $n$ and the number of paths are the same as the
  noise-free case. We generate i.i.d. Gaussian noise with the same
  noise levels as the dependent noise case. For MSLE estimators, we set
  $H_p = 1, 2, \dots, \lfloor n^{5/9} \rfloor$. For pre-averaging
  estimators, the \emph{ex-post} selected version is the estimator with
  the minimum RMSE from windows of 5, 10, 30, 60, and 120. The results
  are presented in Figure~\ref{fig:iid-estimate} and Table~\ref{tab:iid-rmse}.

  \begin{figure}[!ht]
    \centering
    \begin{subfigure}{0.24\textwidth}
      \includegraphics[width=\textwidth]{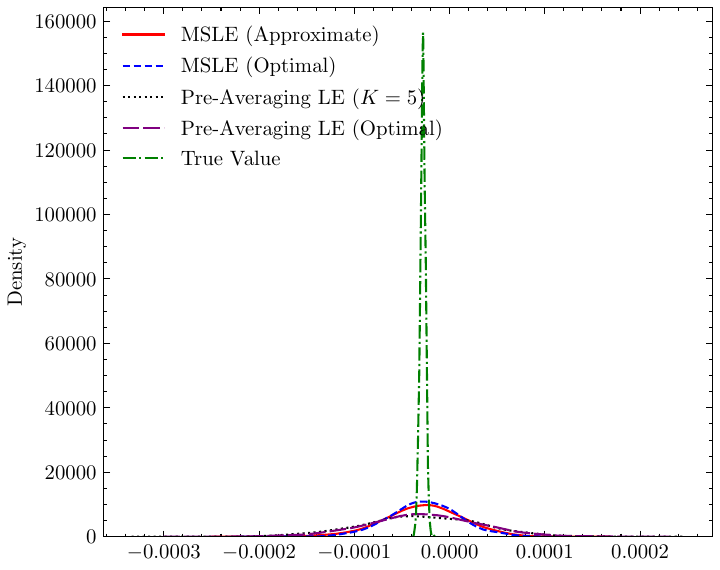}
      \caption{1 day}
    \end{subfigure}
    \hfill
    \begin{subfigure}{0.24\textwidth}
      \includegraphics[width=\textwidth]{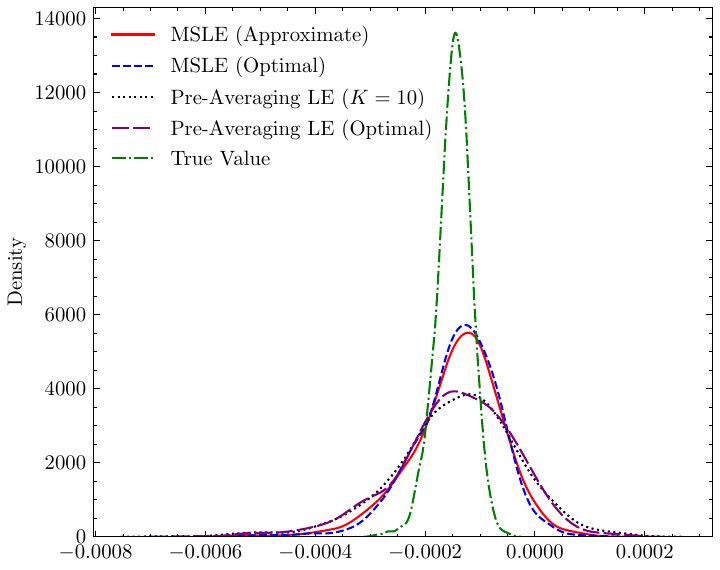}
      \caption{5 days}
    \end{subfigure}
    \hfill
    \begin{subfigure}{0.24\textwidth}
      \includegraphics[width=\textwidth]{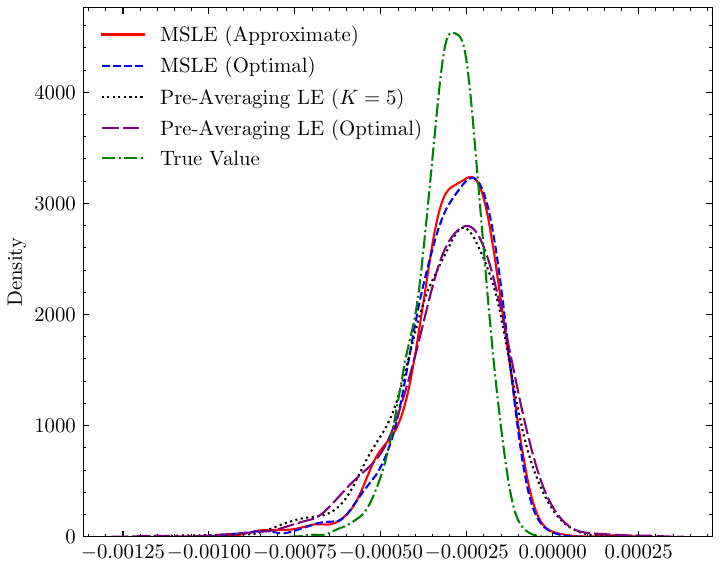}
      \caption{10 days}
    \end{subfigure}
    \hfill
    \begin{subfigure}{0.24\textwidth}
      \includegraphics[width=\textwidth]{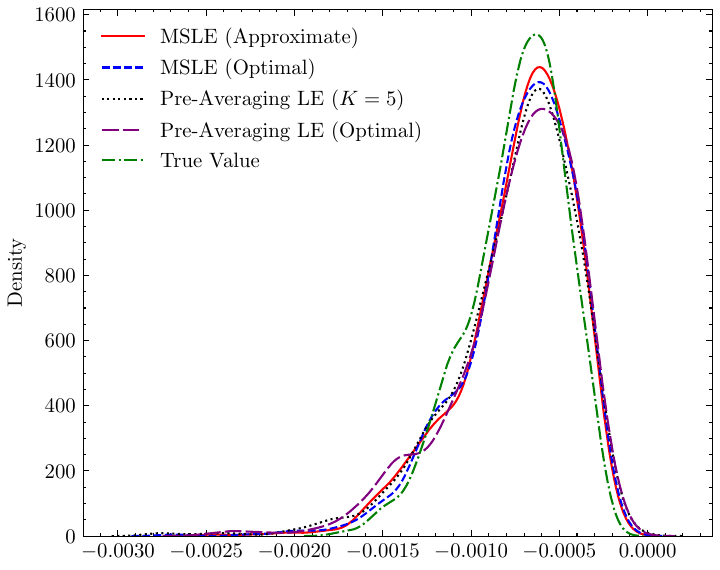}
      \caption{22 days}
    \end{subfigure}
    \begin{subfigure}{0.24\textwidth}
      \includegraphics[width=\textwidth]{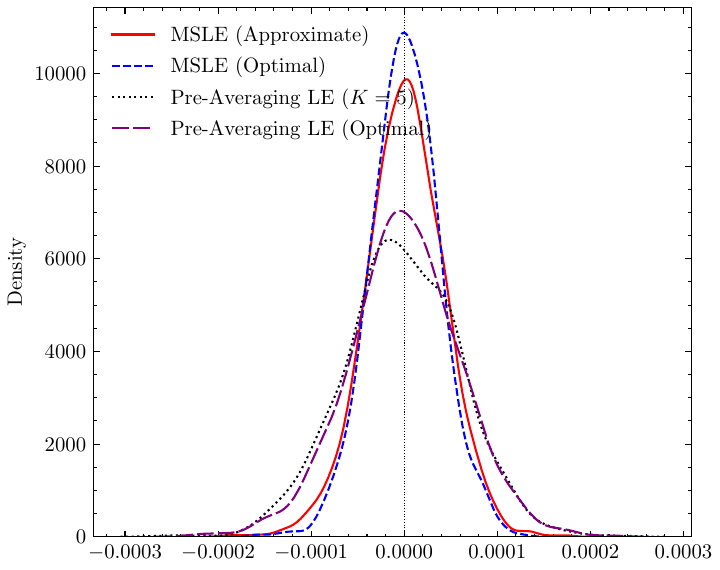}
      \caption{1 day}
    \end{subfigure}
    \hfill
    \begin{subfigure}{0.24\textwidth}
      \includegraphics[width=\textwidth]{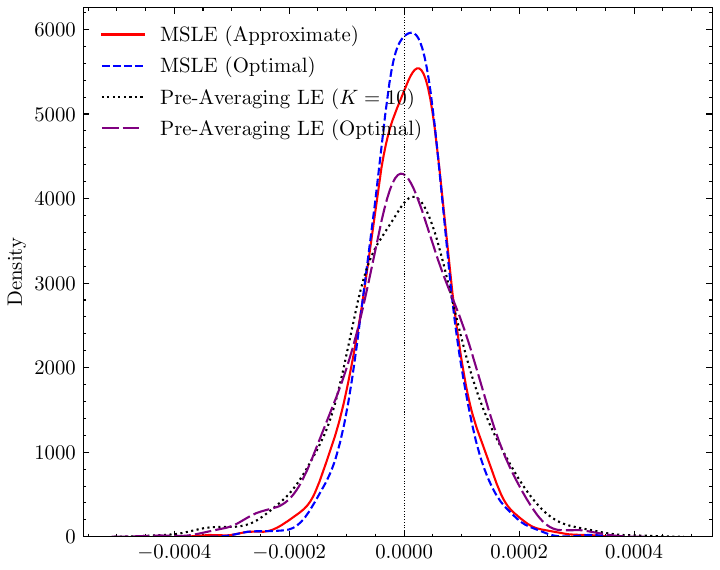}
      \caption{5 days}
    \end{subfigure}
    \hfill
    \begin{subfigure}{0.24\textwidth}
      \includegraphics[width=\textwidth]{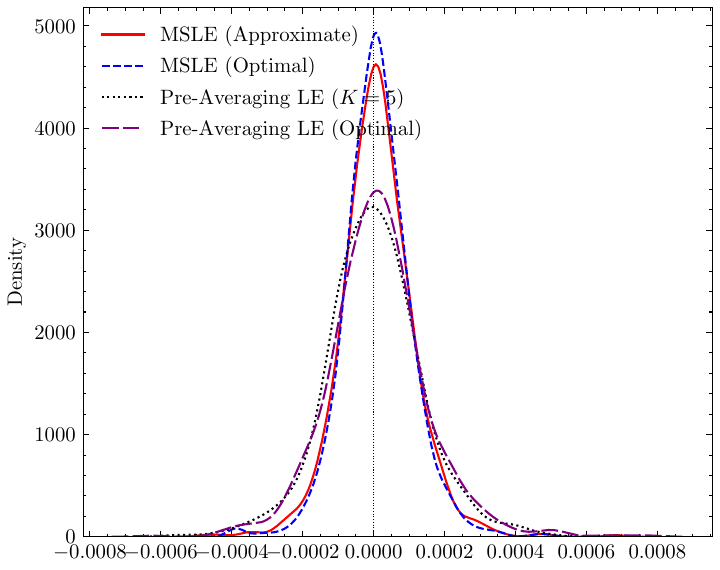}
      \caption{10 days}
    \end{subfigure}
    \hfill
    \begin{subfigure}{0.24\textwidth}
      \includegraphics[width=\textwidth]{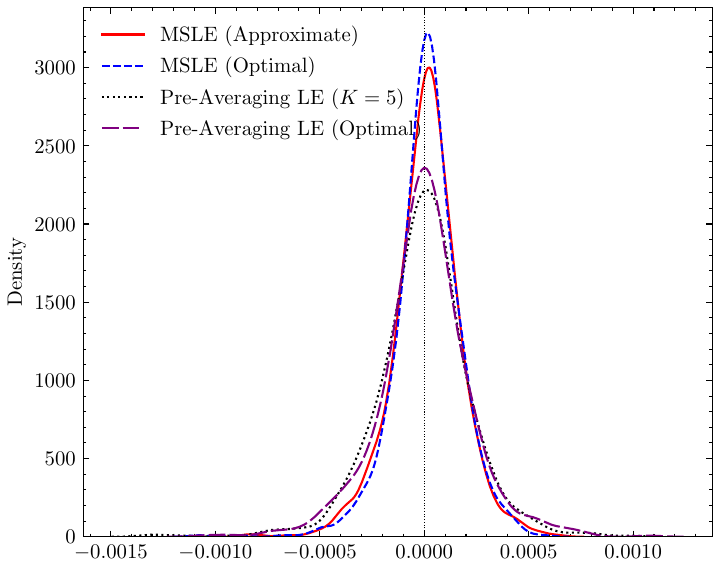}
      \caption{22 days}
    \end{subfigure}
    \caption{The performances of the MSLE and pre-averaging LE
      estimators for each setting of $T$ in the i.i.d. noise setting
      ($\varsigma=10^{-4}$). The first row shows the true and estimated
    values of leverage effect, and the second row shows the estimation error.}
    \label{fig:iid-estimate}
  \end{figure}

  \begin{table}[!ht]
    \centering
    \caption{
      Finite-sample performances of the MSLE and pre-averaging LE
      estimators in the i.i.d. noise setting. The finite-sample
      relative efficiency is compared with the \emph{ex-post} selected
      pre-averaging estimator.
    }
    \label{tab:iid-rmse}
    \resizebox{0.95\textwidth}{!}{
      \begin{tabular}{lrrrrrrrrrrr}
        \toprule
        \multirow{3.5}{*}{$\bm{\varsigma}$} &
        \multirow{3.5}{*}{\textbf{Days}} &
        \multicolumn{2}{c}{\multirow{2.5}{*}{\textbf{True Value}}} &
        \multicolumn{4}{c}{\textbf{RMSE}} &
        \multicolumn{4}{c}{\textbf{Relative Efficiency}} \\
        \cmidrule(lr){5-8} \cmidrule(lr){9-12}
        & & & & \multicolumn{2}{c}{MSLE} &
        \multicolumn{2}{c}{Pre-Averaging LE} & \multicolumn{2}{c}{MSLE}
        & \multicolumn{2}{c}{Pre-Averaging LE} \\
        \cmidrule(lr){3-4} \cmidrule(lr){5-6} \cmidrule(lr){7-8}
        \cmidrule(lr){9-10} \cmidrule(lr){11-12}
        & & Mean & Std & Approximate & Optimal & Optimal & Selected &
        Approximate & Optimal & Optimal & Selected \\
        \midrule
        \multirow{4}{*}{$10^{-4}$} & 1 & \num{-2.80e-05} &
        \num{2.77e-06} & \num{4.20e-05} & \num{3.71e-05} &
        \num{5.93e-05} & \num{6.16e-05} & 2.16 & 2.77 & 1.08 & 1.00 \\
        & 5 & \num{-1.48e-04} & \num{3.10e-05} & \num{7.64e-05} &
        \num{6.98e-05} & \num{1.02e-04} & \num{1.07e-04} & 1.95 & 2.34
        & 1.10 & 1.00 \\
        & 10 & \num{-3.06e-04} & \num{8.98e-05} & \num{1.04e-04} &
        \num{9.61e-05} & \num{1.42e-04} & \num{1.36e-04} & 1.73 & 2.01
        & 0.93 & 1.00 \\
        & 22 & \num{-7.44e-04} & \num{2.72e-04} & \num{1.64e-04} &
        \num{1.58e-04} & \num{2.22e-04} & \num{2.25e-04} & 1.89 & 2.02
        & 1.03 & 1.00 \\
        \cmidrule(lr){1-12}
        \multirow{4}{*}{$10^{-3.5}$} & 1 & \num{-2.80e-05} &
        \num{2.77e-06} & \num{5.86e-05} & \num{6.04e-05} &
        \num{7.79e-05} & \num{7.94e-05} & 1.83 & 1.73 & 1.04 & 1.00 \\
        & 5 & \num{-1.48e-04} & \num{3.10e-05} & \num{1.05e-04} &
        \num{1.07e-04} & \num{1.32e-04} & \num{1.35e-04} & 1.68 & 1.61
        & 1.05 & 1.00 \\
        & 10 & \num{-3.06e-04} & \num{8.98e-05} & \num{1.43e-04} &
        \num{1.40e-04} & \num{1.77e-04} & \num{1.81e-04} & 1.60 & 1.67
        & 1.05 & 1.00 \\
        & 22 & \num{-7.44e-04} & \num{2.72e-04} & \num{2.30e-04} &
        \num{2.19e-04} & \num{2.92e-04} & \num{3.13e-04} & 1.84 & 2.04
        & 1.15 & 1.00 \\
        \cmidrule(lr){1-12}
        \multirow{4}{*}{$10^{-3}$} & 1 & \num{-2.80e-05} &
        \num{2.77e-06} & \num{8.60e-05} & \num{1.12e-04} &
        \num{9.85e-05} & \num{1.01e-04} & 1.36 & 0.81 & 1.04 & 1.00 \\
        & 5 & \num{-1.48e-04} & \num{3.10e-05} & \num{1.62e-04} &
        \num{1.86e-04} & \num{1.80e-04} & \num{1.72e-04} & 1.12 & 0.85
        & 0.92 & 1.00 \\
        & 10 & \num{-3.06e-04} & \num{8.98e-05} & \num{2.26e-04} &
        \num{2.53e-04} & \num{2.43e-04} & \num{2.39e-04} & 1.11 & 0.89
        & 0.96 & 1.00 \\
        & 22 & \num{-7.44e-04} & \num{2.72e-04} & \num{3.45e-04} &
        \num{3.90e-04} & \num{4.08e-04} & \num{4.04e-04} & 1.37 & 1.07
        & 0.98 & 1.00 \\
        \bottomrule
      \end{tabular}
    }
  \end{table}

  \section{Empirical Study}

  \paragraph*{The Estimation of MMS Noise}

  The second-order to fourth-order autocovariances with lags $0, \pm 1,
  \dots, \pm 5$ are estimated with the ReMeDI estimator. While the
  estimator is consistent, its application in this empirical study
  faces a few challenges: (i) ``unreasonable'' moment estimates (for
  example, a minus variance estimate) may appear in the finite-sample
  scenario; and (ii) an operational definition of the ``existence'' of
  the noise is needed. To mitigate these challenges, we adopt a
  conservative procedure for this part, avoiding underestimating the
  noise level:
  \begin{enumerate}
    \item These ``unreasonable'' estimates are handled as follows.
      \begin{enumerate}
        \item The case with $\widehat \nu_2 \leq 0$: it is treated as
          noise-free.
        \item The case with $\widehat \nu_2 > 0$ but $\widehat \nu_4
          \leq 0$: the fourth-order autocovariances are ``corrected''
          by $\widehat \E [\eps_i^2 \eps_{i+l}^2] \leftarrow \widehat
          \nu_2^2 \bigl( \widehat \E [\eps_i \eps_{i+l}] \bigr)^2$ for
          all $l = 0, \pm 1, \dots, \pm 5$.
        \item The case with $\widehat \E [\eps_i^2 \eps_{i+l}^2] <
          \bigl( \widehat \E [\eps_i \eps_{i+l}] \bigr)^2$ for any $l
          \in \{0, \pm 1 \dots, \pm 5\}$: the corresponding
          fourth-order autocovariance is ``corrected'' by $\widehat \E
          [\eps_i^2 \eps_{i+l}^2] \leftarrow \bigl( \widehat \E [\eps_i
          \eps_{i+l}] \bigr)^2 \lor \widehat \E [\eps_i^2 \eps_{i+l}^2]$.
      \end{enumerate}
    \item A conservative $90\%$ confidence interval is used to test the
      significance of each lag of second-order autocovariance. The
      noise is treated as existing if there exists an $l \in \{0, \pm
      1, \dots, \pm 5\}$ such that $\widehat \E [\eps_i \eps_{i+l}]$ is
      significant. In this case, the dependence level is then defined
      as the largest $l$ such that any of $\widehat \E [\eps_i
      \eps_{i+l}]$ and $\widehat \E [\eps_i \eps_{i-l}]$ is significant.
    \item The case with $\sqrt{\widehat \nu_2} > 10^{-4}$ is always
      treated as noisy.
  \end{enumerate}

\end{appendices}

\end{document}